\documentclass[11pt,letterpaper]{article}
\usepackage[letterpaper, margin=1in]{geometry}

\usepackage{cmap} % Load before fontenc
\usepackage[utf8]{inputenc}
\usepackage[T1]{fontenc}
\usepackage{amsmath,amssymb,amsthm,mathrsfs,amsfonts}
\usepackage{bm}
\usepackage{enumerate}
\usepackage{amssymb}
\usepackage{physics}
\usepackage{complexity}
\usepackage{caption}
\usepackage{subcaption}
\usepackage{appendix}
\usepackage{mathtools}
\usepackage{amsfonts}
\usepackage{bm}
\usepackage[dvipsnames]{xcolor}
\definecolor{blueviolet}{rgb}{0.2, 0.2, 0.6}
\definecolor{webgreen}{rgb}{0,.5,0}
\definecolor{webbrown}{rgb}{.6,0,0}
\usepackage[pdftex,
	bookmarks=false,
	colorlinks=true, %allcolors=blueviolet,
	urlcolor=webbrown,
	linkcolor=blueviolet,
	citecolor=webgreen,
	pdfstartpage=1,
	pdfstartview={FitH},  % FitBH
	bookmarksopen=false
	]{hyperref}
\definecolor{darkred}{HTML}{CC333A}
\definecolor{brightblue}{HTML}{6F6FCF}
\definecolor{brightred}{HTML}{EF4F4F}
\definecolor{brightgreen}{HTML}{4FC888}
\allowdisplaybreaks
\usepackage[dvipsnames]{xcolor}
\usepackage{tikz}
\usetikzlibrary{decorations.pathreplacing}
\pgfmathsetmacro\MathAxis{height("$\vcenter{}$")}
\newcommand{\tikzlength}{0.5}

\usepackage{authblk}

\usepackage{dsfont}
\usepackage{braket}
\usepackage{amsthm}
\usepackage{listings}

\newtheorem{corollary}{Corollary}
\newtheorem{definition}{Definition}
\newtheorem{lemma}{Lemma}
\newtheorem{fact}{Fact}
\newtheorem{proposition}{Proposition}
\newtheorem{remark}{Remark}

\newtheorem{theorem}{Theorem}

\newcommand{\mc}{\mathcal}

\usepackage[backend=biber,style=nature,url=false,isbn=false,maxnames=10,minnames=3,maxalphanames=4,minalphanames=3]{biblatex}
\DeclareFixedFont{\ttb}{T1}{txtt}{bx}{n}{9} % for bold
\DeclareFixedFont{\ttm}{T1}{txtt}{m}{n}{9}  % for normal

\usepackage{color}
\definecolor{deepblue}{rgb}{0,0,0.5}
\definecolor{deepred}{rgb}{0.6,0,0}
\definecolor{deepgreen}{rgb}{0,0.5,0}

\newcommand\pythonstyle{\lstset{
language=Python,
basicstyle=\ttm,
morekeywords={self},              % Add keywords here
keywordstyle=\ttb\color{deepblue},
emph={MyClass,__init__},          % Custom highlighting
emphstyle=\ttb\color{deepred},    % Custom highlighting style
stringstyle=\color{deepgreen},
frame=tb,                         % Any extra options here
showstringspaces=false
}}

\lstnewenvironment{python}[1][]
{
\pythonstyle
\lstset{#1}
}
{}

\newcommand\pythoninline[1]{{\pythonstyle\lstinline!#1!}}

\DeclareMathOperator*{\Exp}{\mathop{\mathbb{E}}}
\renewcommand{\E}{\mathop{\mathbb{E}}}

\DeclareMathOperator*{\argmin}{arg\,min}

\DeclareMathOperator{\SWAP}{SWAP}

\newtheorem*{theorem*}{Theorem}

    \makeatletter
\def\@fnsymbol#1{\ensuremath{\ifcase#1\or \dagger\or \ddagger\or
   \mathsection\or \mathparagraph\or \|\or **\or \dagger\dagger
   \or \ddagger\ddagger \else\@ctrerr\fi}}
    \makeatother

\usepackage[noend]{algpseudocode}
\usepackage{algorithm,algorithmicx}

\algrenewcommand\alglinenumber[1]{\sf\scriptsize\color{blue}{#1}}
\algrenewcommand\algorithmicrequire{\textbf{Input:}}
\algrenewcommand\algorithmicensure{\textbf{Output:}}

\input{Qcircuit}

\begin{document}

\title{Learning shallow quantum circuits}
\date{\vspace{-6ex}}

\author[1,2,3]{Hsin-Yuan Huang\thanks{Co-first author. Both authors contributed equally (listed in alphabetical order).}}
\author[4]{Yunchao Liu$^\dagger$}
\author[3]{Michael Broughton}
\author[5]{Isaac Kim}
\author[6]{\\Anurag Anshu}
\author[4]{Zeph Landau}
\author[3]{Jarrod R. McClean}
\affil[1]{California Institute of Technology}
\affil[2]{Massachusetts Institute of Technology}
\affil[3]{Google Quantum AI}
\affil[4]{University of California, Berkeley}
\affil[5]{University of California, Davis}
\affil[6]{Harvard University}

\maketitle

\begin{abstract}
\normalsize
Despite fundamental interests in learning quantum circuits, the existence of a computationally efficient algorithm for learning shallow quantum circuits remains an open question. Because shallow quantum circuits can generate distributions that are classically hard to sample from, existing learning algorithms do not apply. In this work, we present a polynomial-time classical algorithm for learning the description of any unknown $n$-qubit shallow quantum circuit $U$ (with arbitrary unknown architecture) within a small diamond distance using single-qubit measurement data on the output states of $U$. We also provide a polynomial-time classical algorithm for learning the description of any unknown $n$-qubit state $\lvert \psi \rangle = U \lvert 0^n \rangle$ prepared by a shallow quantum circuit $U$ (on a 2D lattice) within a small trace distance using single-qubit measurements on copies of $\lvert \psi \rangle$. Our approach uses a quantum circuit representation based on local inversions and a technique to combine these inversions. This circuit representation yields an optimization landscape that can be efficiently navigated and enables efficient learning of quantum circuits that are classically hard to simulate.
\end{abstract}
\thispagestyle{empty}
\tableofcontents
\thispagestyle{empty}
\clearpage
\pagenumbering{arabic}

\section{Introduction}
\label{sec:intro}

The question of how to efficiently learn expressive classes of quantum states and circuits features prominently in quantum complexity theory, quantum algorithm design, and the experimental characterization of quantum devices.
As a first step, one might consider the efficiency of learning shallow (constant depth) quantum circuits, where, to date, there has been no resolution despite considerable interest from a number of angles. From a complexity perspective, shallow quantum circuits are known to be more powerful than their classical counterparts~\cite{bravyi2018quantum, bravyi2020quantum, watts2023unconditional, watts2019exponential}, and under widely accepted complexity assumptions, sampling from the output distribution of shallow quantum circuits is classically hard to simulate~\cite{terhal2004adaptive,gao2017quantum,bermejo2018architectures, haferkamp2020closing, hangleiter2023computational}. This computational power provides the basis for quantum computational advantage with NISQ (noisy intermediate-scale quantum) devices and supports the quest for developing quantum algorithms based on learning parameterized shallow quantum circuits~\cite{farhi2018classification, benedetti2019parameterized, beer2020training, bausch2020recurrent, skolik2021layerwise, abbas2021power, caro2022generalization, cerezo2021cost, ostaszewski2021structure, pesah2021absence, du2021learnability, holmes2022connecting, sharma2022trainability, anschuetz2022quantum, cerezo2022challenges}. Within an experimental setting focused on coherent errors or gate calibration, characterizing a NISQ device can be modeled as learning what shallow quantum circuit the device is performing. Despite substantial interest in the question of learning shallow quantum circuits from these directions, to date, no polynomial time algorithm for learning shallow quantum circuits has been found.
In this work, we introduce several efficient algorithms for two related tasks.

\begin{theorem*}[Summary of main results]
There are polynomial time algorithms for (1) learning the description of an unknown $n$-qubit shallow quantum circuit $U$ (with arbitrary unknown architecture) within a small diamond distance, given access to $U$; (2) learning the description of an unknown $n$-qubit state $\ket{\psi}=U \ket{0^n}$ prepared by a shallow quantum circuit $U$ (on a 2D lattice) within a small trace distance, given copies of $\ket{\psi}$.
\end{theorem*}

The main challenges in learning shallow quantum circuits are twofold.
While foundational results in computational learning theory have established the efficient learnability of shallow classical circuits \cite{linial1993constant, mossel2003learning, carmosino2016learning}, these techniques may not apply to shallow quantum circuits, as these circuits can generate distributions with nontrivial correlations over the entire system that are classically hard to simulate \cite{bermejo2018architectures, haferkamp2020closing, hangleiter2023computational}.
Furthermore, even when the structure of a shallow quantum circuit is known up to parameterization, the optimization landscape for learning shallow quantum circuits is swamped with exponentially many suboptimal local minima \cite{anschuetz2022quantum}. The bad optimization landscape causes standard optimization methods, such as gradient descent algorithms and Newton methods, to fail in learning shallow quantum circuits.

To address these challenges, we consider a quantum circuit representation based on \emph{local inversions}, which yields an optimization landscape that can be efficiently navigated.
The local inversions disentangle qubits in each local region in a way that does not perturb the remaining system.
We then show how these local inversions may be combined to build up the entire circuit without having to solve a computationally hard problem.
Together, this new technique enables us to learn a natural class of quantum circuits that are classically hard to simulate.

\subsection{Background}
\label{sec:background}

\paragraph{Learning shallow classical circuits}
Although the shallow quantum case has many conceptual challenges resulting from non-locality, the learnability of shallow classical circuits is a fundamental question in computational learning theory that has been well-studied and resolved in many cases.
Learning constant-depth classical circuits with bounded fan-in gates ($\mathsf{NC}^0$) is equivalent to learning juntas and can be performed in polynomial time from uniform samples \cite{mossel2003learning}. In addition, quasi-polynomial time algorithms are known for learning constant-depth classical circuits with unbounded fan-in AND/OR gates ($\mathsf{AC}^0$)~\cite{linial1993constant}, as well as $\mathrm{mod}\,\,p$ gates ($\mathsf{AC}^0[p]$)~\cite{carmosino2016learning} in the PAC model. The problem of learning shallow quantum circuits $(\mathsf{QNC}^0)$ and their output states are natural quantum analogs of learning Boolean circuits. As $\mathsf{QNC}^0$ can be exponentially more powerful than $\mathsf{AC}^0$ for some computational problems \cite{watts2019exponential}, it is natural to ask if shallow quantum circuits can be learned efficiently from random data samples.

\paragraph{Quantum machine learning}
When one parameterizes the gates in a quantum circuit, the parameterized quantum circuit forms an ML model, known as a \emph{quantum neural network}, that can learn from data and make predictions on new inputs \cite{farhi2018classification, benedetti2019parameterized, beer2020training, bausch2020recurrent, skolik2021layerwise, abbas2021power, caro2022generalization}.
Since deep parameterized quantum circuits suffer from having \emph{barren plateaus} in the optimization landscape \cite{mcclean2018barren, holmes2021barren} and are challenging to implement on noisy quantum devices \cite{wang2021noise, chen2022complexity}, shallow quantum circuits have been subject to extensive study in recent years \cite{cerezo2021cost, ostaszewski2021structure, pesah2021absence, du2021learnability, holmes2022connecting, sharma2022trainability, anschuetz2022quantum, cerezo2022challenges}.
Various applications of learning shallow quantum circuits have been explored, ranging from compressing quantum circuits for implementing a unitary \cite{cincio2018learning, khatri2019quantum, sharma2020noise, caro2022generalization, jones2022robust}, speeding up quantum dynamics \cite{cirstoiu2020variational, yao2021adaptive, gibbs2022dynamical, caro2023out, jerbi2023power}, to learning generative models for sampling from predicted distributions \cite{lloyd2018quantum, benedetti2019generative, coyle2020born, gao2022enhancing, rudolph2022generation, zhu2022generative}.
While the optimization landscape for learning shallow quantum circuits is free from barren plateau \cite{cerezo2021cost}, the landscape is swamped with exponentially many suboptimal local minima; see Section~\ref{sec:exp-many-local-minima-param-shallow} and \cite{anschuetz2022quantum} for a study of this phenomenon.
The presence of a large number of suboptimal local minima causes standard local optimization methods, such as gradient descent or Newton's method, to fail in learning parameterized shallow quantum circuits.

\paragraph{Efficient quantum tomography}
While quantum state and process tomography generally require exponential resources, performing tomography over some restricted families of states or processes can be made computationally efficient.
Examples of such families include matrix product states \cite{cramer2010efficient, lanyon2017efficient, gebhart2023learning}, high-temperature Gibbs states \cite{anshu2020sample, rouze2021learning, haah2022optimal}, stabilizer states \cite{montanaro2017learning, gross2021schur, grewal2022low, grewal2023improved}, quantum phase states \cite{arunachalam2022optimal}, noninteracting Fermionic states \cite{aaronson2023efficient}, Clifford circuits with a small number of T gates \cite{gross2021schur, lai2022learning, grewal2023improved}, Pauli channels under structural assumptions \cite{flammia2020efficient, flammia2021pauli, chen2022quantum, van2023probabilistic}, and interacting Hamiltonian dynamics \cite{li2020hamiltonian, che2021learning,yu2022,hangleiter2021,FrancaMarkovichEtAl2022efficient,ZubidaYitzhakiEtAl2021optimal,BaireyAradEtAl2019learning,GranadeFerrieWiebeCory2012robust,gu2022practical,wilde2022learnH, huang2023learning} (see~\cite{anshu2023survey} for a recent survey). Most of these examples correspond to quantum circuit families that are classically easy to simulate \cite{terhal2002classical, aaronson2004improved, cirac2021matrix, wild2023classical, yin2023polynomial}. In contrast, sampling from the output distribution of constant-depth quantum circuits is classically hard even when restricted to a 2D lattice~\cite{gao2017quantum,BermejoVega2018architecture}.
The experimental effort to characterize NISQ devices motivates the question of how to perform tomography for states and processes generated by shallow quantum circuits.
While these states can be learned sample-efficiently using shadow tomography \cite{aaronson2018shadow, badescu2020improved, huang2020predicting}, no computationally efficient algorithms are known.

\subsection{Our Results}

We first focus on cases where one is given black-box access to the unknown unitary in (1) learning general shallow quantum circuits and (2) learning geometrically-local shallow quantum circuits.  We then consider the more restricted model where one is only provided access to copies of an unknown state and focus on (3) learning quantum states prepared by geometrically-local shallow quantum circuits on 2-dimensional lattices.

\subsubsection{Learning general shallow quantum circuits}

Let $U$ be an unknown $n$-qubit unitary generated by a shallow quantum circuit. The learning algorithm uses a randomized measurement dataset consisting of $N$ samples about $U$ \cite{levy2021classical, huang2022learning, kunjummen2023shadow, elben2022randomized, caro2022generalization, caro2023out, jerbi2023power}.
This dataset has been proposed as the classical shadow of $U$ \cite{levy2021classical, huang2022learning, kunjummen2023shadow}.
Each classical data sample specifies a random $n$-qubit product input state $\ket{\psi_\ell} = \bigotimes_{i=1}^n \ket{\psi_{\ell, i}}$ and a randomized Pauli measurement outcome $\ket{\phi_\ell} = \bigotimes_{i=1}^n \ket{\phi_{\ell, i}}$ on the output states $U \ket{\psi_\ell}$, where $\ket{\psi_{\ell, i}},\ket{\phi_{\ell, i}} \in \{\ket{0}, \ket{1}, \ket{+}, \ket{-}, \ket{y+}, \ket{y-}\}$ are single-qubit stabilizer states.
Each data sample can be generated by a single query to $U$. Our goal is to learn $U$ within a small diamond distance. The following results have the form of learning a circuit $V$ acting on $2n$ qubits, such that $\|V-U\otimes U^\dag\|_\diamond\leq \varepsilon$. Hence, $V$ can be used to implement $U$ by tracing out the $n$-qubit ancilla system.

Our first main result shows that one can learn $U$ with a polynomial sample and computational complexity, with only the assumption that $U$ is constant-depth (i.e., $U$ has arbitrary unknown connectivity).
Furthermore, the result applies even when the circuit generating $U$ can have any number $m$ of ancilla qubits used as working space and can have arbitrary two-qubit gates in $\mathrm{SU}(4)$ between any pair of the $n+m$ qubits so long as the resulting operation on the $n$ system qubits is unitary. The learning algorithm is fully classical given the randomized measurement dataset.

\begin{theorem}[Learning shallow quantum circuits; see Theorem~\ref{thm:shallow-SU4-gates}] \label{thm:learn-shallow-QC-main}
Given an unknown $n$-qubit unitary $U$ generated by a constant-depth circuit over any two-qubit gates between any pair of qubits.
One can learn a constant-depth circuit approximating $U$ to diamond distance $\varepsilon$ with high probability from $N = \mathcal{O}(n^2 \log (n) / \varepsilon^2)$ samples about $U$ and $\mathrm{poly}(n) / \varepsilon^2$ classical running time.

When the circuit is over a finite gate set, $U$ can be learned to zero error with high probability from $N = \mathcal{O}(\log n)$ samples and $\mathrm{poly}(n)$ time.
\end{theorem}

\subsubsection{Learning geometrically-local shallow quantum circuits}

The algorithm for learning general shallow quantum circuits runs in polynomial time but with a large exponent. Furthermore, the depth of the learned circuit $V$, while constant, could be substantially greater than the depth of $U$. Motivated by the fact that most realistic quantum systems are geometrically local on a finite-dimensional lattice, it is natural to wonder if these aspects can be improved when learning geometrically-local quantum circuits on lattices. Next, we show that this is indeed the case.

See Theorem~\ref{thm:geo-SU4-gates} for a related result on learning shallow circuits over any geometry represented by a bounded-degree graph.

\begin{theorem}[Learning geometrically-local shallow circuits; see Theorem~\ref{thm:geo-kD-lattice-optimized}] \label{thm:learn-geo-main}
Given an unknown $n$-qubit geometrically-local depth-$d$ quantum circuit $U$ over a $k$-dimensional lattice with $d, k = \mathcal{O}(1)$.
One can learn a geometrically-local shallow circuit that approximates $U$ to diamond distance $\varepsilon$ with high probability from $N = \mathcal{O}(n^2 \log (n) / \varepsilon^2)$ classical data samples and either
\begin{itemize}
    \item $\mathcal{O}(n^3 \log (n) / \varepsilon^2)$ classical running time with a learned circuit depth of $(k+1) 4^{4 (8kd)^{k}} + 1$.
    \item $(n/\varepsilon)^{\mc O((8kd)^{k+1})}$ classical running time with a learned circuit depth of $(k+1)(2d+1)+1$.
\end{itemize}
When the circuit is over a finite gate set, $U$ can be learned to zero error with high probability from $N = \mathcal{O}(\log n)$ samples and $\mathcal{O}(n \log (n))$ time with a learned circuit depth of $(k+1)(2d+1)+1$.
\end{theorem}

This shows that in the geometrically local setting, the learned circuit depth can achieve a linear blow-up. Furthermore, the learning algorithm works for $d=\mathrm{polylog}(n)$ depth circuits at the cost of quasipolynomial running time.

We remark that the more formal statement of the above theorem, which is labeled in this work as Theorem~\ref{thm:geo-kD-lattice-optimized}, can be straightforwardly generalized to a larger class of unitaries called \emph{quantum cellular automata} (QCA), which play an important role in understanding quantum phases of matter~\cite{schumacher2004reversible,Gross_2012,Haah_2022,Shirley2022}.
These are unitaries that map any geometrically local operator to a geometrically local operator in the Heisenberg picture.
For any such unitary, our proof technique applies without any modification, yielding an efficient algorithm for learning any QCAs.
Interestingly, while shallow quantum circuits are QCAs by definition, the converse statement is not necessarily true.
For instance, shifting a set of qubits on a one-dimensional lattice trivially maps local operators to local operators.
However, it is impossible to decompose this unitary into a geometrically local shallow quantum circuit~\cite{Gross_2012}; see Ref.~\cite{Haah_2022,Shirley2022} for other nontrivial examples of QCA.
Therefore, our algorithm is applicable beyond shallow quantum circuits.

So far, we have been focusing on learning a shallow quantum circuit from a classical randomized measurement dataset.
A natural question asks if further improvement is possible when we allow more general quantum query access to $U$.
In the following, we show that by using quantum queries to $U$, an exponential improvement in query complexity is possible and this result is \emph{asymptotically-optimal} in both time and query complexity for learning geometrically-local shallow circuits over finite gate sets.
Surprisingly, quantum access also allows these circuits to be \emph{with certainty}, dropping the familiar qualifier of high probability.
The matching lower bounds stem from the need to query at least $\Omega(1)$ times to obtain any information about $U$ and to write down the learned $n$-qubit circuit, which requires $\Omega(n)$ time.

\begin{theorem}[Learning shallow circuits with quantum queries; see Theorem~\ref{thm:geo-finite-gates}] \label{thm:geo-finite-gates-main}
An unknown $n$-qubit geometrically-local shallow quantum circuit $U$ over a finite gate set can be learned to zero error with zero failure probability using $\Theta(1)$ queries to $U$ and $\Theta(n)$ quantum computational time.
\end{theorem}

\subsubsection{Learning output states of geometrically-local shallow quantum circuits}

Besides learning the $n$-qubit unitary $U$ using input-output queries, it is natural to study the problem of learning a pure quantum state $\ket{\psi}$ prepared by a shallow quantum circuit $U$, i.e., $\ket{\psi}=U\ket{0^n}$. Here, instead of given access to $U$, we are only given copies of the pure state $\ket{\psi}$ as in quantum state tomography~\cite{gross2010quantum, cramer2010efficient}. As discussed in Section~\ref{sec:background}, most families of efficient learnable quantum states, such as matrix product states~\cite{cramer2010efficient, lanyon2017efficient, gebhart2023learning} and stabilizer states~\cite{montanaro2017learning, gross2021schur, grewal2022low, grewal2023improved}, correspond to quantum circuit families that are classically easy to simulate~\cite{aaronson2004improved, cirac2021matrix}. In contrast, constant-depth quantum circuits are classically hard to simulate even when restricted to a 2D lattice~\cite{gao2017quantum,bermejo2018architectures}.

Learning $U \ket{0^n}$ from copies of $U \ket{0^n}$ has an incomparable difficulty to the earlier results because it has a less stringent requirement (learning an output state of $U$) but a more restricted access model (accessing copies of $U \ket{0^n}$ instead of $U$). While $\ket{\psi}=U\ket{0^n}$ can be learned from polynomially many copies~\cite{rouze2021learning,yu2023learning}, the restricted access model makes the problem computationally more challenging, and the question of whether there exists a polynomial time algorithm remains open. We give an efficient algorithm when $U$ is restricted to a 2D lattice.

\begin{theorem}[Learning quantum states prepared by 2D shallow circuits; see Theorem~\ref{thm:2dstate}] Given copies of an unknown  pure state $\ket{\psi}$, with the promise that $\ket{\psi}=U\ket{0^n}$ for an unknown geometrically-local circuit $U$ with depth $d$ over a 2-dimensional lattice.
One can learn a geometrically-local shallow circuit with depth $3d$ that prepares $\ket{\psi}$ to trace distance $\varepsilon$ with high probability, using $2^{\mc O(d^2)}\cdot (n/\varepsilon)^{\mc O(1)}$ copies of $\ket{\psi}$, in time $\left(n d^3/\varepsilon\right)^{\mc O(d^3)}$. When the circuit $U$ is over a finite gate set, $\ket{\psi}$ can be learned to zero error with high probability from $\mathcal{O}(\log(n))$ copies and $\mathcal{O}(n \log n )$ time.
\end{theorem}

Similarly, this result applies to $d=\mathrm{polylog}(n)$ depth at the cost of quasipolynomial running time. The efficient learnability of quantum states prepared by a shallow quantum circuit acting on 3D lattices (or on more general geometries) remains a challenging and interesting open problem.

\subsection{Discussion}

\paragraph{Higher circuit depth} In the general setting without geometric locality, we show that log-depth circuits require exponentially many quantum queries to learn within a small diamond distance (see Prop.~\ref{prop:hardnes-log-depth}), which is proven by showing that log-depth circuits can implement Grover's oracle over $2^n$ elements and applying the Grover lower bound \cite{zalka1999grover}. Therefore, our result for efficiently learning general constant-depth quantum circuits cannot be extended to much higher depth.

In the geometrically-local setting, Theorem~\ref{thm:geo-kD-lattice-optimized} implies polynomial-time learnability for quantum circuits on a $k$-dimensional lattice up to $\log(n)^{1/k}$ depth, and quasi-polynomial time for up to $\mathrm{polylog}(n)$ depth. What structural assumptions allow us to efficiently learn quantum circuits beyond polylog-depth remains an important open question.

\paragraph{Worst-case vs average-case distance} Motivated by the above discussion, it is natural to consider learning quantum circuits under weaker notions of distance, analogous to the classical notion of PAC learning. The standard notion of average-case distance in the literature~\cite{nielsen2002simple, montanaro2013survey} is defined as the distance between output states when averaging over input states generated by Haar random unitaries. While learning polynomial-size quantum circuits to small average-case distance can be achieved with polynomial sample complexity \cite{caro2022generalization, caro2023out}, the computational complexity of achieving a small average-case distance remains an open question.

In addition, Ref.~\cite{huang2022learning} considered a weaker notion of an average-case error where the goal is to learn observables of the output state for random input states and showed that under this notion, any quantum circuit (even those with exponential depth) could be learned in quasi-polynomial time.

\paragraph{Verifying the learned shallow quantum circuit}

Our learning algorithm provably works under the promise that the unknown $n$-qubit channel $\mathcal{C}$ corresponds to a unitary $\mathcal{C}(\rho) = U \rho U^\dagger$ and the unitary $U$ is generated by a shallow quantum circuit.
This promise does not necessarily hold: $U$ could be a deep quantum circuit that may or may not have a shallow quantum circuit implementation, and $\mathcal{C}$ may not be close to a unitary due to the noise in the quantum device.
Even if there is no promise of $\mc C$, one can still bluntly apply our learning algorithm to learn an $n$-qubit channel $\mathcal{E}$ generated by a shallow quantum circuit.
However, the learned circuit $\mc E$ is no longer guaranteed to be close to the true unknown channel $\mc C$.
This raises the question of whether we can verify the learned circuit $\mc E$ or the promise on $\mc C$.

In Section~\ref{sec:verify-learned-shallow-circuit}, we give an efficient verification algorithm that outputs $\textsc{pass}$ if $\mc E$ is close to $\mc C$ in the average-case distance and $\mc C$ is close to unitary. The verification algorithm outputs $\textsc{fail}$ if $\mc E$ is not close to $\mc C$.
Because $\mc E$ is generated by a shallow quantum circuit, the verification algorithm only needs to use the classical dataset consisting of random input product states and randomized Pauli measurement outcomes on the outputs of $\mc C$.

Being able to verify the learned shallow quantum circuits is central to applications such as compressing quantum circuits for a known unitary. In this case, we have a known $n$-qubit unitary $U$ that we know how to implement using a high-depth circuit. The goal is to learn a low-depth circuit that approximates $U$.
If $U$ does have a shallow circuit implementation, then our algorithm will learn a shallow circuit implementation for $U$.
However, $U$ may not have a shallow circuit implementation.
In this case, the verification algorithm can tell us that our learning algorithm has failed.
So far, we are using a simple verification algorithm based on a (weak) approximate local identity test, which only guarantees a small average-case distance.
Whether more advanced verification schemes can be used to achieve stronger guarantees efficiently is an interesting question that requires further exploration.

\section{Technical overview}
\label{sec:overview}

Let $U$ be an unknown $n$-qubit circuit of depth $d=\mc O(1)$.
We consider the following two tasks: (1) Learn a constant-depth circuit $\hat U$ from random data samples from $U$ or query access to $U$, such that $U$ and $\hat U$ are close in diamond distance. (2) Learn a constant-depth circuit $\hat U$ from measuring copies of the $n$-qubit state $\ket{\psi} = U\ket{0^n}$, such that $\hat U\ket{0^n}$ and $U \ket{0^n}$ are close in trace distance.

A basic idea to learn $U$ is to produce a guess $\hat{U}$ and check if $\hat U$ is close to $U$ (i.e., $\hat U^\dag\cdot U$ is close to identity).
While the search space over $\hat{U}$ is exponentially large, the \emph{locality} of shallow circuits allows us to search more efficiently.
For example, in the following figure, we can find a small \emph{local inversion} circuit $V_1$, that disentangles qubit 1 (the rightmost qubit), i.e., $UV_1\approx U'\otimes I_1$. Here, the input wires are at the bottom, and the output wires are at the top; $V_1$ is applied before applying $U$.

\begin{equation}\label{eq:localinversion}
    \begin{tikzpicture}[baseline={(0, 0.5*\tikzlength cm-\MathAxis pt)},x=\tikzlength cm,y=\tikzlength cm]
    \fill[blue!40!white]  (3.2,0) -- (5,0) -- (5,1)
 -- (4.3,1) -- cycle;
        \draw[black] (0,0) rectangle node{$U$} (5,1);
  \foreach \x in {1,2,...,5} {
        \draw[black] (\x-0.5, 0) -- (\x-0.5,-0.5);
        \draw[black] (\x-0.5, 1) -- (\x-0.5,1.5);
    }

\draw[black] (3.2,-0.5) -- (5,-0.5) -- (5,-1.5) -- (4.3,-1.5) -- (3.2,-0.5);
\node at (4.5,-1) {$V_1$};
  \draw[black] (4.5,-1.5) -- (4.5,-2);
  \draw[black] (3.5,-0.77275) -- (3.5,-2);
    \end{tikzpicture}
\quad\approx\quad
\begin{tikzpicture}[baseline={(0, 0.5*\tikzlength cm-\MathAxis pt)},x=\tikzlength cm,y=\tikzlength cm]
        \draw[black] (0,0) rectangle node{$U'$} (4,1);
  \foreach \x in {1,2,...,4} {
        \draw[black] (\x-0.5, 0) -- (\x-0.5,-0.5);
        \draw[black] (\x-0.5, 1) -- (\x-0.5,1.5);
    }
    \draw[black] (4.5, 1.5) -- (4.5,-0.5);
    \end{tikzpicture}
    \vspace{0.5em}
\end{equation}
This follows from a two-step argument. First, the existence of such a local inversion circuit is guaranteed by the locality of $U$, as undoing the gates in the backward lightcone (shaded blue region) of qubit 1 forms such a local inversion. Second, given a guess $V_1$, we develop an efficient procedure to check \emph{approximate local identity}, i.e. $UV_1\approx U'\otimes I_1$ for some $n-1$ qubit unitary $U'$. This allows us to find local inversions via brute force enumerate-and-test since the search space is small (as $V_1$ has depth $d$ and is supported within a constant size region). Note that after this exhaustive process, we may find a list of valid local inversions. The ``ground truth'' local inversion compatible with the unique global inverse of the unitary is among them, but we do not know which one.
Similarly, given copies of a state $\ket{\psi}=U\ket{0^n}$ we can find small local inversion circuits $V_1$ to disentangle qubit 1, $V_1\ket{\psi}\approx\ket{\psi'}\otimes\ket{0}_1$ for some $n-1$ qubit state $\ket{\psi'}$.

The above argument shows a procedure to efficiently learn local inversions for each qubit for both of our learning problems. The central question is whether this suffices to reconstruct the circuit and, if so, whether the reconstruction can be done efficiently. The main obstacle is that local inversions for each qubit are not unique, and two local inversions on neighboring qubits may not be consistent in the overlapping regions. Finding a consistent set of local inversions may require solving a constraint satisfaction problem that is computationally hard. Next, we show how to overcome this obstacle for learning $U$ and $\ket{\psi}=U \ket{0^n}$.

\subsection{Learning $U$ to a small diamond distance}

\subsubsection{Sewing local inversions}
\label{sec:overviewinversion}

Suppose we have learned a set of local inversions $\mc C_i$ for an unknown shallow quantum circuit $U$ for each qubit $i$.
Here, we show how to reconstruct the circuit using the learned local information. Surprisingly, the algorithm only requires an \emph{arbitrary} element $V_i\in\mc C_i$ for each qubit $i$, without the need to search for the element compatible with the global inverse, which could require solving a complicated constraint satisfaction problem.
The formal statements on this algorithmic technique are given in Section~\ref{sec:local-inv}.

For simplicity, here we first assume all the local inversions are found exactly without any approximation.
Take any $V_1\in\mc C_1$, applying it to the unknown circuit $U$ gives $UV_1=U'\otimes I_1$, see Eq.~\eqref{eq:localinversion}, where we imagine qubit 1 to be the rightmost qubit and use a simple 1D geometry for illustration.
This represents some progress: applying $V_1$ reduces the unknown $n$-qubit unitary $U$ to an unknown $(n-1)$-qubit unitary $U'$ (note that $U'$ may not be a shallow circuit).
A natural thought is whether we can keep making this progress by applying local inversion on other qubits.
The main issue here is that now the unitary has changed.
For example, consider qubit 2 which is right next to qubit 1.
Due to the fact that they have overlapping lightcones, some local inversion $V_2\in\mc C_2$ may no longer work for the new circuit $UV_1$.
Separately, we can attempt to find local inversion for qubit 2 with respect to this new circuit $UV_1$; however, doing so might disturb the progress we have made on qubit 1 and therefore requires coordinated effort across different qubits.
This is exactly the type of constraint satisfaction problem that we want to avoid.

Here we introduce a general approach to keep making progress: the idea is to introduce a fresh ancilla qubit, swap it with qubit 1, and then \emph{undo} the local inversion $V_1$. We show this in two steps: first, introduce a fresh ancilla qubit (red) and swap it with qubit 1,

\begin{equation}
    \begin{tikzpicture}[baseline={(0, 0.5*\tikzlength cm-\MathAxis pt)},x=\tikzlength cm,y=\tikzlength cm]
        \draw[black] (0,0) rectangle node{$U$} (5,1);
  \foreach \x in {1,2,...,5} {
        \draw[black] (\x-0.5, 0) -- (\x-0.5,-0.5);
        \draw[black] (\x-0.5, 1) -- (\x-0.5,1.5);
    }
\draw[black] (3.2,-0.5) -- (5,-0.5) -- (5,-1.5) -- (4.3,-1.5) -- (3.2,-0.5);
\node at (4.5,-1) {$V_1$};
  % \draw[black] (4.5,-1.5) -- (4.5,-2);
  \draw[black] (3.5,-0.77275) -- (3.5,-2);
  \draw[red] (5.5, 1.5) -- (5.5,-1.5);
  \draw[red] (5.5,-1.5)..controls (5.5,-2.5) and (4.5,-1.5)..(4.5,-2.5);
  \draw[black] (5.5,-2.5)..controls (5.5,-1.5) and (4.5,-2.5) ..(4.5,-1.5);
    \end{tikzpicture}
\quad=\quad
\begin{tikzpicture}[baseline={(0, 0.5*\tikzlength cm-\MathAxis pt)},x=\tikzlength cm,y=\tikzlength cm]
        \draw[black] (0,0) rectangle node{$U'$} (4,1);
  \foreach \x in {1,2,...,4} {
        \draw[black] (\x-0.5, 0) -- (\x-0.5,-0.5);
        \draw[black] (\x-0.5, 1) -- (\x-0.5,1.5);
    }
    \draw[black] (4.5, 1.5) -- (4.5,-0.5);
    \draw[red] (5.5, 1.5) -- (5.5,-0.5);
    \draw[red] (5.5,-0.5)..controls (5.5,-1.5) and (4.5,-0.5)..(4.5,-1.5);
  \draw[black] (5.5,-1.5)..controls (5.5,-0.5) and (4.5,-1.5) ..(4.5,-0.5);
    \end{tikzpicture}
\end{equation}
and then apply $V_1^\dag$,

\begin{equation}\label{eq:swappicture}
    \begin{tikzpicture}[baseline={(0, 0.5*\tikzlength cm-\MathAxis pt)},x=\tikzlength cm,y=\tikzlength cm]
        \draw[black] (0,0) rectangle node{$U$} (5,1);
  \foreach \x in {1,2,...,5} {
        \draw[black] (\x-0.5, 0) -- (\x-0.5,-0.5);
        \draw[black] (\x-0.5, 1) -- (\x-0.5,1.5);
    }
\draw[black] (3.2,-0.5) -- (5,-0.5) -- (5,-1.5) -- (4.3,-1.5) -- (3.2,-0.5);
\node at (4.5,-1) {$V_1$};
  \draw[black] (3.5,-0.77275) -- (3.5,-3.22725);
  \draw[red] (5.5, 1.5) -- (5.5,-1.5);
  \draw[red] (5.5,-1.5)..controls (5.5,-2.5) and (4.5,-1.5)..(4.5,-2.5);
  \draw[black] (5.5,-2.5)..controls (5.5,-1.5) and (4.5,-2.5) ..(4.5,-1.5);
  \draw[black] (4.3,-2.5) -- (5,-2.5) -- (5,-3.5) -- (3.2,-3.5) -- (4.3,-2.5);
\node at (4.5,-3) {$V_1^\dag$};
\draw[black] (5.5,-2.5) -- (5.5,-4);
\draw[red] (4.5,-3.5) -- (4.5,-4);
\draw[black] (3.5,-3.5) -- (3.5,-4);
    \end{tikzpicture}
\quad=\quad
\begin{tikzpicture}[baseline={(0, 0.5*\tikzlength cm-\MathAxis pt)},x=\tikzlength cm,y=\tikzlength cm]
        \draw[black] (0,0) rectangle node{$U'$} (4,1);
  \foreach \x in {1,2,...,4} {
        \draw[black] (\x-0.5, 0) -- (\x-0.5,-0.5);
        \draw[black] (\x-0.5, 1) -- (\x-0.5,1.5);
    }
    \draw[black] (4.5, 1.5) -- (4.5,-0.5);
    \draw[red] (5.5, 1.5) -- (5.5,-0.5);
    \draw[red] (5.5,-0.5)..controls (5.5,-1.5) and (4.5,-0.5)..(4.5,-1.5);
  \draw[black] (5.5,-1.5)..controls (5.5,-0.5) and (4.5,-1.5) ..(4.5,-0.5);
  \draw[black] (4.3,-1.5) -- (5,-1.5) -- (5,-2.5) -- (3.2,-2.5) -- (4.3,-1.5);
\node at (4.5,-2) {$V_1^\dag$};
\draw[black] (3.5, -2.5) -- (3.5,-3);
\draw[red] (4.5, -2.5) -- (4.5,-3);
\draw[black] (3.5, -0.5) -- (3.5,-2.22725);
\draw[black] (5.5, -1.5) -- (5.5,-3);
    \end{tikzpicture}
\quad=\quad
\begin{tikzpicture}[baseline={(0, 0.5*\tikzlength cm-\MathAxis pt)},x=\tikzlength cm,y=\tikzlength cm]
        \draw[black] (0,0) rectangle node{$U$} (5,1);
  \foreach \x in {1,2,...,4} {
        \draw[black] (\x-0.5, 0) -- (\x-0.5,-0.5);
        \draw[black] (\x-0.5, 1) -- (\x-0.5,1.5);
    }

    \draw[black] (5.5, 1) -- (5.5,-0.5);
    \draw[red] (5.5,2)..controls (5.5,1) and (4.5,2)..(4.5,1);
    \draw[black] (5.5,1)..controls (5.5,2) and (4.5,1)..(4.5,2);
    \draw[red] (4.5,0) -- (4.5,-0.5);
    \end{tikzpicture}
    \vspace{0.6em}
\end{equation}
To explain the second equality of Eq.~\eqref{eq:swappicture}, note that without the swap operation, the above procedure is not doing anything (since we just perform some operation and undo it). In the second picture of Eq.~\eqref{eq:swappicture}, after experiencing $V_1^\dag$, the red wire corresponds to the first output wire of $U$, but then it gets swapped out to the ancilla. Therefore, the overall effect is equivalent to performing a swap at the end after applying $U$.

The key reason that the above procedure is useful is because it \emph{repairs} the circuit. This allows us to continue doing the same operation on qubit 2 because even though a lot of operations were applied before $U$ (see the first picture in Eq.~\eqref{eq:swappicture}), it is equivalent to as if nothing were applied before $U$ (see the last picture in Eq.~\eqref{eq:swappicture}); therefore we can similarly apply $V_2^\dag$, swap with a new fresh qubit, and $V_2$ before $U$, achieving the effect of swapping qubit 2 at the end. Repeating the above procedure for all qubits, we have learned a circuit $\hat U$ acting on $2n$ qubits that satisfies

\begin{equation}
    \begin{tikzpicture}[baseline={(0, 0.5*\tikzlength cm-\MathAxis pt)},x=\tikzlength cm,y=\tikzlength cm]
        \draw[black] (0,0) rectangle node{$U$} (3,1);
  \foreach \x in {1,2,3} {
        \draw[black] (\x-0.5, 0) -- (\x-0.5,-0.5);
        \draw[black] (\x-0.5, 1) -- (\x-0.5,1.5);
    }
    \foreach \x in {4,5,6} {
        \draw[black] (\x-0.5, -0.5) -- (\x-0.5,1.5);
    }
    \draw[black] (0,-0.5) rectangle node{learned circuit $\hat U$} (6,-1.5);
    \foreach \x in {1,...,6} {
        \draw[black] (\x-0.5, -1.5) -- (\x-0.5,-2);
    }
    \end{tikzpicture}
\quad=\quad
\begin{tikzpicture}[baseline={(0, 0.5*\tikzlength cm-\MathAxis pt)},x=\tikzlength cm,y=\tikzlength cm]
        \draw[black] (0,0) rectangle node{$U$} (3,1);
  \foreach \x in {1,2,3} {
        \draw[black] (\x-0.5, 0) -- (\x-0.5,-0.5);
        \draw[black] (\x+2.5, 1) -- (\x+2.5,-0.5);
    }

\foreach \x in {1,2,3} {
        \draw[black] (\x-0.5,1)..controls (\x-0.5,1.5) and (\x+2.5,1.5)..(\x+2.5,2);
        \draw[black] (\x-0.5,2)..controls (\x-0.5,1.5) and (\x+2.5,1.5)..(\x+2.5,1);
    }
    \end{tikzpicture}
    \vspace{0.6em}
\end{equation}
which implies that $\hat U=S\cdot (U\otimes U^\dag)$, where $S$ denotes the global swap operation between the system and ancilla qubits. To implement $U$ using the learned circuit, on input $\rho$ we initialize an ancilla register with some arbitrary state (say $\ket{0^n}$), apply $S\cdot \hat U$ and trace out the ancilla register, and the output state equals $U\rho U^\dag$. We can use a similar procedure to implement $U^\dag$. Thus, the above procedure simultaneously learns to implement $U$ and $U^\dag$, using access only to $U$.

Finally, we remark that the learned circuit $S\cdot \hat U$ is shallow. To see this, note that $S=\mathrm{SWAP}^{\otimes n}$ is depth-1. $\hat U$ consists of unitaries of the form $W_i:=V_i \cdot \mathrm{SWAP} \cdot V_i^\dag$ that are \emph{local}: each of them supports on the lightcone of qubit $i$, as well as an extra ancilla qubit. Therefore we can implement non-overlapping $W_i$s simultaneously, and all of the $W_i$s can be stacked into a constant number of layers since, at most, a constant number of qubits share overlapping lightcones.

To achieve the optimal query and time complexity of $\Theta(1), \Theta(n)$ for learning geometrically-local shallow quantum circuits over finite gate sets in Theorem~\ref{thm:geo-finite-gates-main}, we present a quantum learning algorithm that finds the exact local inversions for all $n$ qubits with zero failure probability by querying $U$ for only $\mathcal{O}(1)$ times.
This surprising scaling is achieved by combining a few ideas: (a) coloring the geometry described by a bounded-degree graph, (b) decoupling the $n$-qubit unitary $U$ into $\mathcal{O}(n)$ few-qubit channels based on the coloring, and (c) designing a tournament to perfectly distinguish between two classes of few-qubit quantum channels: those that form an exact local identity versus those that do not.
The tournament uses the perfect distinguishability of certain pairs of CPTP maps shown in \cite{duan2009perfect}, where we design the few-qubit channels to ensure perfect distinguishability.
Then, the learning algorithm finds a good order to sew the local inversions to produce a constant-depth circuit implementation for the unknown constant-depth $n$-qubit circuit $U$.

\subsubsection{Sewing Heisenberg-evolved Pauli operators}
\label{sec:overviewHeisenberg}

Next, we describe a simpler technique based on directly sewing the Heisenberg-evolved Pauli operators $U^\dagger P_i U$ ($P_i$ is a single-qubit Pauli acting on qubit $i$) and discuss how it is closely related to local inversion. Section~\ref{sec:Hevo-Pauli} provides a detailed discussion of this technique.

We first describe how to learn the Heisenberg-evolved Pauli operators.
Because $U$ is a shallow quantum circuit,
each operator $U^\dagger P_i U$ acts on a constant number of qubits.
The few-qubit observable $U^\dagger P_i U$ can be reconstructed from the randomized measurement dataset. Let the random input product state be $\ket{\psi}=\ket{\psi_1}\otimes\cdots\otimes \ket{\psi_n}$, where $\ket{\psi_i}$ is a random one-qubit stabilizer state.
Because each qubit in the output state is measured in a random $X, Y, Z$ basis with equal probability, we will measure $P_i$ on the output state $U \ketbra{\psi}{\psi} U^\dagger$ with probability $1/3$.
This allows us to estimate $\expval{U^\dag P_i U}{\psi}$.
Then, we show that we can efficiently reconstruct $U^\dag P_i U$ from a small number of different random input states.

After learning the $3n$ Heisenberg-evolved Pauli operators $U^\dagger P_i U$, we present a direct approach for sewing them into a circuit.
This approach uses the identity $\mathrm{SWAP}=\frac{1}{2}\sum_{P\in\{I,X,Y,Z\}}P\otimes P$. Let $S_i$ be the $\mathrm{SWAP}$ gate acting on the $i$-th system qubit and the $i$-th ancilla qubit, let $S=\otimes_{i=1}^n S_i$ be the global swap between system and ancilla, and let $
    W_i := U^\dag S_i U = \frac{1}{2} \sum_{P \in \{I, X, Y, Z\}} U^\dagger P_i U \otimes P, \forall i = 1, \ldots, n$.
From the previous technique for sewing local inversion, we have proven the identity
\begin{equation}\label{eq:overviewimplement}
    U\otimes U^\dag = S\cdot \prod_{i=1}^n \left(V_i \cdot S_i \cdot V_i^\dag\right),
\end{equation}
where $V_i$ satisfies $UV_i=U'^{(i)}\otimes I_i$ is an arbitrary exact local inversion on qubit $i$. We can see that
\begin{equation}
    V_i \cdot S_i \cdot V_i^\dag = U^\dag U V_i \cdot S_i \cdot V_i^\dag U^\dag U = U^\dag S_i U = W_i \Longrightarrow U\otimes U^\dag = S\cdot \prod_{i=1}^n W_i = S\cdot \prod_{i=1}^n \left(U^\dag S_i U\right).
\end{equation}
The new equation can also be seen by itself: simply cancel $U$ with $U^\dag$ in the product so that the right-hand side becomes $S U^\dag S U$, and observe that
\begin{equation}
    \begin{tikzpicture}[baseline={(0, 0.5*\tikzlength cm-\MathAxis pt)},x=\tikzlength cm,y=\tikzlength cm]
        \draw[black] (0.1,0) rectangle node{$U$} (2.9,1);
        \draw[black] (3.1,0) rectangle node{$U^\dag$} (5.9,1);
  \foreach \x in {1,2,...,6} {
        \draw[black] (\x-0.5, 0) -- (\x-0.5,-0.5);
        \draw[black] (\x-0.5, 1) -- (\x-0.5,1.5);
    }
    \end{tikzpicture}
\quad=\quad
\begin{tikzpicture}[baseline={(0, 0.5*\tikzlength cm-\MathAxis pt)},x=\tikzlength cm,y=\tikzlength cm]
        \draw[black] (0.1,0) rectangle node{$U^\dag$} (2.9,1);
        \draw[black] (0.1,-1) rectangle node{$U$} (2.9,-2);
  \foreach \x in {1,2,3} {
        \draw[black] (\x+2.5, 1) -- (\x+2.5,0);
        \draw[black] (\x+2.5, -1) -- (\x+2.5,-2.5);
        \draw[black] (\x-0.5, -2) -- (\x-0.5,-2.5);
    }

\foreach \x in {1,2,3} {
        \draw[black] (\x-0.5,1)..controls (\x-0.5,1.5) and (\x+2.5,1.5)..(\x+2.5,2);
        \draw[black] (\x-0.5,2)..controls (\x-0.5,1.5) and (\x+2.5,1.5)..(\x+2.5,1);
    }
    \foreach \x in {1,2,3} {
        \draw[black] (\x-0.5,-1)..controls (\x-0.5,-0.5) and (\x+2.5,-0.5)..(\x+2.5,0);
        \draw[black] (\x-0.5,0)..controls (\x-0.5,-0.5) and (\x+2.5,-0.5)..(\x+2.5,-1);
    }
    \end{tikzpicture}
\end{equation}
As we can see, the Heisenberg-evolved Pauli operators can be directly sewn into $U \otimes U^\dagger$.

This outlines the following procedure to learn $U$: first learn the Heisenberg-evolved Pauli operators $\{U^\dag P_i U\}_{i=1}^n$, combine them to form $\{W_i\}_{i=1}^n$ according to $W_i=\frac{1}{2}\sum_{P\in\{I,X,Y,Z\}}U^\dag P_i U\otimes P_i$, and reconstruct the circuit using $\{W_i\}_{i=1}^n$. Note that each $W_i$ acts on a constant number $k$ of qubits and can be directly compiled into a circuit of depth $2^{O(k)}$.
To further optimize the depth of the learned circuit, notice that each $W_i$ has the form $W_i=U^\dag S_i U=V_i S_i V_i^\dag$, i.e., it can be represented by a depth-$(2d+1)$ circuit.
We can find such a representation for $W_i$ by brute-force enumerating all depth-$(2d+1)$ circuits acting on $k$ qubits, and the learned circuit has the same form as in Section~\ref{sec:overviewinversion}.
This thus provides a simpler framework for learning an unknown shallow quantum circuit $U$ using a classical dataset containing random samples about $U$.

To prove Theorem~\ref{thm:learn-shallow-QC-main} and \ref{thm:learn-geo-main} on learning general and geometrically-local shallow quantum circuits, we combine this framework with some additional ideas on (a) coloring the $k$-dimensional lattices to ensure all qubits with the same color has nonoverlapping lightcone, (b) truncating small Fourier coefficients to ensure the learned observables acts only on qubits in the support of the true observables, (c) compiling the Heisenberg-evolved Pauli operator when over a finite gate set, and (d) finding a good order to sew the Heisenberg-evolved Pauli operators into a short-depth circuit.

\subsection{Learning $U \ket{0^n}$ to a small trace distance}

Next, we discuss how to learn a quantum state $\ket{\psi}=U\ket{0^n}$ prepared by a shallow circuit $U$, given copies of $\ket{\psi}$.
While this problem appears to be simpler (we need to learn $U\ket{0^n}$ instead of the entire $U$), the weaker access model (we only have access to the output of $U$ for the all-zero input state $\ket{0^n}$) poses new fundamental challenges.
In particular, we can learn local inversions $V_i$ that give $V_i U \ket{0^n} = \ket{\psi'}\otimes \ket{0}_i$ instead of the much stronger $U V_i = U'\otimes I_i$, and the previous approach of ``keep making progress by swapping ancilla qubits'' does not seem to work.

Here, we address these challenges by developing new techniques tailored to a 2D lattice.
The main idea is to \emph{disentangle} the state into many 1D-like states that are easy to learn by leveraging the fact that 1D constraint satisfaction problems can be efficiently solved.

\subsubsection{Disentangling a 2D quantum state}
\label{sec:overview2ddisentangle}

Our starting point is the simpler problem of learning a state $\ket{\psi}=U\ket{0^n}$, with the promise that $U$ is a shallow circuit (white box) acting on a 1D lattice:
\begin{equation}\label{eq:1dargument}
    \begin{tikzpicture}[baseline={(0, 1*\tikzlength cm-\MathAxis pt)},x=\tikzlength cm,y=\tikzlength cm]
  \draw[black] (0,0) rectangle node{$U$} (15,1);
  \draw[blue,thick] (0,1.2) -- (0,2.2) -- (3,2.2) -- (4,1.2) -- (0,1.2);
  \draw[red,thick] (2,1.15) -- (3,2.15) -- (6,2.15) -- (7,1.15) -- (2,1.15);
  \draw[ForestGreen,thick] (5,1.2) -- (6,2.2) -- (9,2.2) -- (10,1.2) -- (5,1.2);

\draw[dashed] (3,-0.5) -- (3,2.7);
\draw[dashed] (6,-0.5) -- (6,2.7);
\draw[dashed] (9,-0.5) -- (9,2.7);
\node[anchor=north] at (1.5,0) {$A$};
\node[anchor=north] at (4.5,0) {$B$};
\node[anchor=north] at (7.5,0) {$C$};
\end{tikzpicture}
\end{equation}
Let $A$, $B$, and $C$ be contiguous regions of constant size. We can find a set of local inversions $\mc C_A$ for $A$ by enumerating over circuits acting on the lightcone of $A$ (blue shape). The question is how to combine different local inversions into a circuit. The key observation is that two neighboring local inversions can be merged together if they are ``consistent'', i.e., sharing the same gates where they overlap. For example, some $V_A\in \mc C_A$ (blue) and $V_B\in\mc C_B$ (red) can be merged into a larger circuit of the same depth $V_{AB}$ if they share the same gates in the overlapping region (intersecting triangle); the merged circuit $V_{AB}$ satisfies $V_{AB}\ket{\psi}=\ket{\psi'}\otimes \ket{0}_{AB}$. This defines a constraint satisfaction problem: we need to find a local inversion for each region such that neighboring local inversions are consistent. Such a solution must exist (since the ``ground truth'' local inversions satisfy these constraints), and we can efficiently find such a solution by simple dynamic programming in time $\mc O(n |\mc C|^2)$ where $|\mc C|$ denotes the maximum number of local inversions for a small region. This gives a circuit $V$ that satisfies $V\ket{\psi}=\ket{0^n}$, so the state $\ket{\psi}$ can be prepared by $\ket{\psi}=V^\dag \ket{0^n}$.

From this perspective, generalizing this approach to 2D may be a difficult task since constraint satisfaction problems on 2D lattices are $\NP$-hard in general.
We address this challenge using an additional insight: instead of solving the constraint satisfaction problem directly in 2D, we first use the 1D argument to disentangle the 2D state.
\begin{equation}\label{eq:2ddisentangle}
    \begin{tikzpicture}[baseline={(0, 4*\tikzlength cm-\MathAxis pt)},x=0.8*\tikzlength cm,y=0.8*\tikzlength cm]

  \fill[blue!40!white] (3.75,7) rectangle (4.25,8) ;
  \fill[red!40!white] (3.75,6) rectangle (4.25,7);
  \fill[ForestGreen!40!white] (3.75,5) rectangle (4.25,6);

  \draw[blue,densely dashed] (3.65,6.9) rectangle (4.35,8);
  \draw[red,densely dashed] (3.64,5.9) rectangle (4.36,7.1);
\draw[ForestGreen,densely dashed] (3.65,4.9) rectangle (4.35,6.1);

  \draw[black] (0,0) rectangle (8,8);
  \draw[black] (3.75,0) -- (3.75,8);
  \draw[black] (4.25,0) -- (4.25,8);

\node at (1.875,4) {$A$};
\node at (4,4) {$B$};
\node at (6.125,4) {$C$};
\end{tikzpicture}\quad\quad\quad\quad
\begin{tikzpicture}[baseline={(0, 4*\tikzlength cm-\MathAxis pt)},x=0.8*\tikzlength cm,y=0.8*\tikzlength cm]
  \foreach \x in {1,2,...,7} {
        \fill[gray!40!white] (\x-0.25,0) rectangle node[black]{{\tiny $B_\x$}} (\x+0.25,8);
    }
\foreach \x in {1,2,...,8} {
        \node at (\x-0.5, 5) {{\tiny $A_\x$}};
    }

  \draw[black] (0,0) rectangle (8,8);

\end{tikzpicture}
\end{equation}

The LHS of \eqref{eq:2ddisentangle} shows a quantum state $\ket{\psi}$ prepared by a depth-$d$ circuit acting on a 2D lattice, divided into three regions $A$, $B$, and $C$. A well-known fact about these states is that they have \emph{finite correlation length}: if the width of $B$ is sufficiently large (say $5d$), then the mutual information between $A$ and $C$ is zero, i.e. the reduced density matrix of $\rho=\ketbra{\psi}$ on $AC$ satisfies $\rho_{AC}=\rho_A\otimes \rho_C$. This fact itself does not simplify the problem because $A$ and $C$ are both entangled with $B$. However, if for some reason we have $\rho_B=\ketbra{0}_B$, then this would force $\rho_A$ and $\rho_C$ to be pure states and not entangled with any outside qubits.

But this is exactly what we can achieve using the 1D argument: we can learn local inversions for a small piece of $B$ (shaded blue) by finding circuits acting on a slightly larger region (dotted blue). We can do this for contiguous small regions (here, the blue, red, and green regions play exactly the same role as in \eqref{eq:1dargument}), and by repeating the 1D argument we can find a depth-$d$ circuit $V$ acting on a region slightly larger than $B$, such that $\Tr_{AC}(V\ketbra{\psi}V^\dag)=\ketbra{0}_B$. After applying $V$, the state becomes $\ket{\phi}_A\otimes \ket{0}_B\otimes \ket{\phi}_C$ for some unknown pure states $\ket{\phi}_A$, $\ket{\phi}_C$.

Finally, note that this argument can be repeated horizontally across the entire system; overall, we can learn a depth-$d$ circuit $V$ such that $V\ket{\psi}$ has the form of RHS in \eqref{eq:2ddisentangle}. Here, all the shaded $B$ regions are inverted and in the state $\ket{0}$. Each of the white regions is in a pure state and disentangled with each other. Now, the problem is reduced to learning each of the states $\ket{\phi}_{A_i}$ on the white regions separately. To prepare $\ket{\psi}$, we first prepare $(\otimes_i\ket{\phi}_{A_i})\otimes \ket{0}_B$, then apply $V^\dag$.

\subsubsection{Learning finite correlated states in 1D}
\label{sec:overview1dfinitecorrelated}

Here we address the final step of learning the 1D-like states $\ket{\phi}_{A_i}$. The main challenge here is that the previous argument in \eqref{eq:1dargument} is not immediately applicable: we do not have the guarantee that $\ket{\phi}_{A_i}$ is prepared by a shallow circuit acting on $\ket{0}_{A_i}$. Instead, what we know is that the global state $(\otimes_i\ket{\phi}_{A_i})\otimes \ket{0}_B$ is prepared by a depth-$2d$ circuit acting on $\ket{0}_{AB}$, because it equals to $V\ket{\psi}$.

Our starting point is to observe the following structure of the state $\ket{\phi}_{A_i}$: it can be prepared by a depth-$2d$ circuit acting on $A_i$ as well as some ancilla qubits $A_i^L$ and $A_i^R$ (see Fig.~\ref{fig:1dancilla} for an illustration). To see this, recall that $\ket{\phi}_{A_i}$ is part of a state that is prepared by a depth-$2d$ circuit. Now, imagine that we \emph{undo} all the gates in that circuit, except for those in the \emph{backward lightcone} of $A_i$. This procedure does not affect the state on $A_i$, and the resulting circuit (denoted as $W_i$) has exactly the same shape as in Fig.~\ref{fig:1dancilla}, where $A_i^L$, $A_i^R$ both have width $2d$. We then develop an algorithm to learn such a depth-$2d$ circuit to prepare $\ket{\phi}_{A_i}$. This problem is different from \eqref{eq:1dargument} in nature due to the existence of ancilla qubits. However, its simple 1D structure allows us to develop a similar argument by solving a 1D constraint satisfaction problem. This implies that we can learn a depth-$2d$ circuit to prepare the entire system in RHS of \eqref{eq:2ddisentangle}. Thus the total learned circuit depth to prepare $\ket{\psi}$ equals $3d$ (see Claim 2 of Theorem~\ref{thm:2dstate}).

In addition, we give a separate argument showing that each of the disentangled states $\ket{\phi}_{A_i}$ in RHS of \eqref{eq:2ddisentangle} can be prepared with a 1D circuit of depth $2^{\mc O(d^2)}$ without any ancilla qubits. This implies an algorithm where the learned circuit for preparing $\ket{\psi}$ has depth $2^{\mc O(d^2)}$ and does not use ancilla qubits (see Claim 3 of Theorem~\ref{thm:2dstate}).

Finally, note that throughout Section~\ref{sec:overview2ddisentangle} and \ref{sec:overview1dfinitecorrelated} we have been working with a simple setting with a finite gate set, which allows each step in the above argument to be performed \emph{exactly} without any approximation error. Generalizing these arguments to arbitrary $\mathrm{SU}(4)$ gates requires each step of the argument to be \emph{robust}, in the sense that small errors in each step do not accumulate significantly. In particular, we can only approximately disentangle the state using the procedure in \eqref{eq:2ddisentangle}, and learning the remaining 1D states poses new technical challenges as they are no longer pure. These issues are addressed in Section~\ref{sec:2dlearningrobustness}, which leads to a robust version of the above result; see Claim 1 of Theorem~\ref{thm:2dstate}.

% \newpage
% \vspace{4em}
% \appendix

% \renewcommand{\appendixname}{APPENDIX}
% \renewcommand{\thesubsection}{\MakeUppercase{\alph{section}}.\arabic{subsection}}
% \makeatletter
% \renewcommand{\p@subsection}{}
% \makeatother

% \noindent
% \textbf{ \LARGE{}Appendices }
% \vspace{-0.5em}

% \tableofcontents

% \vspace{1em}

\section{Preliminaries}
\label{sec:prelim}

Let $\mathrm{stab}_1 = \{\ket{0}, \ket{1}, \ket{+}, \ket{-}, \ket{y+}, \ket{y-}\}$ be the set of single-qubit stabilizer states.
Given an $n$-qubit unitary $U$, we use the Catholic letter $\mathcal{U}$ to denote the corresponding CPTP map $\mathcal{U}(X) = U X U^\dagger$.
We denote $\mathcal{I}$ as the identity CPTP map.
Given a Pauli operator $P \in \{X, Y, Z\}$, we consider $P_i$ to be a multi-qubit operator that is equal to the tensor product of $P$ on the $i$-th qubit and identity on the rest of the qubits.
We also consider the following definitions.

\begin{definition}[Reduced channel] \label{def:reduced-channel}
    Given $n > 0$, $i \in \{1, \ldots, n\}$, and an $n$-qubit CPTP map~$\mathcal{C}$. The reduced channel $\mathcal{E}^{\mathcal{C}}_{\neq i}$ of the CPTP map $\mathcal{C}$ with the $i$-th qubit removed is
    \begin{equation}
        \mathcal{E}^{\mathcal{C}}_{\neq i}(\rho_{\neq i}) =  \Tr_{i}\left( \mathcal{C} \left(\frac{I^{(i)}}{2} \otimes \rho_{\neq i} \right) \right),
    \end{equation}
    where $\rho_{\neq i}$ is a density matrix on all except the $i$-th qubit, $I^{(i)}$ is the identity on the $i$-th qubit, and $\Tr_{i}$ is the partial trace over the $i$-th qubit.
    For $k \in \{0, 1, \ldots, n\}$, we define
    \begin{equation}
        \mathcal{E}^{\mathcal{C}}_{> k}(\rho_{> k}) =  \Tr_{\leq k}\left( \mathcal{C} \left(\frac{I^{(1, \ldots, k)}}{2^k} \otimes \rho_{> k} \right) \right),
    \end{equation}
    where $\rho_{> k}$ is a density matrix on all except the first $k$ qubits, $I^{(1, \ldots, k)}$ is the identity on the first $k$ qubits, and $\Tr_{\leq k}$ is the partial trace over the first $k$ qubits.
    Given a subset of qubits $S \subseteq \{1, \ldots, n\}$, we define
    \begin{equation}
        \mathcal{E}^{\mathcal{C}}_{S}(\cdot) =  \Tr_{\notin S}\left( \mathcal{C} \left(\frac{I^{(\notin S)}}{2^{n - |S|}} \otimes (\cdot) \right) \right),
    \end{equation}
    where $I^{(\notin S)}$ is the identity on qubits not in $S$ and $\Tr_{\notin S}$ is the partial trace over qubits not in $S$.
\end{definition}

\begin{definition}[Fidelity] \label{def:fidelity}
Given two quantum states $\rho, \sigma$. The fidelity $\mathcal{F}(\rho, \sigma) \in [0, 1]$ between the two states is defined as $\Tr(\sqrt{\sqrt{\rho} \sigma \sqrt{\rho}})^2$. If $\sigma = \ketbra{\psi}{\psi}$, then $\mathcal{F}(\rho, \sigma) = \bra{\psi} \rho \ket{\psi}$.
\end{definition}

\begin{fact}[Properties of fidelity \cite{bengtsson2017geometry}] \label{fact:fidelity}
    The function $1 - F(\rho, \sigma)$ satisfies
    \begin{align}
        1 - F(\rho, \sigma) &= 1 - F(\sigma, \rho) & \text{(symmetric)};\\
        1 - F(\rho, \sigma) &\geq 0 & \text{(nonnegative)};\\
        1 - F(\rho, \sigma) &= 0 \iff \rho = \sigma & \text{(identity of indiscernible)}.
    \end{align}
    But $1 - F$ does not satisfy triangle inequality. In contrast, $\Theta(\rho, \sigma) := \arcsin(\sqrt{1 - F(\rho, \sigma)}) \in [0, \pi/2]$ is symmetric, nonnegative, and satisfies identity of indiscernible and triangle inequality,
    \begin{equation}
        \Theta(\rho, \sigma) \leq \Theta(\rho, \tau) + \Theta(\tau, \sigma).
    \end{equation}
    Hence, $\theta(\rho, \sigma)$ is a metric (known as the Fubini-Study metric), but $1 - F(\rho, \sigma)$ is not.
    In addition to the metric properties, we also have
    \begin{equation}
        1 - F(\psi, \rho) \leq \frac{1}{2} \norm{\psi - \rho}_{\mathrm{tr}},
    \end{equation}
    for any state $\rho$ and any pure state $\psi$, where $\norm{\cdot}_{\mathrm{tr}}$ is the trace norm.
    Also, the fidelity is monotonic increasing under CPTP maps,
    \begin{equation}
        F(\mathcal{E}(\rho), \mathcal{E}(\sigma)) \geq F(\rho, \sigma),
    \end{equation}
    for any CPTP map $\mathcal{E}$ and any state $\rho, \sigma$.
\end{fact}

\begin{definition}[Average-case distance] \label{def:ave-dist}
Given two $n$-qubit CPTP maps $\mathcal{E}_1, \mathcal{E}_2$. The average-case distance $\mathcal{D}_{\mathrm{ave}}(\mathcal{E}_1, \mathcal{E}_2)$ between the two CPTP maps is defined as
\begin{equation}
\Exp_{\ket{\psi}: \mathrm{Unif}} \big[ 1 - \mathcal{F}( \mathcal{E}_1(\ketbra{\psi}{\psi}), \mathcal{E}_2(\ketbra{\psi}{\psi}) ) \big],
\end{equation}
where $\Exp_{\ket{\psi}: \mathrm{Unif}}$ considers averaging under the uniform measure over pure states.
\end{definition}

\begin{fact}[Haar average for average-case distance \cite{nielsen2002simple}] \label{fact:ave-dist}
    Given an $n$-qubit CPTP map $\mathcal{E}$ and an $n$-qubit unitary $U$. We have the following identity,
    \begin{equation}
        \mathcal{D}_{\mathrm{ave}}(\mathcal{E}, \mathcal{U}) = \frac{2^n}{2^n + 1} \left( 1 - \frac{1}{4^n} \sum_{i, j} \bra{i} \mathcal{E}\left( U^\dagger \ketbra{i}{j} U \right) \ket{j} \right),
    \end{equation}
    after averaging over the uniform measure over pure states.
\end{fact}

\begin{proposition}[Normalized Frobenius norm] \label{prop:ave-dist-Frob}
    Given two $n$-qubit unitaries $U_1, U_2$. We have
    \begin{equation}
        \frac{1}{3} \min_{\phi \in \mathbb{R}}  \frac{\norm{ e^{i\phi} U_1 - U_2}^2_F}{2^n} \leq \mathcal{D}_{\mathrm{ave}}(\mathcal{U}_1, \mathcal{U}_2) \leq \min_{\phi \in \mathbb{R}} \frac{ \norm{ e^{i\phi} U_1 - U_2}^2_F}{2^n},
    \end{equation}
    where $\norm{X}_F = \sqrt{\Tr(X^\dagger X)}$ is the Frobenius norm of $X$.
\end{proposition}
\begin{proof}
    From \cite{nielsen2002simple}, the average-case distance (also known as the average gate fidelity) satisfies
    \begin{equation}
        \mathcal{D}_{\mathrm{ave}}(\mathcal{U}_1, \mathcal{U}_2) = \frac{2^n}{2^n + 1} \left(1 - \frac{1}{4^n} \left| \Tr(U_1^\dagger U_2) \right|^2 \right).
    \end{equation}
    Expanding the definition of Frobenius norm, we have
    \begin{equation}
        \min_{\phi \in \mathbb{R}} \frac{ \norm{ e^{i\phi} U_1 - U_2}^2_F}{2^n} = 2 \left( 1 - \frac{\left| \Tr(U_1^\dagger U_2) \right|}{2^n} \right).
    \end{equation}
    Recall that
    \begin{equation}
        0 \leq \frac{\left| \Tr(U_1^\dagger U_2) \right|}{2^n} \leq 1.
    \end{equation}
    Hence, we have
    \begin{equation}
        \left( 1 - \frac{\left| \Tr(U_1^\dagger U_2) \right|}{2^n} \right) \leq \left(1 + \frac{\left| \Tr(U_1^\dagger U_2) \right|}{2^n}\right) \left( 1 - \frac{\left| \Tr(U_1^\dagger U_2) \right|}{2^n} \right) \leq 2 \left( 1 - \frac{\left| \Tr(U_1^\dagger U_2) \right|}{2^n} \right).
    \end{equation}
    This immediately implies that
    \begin{equation}
        \frac{2}{3} \left( 1 - \frac{\left| \Tr(U_1^\dagger U_2) \right|}{2^n} \right) \leq \frac{2^n}{2^n + 1} \left(1 - \frac{\left| \Tr(U_1^\dagger U_2) \right|^2}{4^n} \right) \leq 2 \left( 1 - \frac{\left| \Tr(U_1^\dagger U_2) \right|}{2^n} \right)
    \end{equation}
    which is equivalent to
    \begin{equation}
        \frac{1}{3} \min_{\phi \in \mathbb{R}}  \frac{\norm{ e^{i\phi} U_1 - U_2}^2_F}{2^n} \leq \mathcal{D}_{\mathrm{ave}}(\mathcal{U}_1, \mathcal{U}_2) \leq \min_{\phi \in \mathbb{R}} \frac{ \norm{ e^{i\phi} U_1 - U_2}^2_F}{2^n}.
    \end{equation}
    This concludes the proof.
\end{proof}

\begin{definition}[Worse-case distance / diamond distance] \label{def:diamond-dist}
Given two $n$-qubit CPTP maps $\mathcal{E}_1, \mathcal{E}_2$. The worst-case distance $\mathcal{D}_\diamond(\mathcal{E}_1, \mathcal{E}_2)$ between the two CPTP maps is defined as
\begin{equation}
\frac{1}{2} \max_{\rho} \norm{ (\mathcal{E}_1 \otimes \mathcal{I})(\rho) - (\mathcal{E}_2 \otimes \mathcal{I})(\rho) }_1 \triangleq \frac{1}{2} \norm{\mathcal{E}_1 - \mathcal{E}_2}_\diamond,
\end{equation}
where $\rho$ is maximized over $2n$-qubit states and $\mathcal{I}^{(>n)}$ is an identity map acting on the $n$ qubits.
$\mathcal{D}_\diamond(\mathcal{E}_1, \mathcal{E}_2)$ is also known as diamond distance and $\norm{\cdot}_\diamond$ is the diamond norm.
\end{definition}

\begin{fact}[Diamond distance for unitaries; Prop.~1.6 of \cite{haah2023query}] \label{fact:diamond-unitary}
    For any two unitaries $U_1, U_2$, we have
    \begin{equation}
        \min_{\phi \in \mathbb{R}} \norm{e^{i \phi} U_1 - U_2}_\infty \leq \norm{\mathcal{U}_1 - \mathcal{U}_2 }_\diamond \leq 2 \min_{\phi \in \mathbb{R}} \norm{e^{i \phi} U_1 - U_2}_\infty.
    \end{equation}
\end{fact}

\begin{fact}[Exact unitary synthesis; see e.g.~\cite{Barenco1995Elementary, shende2005synthesis}] \label{fact:unitary-synthesis}
Given any unitary $U$ acting on $k$ qubits, there is an algorithm that outputs a circuit (acting on $k$ qubits) consisting of at most $4^{k}$ two-qubit gates, which exactly implements the unitary $U$, in time $2^{O(k)}$.
\end{fact}

\begin{corollary}[Exact unitary synthesis in geometrically-local circuit] \label{cor:unitary-synthesis-geo}
Given any unitary $U$ acting on $k$ qubits and a connected graph $G$ over $k$ qubits, there is an algorithm that outputs a geometrically-local circuit (acting on $k$ qubits and consists only of gates between connected qubits) consisting of at most $2k 4^{k}$ two-qubit gates, which exactly implements the unitary $U$, in time $2^{O(k)}$.
\end{corollary}
\begin{proof}
    For each two-qubit gate in the original synthesis protocol, which may not be geometrically-local under the connectivity graph $G$, we consider at most $k-1$ swap gates to move one of the qubits from the original location to a location next to the other qubit, apply the two-qubit gate, then perform at most $k-1$ swap gates to move the qubit back to the original location.
\end{proof}

\section{Approximate local identity}
\label{sec:approx-local-id}

A central concept that we will use to define local inversion for representing $n$-qubit unitaries is the \emph{$\varepsilon$-approximate local identity}.
In this section, we provide the properties for understanding the concept of approximate local identity.
In particular, we will consider a strong and a weak form of local identity in Section~\ref{sec:strong-localid} and \ref{sec:weak-localid}.
In each section, we state the definition, show how to characterize if a unitary map forms a strong/weak $\varepsilon$-approximate local identity, and prove how local identity relates to global identity.

\subsection{Strong $\varepsilon$-approximate local identity}
\label{sec:strong-localid}

We begin by looking at a strong form of approximate local identity. The idea is that the action of the $n$-qubit unitary $U$ on the $i$-th qubit is close to the identity map, while the action on the other qubits is close to the reduced channel of $U$ with the $i$-th qubit removed (feed in a maximally mixed state on qubit $i$ and trace out qubit $i$ at the end). Recall Definition~\ref{def:reduced-channel} of reduced channel,
\begin{equation}
        \mathcal{E}^{\mathcal{U}}_{\neq i}(\rho_{\neq i}) =  \Tr_{i}\left( \mathcal{U} \left(\frac{I^{(i)}}{2} \otimes \rho_{\neq i} \right) \right),
\end{equation}
where $\rho_{\neq i}$ is a density matrix on all except the $i$-th qubit, $I^{(i)}$ is the identity on the $i$-th qubit, and $\Tr_{i}$ is the partial trace over the $i$-th qubit.

\begin{definition}[Strong $\varepsilon$-approximate local identity]
    Given $n > 0, \varepsilon \geq 0,$ and $i \in \{1, \ldots, n\}$. An $n$-qubit unitary $U$ is a strong $\varepsilon$-approximate local identity on the $i$-th qubit if
    \begin{equation}
        \mathcal{D}_\diamond\left(\mathcal{U}, \, \mathcal{I}^{(i)} \otimes \mathcal{E}^{\mathcal{U}}_{\neq i}\right) \leq \varepsilon,
    \end{equation}
    where $\mathcal{I}^{(i)} \otimes \mathcal{E}^{\mathcal{U}}_{\neq i}$ is an $n$-qubit CPTP map that acts as identity on the $i$-th qubit.
\end{definition}

While diamond distances are typically hard to characterize, the strong $\varepsilon$-approximate local identity can be characterized up to a constant factor by studying the Heisenberg evolution of single-qubit Pauli observables under the $n$-qubit unitary $U$.
Hence, in order to check if an $n$-qubit unitary $U$ strong approximate local identity on the $i$-th qubit, all we need to check is whether the three Pauli observables $X_i, Y_i, Z_i$ remains approximately unchanged after Heisenberg evolution under $U$.

\begin{lemma}[Characterization of strong $\varepsilon$-approximate local identity] \label{lem:char-strong-local-id}
    Given $n > 0$, $\varepsilon \geq 0$, and an $n$-qubit unitary $\mathcal{U}$. If $\mathcal{U}$ is a strong $\varepsilon$-approximate local identity on the $i$-th qubit, then
    \begin{equation}
        \frac{1}{2} \norm{U^\dagger P_i U - P_i}_\infty \leq \varepsilon, \forall P \in \{X, Y, Z\},
    \end{equation}
    where $P_i$ is the Pauli operator $P$ acting only on qubit $i$, and $U^\dagger P_i U$ is the Heisenberg evolution of $P_i$ under $U$.
    Furthermore, if the following holds,
    \begin{equation}
        \frac{1}{2} \sum_{P \in \{X, Y, Z\}} \norm{ U^\dagger P_i U - P_i}_\infty \leq \varepsilon,
    \end{equation}
    then $\mathcal{U}$ is a strong $\varepsilon$-approximate local identity on the $i$-th qubit.
\end{lemma}
\begin{proof}
    We start by showing the first claim.
    Consider any $n$-qubit pure state $\ket{\psi}$. We have
    \begin{equation}
        \norm{U^\dagger P_i U - P_i}_\infty = \max_{\ket{\psi}} \left| \bra{\psi} \left( U^\dagger P_i U - P_i \right) \ket{\psi} \right|.
    \end{equation}
    By the definition of CPTP maps, we have
    \begin{equation}
        \bra{\psi} U^\dagger P_i U\ket{\psi} = \Tr\left( P_i \mathcal{U}\left(\ketbra{\psi}{\psi}\right) \right).
    \end{equation}
    From the definition of diamond distance and of strong $\varepsilon$-approximate local identity on the $i$-th qubit, we have the following inequality,
    \begin{equation}
        \frac{1}{2} \left| \Tr\left( P_i \mathcal{U}\left(\ketbra{\psi}{\psi}\right) \right) - \Tr\left( P_i \left(\mathcal{I}^{(i)} \otimes \mathcal{E}^{\mathcal{U}}_{\neq i}\right)\left(\ketbra{\psi}{\psi}\right) \right) \right| \leq \varepsilon.
    \end{equation}
    By the definition of a CPTP map, we have
    \begin{equation}
        \Tr_{\neq i}\left(\mathcal{E}^{\mathcal{U}}_{\neq i}(\rho) \right) = \rho
    \end{equation}
    for any quantum state $\rho$, where $\Tr_{\neq i}$ traces out all qubits except for qubit $i$.
    Hence, we have $\Tr\left( P_i \left(\mathcal{I}^{(i)} \otimes \mathcal{E}^{\mathcal{U}}_{\neq i}\right)\left(\ketbra{\psi}{\psi}\right) \right) = \Tr(P_i \ketbra{\psi}{\psi})$.
    Together, we obtain the first claim.

    The second claim uses the following equality defined over an $n+1$-qubit system,
    \begin{equation}
        \frac{1}{2} \left(I_{n+1} + \sum_{P \in \{X, Y, Z\}} P_i \otimes P\right) = S_{i, n+1},
    \end{equation}
    where $I_{n+1}$ is an $n+1$-qubit identity, $P_i$ is an $n$-qubit unitary that acts as the Pauli operator $P$ on the $i$-th qubit, and $S_{i, n+1}$ is the swap operator between qubit $i$ in the first $n$ qubits and the last qubit (qubit $n+1$).
    We interpret the error in the Heisenberg-evolved single-qubit Pauli observables as an error in commuting the Pauli observable $P_i$ and the $n$-qubit unitary $U$,
    \begin{equation}
        \norm{ U^\dagger P_i U - P_i}_\infty = \norm{ P_i U - U P_i}_\infty.
    \end{equation}
    From this interpretation, we have the following inequalities,
    \begin{align}
        \norm{S_{i, n+1} (U \otimes I) - (U \otimes I) S_{i, n+1}}_\infty &\leq \frac{1}{2} \sum_{P \in \{X, Y, Z\}} \norm{(P_i \otimes P) (U \otimes I) - (U \otimes I)(P_i \otimes P)}_{\infty}\\
        &\leq \frac{1}{2} \sum_{P \in \{X, Y, Z\}} \norm{ (P_i U - U P_i ) \otimes P}_{\infty}\\
        &= \frac{1}{2} \sum_{P \in \{X, Y, Z\}} \norm{ P_i U - U P_i}_{\infty} \leq \varepsilon.
    \end{align}
    The above inequality can be easily generalized to any of the following,
    \begin{equation} \label{eq:Sij-UotimesIm}
        \norm{S_{i, j} (U \otimes I_m) - (U \otimes I_m) S_{i, j}}_\infty \leq \varepsilon,
    \end{equation}
    where $m \geq 1$, $n+1 \leq j \leq n+m$, and $I_m$ is the identity operator on $m$ qubits.
    Recall the formal definition diamond distance from Definition~\ref{def:diamond-dist},
    \begin{equation}
        \mathcal{D}_\diamond\left(\mathcal{E}_1, \, \mathcal{E}_2\right) = \frac{1}{2} \max_{\rho} \norm{ (\mathcal{E}_1 \otimes \mathcal{I}_n)(\rho) - (\mathcal{E}_2 \otimes \mathcal{I}_n)(\rho) }_1,
    \end{equation}
    where $\rho$ is a density matrix over $2n$ qubits, and $\mathcal{I}_n$ is the identity map over $n$ qubits.
    From Fact~\ref{fact:diamond-unitary}, for any two unitaries $U_1, U_2$, we have $\norm{\mathcal{U}_1 - \mathcal{U}_2 }_\diamond \leq 2 \norm{U_1 - U_2}_\infty$.
    We obtain the following from Eq.~\eqref{eq:Sij-UotimesIm},
    \begin{equation} \label{eq:Sij-UotimesIm-diamond}
        \norm{ \mathcal{S}_{i, j} (\mathcal{U} \otimes I_m) - (\mathcal{U} \otimes I_m) \mathcal{S}_{i, j}\Big)}_\diamond \leq 2 \norm{S_{i, j} (U \otimes I_m) - (U \otimes I_m) S_{i, j}}_\infty \leq 2 \varepsilon.
    \end{equation}
    The strong $\varepsilon$-approximate local identity considers
    \begin{equation}
        \mathcal{D}_\diamond\left(\mathcal{U}, \, \mathcal{I}^{(i)} \otimes \mathcal{E}^{\mathcal{U}}_{\neq i}\right) = \frac{1}{2} \max_{\rho} \norm{ (\mathcal{U} \otimes \mathcal{I}_n)(\rho) - (\mathcal{I}^{(i)} \otimes \mathcal{E}^{\mathcal{U}}_{\neq i} \otimes \mathcal{I}_n)(\rho) }_1.
    \end{equation}
    We add one more qubit to form $2n+1$ qubits. The additional qubit begins in a maximally mixed state $I/2$, stays in $I/2$, and is traced out at the end.
    Let us now consider the following series of analysis,
    \begin{align}
        &\norm{ (\mathcal{U} \otimes \mathcal{I}_n)(\rho) - (\mathcal{I}^{(i)} \otimes \mathcal{E}^{\mathcal{U}}_{\neq i} \otimes \mathcal{I}_n)(\rho) }_1\\
        &= \norm{ \Tr_{2n+1}\left[ (\mathcal{U} \otimes \mathcal{I}_{n+1})(\rho \otimes (I/2)) \right] - (\mathcal{I}^{(i)} \otimes \mathcal{E}^{\mathcal{U}}_{\neq i} \otimes \mathcal{I}_{n})(\rho) }_1 \\
        &= \norm{ \Tr_{i}\left[ \left(\mathcal{S}_{i, 2n+1} \circ (\mathcal{U} \otimes \mathcal{I}_{n+1})\right)(\rho \otimes (I/2)) \right] - (\mathcal{I}^{(i)} \otimes \mathcal{E}^{\mathcal{U}}_{\neq i} \otimes \mathcal{I}_{n+1})(\rho \otimes (I/2)) }_1 \\
        &\leq \norm{ \Tr_{i}\left[ \left((\mathcal{U} \otimes \mathcal{I}_{n+1}) \circ \mathcal{S}_{i, 2n+1} \right)(\rho \otimes (I/2)) \right] - (\mathcal{I}^{(i)} \otimes \mathcal{E}^{\mathcal{U}}_{\neq i} \otimes \mathcal{I}_{n+1})(\rho \otimes (I/2)) }_1 + 2\varepsilon \\
        &= \norm{ (\mathcal{I}^{(i)} \otimes \mathcal{E}^{\mathcal{U}}_{\neq i} \otimes \mathcal{I}_{n+1})(\rho \otimes (I/2)) - (\mathcal{I}^{(i)} \otimes \mathcal{E}^{\mathcal{U}}_{\neq i} \otimes \mathcal{I}_{n+1})(\rho \otimes (I/2)) }_1 + 2\varepsilon = 2\varepsilon.
    \end{align}
    The only inequality above uses Eq.~\eqref{eq:Sij-UotimesIm-diamond}. We have proved the claim.
\end{proof}

The following two lemmas give the relationships between global and local identity checks.
The basic idea is to check whether a map is close to identity by checking whether the map forms approximate local identities on all the $n$ qubits.
If the map is far from identity, then the map is not an approximate local identity for some qubits.
If the map is an approximate local identity for all qubits, then the map is close to the identity.

\begin{lemma}[Global non-identity check from local non-identity checks] \label{lem:nonIDcheck}
    Given an integer $n > 0$ and an $n$-qubit unitary~$U$. If there exists $\varepsilon > 0$ and $i \in \{1, \ldots, n\}$, such that $\mathcal{U}$ is not a strong $\varepsilon$-approximate local identity on the $i$-th qubit, then $\norm{\mathcal{U} - \mathcal{I}}_\diamond \geq \varepsilon / 2.$
\end{lemma}

\begin{lemma}[Global identity check from local identity checks] \label{thm:IDcheck}
    Given an integer $n > 0$ and an $n$-qubit unitary~$U$. If there exists $\varepsilon_1, \ldots, \varepsilon_n > 0$, such that $\mathcal{U}$ is a strong $\varepsilon_i$-approximate local identity on the $i$-th qubit for all $i \in \{1, \ldots, n\}$, then $\norm{\mathcal{U} - \mathcal{I}}_\diamond \leq 3 \sum_{i=1}^n \varepsilon_i.$
\end{lemma}

We give proofs of these two lemmas at the end of this subsection.
Lemma~\ref{lem:nonIDcheck} is proven by contradiction. To prove Lemma~\ref{thm:IDcheck}, we consider a stabilizer decomposition for a single qubit.

\begin{proposition}[Single-qubit stabilizer decomposition] \label{prop:stab-decomp}
    Given an integer $n > 0$ and an $n$-qubit density matrix $\rho$.
    For any $S\subseteq\{1, \ldots, n\}$, $\rho$ can be written as a linear combination of $R = 10^{|S|}$ $n$-qubit density matrices $\rho_1, \ldots, \rho_R$, $\rho = \sum_{r = 1}^{R} \alpha_r \rho_r$, where $\alpha_r\in\mathbb{R}$ and $\rho_r$ is a density matrix that satisfies
    \begin{equation}
        \rho_r=\bigotimes_{j \in S} \ketbra{s_j}{s_j}\otimes \Tr_S (\rho_r),
    \end{equation}
    for some $\ket{s_j} \in \mathrm{stab}_1$. We also have $\sum_{r=1}^R \alpha_r = 1$ and $\sum_{r=1}^R |\alpha_r| = 3^{|S|}$.
\end{proposition}
\begin{proof}
    Given an integer $i \in \{1, \ldots, n\}$,
    consider the following linear map $\mathcal{M}_i$ which equals to the identity channel on $i$-th qubit,
    \begin{align}
        \mathcal{M}_i(\rho) &:= \ketbra{0}{0}_i \otimes \bra{0} \rho \ket{0}_i + \ketbra{1}{1}_i \otimes \bra{1} \rho \ket{1}_i \nonumber \\
        &+ \frac{1}{2} \ketbra{+}{+}_i \otimes \bra{+} \rho \ket{+}_i - \frac{1}{2}\ketbra{+}{+}_i \otimes \bra{-} \rho \ket{-}_i \nonumber \\
        &- \frac{1}{2} \ketbra{-}{-}_i \otimes \bra{+} \rho \ket{+}_i + \frac{1}{2}\ketbra{-}{-}_i \otimes \bra{-} \rho \ket{-}_i \label{eq:singleQBstabdecomp} \\
        &+ \frac{1}{2} \ketbra{y+}{y+}_i \otimes \bra{y+} \rho \ket{y+}_i - \frac{1}{2}\ketbra{y+}{y+}_i \otimes \bra{y-} \rho \ket{y-}_i \nonumber \\
        &- \frac{1}{2} \ketbra{y-}{y-}_i \otimes \bra{y+} \rho \ket{y+}_i + \frac{1}{2}\ketbra{y-}{y-}_i \otimes \bra{y-} \rho \ket{y-}_i, \nonumber \\
        &= \sum_{r=1}^{10} b_r \ketbra{s_r}{s_r}_i \otimes \bra{s'_r} \rho \ket{s'_r}_i.
    \end{align}
    where $\ketbra{s}{s}_i$ is a single-qubit stabilizer state on the $i$-th qubit, $\bra{s} \rho \ket{s}_i$ is a partial inner product on the $i$-th qubit, $s_r$, $s'_r$, $b_r$ takes on the corresponding values in $\mathrm{stab}_1$, $\mathrm{stab}_1$, $\{1, 1/2, -1/2\}$, respectively. The fact that $\mathcal{M}_i$ equals to the identity CPTP map $\mathcal{I}$ is because of the following identity
    \begin{equation}
        \rho=\sum_{P\in\{I,X,Y,Z\}}\Tr_i(P_i \rho)\otimes \frac{P_i}{2},
    \end{equation}
    where $P_i$ acts on the $i$-th qubit, and Eq.~\eqref{eq:singleQBstabdecomp} follows by further decomposing the Pauli operators into their eigenstates.

    Without loss of generality, we consider $k = |S|$ and $S = \{1, \ldots, k\}$.
    The identity $\rho = (\circ_{i \in S} \mathcal{M}_i)(\rho)$ gives rise to the equality
    \begin{equation}
        \rho = \sum_{r_1 =1}^{10} \cdots \sum_{r_{k} = 1}^{10} \left(\prod_{i=1}^k b_{r_i}\right) \ketbra{s_{r_1}, \ldots, s_{r_k}}{s_{r_1}, \ldots, s_{r_k}} \otimes \bra{s'_{r_1}, \ldots, s'_{r_k}} \rho \ket{s'_{r_1}, \ldots, s'_{r_k}}.
    \end{equation}
    We define $r = \sum_{i=1}^k 10^{i-1} r_i$, $R = 10^k$, $Z_r = \Tr(\bra{s'_{r_1}, \ldots, s'_{r_k}} \rho \ket{s'_{r_1}, \ldots, s'_{r_k}})\geq 0$, and
    \begin{equation}
        \rho_r = \begin{cases}
        \ketbra{s_{r_1}, \ldots, s_{r_k}}{s_{r_1}, \ldots, s_{r_k}} \otimes \frac{\bra{s'_{r_1}, \ldots, s'_{r_k}} \rho \ket{s'_{r_1}, \ldots, s'_{r_k}}}{Z_r} & \mbox{if} \,\, Z_r > 0, \\
        \ketbra{s_{r_1}, \ldots, s_{r_k}}{s_{r_1}, \ldots, s_{r_k}} \otimes \frac{I}{2^{n-k}} & \mbox{if} \,\, Z_r = 0,
        \end{cases}
    \end{equation}
    and $\alpha_r = Z_r \prod_{i=1}^k b_{r_i}$. It is not hard to check that $\sum_r |\alpha_r| = 3^k$.
    Together, we have the single-qubit stabilizer decomposition $\rho = \sum_{r=1}^R \alpha_r \rho_r$.
\end{proof}

\begin{proof}[Proof of Lemma~\ref{lem:nonIDcheck}]
    We consider proof by contradiction.
    Assume $\norm{\mathcal{U} - \mathcal{I}}_\diamond < \varepsilon / 2$.
    For any integer $m \geq 0$, for any state $\ket{s}_i \in \mathrm{stab}_1$ on the $i$-th qubit, and for any $(n-1+m)$-qubit density matrix $\rho$,
    \begin{align}
        &\norm{ (\mathcal{U} \otimes \mathcal{I}^{(>n)}) \left(\ketbra{s}{s}_i \otimes \rho \right) - \ketbra{s}{s}_i \otimes (\mathcal{E}^{U}_{\neq i}\otimes \mathcal{I}^{(>n)})(\rho) }_1\\
        &\leq \norm{ (\mathcal{U} \otimes \mathcal{I}^{(>n)}) \left(\ketbra{s}{s}_i \otimes \rho \right) - \ketbra{s}{s}_i \otimes \rho }_1 + \norm{ \ketbra{s}{s}_i \otimes \rho - \ketbra{s}{s}_i \otimes (\mathcal{E}^{U}_{\neq i}\otimes \mathcal{I}^{(>n)})(\rho) }_1 \\
        &\leq \norm{\mathcal{U} - \mathcal{I}}_\diamond + \norm{\mathcal{U} - \mathcal{I}}_\diamond < \varepsilon.
    \end{align}
    The first inequality follows from putting in $\ketbra{s}{s}_i \otimes \rho$ and using triangle inequality.
    The second inequality follows from the definition of diamond distance, the identity
    \begin{align}
        &\norm{ \ketbra{s}{s}_i \otimes \rho - \ketbra{s}{s}_i \otimes (\mathcal{E}^{U}_{\neq i}\otimes \mathcal{I}^{(>n)})(\rho) }_1 \\
        &= \norm{\ketbra{s}{s}_i \otimes \Tr_{i}\left(\frac{I^{(i)}}{2} \otimes \rho\right) - \ketbra{s}{s}_i \otimes \Tr_{i}\left( \left(\mathcal{U} \otimes \mathcal{I}^{(>n)}\right) \left(\frac{I^{(i)}}{2} \otimes \rho \right) \right) }_1,
    \end{align}
    and the two facts: $\norm{\rho_A \otimes \rho_B - \rho_A \otimes \rho_C}_1 = \norm{\rho_B - \rho_C}_1, \norm{\Tr_i(\rho_A)}_1 \leq \norm{\Tr(\rho_A)}_1$ for any density matrix $\rho_A, \rho_B, \rho_C$.
    The above derivation shows that $U$ is an $\varepsilon$-approximate local identity on the $i$-th qubit, which is a contradiction.
    Therefore, $\norm{\mathcal{U} - \mathcal{I}}_\diamond \geq \varepsilon / 2$.
\end{proof}

\begin{proof}[Proof of Lemma~\ref{thm:IDcheck}]
    From Theorem~3.55 in \cite{watrous2018theory}, we have
    \begin{align}
        \norm{\mathcal{U} - \mathcal{I}}_\diamond &= \norm{U \ketbra{\psi}{\psi} U^\dagger - \ketbra{\psi}{\psi}}_1
    \end{align}
    for some $n$-qubit state $\ket{\psi}$.
    Let $\mathcal{I}^{(\leq k)}$ be the identity CPTP map acting on the first $k$ qubit.
    We use a telescoping sum of the form,
    \begin{equation}
        U \ketbra{\psi}{\psi} U^\dagger - \ketbra{\psi}{\psi} = \sum_{k=0}^{n-1} \left[ \left(\mathcal{I}^{(\leq k)} \otimes \mathcal{E}^{U}_{> k}\right)(\ketbra{\psi}{\psi}) - \left(\mathcal{I}^{(\leq k+1)} \otimes \mathcal{E}^{U}_{> k+1}\right)(\ketbra{\psi}{\psi}) \right].
    \end{equation}
    By triangle inequality, we obtain
    \begin{equation} \label{eq:diamond-Ek}
        \norm{\mathcal{U} - \mathcal{I}}_\diamond \leq \sum_{k=0}^{n-1} \norm{ \left(\mathcal{I}^{(\leq k)} \otimes \mathcal{E}^{U}_{> k}\right)(\ketbra{\psi}{\psi}) - \left(\mathcal{I}^{(\leq k+1)} \otimes \mathcal{E}^{U}_{> k+1}\right)(\ketbra{\psi}{\psi}) }_1.
    \end{equation}
    In the next step, we will bound each term in the above telescoping sum.

To bound the term corresponding to $k \in \{0, \ldots, n-1\}$ in Eq.~\eqref{eq:diamond-Ek}, we consider an $(k+(n-k)+k)$-qubit density matrix $\rho^{(k)}$.
    The first $k$ qubits of $\rho^{(k)}$ is the maximally mixed state $\frac{I^{(1, \ldots, k)}}{2^k}$.
    The next $(n-k)$ qubits of $\rho^{(k)}$ corresponds to all except the first $k$ qubits in $\ketbra{\psi}{\psi}$.
    The last $k$ qubits of $\rho^{(k)}$ corresponds to the first $k$ qubits in $\ketbra{\psi}{\psi}$.
    Under this definition of $\rho^{(k)}$, we have
    \begin{align}
        &\norm{\left(\mathcal{I}^{(\leq k)} \otimes \mathcal{E}^{U}_{> k}\right)(\ketbra{\psi}{\psi}) - \left(\mathcal{I}^{(\leq k+1)} \otimes \mathcal{E}^{U}_{> k+1}\right)(\ketbra{\psi}{\psi}) }_1 \\
        &= \norm{ \left(\mathcal{U} \otimes \mathcal{I}^{(>n)}\right)(\rho^{(k)}) - \left( \mathcal{I}^{(k+1)} \otimes \mathcal{E}^U_{\neq k+1} \otimes \mathcal{I}^{(>n)}\right) (\rho^{(k)})}_1, \label{eq:rearrangequbits}
    \end{align}
    where $\left( \mathcal{I}^{(k+1)} \otimes \mathcal{E}^U_{\neq k+1} \otimes \mathcal{I}^{(>n)}\right) (\rho^{(k)})$ is the output state after applying the $(n-1)$-qubit CPTP map $\mathcal{E}^U_{\neq k+1}$ to the first $n$ qubits except the $(k+1)$-th qubit of $\rho^{(k)}$.
    We now use the single-qubit stabilizer decomposition with $S = \{k+1\}$ given in Prop.~\ref{prop:stab-decomp} to obtain $\rho^{(k)} = \sum_{r=1}^{10} \alpha_r \rho^{(k)}_r$ with $\sum_r |\alpha_r| = 3$ and the reduced density matrix of $\rho^{(k)}_r$ on the $(k+1)$-th qubit is a single-qubit stabilizer state.
    We can now bound each term by
    \begin{align}
        &\norm{ \left(\mathcal{U} \otimes \mathcal{I}^{(>n)}\right)(\rho^{(k)}) - \left( \mathcal{I}^{(k+1)} \otimes \mathcal{E}^U_{\neq k+1} \otimes \mathcal{I}^{(>n)}\right) (\rho^{(k)})}_1 \\
        &\leq \sum_{r=1}^{10} |\alpha_r| \norm{ \left(\mathcal{U} \otimes \mathcal{I}^{(>n)}\right)(\rho^{(k)}_r) - \left( \mathcal{I}^{(k+1)} \otimes \mathcal{E}^U_{\neq k+1} \otimes \mathcal{I}^{(>n)}\right) (\rho^{(k)}_r)}_1 \\
        &\leq \sum_{r=1}^{10} |\alpha_r| \varepsilon_{k+1} = 3 \varepsilon_{k+1}. \label{eq:decompose-error}
    \end{align}
    The first line is the triangle inequality. The second line uses the assumption that $U$ is an $\varepsilon_{k+1}$-approximate local identity on the $(k+1)$-th qubit.
    Combining Eq.~\eqref{eq:diamond-Ek}, Eq.~\eqref{eq:rearrangequbits}, Eq.~\eqref{eq:decompose-error},
    \begin{equation}
        \norm{\mathcal{U} - \mathcal{I}}_\diamond \leq 3 \sum_{k=0}^{n-1} \varepsilon_{k+1},
    \end{equation}
    which establishes the stated result.
\end{proof}

\subsection{Weak $\varepsilon$-approximate local identity}
\label{sec:weak-localid}

We next look at another definition of approximate local identity: the reduced channel of $U$ on the $i$-th qubit is close to the identity map.
This definition is very easy to check but only guarantees that the unitary $U$ is close to the identity in the average-case distance (instead of the worst-case distance, i.e., the diamond distance).
Hence, we will refer to this as the weak $\varepsilon$-approximate local identity.
Recall Definition~\ref{def:reduced-channel} of reduced channel,
\begin{equation}
        \mathcal{E}^{\mathcal{U}}_{i}(\rho_{i}) =  \Tr_{\neq i}\left( U \left(\frac{I^{(\neq i)}}{2^{n-1}} \otimes \rho_{i} \right) U^\dagger \right),
\end{equation}
where $\rho_{i}$ is a density matrix on the $i$-th qubit, $I^{(\neq i)}$ is the identity on all except the $i$-th qubit, and $\Tr_{\neq i}$ is the partial trace over all except the $i$-th qubit.

\begin{definition}[Weak $\varepsilon$-approximate local identity; unitary version]
    Given $n > 0, \varepsilon \geq 0,$ and $i \in \{1, \ldots, n\}$. An $n$-qubit unitary $U$ is a weak $\varepsilon$-approximate local identity on the $i$-th qubit if
    \begin{equation}
        \mathcal{D}_{\mathrm{ave}}\left(\mathcal{E}^{\mathcal{U}}_{i}, \, \mathcal{I}\right) \leq \varepsilon,
    \end{equation}
    where $\mathcal{I}$ is a $1$-qubit CPTP map that acts as an identity.
\end{definition}

In the literature of quantum junta learning \cite{chen2023testing}, one defines the influence of a qubit $i$ in an $n$-qubit unitary $U = \sum_{P \in \{I, X, Y, Z\}^{\otimes n}} \alpha_P P$, where $\alpha_P \in \mathbb{C}$ to be
\begin{equation}
    \sum_{\substack{P \in \{I, X, Y, Z\}^{\otimes n} \\ P_i \neq I}} \left| \alpha_P\right|^2.
\end{equation}
The following lemma shows that weak approximate local identity is equivalent to low influence.

\begin{lemma}[Characterization of weak $\varepsilon$-approximate local identity] \label{lem:char-weak-local-id}
    Given $n > 0$, $\varepsilon \geq 0$, and an $n$-qubit unitary $U$.
    Consider the Pauli representation of $U = \sum_{P \in \{I, X, Y, Z\}^{\otimes n}} \alpha_P P$, where $\alpha_P \in \mathbb{C}$.
    $\mathcal{U}$ is a weak $\varepsilon$-approximate local identity on the $i$-th qubit if and only if
    \begin{equation}
        \sum_{\substack{P \in \{I, X, Y, Z\}^{\otimes n} \\ P_i \neq I}} \left| \alpha_P\right|^2 \leq \frac{3}{2} \varepsilon.
    \end{equation}
    From the definition of influence in quantum junta learning \cite{chen2023testing}, we have qubit $i$ has influence bounded above by $1.5 \varepsilon$ in the unitary $U$.
\end{lemma}
\begin{proof}
    From the definition of the reduced channel, we have
    \begin{equation}
        \mathcal{E}^{\mathcal{U}}_{i}(\rho_{i}) = \sum_{s_1, s_2 \in \{I, X, Y, Z\}} \left(\sum_{\substack{P, Q \in \{I, X, Y, Z\}^{\otimes n} \\ P_i = s_1, Q_i = s_2, P_{\neq i} = Q_{\neq i}}} \alpha_{P}^* \alpha_{Q} \right) s_1 \rho_i s_2,
    \end{equation}
    where $P_{\neq i}, Q_{\neq i}$ is an $(n-1)$-qubit Pauli observable equal to $P$, $Q$ with qubit $i$ removed.
    From Fact~\ref{fact:ave-dist} characterizing the average-case distance $\mathcal{D}_{\mathrm{ave}}$, we have
    \begin{equation}
        \mathcal{D}_{\mathrm{ave}}\left(\mathcal{E}^{\mathcal{U}}_{i}, \mathcal{I}\right) = \frac{2}{3}\left(1 - \sum_{\substack{P, Q \in \{I, X, Y, Z\}^{\otimes n} \\ P_i = I, Q_i = I, P_{\neq i} = Q_{\neq i}}} \alpha_{P}^* \alpha_{Q} \right) = \frac{2}{3}\left(1 - \sum_{\substack{P \in \{I, X, Y, Z\}^{\otimes n} \\ P_i = I}} \left| \alpha_P\right|^2\right).
    \end{equation}
    Furthermore, we note that $\Tr(U^\dagger U) = 2^n = 2^n \sum_{P \in \{I, X, Y, Z\}^{\otimes n}} |\alpha_P|^2$. Hence, we have
    \begin{equation}
        1 - \sum_{\substack{P \in \{I, X, Y, Z\}^{\otimes n} \\ P_i = I}} \left| \alpha_P\right|^2 = \sum_{\substack{P \in \{I, X, Y, Z\}^{\otimes n} \\ P_i \neq I}} \left| \alpha_P\right|^2.
    \end{equation}
    The lemma follows from the two identities given above.
\end{proof}

Weak $\varepsilon$-approximate local identity naturally generalizes to any quantum process (channel) by using the definition of reduced channels for channels. The formal definition is given below.

\begin{definition}[Weak $\varepsilon$-approximate local identity; channel version]
    Given $n > 0, \varepsilon \geq 0,$ and $i \in \{1, \ldots, n\}$. An $n$-qubit CPTP map $\mathcal{C}$ is a weak $\varepsilon$-approximate local identity on the $i$-th qubit if
    \begin{equation}
        \mathcal{D}_{\mathrm{ave}}\left(\mathcal{E}^{\mathcal{C}}_{i}, \, \mathcal{I}\right) \leq \varepsilon,
    \end{equation}
    where $\mathcal{I}$ is a $1$-qubit CPTP map that acts as an identity.
\end{definition}

The following two lemmas give the relationships between global and local identity checks.
The basic idea is to check whether a map is close to identity by checking whether the map forms approximate local identities on all the $n$ qubits.

\begin{lemma}[Global non-identity check from local non-identity checks] \label{lem:nonIDcheck-weak}
    Given an integer $n > 0$ and an $n$-qubit CPTP map~$\mathcal{C}$. If there exists $\varepsilon > 0$ and $i \in \{1, \ldots, n\}$, such that $\mathcal{C}$ is not a weak $\varepsilon$-approximate local identity on the $i$-th qubit, then $\mathcal{D}_{\mathrm{ave}}(\mathcal{C}, \mathcal{I}) \geq \varepsilon.$
\end{lemma}

\begin{lemma}[Global identity check from local identity checks] \label{thm:IDcheck-weak}
    Given an integer $n > 0$ and an $n$-qubit CPTP map~$\mathcal{C}$. If there exists $\varepsilon_1, \ldots, \varepsilon_n > 0$, such that $\mathcal{C}$ is a weak $\varepsilon_i$-approximate local identity on the $i$-th qubit for all $i \in \{1, \ldots, n\}$, then $\mathcal{D}_{\mathrm{ave}}(\mathcal{C}, \mathcal{I}) \leq \frac{3}{2} \sum_{i=1}^n \varepsilon_i.$
\end{lemma}

\begin{proof}[Proof of Lemma~\ref{lem:nonIDcheck-weak}~and~\ref{thm:IDcheck-weak}]
    Let us define $\ket{\Omega_1} = \frac{1}{\sqrt{2}}(\ket{00}+\ket{11})$, and $\ket{\Omega_n} = \ket{\Omega_1}^{\otimes n}$.
    From Fact~\ref{fact:ave-dist} characterizing the average-case distance $\mathcal{D}_{\mathrm{ave}}$, we have
    \begin{equation}
        \mathcal{D}_{\mathrm{ave}}(\mathcal{C}, \mathcal{I}) = \frac{2^n}{2^n + 1} \left( 1 - \bra{\Omega_n} \left(\mathcal{C}\otimes \mathcal{I}\right)\left(\ketbra{\Omega_n}{\Omega_n}\right) \ket{\Omega_n} \right).
    \end{equation}
    We can think of the term $\bra{\Omega_n} \left(\mathcal{C}\otimes \mathcal{I}\right)\left(\ketbra{\Omega_n}{\Omega_n}\right) \ket{\Omega_n}$ as the probability of getting $\ket{\Omega_1}$ on all $n$ parallel two-qubit Bell-basis measurements on the $2n$-qubit state $\left(\mathcal{C}\otimes \mathcal{I}\right)\left(\ketbra{\Omega_n}{\Omega_n}\right)$.
    From standard probability theory, we have the following inequality,
    \begin{equation}
        1 - \bra{\Omega_n} \left(\mathcal{E}\otimes \mathcal{I}\right)\left(\ketbra{\Omega_n}{\Omega_n}\right) \ket{\Omega_n} \geq 1 - \Tr\left( \left(\ketbra{\Omega_1}{\Omega_1} \otimes I_{\neq i}^{\otimes 2} \right) \left(\mathcal{C}\otimes \mathcal{I}\right)\left(\ketbra{\Omega_n}{\Omega_n}\right) \right),
    \end{equation}
    where $\ketbra{\Omega_1}{\Omega_1} \otimes I_{\neq i}^{\otimes 2}$ is a projection onto $\ketbra{\Omega_1}{\Omega_1}$ on the $i$-th and $(n+i)$-th qubit for any $i$.
    Also, from union bound, we have
    \begin{equation}
        1 - \bra{\Omega_n} \left(\mathcal{E}\otimes \mathcal{I}\right)\left(\ketbra{\Omega_n}{\Omega_n}\right) \ket{\Omega_n} \leq 1 - \sum_{i=1}^n \left(1 - \Tr\left( \left(\ketbra{\Omega_1}{\Omega_1} \otimes I_{\neq i}^{\otimes 2} \right) \left(\mathcal{C}\otimes \mathcal{I}\right)\left(\ketbra{\Omega_n}{\Omega_n}\right) \right) \right).
    \end{equation}
    By reorganizing using the reduced channel of $\mathcal{C}$ on the $i$-th qubit, we have
    \begin{equation}
        \Tr\left( \left(\ketbra{\Omega_1}{\Omega_1} \otimes I_{\neq i}^{\otimes 2} \right) \left(\mathcal{C}\otimes \mathcal{I}\right)\left(\ketbra{\Omega_n}{\Omega_n}\right) \right) = \bra{\Omega_1} \left(\mathcal{E}^{\mathcal{C}}_i \otimes \mathcal{I}\right)\left(\ketbra{\Omega_1}{\Omega_1}\right) \ket{\Omega_1}.
    \end{equation}
    Therefore, we have
    \begin{equation}
        \frac{3}{2} \times \frac{2^n}{2^n + 1} \sum_{i=1}^n \mathcal{D}_{\mathrm{ave}}(\mathcal{E}^{\mathcal{C}}_i, \mathcal{I}) \geq \mathcal{D}_{\mathrm{ave}}(\mathcal{C}, \mathcal{I}) \geq \frac{3}{2} \times \frac{2^{n}}{2^n + 1} \mathcal{D}_{\mathrm{ave}}(\mathcal{E}^{\mathcal{C}}_i, \mathcal{I}).
    \end{equation}
    By noting that $\frac{3}{2} \geq \frac{3}{2} \times \frac{2^n}{2^n + 1}$ and $\frac{3}{2} \times \frac{2^n}{2^n + 1} \geq 1$, we obtain Lemma~\ref{lem:nonIDcheck-weak}~and~\ref{thm:IDcheck-weak}.
\end{proof}

\section{Learning shallow quantum circuits from a classical dataset}
\label{sec:learning-unitary-diamond}

In this section, we present algorithms for learning shallow quantum circuits that achieve a small diamond distance. All algorithms in this section use a classical dataset obtained from performing randomized measurements on the unknown shallow quantum circuit (defined below) to classically reconstruct the unknown circuit.
The learning algorithms only require classical computation.

\begin{definition}[Randomized measurement dataset for an unknown unitary] \label{def:random-measure-data}
    The learning algorithm accesses an unknown $n$-qubit unitary $U$ via a randomized measurement dataset of the following form,
    \begin{equation} \label{eq:random-measure-data}
    \mathcal{T}_U(N) = \left\{ \ket{\psi_\ell} = \bigotimes_{i=1}^n \ket{\psi_{\ell, i}}, \ket{\phi_\ell} = \bigotimes_{i=1}^n \ket{\phi_{\ell, i}} \right\}_{\ell=1}^N.
    \end{equation}
    A randomized measurement dataset of size $N$ is constructed by obtaining $N$ samples from the unknown unitary $U$.
    One sample is obtained from one experiment given as follows.
    \begin{enumerate}
        \item Sample an input state $\ket{\psi_\ell} = \bigotimes_{i=1}^n \ket{\psi_{\ell, i}}$, which is a product state consisting of uniformly random single-qubit stabilizer states in $\mathrm{stab}_1$.
        \item Apply the unknown unitary $U$ to $\ket{\psi_\ell}$.
        \item Measure every qubit of $U \ket{\psi_\ell}$ under a random Pauli basis. The measurement collapses the state $U \ket{\psi_\ell}$ to a state $\ket{\phi_\ell} = \bigotimes_{i=1}^n \ket{\phi_{\ell, i}}$, where $\ket{\phi_{\ell, i}}$ is a single-qubit stabilizer state $\mathrm{stab}_1$.
    \end{enumerate}
    Together, $N$ queries to $U$ construct a dataset $\mathcal{T}_U(N)$ with $N$ samples.
    The dataset can be represented efficiently on a classical computer with $\mathcal{O}(Nn)$ bits.
\end{definition}

An interesting question is whether quantum learning algorithms that have access to the unknown quantum circuit $U$ could be much more efficient.
In Section~\ref{sec:learning-unitary-coherent}, we present a quantum learning algorithm that achieves the optimal scaling in query complexity and computational time for learning geometrically-local shallow quantum circuits over finite gate sets.

\subsection{Results}

We present the results for learning general and geometrically-local shallow quantum circuits consisting of two-qubit gates over $\mathrm{SU}(4)$ and over a finite gate set using a classical dataset.

\subsubsection{Learning general shallow quantum circuits}

We consider the problem of learning an $n$-qubit unitary $U$ created by a general shallow quantum circuit $C$ with arbitrary circuit connectivity, i.e., every qubit can be connected to any other qubit by a quantum gate, and an arbitrary number $m$ of ancilla qubits initialized in $\ket{0^m}$ and ended up in $\ket{0^m}$ after $C$. Formally, we have the following identity for $U$,
\begin{equation}
    U \otimes \ket{0^m} = C (I_n \otimes \ket{0^m}),
\end{equation}
where $I_n$ is an identity on $n$ qubits.

We have the following theorems for learning the unknown unitary $U$.
We can see that the sample/query complexity is very similar to learning geometric-local circuits. However, the computational complexity becomes higher, and we can only guarantee a polynomial scaling with system size $n$. The learning algorithm and proof are given in Section~\ref{sec:shallow-SU4-gates-proof}.

\begin{theorem}[Learning general shallow quantum circuits] \label{thm:shallow-SU4-gates}
    Given a failure probability~$\delta$, an approximation error $\varepsilon$, and an unknown $n$-qubit unitary $U$ generated by a constant-depth circuit over any two-qubit gates in $\mathrm{SU}(4)$ with an arbitrary number of ancilla qubits.
    With a randomized measurement dataset $\mathcal{T}_U(N)$ of size
    \begin{equation}
        N = \mathcal{O}\left( \frac{n^2 \log(n / \delta)}{\varepsilon^2}  \right),
    \end{equation}
    we can learn an $n$-qubit quantum channel $\hat{\mathcal{E}}$ that can be implemented by a constant-depth quantum circuit over $2n$ qubits, such that
    \begin{equation}
        \norm{\hat{\mathcal{E}} - \mathcal{U}}_\diamond \leq \varepsilon,
    \end{equation}
    with probability at least $1 - \delta$. The classical computational time to learn $\hat{\mathcal{E}}$ is $\mathrm{poly}(n) \log(1 / \delta) / \varepsilon^2$.

    In addition, if each two-qubit gate in the unknown circuit is chosen from a finite gate set of a constant size, then the algorithm learns an exact description $\hat{\mathcal{E}} = \mathcal{U}$ with probability $1-\delta$, using $N = \mathcal{O}( \log(n / \delta) )$ samples and $\mathcal{O}(\mathrm{poly}(n) \log (1 /\delta))$ time.
\end{theorem}

\begin{remark}[Implementation of learned $n$-qubit channel]
    The $n$-qubit channel $\hat{\mathcal{E}}$ is the reduced channel $\mathcal{E}^{\hat{V}}_{\leq n}$ of the constant-depth $2n$-qubit circuit $\hat{V}$ on the first $n$ qubits.
\end{remark}

\subsubsection{Learning geometrically-local shallow quantum circuits}

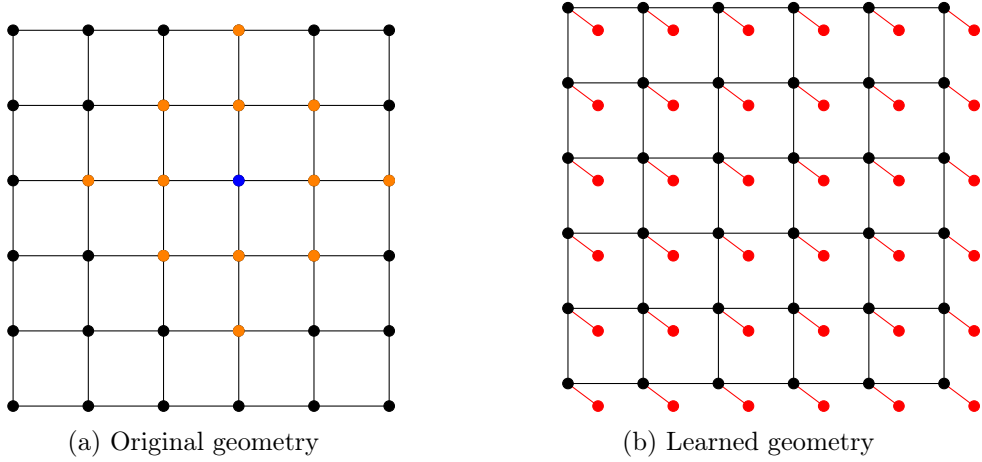
\begin{figure}[t]
    \centering
    \begin{subfigure}[b]{0.3\textwidth}
    \begin{tikzpicture}
  \draw[step=1cm,black] (0,0) grid (5,5);
  \foreach \x in {0,...,5}
  \foreach \y in {0,...,5}
  {
  \filldraw[black] (\x,\y) circle (2pt);}
  \filldraw[blue] (3,3) circle (2pt);
  \filldraw[orange] (3,2) circle (2pt);
  \filldraw[orange] (3,4) circle (2pt);
  \filldraw[orange] (2,3) circle (2pt);
  \filldraw[orange] (4,3) circle (2pt);
  \filldraw[orange] (1,3) circle (2pt);
  \filldraw[orange] (3,1) circle (2pt);
  \filldraw[orange] (5,3) circle (2pt);
  \filldraw[orange] (3,5) circle (2pt);
  \filldraw[orange] (2,2) circle (2pt);
  \filldraw[orange] (2,4) circle (2pt);
  \filldraw[orange] (4,2) circle (2pt);
  \filldraw[orange] (4,4) circle (2pt);
\end{tikzpicture}
    \caption{Original geometry}
    \end{subfigure}
\quad\quad\quad\quad\quad\quad
\begin{subfigure}[b]{0.3\textwidth}
    \begin{tikzpicture}
  \draw[step=1cm,black] (0,0) grid (5,5);
  \foreach \x in {0,...,5}
  \foreach \y in {0,...,5}
  {\filldraw[red] (\x+0.4,\y-0.3) circle (2pt);
  \draw[red] (\x,\y) -- (\x+0.4,\y-0.3);
  \filldraw[black] (\x,\y) circle (2pt);}
\end{tikzpicture}
    \caption{Learned geometry}
    \end{subfigure}

    \caption{Learning geometrically-local shallow quantum circuits. (a) In this example, the geometry is a 2D lattice where each vertex has a degree at most 4. The lightcone of the blue qubit (for depth $d=2$) is the union of the blue and orange qubits. (b) The learned circuit acts on an extended geometry with $2n$ qubits, where each system qubit (black) is attached to an ancilla qubit (red). Note that each ancilla qubit is connected only with its corresponding system qubit (red edges).}
    \label{fig:geometry}
\end{figure}

We consider the problem of learning geometrically-local shallow quantum circuits.
Here, we consider a generalized definition of geometric locality, which includes quantum circuits over 1D, 2D, and 3D geometry.
The generalization enables more exotic geometry over the qubits and is formally represented by a fixed constant-degree graph.
See Fig.~\ref{fig:geometry}(a) for an illustration of the definitions.

\begin{definition}[Geometric locality] \label{def:geo-local-circuit}
    A geometry over $n$ qubits is defined by a graph $G=(V, E)$ with $n=|V|$ vertices, and each vertex has a degree of at most $\kappa=O(1)$.
    A geometrically-local two-qubit gate can only act on an edge of $G$. A geometrically-local quantum circuit is a circuit with only geometrically-local two-qubit quantum gates.
    A depth-$d$ geometrically-local quantum circuit has $d$ layers, where each layer consists of non-overlapping geometrically-local two-qubit gates.
\end{definition}

\begin{definition}[Lightcone in a geometry] \label{def:light-cone-geo}
    Given a geometry over $n$ qubits represented by a graph $G = (V, E)$ with degree $\kappa$ and an integer $d$.
    The lightcone $L_d(i)$ of a qubit $i$ with depth $d$ is the set of qubits with distance at most $d$ from qubit $i$ in the graph $G$. We have $|L_d(i)|\leq (\kappa+1)^d$.
\end{definition}

\begin{definition}[Geometrically-local set] \label{def:geo-local-set}
    Given a geometry over $n$ qubits represented by a graph $G = (V, E)$.
    A set $S$ of qubits is geometrically local if all qubits in $S$ are of $\mathcal{O}(1)$ distance in $G$.
\end{definition}

Under this more general definition of geometry, our proposed algorithm can still learn very efficiently.
The following theorem quantifies the efficiency in terms of both the query complexity and the computational complexity. The learning algorithm and proof are given in Section~\ref{sec:geo-finite-gates-proof}.

\begin{theorem}[Learning geometrically-local shallow quantum circuits] \label{thm:geo-SU4-gates}
    Given an unknown geometrically local constant-depth $n$-qubit circuit $U$ over any two-qubit gates in $\mathrm{SU}(4)$.
    With a randomized measurement dataset $\mathcal{T}_U(N)$ of size
    \begin{equation}
        N = \mathcal{O}\left( \frac{n^2 \log(n / \delta)}{\varepsilon^2}  \right),
    \end{equation}
    we can learn an $n$-qubit quantum channel $\hat{\mathcal{E}}$ that can be implemented by a geometrically local constant-depth quantum circuit over $2n$ qubits, such that
    \begin{equation}
        \norm{\hat{\mathcal{E}} - \mathcal{U}}_\diamond \leq \varepsilon,
    \end{equation}
    with probability at least $1 - \delta$. The computational time to learn $\hat{\mathcal{E}}$ is $\mathcal{O}(n^3 \log(n / \delta) / \varepsilon^2)$.

    In addition, if each two-qubit gate in the unknown circuit is chosen from a finite gate set of a constant size, then the algorithm learns an exact description $\hat{\mathcal{E}} = \mathcal{U}$ with probability $1-\delta$, using $N = \mathcal{O}( \log(n / \delta) )$ samples and $\mathcal{O}(n \log (n /\delta))$ time.
\end{theorem}

\begin{remark}[Implementation of learned $n$-qubit channel]
    The $n$-qubit channel $\hat{\mathcal{E}}$ is equal to the reduced channel $\mathcal{E}^{\hat{V}}_{\leq n}$ of the geometrically-local constant-depth $2n$-qubit circuit $\hat{V}$ on the first $n$ qubits.
\end{remark}

Next, we look at a result, where we optimize the circuit depth in the learned circuit for implementing $\hat{\mathcal{E}}$.
While the depth in the learned circuit can be controlled, the computational complexity becomes substantially worse.
The learning algorithm and proof are given in Section~\ref{sec:geo-kD-lattice-optimized-proof}.

\begin{theorem}[Learning geometrically-local shallow circuits on $k$-dimensional lattice with optimized circuit depth] \label{thm:geo-kD-lattice-optimized}
    Given an unknown $n$-qubit circuit $U$ over any two-qubit gates in $\mathrm{SU}(4)$ with circuit depth $d = \mathcal{O}(1)$ acting on a $k$-dimensional lattice with $k = \mathcal{O}(1)$.
    With a randomized measurement dataset $\mathcal{T}_U(N)$ of size
    \begin{equation}
        N = 2^{\mathcal{O}((8kd)^{k})} \frac{ n^2\log(n/\delta) }{\varepsilon^2},
    \end{equation}
    we can learn an $n$-qubit quantum channel $\hat{\mathcal{E}}$ that can be implemented by a quantum circuit over $2n$ qubits on an extended $k$-dimensional lattice (see Fig.~\ref{fig:geometry}(b)), such that
    \begin{equation}
        \norm{\hat{\mathcal{E}} - \mathcal{U}}_\diamond \leq \varepsilon,
    \end{equation}
    with probability at least $1 - \delta$.
    \begin{itemize}
        \item With computational time $\mathcal{O}(n) \cdot N$, the learned circuit has depth at most
        \begin{equation}
            (k+1) 4^{4 (8kd)^{k}}+1.
        \end{equation}
        \item With computational time $\mathcal{O}(n) \cdot N + \left(n/\varepsilon\right)^{\mc O( (8kd)^{k+1})}$, the learned circuit has depth at most
        \begin{equation}
            (k+1)(2d+1)+1.
        \end{equation}
    \end{itemize}
    In addition, if each two-qubit gate in the unknown circuit is chosen from a finite gate set of a constant size, then the algorithm learns an exact description $\hat{\mathcal{E}} = \mathcal{U}$ with probability $1-\delta$, using $N = \mathcal{O}( \log(n / \delta) )$ samples, $\mathcal{O}(n \log (n /\delta))$ time, and a learned circuit of depth $(k+1)(2d+1)+1$.
\end{theorem}

\begin{remark}[The geometry in the doubled system] \label{rem:geo-double}
    In the two theorems given above, we mentioned geometrically-local circuits over $2n$ qubits, while the geometry is defined over $n$ qubits.
    Given the geometry represented as a graph $G = (V, E)$ over $n$ qubits with $V = \{1, \ldots, n\}$.
    We extend the graph to $2n$ qubits $G_{\mathrm{ext}} = (V_{\mathrm{ext}}, E_{\mathrm{ext}})$ as follows.
    \begin{equation}
        V_{\mathrm{ext}} = \{1, \ldots, n, n+1, \ldots, 2n\}, \quad E_{\mathrm{ext}} = E \cup \{ (i, n+i) | 1 \leq i \leq n \}.
    \end{equation}
    Each qubit $n+i$ in the added system is connected only to qubit $i$ in the original system; See Fig.~\ref{fig:geometry}(b).
\end{remark}

\subsection{Techniques}
We present two sets of closely related techniques for learning an $n$-qubit unitary $U$. The first set in Section~\ref{sec:local-inv} uses an idea called local inversion unitary, which follows from the concept of strong approximate local identity given in Section~\ref{sec:approx-local-id}.
As we have shown earlier, strong local identity checks can be performed by using Heisenberg-evolved single-qubit Pauli observables $U^\dagger P_i U$.
The second set in Section~\ref{sec:Hevo-Pauli} directly uses the Heisenberg-evolved Pauli observables $U^\dagger P_i U$.

\subsubsection{Learning using local inversion}
\label{sec:local-inv}

We begin by defining the concept of an approximate local inversion unitary.

\begin{definition}[Strong $\varepsilon$-approximate local inversion]
    Given $n \in \mathbb{N}, \varepsilon \in (0, 1),$ $i \in \{1, \ldots, n\}$, and $n$-qubit unitaries $U$ and $V_i$.
    We say $V_i$ is a strong $\varepsilon$-approximate local inversion of $U$ on the $i$-th qubit if $\mathcal{U} \mathcal{V}_i$ is a strong $\varepsilon$-approximate local identity on the $i$-th qubit.
\end{definition}

\begin{corollary}[Local inversion from Heisenberg-evolved Pauli observables]
    Given $n \in \mathbb{N}, \varepsilon \in (0, 1),$ $i \in \{1, \ldots, n\}$, and $n$-qubit unitaries $U$ and $V_i$. If $V_i$ satisfies
    \begin{equation}
        \sum_{P \in \{X, Y, Z\}} \norm{V_i^\dagger U^\dagger P_i U V_i - P_i}_\infty \leq \varepsilon,
    \end{equation}
    where $P_i$ acts as $P \in \{X, Y, Z\}$ on the $i$-th qubit and as identity on the rest of the qubits, then $V_i$ is a strong $\varepsilon$-approximate local inversion of $U$ on the $i$-th qubit.
\end{corollary}
\begin{proof}
    This corollary follows from Lemma~\ref{lem:char-strong-local-id}, which characterizes the strong $\varepsilon$-approximate local identity with Heisenberg evolution of single-qubit Pauli observables.
\end{proof}

Instead of learning the unitary $U$ alone, we consider learning the $n$ local inversion unitaries $V_1, \ldots, V_n$. From the corollary given above, a straightforward way to learn $V_i$ is to first learn the Heisenberg-evolved single-qubit Pauli observable $U^\dagger P_i U$ for all $P = X, Y, Z$, then try to find a unitary $V_i$ that evolves $U^\dagger P_i U$ approximately back to $P_i$.
This could be a much simpler task than learning the entire $n$-qubit unitary altogether.

While local inversion could potentially make the learning easier, it is \emph{a priori} unclear if learning these local inversions is sufficient to learn $U$.
In the following, we define a formalism for sewing these local inversion unitaries into a $2n$-qubit unitary (instead of $n$ qubits).

\begin{definition}[Sewing the local inversions] \label{def:sew-local-inv}
    Given $n \in \mathbb{N}$ and $n$-qubit unitaries $V_1, \ldots, V_n$.
    We define the sewed $2n$-qubit unitary consisting of two sets of $n$ qubits to be the following,
    \begin{equation}
        U_{\mathrm{sew}}(V_1, \ldots, V_n) := S \left[\prod_{i=1}^n \left(V_i^{(1)}\right) S_{i} \left(V_i^{(1)}\right)^\dagger \right],
    \end{equation}
    where $V^{(1)}_i$ corresponds to applying the $n$-qubit unitary $V_i$ on the first $n$ qubits, $S_{i}$ is the swap operator for the $i$-th qubit between the two sets of $n$ qubits, $S$ is the swap operator for all $n$ qubits.
\end{definition}

\begin{remark}[Sewing order] \label{rem:sewing-order}
    The order for $\left(V_i^{(1)}\right) \mathrm{S}_{i} \left(V_i^{(1)}\right)^\dagger$ in sewing the local inversions does not matter. We can choose the order to optimize the resulting circuit, e.g., to minimize the circuit depth.
\end{remark}

\begin{lemma}[Form of the sewed local inversions] \label{lem:sewedlocalinv}
    Given $n \in \mathbb{N}$ and $n$-qubit unitaries $U, V_1, \ldots, V_n$.
    Assume $V_i$ is a strong $\varepsilon_i$-approximate local inversion of $U$ on the $i$-th qubit.
    Let $U_{\mathrm{sew}} = U_{\mathrm{sew}}(V_1, \ldots, V_n)$.
    \begin{equation}
        \mathcal{D}_{\diamond}(\mathcal{U}_{\mathrm{sew}}, \mathcal{U} \otimes \mathcal{U}^\dagger) = \frac{1}{2} \norm{ \mathcal{U}_{\mathrm{sew}} - \mathcal{U} \otimes \mathcal{U}^\dagger }_\diamond \leq \sum_{i=1}^n \varepsilon_i,
    \end{equation}
    where the first/second set of $n$ qubits is on the left/right of the tensor product.
\end{lemma}
\begin{proof}
    From Theorem~3.55 in \cite{watrous2018theory}, we have
    \begin{align}
        \norm{ \mathcal{U}_{\mathrm{sew}} -  \mathcal{U} \otimes \mathcal{U}^\dagger }_\diamond &= \norm{ (U^\dagger \otimes U) U_{\mathrm{sew}} \ketbra{\psi}{\psi} U_{\mathrm{sew}}^\dagger (U \otimes U^\dagger) - \ketbra{\psi}{\psi}}_1
    \end{align}
    for some $2n$-qubit state $\ket{\psi}$.
    We define the following mathematical object,
    \begin{equation}
        \ketbra{\psi_i}{\psi_i} := \left[ (\mathcal{U}^\dagger \otimes \mathcal{I})  \left( \mathcal{S}_{1} \ldots \mathcal{S}_{i}\right) (\mathcal{I} \otimes \mathcal{U}) \mathcal{S} \left( \left(\mathcal{V}^{(1)}_{i+1}\right) \mathcal{S}_{i+1} \left(\mathcal{V}^{(1)}_{i+1}\right)^\dagger \ldots \left(\mathcal{V}^{(1)}_{n} \right) \mathcal{S}_{n} \left(\mathcal{V}^{(1)}_{n} \right)^\dagger \right) \right](\ketbra{\psi}{\psi})
    \end{equation}
    for each $i = 0, \ldots, n$.
    Note that we have the following identities,
    \begin{align}
        \ketbra{\psi_0}{\psi_0} &= (U^\dagger \otimes U) U_{\mathrm{sew}} \ketbra{\psi}{\psi} U_{\mathrm{sew}}^\dagger (U \otimes U^\dagger), \\
        \ketbra{\psi_n}{\psi_n} &=
        \left[(\mathcal{U}^\dagger \otimes \mathcal{I}) \mathcal{S} (\mathcal{I} \otimes \mathcal{U} ) \mathcal{S}\right](\ketbra{\psi}{\psi}) = \ketbra{\psi}{\psi}.
    \end{align}
    By the triangle inequality, we can obtain the following telescoping sum,
    \begin{equation}
        \norm{ \mathcal{U}_{\mathrm{sew}} - \mathcal{U} \otimes \mathcal{U}^\dagger }_\diamond = \norm{\ketbra{\psi_0}{\psi_0} - \ketbra{\psi_n}{\psi_n}}_1 \leq \sum_{i=1}^n \norm{ \ketbra{\psi_i}{\psi_i} - \ketbra{\psi_{i-1}}{\psi_{i-1}} }_1.
    \end{equation}
    Each summand can be bounded as follows,
    \begin{align}
        \norm{ \ketbra{\psi_i}{\psi_i} - \ketbra{\psi_{i-1}}{\psi_{i-1}} }_1 & \leq \norm{  \mathcal{S}_i (\mathcal{I} \otimes \mathcal{U}) \mathcal{S} -
        (\mathcal{I} \otimes \mathcal{U}) \mathcal{S} \left(\mathcal{V}_{i} \otimes \mathcal{I} \right) \mathcal{S}_{i} \left(\mathcal{V}_{i} \otimes \mathcal{I} \right)^\dagger }_\diamond \\
        &= \norm{  \mathcal{S} \mathcal{S}_i (\mathcal{U} \otimes \mathcal{I}) -
        \mathcal{S} (\mathcal{U} \otimes \mathcal{I}) \left(\mathcal{V}_{i} \otimes \mathcal{I} \right) \mathcal{S}_{i} \left(\mathcal{V}_{i} \otimes \mathcal{I} \right)^\dagger }_\diamond \\
        &\leq \norm{ \mathcal{S}_i (\mathcal{U} \otimes \mathcal{I}) -
        \left( \left(\mathcal{I}_i \otimes \mathcal{E}^{\mathcal{U} \mathcal{V}_{i}}_{\neq i} \right) \otimes \mathcal{I}\right) \mathcal{S}_{i} \left(\mathcal{V}_{i} \otimes \mathcal{I} \right)^\dagger }_\diamond  + \varepsilon_i \\
        &= \norm{ \mathcal{S}_i (\mathcal{U} \otimes \mathcal{I}) -
        \mathcal{S}_{i} \left( \left(\mathcal{I}_i \otimes \mathcal{E}^{\mathcal{U} \mathcal{V}_{i}}_{\neq i} \right) \otimes \mathcal{I}\right) \left(\mathcal{V}_{i} \otimes \mathcal{I} \right)^\dagger }_\diamond  + \varepsilon_i\\
        &= \norm{ \left(\mathcal{U}\mathcal{V}_{i} \otimes \mathcal{I}\right) - \left( \left(\mathcal{I}_i \otimes \mathcal{E}^{ \mathcal{U} \mathcal{V}_{i}}_{\neq i} \right) \otimes \mathcal{I}\right) }_\diamond  + \varepsilon_i \leq 2\varepsilon_i.
    \end{align}
    Together, we obtain the desired statement.
\end{proof}

\begin{remark}[A basic identity for $U \otimes U^\dagger$]
    A trivial example of an exact local inversion of $U$ on the $i$-th qubit is $V_i = U^\dagger$. In this case, Lemma~\ref{lem:sewedlocalinv} yields the following basic identity,
    \begin{equation} \label{eq:UdagU-id}
        U \otimes U^\dagger = S \left[ \prod_{i=1}^n \left(U^\dagger \otimes I \right) S_i \left(U \otimes I\right)\right],
    \end{equation}
    which can also be shown by canceling all the intermediate $\left(U \otimes I\right) \left(U^\dagger \otimes I \right)$.
\end{remark}

\subsubsection{Learning using Heisenberg-evolved Pauli observables}
\label{sec:Hevo-Pauli}

We have seen earlier that one direct approach to learning local inversion is to first learn the Heisenberg-evolved single-qubit Pauli observables $U^\dagger P_i U$.
In the following, we define an alternative formalism that directly sews the Heisenberg-evolved Pauli observables into a $2n$-qubit unitary (instead of $n$ qubits) that approximates $U \otimes U^\dagger$.
One can flexibly choose either approach. Typically, learning the Heisenberg-evolved Pauli observables is computationally simpler, but yields higher depth in the learned circuit.

\begin{definition}[Approximate Heisenberg-evolved Paui observables]
    Given $n \in \mathbb{N}, \varepsilon \in (0, 1),$ $i \in \{1, \ldots, n\}$, $P \in \{X, Y, Z\}$, an $n$-qubit unitary $U$, and an $n$-qubit observable $O_{i, P}$.
    We say $O_{i, P}$ is an $\varepsilon$-approximate Heisenberg-evolved Pauli observable $P$ on qubit $i$ under $U$ if $\norm{ O_{i, P} - U^\dagger P_i U}_\infty \leq \varepsilon$.
\end{definition}

Given a set of $3n$ Heisenberg-evolved Pauli observables, we use the following definition to sew them into a $2n$-qubit unitary.

\begin{definition}[Sewing the Heisenberg-evolved observables] \label{def:sewing-Hei-obs}
    Given $n \in \mathbb{N}$ and $3 \times n$ $n$-qubit observables $O_{i, P}, \forall i = 1, \ldots, n, P \in \{X, Y, Z\}$.
    Let $\mathrm{Proj}_{\mathrm{U}}(A)$ be the projection of a matrix $A$ to a unitary matrix minimizing the operator norm $\norm{\cdot}_\infty$, i.e.,
    \begin{equation}
        \mathrm{Proj}_{\mathrm{U}}(A) := \argmin_{B: \text{unitary}} \norm{A - B}_\infty.
    \end{equation}
    We define the sewed $2n$-qubit unitary consisting of two sets of $n$ qubits to be the following,
    \begin{equation} \label{eq:Usew-def}
        U_{\mathrm{sew}}(\{O_{i, P}\}_{i, P}) := S \prod_{i=1}^n \left[\mathrm{Proj}_{\mathrm{U}}\left( \frac{1}{2} \, I \otimes I + \frac{1}{2}\sum_{P \in \{X, Y, Z\}} O_{i, P} \otimes P_i \right) \right],
    \end{equation}
    where $V^{(1)}_i$ corresponds to applying the $n$-qubit unitary $V_i$ on the first $n$ qubits, $S_{i}$ is the swap operator for the $i$-th qubit between the two sets of $n$ qubits, $S$ is the swap operator for all $n$ qubits.
\end{definition}

\begin{remark}[Sewing order]
    The order for sewing $\mathrm{Proj}_{\mathrm{U}}\left( \frac{1}{2} \, I \otimes I + \frac{1}{2}\sum_{P} O_{i, P} \otimes P_i \right)$ is arbitrary.
\end{remark}

In the above, we have utilized the projection function $\mathrm{Proj}_{\mathrm{U}}$.
In the following lemma, we show that this function can be computed efficiently on a classical computer.

\begin{lemma}[Projection onto unitary matrices] \label{lem:projection-unitary}
    Consider the singular value decomposition $A = U \Sigma V^\dagger$, where $\Sigma$ is diagonal, nonnegative, and $U, V$ is unitary.
    The projection can be defined as
    \begin{equation}
        \mathrm{Proj}_{\mathrm{U}}(A) = U V^\dagger.
    \end{equation}
    The computational time is polynomial in the dimension of $A$.
\end{lemma}
\begin{proof}
    Consider any unitary $B$. We have $\norm{A - B}_\infty = \norm{\Sigma - U^\dagger B V}_\infty.$ Let $W$ be the unitary $U^\dagger B V$. We can use the definition of $\norm{M}_\infty = \sup_{v} \norm{M v}_2 / \norm{v}_2$ to see that
    \begin{equation}
        \norm{\Sigma - W}_\infty \geq \max_{i} \norm{\Sigma_{ii} \hat{e}_i - W \hat{e}_i}_2 \geq \max_{i} \sqrt{1 + \Sigma_{ii}^2 - 2 \Sigma_{ii} \mathrm{Re}[\hat{e}_i^T W \hat{e}_i]} \geq \max_{i} |1 - \Sigma_{ii}| = \norm{\Sigma - I}_\infty,
    \end{equation}
    where $\hat{e}_i$ is the unit vector with a nonzero entry on the $i$-th coordinate.
    Because $\norm{\Sigma - I}_\infty = \norm{A - U V^\dagger}_\infty$, we have obtained $\norm{A - B}_\infty \geq \norm{A - U V^\dagger}_\infty$.
\end{proof}

Similar to sewing local inversions, the sewed unitary accurately approximates $U \otimes U^\dagger$.

\begin{lemma}[Form of the sewed Heisenberg-evolved observables] \label{lem:sew-form-Heisenberg}
    Given $n \in \mathbb{N}$, an $n$-qubit unitary $U$, and $3 \times n$ $n$-qubit observables $O_{i, P}, \forall i = 1, \ldots, n, P \in \{X, Y, Z\}$.
    Assume $O_{i, P}$ is an $\varepsilon_{i, P}$-approximate Heisenberg-evolved Pauli observable $P$ on qubit $i$ under $U$.
    Let $U_{\mathrm{sew}} = U_{\mathrm{sew}}(\{O_{i, P}\}_{i, P})$. Then
    \begin{equation}
        \mathcal{D}_{\diamond}(\mathcal{U}_{\mathrm{sew}}, \mathcal{U} \otimes \mathcal{U}^\dagger) = \frac{1}{2} \norm{ \mathcal{U}_{\mathrm{sew}} - \mathcal{U} \otimes \mathcal{U}^\dagger }_\diamond \leq \sum_{i=1}^n \sum_{P \in \{X, Y, Z\}} \varepsilon_{i, P},
    \end{equation}
    where the first/second set of $n$ qubits is on the left/right of the tensor product.
\end{lemma}
\begin{proof}
    From Eq.~\eqref{eq:UdagU-id}, we have the following identity,
    \begin{equation}
        U \otimes U^\dagger = S \left[ \prod_{i=1}^n (U^\dagger \otimes I) S_i (U \otimes I) \right].
    \end{equation}
    Using the fact that $S_i = \tfrac{1}{2} I \otimes I + \tfrac{1}{2} \sum_{P \in \{X, Y, Z\}} P_i \otimes P_i$, we can rewrite the above identity as
    \begin{equation}
        U \otimes U^\dagger = S \prod_{i=1}^n \left[ \frac{1}{2} \, I \otimes I + \frac{1}{2}\sum_{P \in \{X, Y, Z\}} (U^\dagger P_i U) \otimes P_i \right].
    \end{equation}
    Let us denote the following unitaries,
    \begin{align}
        V_i &:= \frac{1}{2} \, I \otimes I + \frac{1}{2}\sum_{P \in \{X, Y, Z\}} (U^\dagger P_i U) \otimes P_i,\\
        \widetilde{W}_i &:= \frac{1}{2} \, I \otimes I + \frac{1}{2}\sum_{P \in \{X, Y, Z\}} O_{i, P} \otimes P_i\\
        W_i &:= \mathrm{Proj}_{\mathrm{U}}\left( \widetilde{W}_i \right).
    \end{align}
    We can upper bound the diamond distance as follows,
    \begin{align}
        \norm{ \mathcal{U}_{\mathrm{sew}} - \mathcal{U} \otimes \mathcal{U}^\dagger }_\diamond &= \norm{ \mathcal{V}_n \ldots \mathcal{V}_1 - \mathcal{W}_n \ldots \mathcal{W}_1 }_\diamond \\
        &\leq \sum_{i=1}^n \norm{ \mathcal{V}_n \ldots \mathcal{V}_{i+1} \mathcal{W}_i \ldots \mathcal{W}_1 - \mathcal{V}_n \ldots\mathcal{V}_{i} \mathcal{W}_{i-1} \ldots \mathcal{W}_1 }_\diamond \\
        &\leq \sum_{i=1}^n \norm{ \mathcal{V}_n \ldots \mathcal{V}_{i+1} \mathcal{W}_i \ldots \mathcal{W}_1 - \mathcal{V}_n \ldots\mathcal{V}_{i} \mathcal{W}_{i-1} \ldots \mathcal{W}_1 }_\diamond \\
        &= \sum_{i=1}^n \norm{ \mathcal{W}_i - \mathcal{V}_{i} }_\diamond \leq 2 \sum_{i=1}^n \norm{ W_i - V_{i} }_\infty. \label{eq:Heisenberg-sewing-last-line}
    \end{align}
    The last inequality uses the fact that $\mathcal{W}_i$ and $\mathcal{V}_{i}$ are unitary channels.
    From triangle inequality and the definition of $\mathrm{Proj}_{\mathrm{U}}(\cdot)$, we have the following inequality,
    \begin{equation}\label{eq:projuerror}
    \begin{aligned}
        \norm{W_i - V_i}_\infty &\leq \norm{W_i - \widetilde{W}_i}_\infty + \norm{\widetilde{W}_i - V_i}_\infty \\
        &= \min_{V: \text{unitary}} \norm{\widetilde{W}_i - V}_\infty + \norm{\widetilde{W}_i - V_i}_\infty \\&\leq 2 \norm{\widetilde{W}_i - V_i}_\infty.
    \end{aligned}
    \end{equation}
    We now use the specific form of $\widetilde{W}_i, V_i$ to upper bound the summand,
    \begin{equation}
        \norm{W_i - V_i}_\infty \leq \sum_{P \in \{X, Y, Z\}} \norm{ O_{i, P} - U^\dagger P_i U }_\infty \leq \sum_{P} \varepsilon_{i, P}.
    \end{equation}
    Together with Eq.~\eqref{eq:Heisenberg-sewing-last-line}, we can obtain the desired statement.
\end{proof}

Given an $n$-qubit observable $O$, we define $\mathrm{supp}(O)$ to be the set of qubits that the observable $O$ acts on.
We also define $|O|$ to be the size of $\mathrm{supp}(O)$.
We have the following lemma for learning a few-body observable.
The learned observable $\hat{O}$ has the property that it only acts on qubits that $O$ acts on, hence $\mathrm{supp}(\hat{O}) \subseteq \mathrm{supp}(O)$.

\begin{lemma}[Learning a few-body observable with an unknown support]
    \label{lem:learn-few-body-obs-unk-supp}
    Given an error $\varepsilon$, failure probability $\delta$, an unknown $n$-qubit observable $O$ with $\norm{O}_\infty \leq 1$ that acts on an unknown set of $k$ qubits, and a dataset $\mathcal{T}_O(N) = \left\{ \ket{\psi_\ell} = \bigotimes_{i=1}^n \ket{\psi_{\ell, i}}, v_\ell \right\}_{\ell=1}^N$, where $\ket{\psi_{\ell, i}}$ is sampled uniformly from $\mathrm{stab}_1$ and $v_\ell$ is a random variable with $\mathbb{E}[v_\ell] = \bra{\psi_\ell} O \ket{\psi_\ell}$, $|v_\ell| = \mathcal{O}(1)$.
    Given a dataset size of
    \begin{equation}
        N = \frac{2^{\mathcal{O}(k)} \log(n / \delta)}{\varepsilon^2},
    \end{equation}
    with probability at least $1 - \delta$, we can learn an observable $\hat{O}$ such that $\norm{\hat{O} - O}_\infty \leq \varepsilon$ and $\mathrm{supp}(\hat{O}) \subseteq \mathrm{supp}(O)$.
    The computational complexity is $\mathcal{O}(n^{k} \log(n / \delta) / \varepsilon^2)$.
\end{lemma}
\begin{proof}
Consider the observable $O$ under the Pauli basis, $O = \sum_{P} \alpha_P P$. The $\alpha_P$ coefficients satisfy
\begin{equation}
    \alpha_P=3^{|P|}\Exp_{\ket{\psi}\sim\mathrm{stab}_1^{\otimes n}}\expval{O}{\psi}\expval{P}{\psi},
\end{equation}
which can be learned by replacing the expectation with averaging over the dataset.

    We begin by defining the learned observable $\hat{O}$.
    \begin{align}
        \hat{\alpha}_P &:= \frac{3^{|P|}}{N} \sum_{\ell=1}^N v_\ell \bra{\psi_\ell} P \ket{\psi_\ell}, \quad\quad \forall P \in \{I, X, Y, Z\}^{\otimes n}: |P| \leq k, \\
        \hat{\beta}_P &:= \begin{cases}
            \hat{\alpha}_P, & |\hat{\alpha}_P| \geq 0.5 \varepsilon / (2 \sqrt{2})^{k},\\
            0, & |\hat{\alpha}_P| < 0.5 \varepsilon / (2 \sqrt{2})^{k},
        \end{cases}\\
        \hat{O} &:= \sum_{P \in \{I, X, Y, Z\}^{\otimes n}: |P| \leq k} \hat{\beta}_P P.
    \end{align}
     Because $O$ acts on at most $k$ qubits, $\alpha_P = 0$ for $|P| > k$.
    From Bernstein's inequality, given a dataset size of
    \begin{equation}
        N = \frac{2^{\mathcal{O}(k)} \log(n / \delta)}{\varepsilon^2},
    \end{equation}
    with probability at least $1 - \delta$, we have
    \begin{equation}
        \left| \alpha_P - \hat{\alpha}_P \right| < 0.5 \varepsilon / (2 \sqrt{2})^{k}, \quad\quad \forall P \in \{I, X, Y, Z\}^{\otimes n}: |P| \leq k. \label{eq:Bernstein-err}
    \end{equation}
    In the following, we assume the above event holds, which happens with probability at least $1 - \delta$.
    We separately prove the following two  statements.

    \paragraph{$\mathrm{supp}(\hat{O}) \subseteq \mathrm{supp}(O)$}: For a Pauli observable $P$ with $\alpha_P = 0$, we have $|\hat{\alpha}_P| < 0.5 \varepsilon / (2 \sqrt{2})^{k}$ from Eq.~\eqref{eq:Bernstein-err}. Hence, $\hat{\beta}_P = 0$.
    As a result, the set of qubits acted by $\hat{O}$ is a subset of $\mathrm{supp}(O)$.

    \vspace{1em}

    \paragraph{$\norm{\hat{O} - O}_\infty \leq \varepsilon$}:
    From the fact that $\alpha_P = 0$ implies $\hat{\beta}_P = 0$,
    we have
    \begin{align}
        \hat{O} - O &= \sum_{P \in \{I, X, Y Z\}^{\otimes n}: \mathrm{supp}(P) \subseteq \mathrm{supp}(O)} \left(\hat{\beta}_P - \alpha_{P}\right) P \label{eq:hatOmnsOeps-first}\\
        &= \sum_{Q \in \{I, X, Y Z\}^{\otimes k}} \left(\hat{\beta}_{P(Q)} - \alpha_{P(Q)}\right) P(Q),
    \end{align}
    where $P(Q) := Q \otimes I_{\{1, \ldots, n\} \setminus \mathrm{supp}(O)}$ and $k = |\mathrm{supp}(O)|$.
    Therefore, we can upper bound the spectral norm by
    \begin{equation}
        \norm{\hat{O} - O}_\infty \leq \norm{\sum_{Q \in \{I, X, Y Z\}^{\otimes k}} \left(\hat{\beta}_{P(Q)} - \alpha_{P(Q)}\right) P(Q)}_\infty = \norm{\sum_{Q \in \{I, X, Y Z\}^{\otimes k}} \left(\hat{\beta}_{P(Q)} - \alpha_{P(Q)}\right) Q}_\infty.
    \end{equation}
    Recall that $\norm{A}_\infty \leq \sqrt{\Tr(A^2)}$ for any Hermitian matrix $A$, we have
    \begin{equation}
        \norm{\hat{O} - O}_\infty \leq \sqrt{ \sum_{Q \in \{I, X, Y Z\}^{\otimes k}} \left(\hat{\beta}_{P(Q)} - \alpha_{P(Q)}\right)^2 \Tr(Q^2) } \leq (2 \sqrt{2})^k \max_{|P| \leq k} \left| \hat{\beta}_P - \alpha_P \right|. \label{eq:hatOmnsOeps-last}
    \end{equation}
    By the triangle inequality and Eq.~\eqref{eq:Bernstein-err}, we have
    \begin{equation}
        \left| \hat{\beta}_P - \alpha_P \right| \leq \left| \hat{\beta}_P - \hat{\alpha}_P \right| + \left| \hat{\alpha}_P - \alpha_P \right| < \varepsilon / (2 \sqrt{2})^k, \quad\quad \forall |P| \leq k.
    \end{equation}
    Therefore, we have obtained the desired inequality $\norm{\hat{O} - O}_\infty \leq \varepsilon$.
\end{proof}

\begin{lemma}[Learning a few-body observable with a known support]
    \label{lem:learn-few-body-obs-kno-supp}
    Given an error $\varepsilon$, failure probability $\delta$, an unknown $n$-qubit observable $O$ with $\norm{O}_\infty \leq 1$ that acts on an known set $S$ of $k$ qubits, and a dataset $\mathcal{T}_O(N) = \left\{ \ket{\psi_\ell} = \bigotimes_{i=1}^n \ket{\psi_{\ell, i}}, v_\ell \right\}_{\ell=1}^N$, where $\ket{\psi_{\ell, i}}$ is sampled uniformly from $\mathrm{stab}_1$ and $v_\ell$ is a random variable with $\mathbb{E}[v_\ell] = \bra{\psi_\ell} O \ket{\psi_\ell}$, $|v_\ell| = \mathcal{O}(1)$.
    Given a dataset size of
    \begin{equation}
        N = \frac{2^{\mathcal{O}(k)} \log(1 / \delta)}{\varepsilon^2},
    \end{equation}
    with probability at least $1 - \delta$, we can learn an observable $\hat{O}$ such that $\norm{\hat{O} - O}_\infty \leq \varepsilon$ and $\mathrm{supp}(\hat{O}) \subseteq S$.
    The computational complexity is $\mathcal{O}(2^{\mathcal{O}(k)} \log(1 / \delta) / \varepsilon^2)$.
\end{lemma}
\begin{proof}
    We begin by defining the learned observable $\hat{O}$.
    \begin{align}
        \hat{\alpha}_P &:= \frac{3^{|P|}}{N} \sum_{\ell=1}^N v_\ell \bra{\psi_\ell} P \ket{\psi_\ell}, \quad\quad \forall P \in \{I, X, Y, Z\}^{\otimes n}: \mathrm{supp}(P) \subseteq S, \\
        \hat{O} &:= \sum_{P \in \{I, X, Y, Z\}^{\otimes n}: \, \mathrm{supp}(P) \subseteq S} \hat{\alpha}_P P.
    \end{align}
    By definition, we can see that $\mathrm{supp}(\hat{O}) \subseteq S$.
    Consider the observable $O$ under the Pauli basis, $O = \sum_{P} \alpha_P P$. Because $O$ acts on the qubits in the set $S$, $\alpha_P = 0$ for $\mathrm{supp}(P) \not\subseteq S$.
    From Bernstein's inequality, given a dataset of size
    \begin{equation}
        N = \frac{2^{\mathcal{O}(k)} \log(1 / \delta)}{\varepsilon^2},
    \end{equation}
    with probability at least $1 - \delta$, we have
    \begin{equation}
        \left| \alpha_P - \hat{\alpha}_P \right| < \varepsilon / (2 \sqrt{2})^{k}, \quad\quad \forall P \in \{I, X, Y, Z\}^{\otimes n}: \mathrm{supp}(P) \subseteq S. \label{eq:Bernstein-err2}
    \end{equation}
    In the following, we assume the above event holds, which happens with probability at least $1 - \delta$.
    Using the same derivation as in Eq.~\eqref{eq:hatOmnsOeps-first} to Eq.~\eqref{eq:hatOmnsOeps-last} for the proof of Lemma~\ref{lem:learn-few-body-obs-unk-supp}, we have
    \begin{equation}
        \norm{\hat{O} - O}_\infty \leq (2 \sqrt{2})^k \max_{P: \mathrm{supp}(P) \subseteq S} \left| \hat{\alpha}_P - \alpha_P \right| < \varepsilon,
    \end{equation}
    hence we have arrived at the desired statement.
\end{proof}

\begin{remark}[Relation to learning quantum juntas]
    The two lemmas given above are related to quantum junta learning \cite{chen2023testing} but consider a much weaker access model.
    \cite{chen2023testing} requires that the unknown observable $O$ be a unitary, and the learning algorithm can access the unitary coherently.
    In particular, \cite{chen2023testing} requires inputting half of the maximally entangled state to the unitary.
    Here, we consider access to $O$ through a simple classical dataset consisting of random product input states and the outcome when measuring the input states with observable $O$.
    When the lemmas are used as a subroutine in learning algorithms given in Section~\ref{sec:learning-unitary-diamond}, we do not have access to $O$ as a unitary, so \cite{chen2023testing} cannot be used.
\end{remark}

\subsection{Learning general shallow circuits (Proof of Theorem~\ref{thm:shallow-SU4-gates})} \label{sec:shallow-SU4-gates-proof}

We present the algorithm for learning an unknown $n$-qubit unitary $U$ generated by an arbitrary constant-depth quantum circuit $C$ with arbitrarily many ancilla qubits.
We separate the proof into two-qubit gates over $\mathrm{SU}(4)$ and over a finite gate set.

\subsubsection{Arbitrary $\mathrm{SU}(4)$ gates}

The algorithm utilizes a randomized measurement dataset $\mathcal{T}_U(N)$.
The key ideas are using Lemma~\ref{lem:learn-few-body-obs-unk-supp} to learn approximate Heisenberg-evolved Pauli observables, using Lemma~\ref{lem:sewing-const-depth} to sew the Heisenberg-evolved Pauli observables into a constant-depth quantum circuit, and using Lemma~\ref{lem:sew-form-Heisenberg} to obtain the rigorous performance guarantee.

The following lemma shows how to reuse the randomized measurement dataset $\mathcal{T}_U(N)$ to create the datasets needed to learn approximate Heisenberg-evolved Pauli observables using Lemma~\ref{lem:learn-few-body-obs-unk-supp}.

\begin{lemma}[Reusing the randomized measurement dataset] \label{lem:reusing-RMdata}
    Given an unknown $n$-qubit unitary $U$, and a randomized measurement dataset $\mathcal{T}_U(N)$ given in Eq.~\eqref{eq:random-measure-data}.
    We can create $3n$ datasets $\mathcal{T}_{U^\dagger P_i U}(N)$, for each Pauli observable $P \in \{X, Y, Z\}$ and each qubit $i$,
    \begin{equation}
        \mathcal{T}_{U^\dagger P_i U}(N) := \left\{ \ket{\psi_\ell} = \bigotimes_{j=1}^n \ket{\psi_{\ell, j}}, v^{U^\dagger P_i U}_\ell \right\}_{\ell=1}^N,
    \end{equation}
    where $\ket{\psi_{\ell, i}}$ is sampled uniformly and independently from $\mathrm{stab}_1$ and $v^{U^\dagger P_i U}_\ell$ is a random variable with $\mathbb{E}[v^{U^\dagger P_i U}_\ell] = \bra{\psi_\ell} U^\dagger P_i U \ket{\psi_\ell}$ and $|v^{U^\dagger P_i U}_\ell| = \mathcal{O}(1)$.
\end{lemma}
\begin{proof}
    Recall that from Eq.~\eqref{eq:random-measure-data}, we have
    \begin{equation}
        \mathcal{T}_U(N) = \left\{ \ket{\psi_\ell} = \bigotimes_{i=1}^n \ket{\psi_{\ell, i}}, \ket{\phi_\ell} = \bigotimes_{i=1}^n \ket{\phi_{\ell, i}} \right\}_{\ell=1}^N.
    \end{equation}
    The input states are reused over the $3n$ datasets.
    For each Pauli observable $P \in \{X, Y, Z\}$ and each qubit $i$,
    we define the output value to be
    \begin{equation}
        v^{U^\dagger P_i U}_\ell := 3 \bra{\phi_{\ell, i}} P \ket{\phi_{\ell, i}}.
    \end{equation}
    We have $| v^{U^\dagger P_i U}_\ell| = |3 \bra{\phi_{\ell, i}} P \ket{\phi_{\ell, i}}| \leq 3 = \mathcal{O}(1)$.
    Now, recall how $\ket{\phi_{\ell, i}}$ is defined.
    $\ket{\phi_{\ell, i}}$ is the measurement outcome when we measure the $i$-th qubit of the $n$-qubit state $U \ket{\psi_\ell}$ in a random Pauli basis: $X$ basis gives $\ket{X, 0} := \ket{+}, \ket{X, 1} := \ket{-}$; $Y$ basis gives $\ket{Y, 0} := \ket{y+}, \ket{Y, 1} := \ket{y-}$; $Z$ basis gives $\ket{Z, 0} := \ket{0}, \ket{Z, 1} := \ket{1}$.
    Using the fact that
    \begin{align}
        0 &= \bra{Q, b} P \ket{Q, b}, &\forall P \neq Q \in \{X, Y, Z\}, b \in \{0, 1\},\\
        P &= \sum_{b \in \{0, 1\}} (-1)^b \ketbra{P, b}{P, b}, &\forall P \in \{X, Y, Z\}.
    \end{align}
    and that the randomized measurement measures $X, Y, Z$ bases equally likely, we have
    \begin{equation}
        \Exp\left[3 \bra{\phi_{\ell, i}} P \ket{\phi_{\ell, i}}\right] = \bra{\psi_\ell} U^\dagger P_i U \ket{\psi_\ell}.
    \end{equation}
    This concludes the proof.
\end{proof}

From Lemma~\ref{lem:lightcone-super-shallow} and the fact that $\mathrm{supp}\left(U^\dagger P_i U \right) \subseteq A(i) = \bigcup_{P \in \{X, Y, Z\}} \mathrm{supp}\left(U^\dagger P_i U \right)$, we have
\begin{equation}
    \left| \mathrm{supp}\left(U^\dagger P_i U \right) \right| \leq |A(i)| = \mathcal{O}(1).
\end{equation}
This enables us to combine Lemma~\ref{lem:reusing-RMdata} for constructing $\mathcal{T}_{U^\dagger P_i U}(N), \forall i, P$ from $\mathcal{T}_{U}(N)$ and Lemma~\ref{lem:learn-few-body-obs-unk-supp} for learning few-body observables with unknown supports (since $A(i)$ is unknown) to show the following.
For any constant value $\tilde{\varepsilon} = \mathcal{O}(1)$, given a dataset size of
\begin{equation}
    N = \mathcal{O}\left( \frac{n^2 \log(n / \delta)}{\varepsilon^2} \right),
\end{equation}
we can learn $\hat{O}_{i, P}, \forall i, P$, such that with probability at least $1 - \delta$, for all $i \in \{1, \ldots, n\}$ and Pauli observable $P \in \{X, Y, Z\}$, we have
\begin{equation} \label{eq:Oip-error-learned-all2qb}
    \norm{\hat{O}_{i, P} - U^\dagger P_i U}_\infty \leq \frac{\varepsilon}{6n}, \quad \mbox{and} \quad \mathrm{supp}(\hat{O}_{i, P}) \subseteq \mathrm{supp}(U^\dagger P_i U) \subseteq A(i).
\end{equation}
The computational time for learning all $\hat{O}_{i, P}$ is $\mathcal{O}(n^{\mathcal{O}(1)} \log(n / \delta) / \varepsilon^2) = \mathrm{poly}(n) \log(1 / \delta  / \varepsilon^2)$.
From Lemma~\ref{lem:lightcone-super-shallow}, we can characterize $\mathrm{supp}(\hat{O}_{i, P}) \subseteq \mathrm{supp}(U^\dagger P_i U)$ to apply Lemma~\ref{lem:sewing-const-depth}.

\begin{lemma}[Sewing into a constant-depth quantum circuit] \label{lem:sewing-const-depth}
    Given $3n$ $n$-qubit observables $\hat{O}_{i, P}, \forall i \in \{1, \ldots, n\}, P \in \{X, Y, Z\}$, such that for any qubit $i$, $\left|\bigcup_{P} \mathrm{supp}\left(\hat{O}_{i, P}\right)\right| = \mathcal{O}(1)$ and there is only a constant number of qubit $j$ with
    \begin{equation}
        \left(\bigcup_{P} \mathrm{supp}\left(\hat{O}_{i, P}\right)\right) \cap \left(\bigcup_{P} \mathrm{supp}\left(\hat{O}_{j, P}\right)\right) \neq \varnothing.
    \end{equation}
    There exists a sewing ordering for $U_{\mathrm{sew}}(\{\hat{O}_{i, P}\}_{i, P})$ given in Definition~\ref{def:sewing-Hei-obs}, such that $U_{\mathrm{sew}}(\{\hat{O}_{i, P}\}_{i, P})$ can be implemented by a constant-depth quantum circuit.
    The constant-depth quantum circuit is geometrically-local (see Definition~\ref{def:geo-local-circuit}) if $\bigcup_{P} \mathrm{supp}\left(\hat{O}_{i, P}\right), \forall i$ are geometrically-local sets (see Definition~\ref{def:geo-local-set}).
    The computational time for finding the circuit implementation is $\mathcal{O}(n)$.
\end{lemma}
\begin{proof}
    For simplicity of notations, we define $A(i) := \bigcup_{P} \mathrm{supp}\left(\hat{O}_{i, P}\right)$.
    We can see that
    \begin{equation}
        \mathrm{supp}\left(\mathrm{Proj}_{\mathrm{U}}\left( \frac{1}{2} \, I \otimes I + \frac{1}{2}\sum_{P \in \{X, Y, Z\}} \hat{O}_{i, P} \otimes P_i \right)\right) \subseteq A(i) \cup \{n+i\},
    \end{equation}
    Because $|A(i)| = \left|\bigcup_{P} \mathrm{supp}\left(\hat{O}_{i, P}\right)\right| = \mathcal{O}(1)$ and $\mathrm{Proj}_\mathrm{U}$ can be implemented in time polynomial in $2^{|A(i) \cup \{n+i\}|} = \mathcal{O}(1)$ as shown in Lemma~\ref{lem:projection-unitary}, the following unitary
    \begin{equation}
        \mathrm{Proj}_{\mathrm{U}}\left( \frac{1}{2} \, I \otimes I + \frac{1}{2}\sum_{P \in \{X, Y, Z\}} \hat{O}_{i, P} \otimes P_i \right)
    \end{equation}
    can be implemented by a constant-depth circuit acting only on qubits in $A(i) \cup \{n+i\}$; see Fact~\ref{fact:unitary-synthesis} for exact unitary synthesis.
    Furthermore, if $A(i) = \bigcup_{P} \mathrm{supp}\left(\hat{O}_{i, P}\right)$ is a geometrically-local set, the constant-depth circuit is geometrically-local; see Corollary~\ref{cor:unitary-synthesis-geo} for exact unitary synthesis given a connectivity graph.
    The geometric locality for the $2n$-qubit system is defined in Remark~\ref{rem:geo-double}.

    Consider an $n$-node graph (equivalently, an $n$-qubit graph), where each pair $(i, j)$ of nodes (qubits) is connected by an edge if
    \begin{equation}
        A(i) \cap A(j) \neq \varnothing.
    \end{equation}
    The graph only has $\mathcal{O}(n)$ edges and can be constructed as an adjacency list in time $\mathcal{O}(n)$.
    Because the graph has a constant degree, we can use a $\mathcal{O}(n)$-time greedy graph coloring algorithm to color the $n$-qubit graph using only a constant number $\chi = \mathcal{O}(1)$ of colors.
    For each node/qubit $i$, we consider $c(i)$ to be the color labeled from $1$ to $\chi$.
    The sewing order for the $3n$ observables $\hat{O}_{i, P}$ in Definition~\ref{def:sewing-Hei-obs} are given by the greedy graph coloring, where we order from the smallest color to the largest color.
    By the definition of graph coloring, for any pair $i, j$ of qubits with the same color, we have
    \begin{equation}
        A(i) \cap A(j) = \varnothing.
    \end{equation}
    Therefore, for any color $c'$, we can find an implementation of the $2n$-qubit unitary
    \begin{equation}
        \prod_{i: c(i) = c'} \left[\mathrm{Proj}_{\mathrm{U}}\left( \frac{1}{2} \, I \otimes I + \frac{1}{2}\sum_{P \in \{X, Y, Z\}} O_{i, P} \otimes P_i \right) \right]
    \end{equation}
    with a constant-depth (and geometrically-local if $A(i), \forall i$ are geometrically-local) quantum circuit in time $\mathcal{O}(n)$.
    Since there is only a constant number of colors, the $2n$-qubit unitary $U_{\mathrm{sew}}(\{\hat{O}_{i, P}\}_{i, P})$ in Eq.~\eqref{eq:Usew-def} with the color-based ordering can be implemented with a constant-depth (and geometrically-local if $A(i), \forall i$ are geometrically-local) quantum circuit in time $\mathcal{O}(n)$.
\end{proof}

Lemma~\ref{lem:sewing-const-depth} shows that there exists an ordering for sewing the approximate Heisenberg-evolved Pauli observables $\hat{O}_{i, P}$ to create $U_{\mathrm{sew}}(\{\hat{O}_{i, P}\}_{i, P})$ given in Definition~\ref{def:sewing-Hei-obs}, such that $U_{\mathrm{sew}}(\{\hat{O}_{i, P}\}_{i, P})$ can be implemented by a constant-depth quantum circuit.
Given Eq.~\eqref{eq:Oip-error-learned-all2qb}, we can use Lemma~\ref{lem:sew-form-Heisenberg} on the form of the sewed Heisenberg-evolved Pauli observables to yield
\begin{equation}
    \norm{ \mathcal{U}_{\mathrm{sew}}(\{\hat{O}_{i, P}\}_{i, P}) - \mathcal{U} \otimes \mathcal{U}^\dagger }_\diamond \leq \varepsilon.
\end{equation}
Finally, define an $n$-qubit channel $\hat{\mathcal{E}}$ as follows,
\begin{equation}
    \hat{\mathcal{E}}(\rho) := \Tr_{> n}\left(\mathcal{U}_{\mathrm{sew}}(\{\hat{O}_{i, P}\}_{i, P})(\rho \otimes \ketbra{0^n}{0^n})\right),
\end{equation}
which can be implemented as a constant-depth quantum circuit over $2n$ qubits.
Because Eq.~\eqref{eq:Oip-error-learned-all2qb} holds with probability at least $1 - \delta$, we have
\begin{equation}
    \norm{\hat{\mathcal{E}} - \mathcal{U}}_\diamond \leq \varepsilon
\end{equation}
with probability at least $1 - \delta$.
This concludes the proof of the first part of Theorem~\ref{thm:shallow-SU4-gates}.

\subsubsection{Finite gate sets}

Let the circuit depth be $d = \mathcal{O}(1)$, the finite gate set be $\mathcal{G}$ with $|\mathcal{G}| = \mathcal{O}(1)$, and the number of ancilla qubits be $m$.
The ancilla qubits are initialized as $\ket{0}$ and end up at $\ket{0}$ after applying $C$, i.e.,
\begin{equation}
    U \otimes \ket{0^m} = C (I_n \otimes \ket{0^m}).
\end{equation}
The Schrodinger evolution of an $n$-qubit state $\rho$ under $U$ is
\begin{equation}
    U \rho U^\dagger = \Tr_{> n}( C (\rho \otimes \ketbra{0^m}{0^m}) C^\dagger),
\end{equation}
where $C$ is a shallow quantum circuit over $n+m$ qubits and $\Tr_{> n}$ traces out the ancilla qubits.
The Heisenberg evolution of an $n$-qubit observable $O$ under $U$ is
\begin{equation}
    U^\dagger O U = (I_n \otimes \bra{0^m}) C^\dagger (O \otimes I_m) C (I_n \otimes \ket{0^m}),
\end{equation}
where $I_n$ is an identity on $n$ qubits and $I_m$ is an identity on $m$ qubits.

The algorithm utilizes a randomized measurement dataset $\mathcal{T}_U(N)$.
The key ideas are using Lemma~\ref{lem:learn-few-body-obs-unk-supp} and a brute-force search algorithm over a constant number of choices to find the exact Heisenberg-evolved Pauli observables, using Lemma~\ref{lem:sewing-const-depth} to sew the Heisenberg-evolved Pauli observables into a constant-depth quantum circuit, and using Lemma~\ref{lem:sew-form-Heisenberg} to obtain the rigorous guarantee.

\begin{lemma}[Characterizing the support] \label{lem:lightcone-super-shallow}
Given an $n$-qubit unitary $U$ generated by a constant-depth quantum circuit $C$ with $m$ ancilla qubits. For each qubit~$i \in \{1, \ldots, n\}$, let us define a set of qubits
\begin{equation}
    A(i) := \bigcup_{P \in \{X, Y, Z\}} \mathrm{supp}\left(U^\dagger P_i U \right).
\end{equation}
We have $|A(i)| = \mathcal{O}(1)$ and the number of qubits $j$ such that $A(i) \cap A(j) \neq \varnothing$ is at most a constant.
\end{lemma}
\begin{proof}
    From the definition of $U$, $U \otimes \ket{0^m} = C (I_n \otimes \ket{0^m})$, we have
    \begin{equation}
        A(i) \subseteq \bigcup_{P \in \{X, Y, Z\}} \mathrm{supp}\left(C^\dagger P_i C \right).
    \end{equation}
    Let $d = \mathcal{O}(1)$ be the depth of the circuit $C$. We say qubit $i$ is connected to qubit $j$ in the circuit $C$ if there is a sequence of gates in $C$ with strictly decreasing layers, such that each pair of consecutive gates share a qubit and the first gate acts on qubit $i$ and the last gate acts on qubit $j$.
    Let $B(i)$ be the set of qubits connected to $i$.
    Because each pair of consecutive two-qubit gates share a qubit, the number of possible gate sequences for a fixed $i$ grows at most twice as large at every step.
    Hence, $|B(i)| \leq 2^d$.
    Furthermore, for any Pauli operator $P$, $\mathrm{supp}\left(C^\dagger P_i C \right)$ only contains qubits connected to $i$, so $A(i) \subseteq B(i)$. Together, $|A(i)| \leq |B(i)| \leq 2^d = \mathcal{O}(1).$
    This establishes the first claim.

    Now, we show that for any $i$, the number of $j$ such that $B(i) \cap B(j) \neq \varnothing$ is at most a constant.
    If $B(i) \cap B(j) \neq \varnothing$, we know that there is a sequence of gates in $C$ with strictly decreasing layers and then strictly increasing layers, such that each pair of consecutive gates share a qubit and the first gate acts on qubit $i$ and the last gate acts on qubit $j$.
    Similar to before, The number of possible gate sequences for a fixed $i$ grows at most twice as large at every step.
    Hence the number of $j$ with $B(i) \cap B(j) \neq \varnothing$ is at most $2^{2d} = \mathcal{O}(1)$.
    Because $A(i) \subseteq B(i)$, any $j$ with $A(i) \cap A(j) \neq \varnothing$ satisfies $B(i) \cap B(j) \neq \varnothing$.
    Therefore, the number of qubits $j$ such that $A(i) \cap A(j) \neq \varnothing$ is at most a constant.
    This establishes the second claim of the lemma.
\end{proof}

From the above lemma and the fact that $\mathrm{supp}\left(U^\dagger P_i U \right) \subseteq A(i)$, we have
\begin{equation}
    \left| \mathrm{supp}\left(U^\dagger P_i U \right) \right| \leq |A(i)| = \mathcal{O}(1).
\end{equation}
This enables us to combine Lemma~\ref{lem:reusing-RMdata} for constructing $\mathcal{T}_{U^\dagger P_i U}(N), \forall i, P$ from $\mathcal{T}_{U}(N)$ and Lemma~\ref{lem:learn-few-body-obs-unk-supp} for learning few-body observables with unknown supports (since $A(i)$ is unknown) to show the following.
For any constant value $\tilde{\varepsilon} = \mathcal{O}(1)$, given a dataset size of
\begin{equation}
    N = \mathcal{O}\left( \log(n / \delta) \right),
\end{equation}
we can learn $\hat{O}_{i, P}, \forall i, P$, such that with probability at least $1 - \delta$, for all $i \in \{1, \ldots, n\}$ and Pauli observable $P \in \{X, Y, Z\}$, we have
\begin{equation} \label{eq:Oip-error-learned-finite-gate-set}
    \norm{\hat{O}_{i, P} - U^\dagger P_i U}_\infty \leq \tilde{\varepsilon}, \quad \mbox{and} \quad \mathrm{supp}(\hat{O}_{i, P}) \subseteq \mathrm{supp}(U^\dagger P_i U).
\end{equation}
The computational time for learning all $\hat{O}_{i, P}$ is $\mathcal{O}(n^{\mathcal{O}(1)} \log(n / \delta)) = \mathrm{poly}(n) \log(1 / \delta)$.

Our goal now is to find $U^\dagger P_i U$ exactly using the approximate observable $\hat{O}_{i, P}$ satisfying Eq.~\eqref{eq:Oip-error-learned-finite-gate-set} by choosing a sufficiently small $\tilde{\varepsilon}$ that is constant in system size $n$.
To do so, we need to consider the backward lightcone of qubit $i$ in circuit $C$ defined below.

\begin{definition}[Backward lightcone in a circuit]
    We say a gate $g$ in circuit $U$ is in the backward lightcone of qubit $i$ in $C$ if there is a sequence of gates in $C$ with strictly decreasing layers, such that each pair of consecutive gates share a qubit, the first gate acts on qubit $i$, and the last gate is $g$.

    The circuit $C_i$ corresponding to the backward lightcone of qubit $i$ in circuit $C$ is the circuit with all gates in the backward lightcone of qubit $i$ in circuit $C$.

    The set $S_i$ of qubits corresponding to the backward lightcone of qubit $i$ in circuit $C$ is the set of all qubits acted by at least one of the gates in the backward lightcone of qubit $i$ in circuit $C$.
\end{definition}

From the definition of $C_i, S_i$ corresponding to the backward lightcones given above, we have
\begin{equation}
    \mathrm{supp}(U^\dagger P_i U) \subseteq \mathrm{supp}(C^\dagger P_i C) \subseteq S_i \quad \mbox{and} \quad U^\dagger P_i U = (I_n \otimes \bra{0^m}) C_i^\dagger P_i C_i(I_n \otimes \ket{0^m}).
\end{equation}
Note one cannot guarantee $S_i = \mathrm{supp}(U^\dagger P_i U)$.
By a counting argument similar to the proof of Lemma~\ref{lem:lightcone-super-shallow}, we have the following fact.

\begin{fact}[Size of backward lightcone] \label{fact:size-Bck-light}
    Given a depth-$d$ circuit $C$.
    The circuit $C_i$ corresponding to the backward lightcone of qubit $i$ in $C$ consists of at most $2^{d-1}$ gates.
    The set $S_i$ of qubits corresponding to the backward lightcone of qubit $i$ in $C$ contains at most $2^d$ qubits.
\end{fact}

Recall that the depth of $C$ is $d = \mathcal{O}(1)$, and the gate set is $\mathcal{G}$ with $|\mathcal{G}| = \mathcal{O}(1)$.
Because $d = \mathcal{O}(1)$, $|S_i| \leq 2^d = \mathcal{O}(1)$.
For any $(n+m)$-qubit constant-depth circuit $\tilde{C}$ over a finite gate set, given a fixed set $\tilde{S}_i$ of qubits corresponding to the backward lightcone of qubit $i$ in $\tilde{C}$, the number of possible circuit $\tilde{C}_i$ corresponding to the backward lightcone of qubit $i$ in circuit $\tilde{C}$ is a constant independent of $n, m$ and $1/\delta$.
Hence, there is a constant number of $\tilde{C}_i^\dagger P_i \tilde{C}_i = \tilde{C}^\dagger P_i \tilde{C}$.
We denote the possible choices of the $n$-qubit observable given the set $\tilde{S}_i$ and qubit $i \in \{1, \ldots, n\}$ to be $\mathcal{S}_{\mathrm{obs}}(i, \tilde{S}_i)$,
 \begin{align}
     &\mathcal{S}_{\mathrm{obs}}(i, \tilde{S}_i) := \Big\{ (I_n \otimes \bra{0^m}) \tilde{C}^\dagger P_i \tilde{C} (I_n \otimes \ket{0^m}) \,\, \Big| \,\, \text{$\tilde{C}$ is a depth-$d$ circuit over gate set $\mathcal{G}$,} \\
     &\,\,\text{such that $\tilde{S}_i$ is the set of qubits corresponding to the backward lightcone of qubit $i$ in $\tilde{C}$} \Big\}
 \end{align}
We have $|\mathcal{S}_{\mathrm{obs}}(i, \tilde{S}_i)| = \mathcal{O}(1)$.
Furthermore, we can always consider a permutation $\Pi_{i, \tilde{S}_i}$ over the qubits that implements the following permutation mapping,
\begin{equation}
    1 \rightarrow_{\Pi_{i, \tilde{S}_i}} i, \quad \{1, \ldots, |\tilde{S}_i|\} \rightarrow_{\Pi_{i, \tilde{S}_i}} \tilde{S}_i,
\end{equation}
and $\Pi_{i, \tilde{S}_i}$ acts as identity on the $m$ ancilla qubits.
Given a permutation $\Pi_{i, \tilde{S}_i}$ over the qubits (which is itself a unitary), we have
\begin{equation}
    \mathcal{S}_{\mathrm{obs}}(i, \tilde{S}_i) = \Big\{ \Pi_{i, \tilde{S}_i} O \Pi_{i, \tilde{S}_i} \,\, \Big| \,\, O \in \mathcal{S}_{\mathrm{obs}}(1, \{1, \ldots, |\tilde{S}_i|\}) \Big\}.
\end{equation}
We note that $O$ acts on $n$ qubits, while $\Pi_{i, \tilde{S}_i}$ acts on $n+m$ qubits; hence, we implicitly extend $O$ to $n+m$ qubits by acting as identity on the $m$ ancilla qubits.
The set $\mathcal{S}_{\mathrm{obs}}(1, \{1, \ldots, |\tilde{S}_i|\})$ contains all the possible observables (up to permutation of the qubits) with $|\tilde{S}_i|$ qubits in the backward lightcone of qubit $i \in \{1, \ldots, n\}$ in a depth-$d$ circuit.

Recall from Fact~\ref{fact:size-Bck-light} that the set $\tilde{S}_i$ of qubits corresponding to the backward lightcone of qubit $i$ in a depth-$d$ circuit satisfies $1 \leq |\tilde{S}_i| \leq 2^d$.
We take the union over all possible values of $|\tilde{S}_i|$ to define
\begin{equation}
    \mathcal{S}^*_{\mathrm{obs}} := \bigcup_{k = 1}^{2^d} \mathcal{S}_{\mathrm{obs}}(1, \{1, \ldots, k\}).
\end{equation}
Because $2^d = \mathcal{O}(1)$ and for all $k = \mathcal{O}(1)$, $|\mathcal{S}_{\mathrm{obs}}(1, \{1, \ldots, k\}| = \mathcal{O}(1)$, we have $|\mathcal{S}^*_{\mathrm{obs}}| = \mathcal{O}(1)$.
We define the minimum distance between every pair of distinct observables in $\mathcal{S}^*_{\mathrm{obs}}$ as follows,
\begin{equation}
    \varepsilon^{\mathrm{dist}} := \min_{O_1 \neq O_2 \in \mathcal{S}^*_{\mathrm{obs}}} \norm{O_1 - O_2}_\infty.
\end{equation}
The minimum distance $\varepsilon^{\mathrm{dist}}$ depends on the depth $d = \mathcal{O}(1)$ and the finite gate set $\mathcal{G}$ with $|\mathcal{G}| = \mathcal{O}(1)$, so $\varepsilon^*$ is a constant independent of the system size $n$ and failure probability $\delta$.
We also define the minimum distance to an observable with a strictly smaller support.
\begin{equation}
    \varepsilon^{\mathrm{supp}} := \min_{O_1 \in \mathcal{S}^*_{\mathrm{obs}}} \min_{\substack{O_2, \,\, \mathrm{such}\,\mathrm{that}\\ \mathrm{supp}(O_2) \subseteq \mathrm{supp}(O_1)\\ \mathrm{supp}(O_2) \neq \mathrm{supp}(O_1)}} \norm{O_1 - O_2}_\infty.
\end{equation}
Because the support of $O_2$ is strictly contained in the support of $O_1$, we have $\norm{O_1 - O_2}_\infty > 0$. And since $|\mathcal{S}^*_{\mathrm{obs}}| = \mathcal{O}(1)$, we have $\varepsilon^{\mathrm{supp}}$ is a constant independent of $n$ and $\delta$.

Let $\tilde{\varepsilon} = \min(\varepsilon^{\mathrm{dist}}, \varepsilon^{\mathrm{supp}}) / 3$ in Eq.~\eqref{eq:Oip-error-learned-finite-gate-set}, and define $\hat{S}_i := \{i\} \cup \mathrm{supp}(\hat{O}_{i, P})$.
Consider any permutation $\Pi_{i, \hat{S}_i}$ over $n$ qubits that implements the following permutation mapping,
\begin{equation}
    1 \rightarrow_{\Pi_{i, \hat{S}_i}} i, \quad \{1, \ldots, |\hat{S}_i|\} \rightarrow_{\Pi_{i, \hat{S}_i}} \hat{S}_i.
\end{equation}
We consider the following observable
\begin{equation}
    O^*_{i, P} := \Pi_{i, \hat{S}_i} \left( \argmin_{O \in \mathcal{S}^*_{\mathrm{obs}}} \norm{\Pi_{i, \hat{S}_i}^{-1} \hat{O}_{i, P} \Pi_{i, \hat{S}_i}^{-1} - O}_\infty \right) \Pi_{i, \hat{S}_i}.
\end{equation}
Because $|\mathcal{S}^*_{\mathrm{obs}}| = \mathcal{O}(1)$ and the dimension of $O \in \mathcal{S}^*_{\mathrm{obs}}$ is a constant, the brute-force minimum over $\mathcal{S}^*_{\mathrm{obs}}$ takes $\mathcal{O}(1)$ time.
Because there are $3n$ observables $O^*_{i, P}$, the computational time to find all $3n$ observables $O^*_{i, P}$ is $\mathcal{O}(n)$.
The following lemma shows that $O^*_{i, P}$ is exactly equal to the desired Heisenberg-evolved Pauli observable $U^\dagger P_i U$.

\begin{lemma}[Exact reconstruction] \label{lem:exactHeisenberg}
    Given the definitions above, with probability at least $1-\delta$, we have $O^*_{i, P} = U^\dagger P_i U$ for all qubits $i$ and Pauli observable $P$.
\end{lemma}
\begin{proof}
    We condition on the event that Eq.~\eqref{eq:Oip-error-learned-finite-gate-set} is true, which happens with probability at least $1 - \delta$.
    Recall that $\mathrm{supp}(\hat{O}_{i, P}) \subseteq \mathrm{supp}(U^\dagger P_i U)$ and $\norm{\hat{O}_{i, P} - U^\dagger P_i U}_\infty \leq \tilde{\varepsilon} \leq \varepsilon^{\mathrm{supp}} / 3$. From the definition of $\varepsilon^{\mathrm{supp}}$, we have $\mathrm{supp}(\hat{O}_{i, P}) = \mathrm{supp}(U^\dagger P_i U)$. Hence,
    \begin{equation}
        \hat{S}_i = \left(\{i\} \cup \mathrm{supp}(U^\dagger P_i U)\right) \subseteq S_i,
    \end{equation}
    where $S_i$ is the set of qubits corresponding to the backward lightcone of qubit $i$ in circuit $C$.
    Consider any permutation $\Pi_{i, \hat{S}_i, S_i}$ over $n$ qubits that is equal to $\Pi_{i, \hat{S}_i}$ for inputs $1, \ldots, |\hat{S}_i|$ and implements the following permutation mapping,
    \begin{equation}
        \left\{ |\hat{S}_i|+1, \ldots, |S_i| \right\} \rightarrow_{\Pi_{i, \hat{S}_i, S_i}} S_i \setminus \hat{S}_i,
    \end{equation}
    and $\Pi_{i, \tilde{S}_i, S_i}$ acts as identity on the $m$ ancilla qubits.
    Because $\mathrm{supp}(U^\dagger P_i U) \subseteq \hat{S}_i$, we have
    \begin{align}
        \Pi_{i, \hat{S}_i}^{-1} U^\dagger P_i U \Pi_{i, \hat{S}_i}^{-1} &= \Pi_{i, \hat{S}_i}^{-1} (I_n \otimes \bra{0^m}) C^\dagger (P_i \otimes I_m) C (I_n \otimes \ket{0^m}) \Pi_{i, \hat{S}_i}^{-1} \nonumber \\
        &= (I_n \otimes \bra{0^m}) \left(\Pi_{i, \hat{S}_i, S_i}^{-1} C^\dagger \Pi_{i, \hat{S}_i, S_i}^{-1} \right) P_1 \left(\Pi_{i, \hat{S}_i, S_i}^{-1} C \Pi_{i, \hat{S}_i, S_i}^{-1} \right) (I_n \otimes \ket{0^m}). \label{eq:PiUdagPiUPi}
    \end{align}
    By the definition of the permutation $\Pi_{i, \hat{S}_i, S_i}^{-1}$, $\{1, \ldots, |S_i|\}$ is the set of qubits corresponding to the backward lightcone of qubit $1$ in the circuit $\Pi_{i, \hat{S}_i, S_i}^{-1} C \Pi_{i, \hat{S}_i, S_i}^{-1}$.
    As a result, we have
    \begin{align}
        O^* &:= (I_n \otimes \bra{0^m}) \left(\Pi_{i, \hat{S}_i, S_i}^{-1} C^\dagger \Pi_{i, \hat{S}_i, S_i}^{-1} \right) P_1 \left(\Pi_{i, \hat{S}_i, S_i}^{-1} C \Pi_{i, \hat{S}_i, S_i}^{-1} \right) (I_n \otimes \ket{0^m})\\
        &\in \mathcal{S}_{\mathrm{obs}}(1, \{1, \ldots, |S_i|\}) \subseteq \mathcal{S}^*_{\mathrm{obs}}.
    \end{align}
    The last $\subseteq$ follows from the fact that $|S_i| \leq 2^d$ in Fact~\ref{fact:size-Bck-light}.
    We can use Eq.~\eqref{eq:PiUdagPiUPi} and
    \begin{equation}
        \norm{\hat{O}_{i, P} - U^\dagger P_i U}_\infty \leq \tilde{\varepsilon} \leq \varepsilon^{\mathrm{dist}} / 3
    \end{equation}
    to see that
    \begin{equation}
        \norm{\Pi_{i, \hat{S}_i}^{-1} \hat{O}_{i, P} \Pi_{i, \hat{S}_i}^{-1} - O^*}_\infty \leq \varepsilon^{\mathrm{dist}} / 3.
    \end{equation}
    For any $O \in \mathcal{S}^*_{\mathrm{obs}}$ with $O \neq O^*$, we have $\norm{O - O^*}_\infty \geq \varepsilon^{\mathrm{dist}}$.
    By the triangle inequality, we have
    \begin{equation}
        \norm{\Pi_{i, \hat{S}_i}^{-1} \hat{O}_{i, P} \Pi_{i, \hat{S}_i}^{-1} - O}_\infty \geq \norm{O - O^*}_\infty - \norm{\Pi_{i, \hat{S}_i}^{-1} \hat{O}_{i, P} \Pi_{i, \hat{S}_i}^{-1} - O^*}_\infty \geq 2  \varepsilon^{\mathrm{dist}} / 3.
    \end{equation}
    Together, we can show that $O^*$ is the unique global minimum,
    \begin{equation}
        O^* = \argmin_{O \in \mathcal{S}^*_{\mathrm{obs}}} \norm{\Pi_{i, \hat{S}_i}^{-1} \hat{O}_{i, P} \Pi_{i, \hat{S}_i}^{-1} - O}_\infty.
    \end{equation}
    Using Eq.~\eqref{eq:PiUdagPiUPi} again shows that
    \begin{equation}
         O^*_{i, P} = \Pi_{i, \hat{S}_i} \left( \argmin_{O \in \mathcal{S}^*_{\mathrm{obs}}} \norm{\Pi_{i, \hat{S}_i}^{-1} \hat{O}_{i, P} \Pi_{i, \hat{S}_i}^{-1} - O}_\infty \right) \Pi_{i, \hat{S}_i} = U^\dagger P_i U.
    \end{equation}
    This concludes the proof.
\end{proof}

From Lemma~\ref{lem:lightcone-super-shallow}, we can characterize the support of $O^*_{i, P} = U^\dagger P_i U$ to apply Lemma~\ref{lem:sewing-const-depth}.
Lemma~\ref{lem:sewing-const-depth} shows that there exists an ordering for sewing the Heisenberg-evolved Pauli observables $O^*_{i, P} = U^\dagger P_i U$ to create $U_{\mathrm{sew}}(\{O^*_{i, P}\}_{i, P})$ given in Definition~\ref{def:sewing-Hei-obs}, such that $U_{\mathrm{sew}}(\{O^*_{i, P}\}_{i, P})$ can be implemented by a constant-depth quantum circuit.
Under the event that $O^*_{i, P} = U^\dagger P_i U$ (think of $O^*_{i, P}$ as $0$-approximate Heisenberg-evolved Pauli observable $P$ on qubit $i$ under $U$) for all Pauli observable $P$ and qubit $i$, Lemma~\ref{lem:sew-form-Heisenberg} shows that
\begin{equation}
    U_{\mathrm{sew}}(\{O^*_{i, P}\}_{i, P}) = U \otimes U^\dagger.
\end{equation}
Finally, define an $n$-qubit channel $\hat{\mathcal{E}}$ as follows,
\begin{equation}
    \hat{\mathcal{E}}(\rho) := \Tr_{> n}\left(\mathcal{U}_{\mathrm{sew}}(\{O^*_{i, P}\}_{i, P})(\rho \otimes \ketbra{0^n}{0^n})\right),
\end{equation}
which can be implemented as a constant-depth $2n$ qubits circuit.
Using Lemma~\ref{lem:exactHeisenberg}, we have
\begin{equation}
    \hat{\mathcal{E}} = \mathcal{U}
\end{equation}
with probability at least $1 - \delta$.
This concludes the proof of Theorem~\ref{thm:shallow-SU4-gates}.

\subsection{Learning geometrically-local shallow circuits (Proof of Theorem~\ref{thm:geo-SU4-gates})} \label{sec:geo-SU4-gates-proof}

We present the algorithm for learning an unknown geometrically-local shallow quantum circuit $U$.
We separate the proof into two-qubit gates over $\mathrm{SU}(4)$ and over a finite gate set.

\subsubsection{Arbitrary $\mathrm{SU}(4)$ gates}

We present the algorithm for learning an unknown geometrically-local shallow quantum circuit $U$ over any two-qubit gate in $\mathrm{SU}(4)$.
The algorithm uses the randomized measurement dataset $\mathcal{T}_U(N)$.
The key ideas are constructing a superset of the support of the Heisenberg-evolved Pauli observables using Lemma~\ref{lem:lightcone-property-geo}, finding the Heisenberg-evolved Pauli observables for every qubit using Lemma~\ref{lem:learn-few-body-obs-kno-supp}, and sewing the Heisenberg-evolved Pauli observables together using Definition~\ref{def:sewing-Hei-obs} and Lemma~\ref{lem:sew-form-Heisenberg}.

Consider the lightcones $L_d(i)$ for each qubit $i$ with depth $d$ as given in Definition~\ref{def:light-cone-geo}. We have the following lemma for characterizing the properties of $L_d(i)$.

\begin{lemma}[Properties of lightcones] \label{lem:lightcone-property-geo}
Given a geometry over $n$ qubits represented by a graph $G = (V, E)$ with a degree $\kappa = \mathcal{O}(1)$, a depth-$d$ geometrically-local circuit $U$ as given in Definition~\ref{def:geo-local-circuit} with $d = \mathcal{O}(1)$, and the lightcones $L_d(i)$ for each qubit $i$ with depth $d$ as given in Definition~\ref{def:light-cone-geo}. For each qubit $i$, we have
\begin{equation}
    \mathrm{supp}\left( U^\dagger P_i U \right) \subseteq L_d(i),
\end{equation}
for any Pauli operator $P \in \{X, Y, Z\}$. Furthermore, $L_d(i)$ is geometrically local (see Definition~\ref{def:geo-local-set}), $|L_d(i)| = \mathcal{O}(1)$, $L_d(i)$ is known, and the number of qubits $j$ such that $L_d(i) \cap L_d(j) \neq \varnothing$ is at most a constant.
\end{lemma}
\begin{proof}
Because $U$ is of depth $d$ and $P_i$ acts only on qubit $i$, $U^\dagger P_i U$ only acts only on qubits that are distance $d$ away from qubit $i$ according to the graph $G$.
By the definition of $L_d(i)$, we have $\mathrm{supp}\left( U^\dagger P_i U \right) \subseteq L_d(i)$.
Recall that $|L_d(i)| \leq (\kappa + 1)^d = \mathcal{O}(1)$.
Furthermore, since $G$ is known, $L_d(i)$ is known.
Now, consider a qubit $j$ such that $L_d(i) \cap L_d(j) \neq \varnothing$.
This condition shows that qubit $j$ must be of distance at most $2d$ from qubit $i$ in the graph $G$.
Hence, the number of such $j$ is bounded above by $(\kappa + 1)^{2d} = \mathcal{O}(1)$.
This concludes the proof of the lemma.
\end{proof}

Lemma~\ref{lem:lightcone-property-geo} shows that $L_d(i)$ is a geometrically-local set, $|L_d(i)| = \mathcal{O}(1)$, $L_d(i)$ is known, and the number of qubits $j$ such that $L_d(i) \cap L_d(j) \neq \varnothing$ is at most a constant.

Recall that we can use Lemma~\ref{lem:reusing-RMdata} to constructing $\mathcal{T}_{U^\dagger P_i U}(N), \forall i, P$ from the classical dataset $\mathcal{T}_{U}(N)$ given in Definition~\ref{def:random-measure-data}.
Because $|L_d(i)| = \mathcal{O}(1)$ and $L_d(i)$ is known, from Lemma~\ref{lem:learn-few-body-obs-kno-supp}, with a dataset size of
\begin{equation}
    N = \mathcal{O}\left(\frac{n^2 \log(3n / \delta)}{\varepsilon^2}\right),
\end{equation}
we can use $\mathcal{T}_{U^\dagger P_i U}(N), \forall i, P$ constructed from $\mathcal{T}_{U}(N)$ to learn $\hat{O}_{i, P}, \forall i, P$ such that, with probability at least $1 - \delta$, for all $i \in \{1, \ldots, n\}$ and Pauli observable $P \in \{X, Y, Z\}$, we have
\begin{equation} \label{eq:Oip-error-learned}
    \norm{\hat{O}_{i, P} - U^\dagger P_i U}_\infty \leq \frac{\varepsilon}{6n} \quad \mbox{and} \quad \mathrm{supp}(\hat{O}_{i, P}) \subseteq \mathrm{supp}\left(U^\dagger P_i U \right) \subseteq L_d(i).
\end{equation}
The computational time for learning all $\hat{O}_{i, P}$ is $\mathcal{O}(n^3 \log(n / \delta) / \varepsilon^2)$.

We now utilize Lemm~\ref{lem:sewing-const-depth} to sew the learned observables into a geometrically-local constant-depth quantum circuit.
To use the lemma, we note the following relations from Eq.~\eqref{eq:Oip-error-learned},
\begin{equation}
    A(i) := \bigcup_{P} \mathrm{supp}(\hat{O}_{i, P}) \subseteq \bigcup_{P} \mathrm{supp}(U^\dagger P_i U) \subseteq L_d(i).
\end{equation}
Because $L_d(i)$ is a geometrically-local set, $|L_d(i)| = \mathcal{O}(1)$ and the number of qubits $j$ such that $L_d(i) \cap L_d(j) \neq \varnothing$ is at most a constant, we have $A(i)$ is a geometrically-local set, $|A(i)| = \mathcal{O}(1)$ and the number of qubits $j$ such that $A(i) \cap A(j) \neq \varnothing$ is at most a constant.
Hence Lemma~\ref{lem:sewing-const-depth} given above shows that we can find an implementation of $U_{\mathrm{sew}}(\{\hat{O}_{i, P}\}_{i, P})$ as a geometrically-local constant-depth $2n$-qubit circuit in time $\mathcal{O}(n)$.
Given Eq.~\eqref{eq:Oip-error-learned}, we can use Lemma~\ref{lem:sew-form-Heisenberg} on the form of the sewed Heisenberg-evolved Pauli observables to yield
\begin{equation}
    \norm{ \mathcal{U}_{\mathrm{sew}}(\{\hat{O}_{i, P}\}_{i, P}) - \mathcal{U} \otimes \mathcal{U}^\dagger }_\diamond \leq \varepsilon.
\end{equation}
Finally, define an $n$-qubit channel $\hat{\mathcal{E}}$ as follows,
\begin{equation}
    \hat{\mathcal{E}}(\rho) := \Tr_{> n}\left(\mathcal{U}_{\mathrm{sew}}(\{\hat{O}_{i, P}\}_{i, P})(\rho \otimes \ketbra{0^n}{0^n})\right),
\end{equation}
which can be implemented as a geometrically-local constant-depth quantum circuit over $2n$ qubits.
Because Eq.~\eqref{eq:Oip-error-learned} holds with probability at least $1 - \delta$, we have
\begin{equation}
    \norm{\hat{\mathcal{E}} - \mathcal{U}}_\diamond \leq \varepsilon
\end{equation}
with probability at least $1 - \delta$.
This concludes the proof of the first part of Theorem~\ref{thm:geo-SU4-gates}.

\subsubsection{Finite gate sets}

We present the algorithm for learning an unknown geometrically-local shallow quantum circuit $U$ over a finite gate set.
Let the depth of the unknown shallow quantum circuit be $d = \mathcal{O}(1)$ and the finite gate set be $\mathcal{G}$ with $|\mathcal{G}| = \mathcal{O}(1)$.
The algorithm uses the randomized measurement dataset $\mathcal{T}_U(N)$.
The algorithm constructs a superset of the support of the Heisenberg-evolved Pauli observables using Lemma~\ref{lem:lightcone-property-geo}, finds the Heisenberg-evolved Pauli observables for every qubit exactly using Lemma~\ref{lem:learn-few-body-obs-kno-supp} and the information about the finite gate set $\mathcal{G}$, and sew the Heisenberg-evolved Pauli observables together using Definition~\ref{def:sewing-Hei-obs} and Lemma~\ref{lem:sew-form-Heisenberg}.

Consider the lightcones $L_d(i)$ for each qubit $i$ with depth $d$ as given in Definition~\ref{def:light-cone-geo}.
Lemma~\ref{lem:lightcone-property-geo} shows that $L_d(i)$ is a geometrically-local set, $|L_d(i)| = \mathcal{O}(1)$, $L_d(i)$ is known, and the number of qubits $j$ such that $L_d(i) \cap L_d(j) \neq \varnothing$ is at most a constant.
The algorithm and the proof proceed similarly to the case of having arbitrary two-qubit gates in $\mathrm{SU}(4)$.
The main difference is in defining the following set $\mathcal{S}_{\mathrm{obs}}(P_i)$ for all $i \in \{1, \ldots, n\}$ and Pauli observable $P \in \{X, Y, Z\}$,
\begin{equation}
    \mathcal{S}_{\mathrm{obs}}(P_i) := \left\{ U^\dagger P_i U \,\, | \,\, U \,\, \text{is a geometrically-local depth-$d$ circuit over the gate set} \,\, \mathcal{G}  \right\}.
\end{equation}
Because $|\mathcal{G}| = \mathcal{O}(1)$ and $d = \mathcal{O}(1)$, the set $\mathcal{S}_{\mathrm{obs}}(P_i)$ contains a constant number of observables that only act on qubits in $L_d(i)$.
We can define the minimum distance to be
\begin{equation}
    \varepsilon_0(P_i) := \min\left\{ \,\, \norm{O_1 - O_2}_\infty \,\, | \,\, O_1 \neq O_2 \in \mathcal{S}_{\mathrm{obs}}(P_i) \,\, \right\} = \Omega(1).
\end{equation}
We also define $\varepsilon_0 = \min_{i, P} \varepsilon_0(P_i) = \Omega(1)$, which is a constant.

Recall that we can use Lemma~\ref{lem:reusing-RMdata} to constructing $\mathcal{T}_{U^\dagger P_i U}(N), \forall i, P$ from the classical dataset $\mathcal{T}_{U}(N)$ given in Definition~\ref{def:random-measure-data}.
Because $|L_d(i)| = \mathcal{O}(1)$ and $L_d(i)$ is known, from Lemma~\ref{lem:learn-few-body-obs-kno-supp}, with a dataset size of
\begin{equation}
    N = \mathcal{O}\left(\frac{\log(3n / \delta)}{\varepsilon_0^2}\right) = \mathcal{O}(\log(n / \delta)),
\end{equation}
we can use $\mathcal{T}_{U^\dagger P_i U}(N), \forall i, P$ constructed from $\mathcal{T}_{U}(N)$ to learn $\hat{O}_{i, P}, \forall i, P$ such that, with probability at least $1 - \delta$, for all $i \in \{1, \ldots, n\}$ and Pauli observable $P \in \{X, Y, Z\}$, we have
\begin{equation}
    \norm{\hat{O}_{i, P} - U^\dagger P_i U}_\infty \leq \frac{\varepsilon_0}{3} \quad \mbox{and} \quad \mathrm{supp}(\hat{O}_{i, P}) \subseteq \mathrm{supp}\left(U^\dagger P_i U \right) \subseteq L_d(i).
\end{equation}
The computational time for learning all $\hat{O}_{i, P}$ is $\mathcal{O}(n \log(n / \delta) / \varepsilon_0^2) = \mathcal{O}(n \log (n / \delta))$.
Because $U^\dagger P_i U \in \mathcal{S}_{\mathrm{obs}}(P_i)$ only has a constant number of possibilities, we can find
\begin{equation}
    O^*_{i, P} := \argmin_{O \in \mathcal{S}_{\mathrm{obs}}(P_i)} \norm{O - \hat{O}_{i, P}}_\infty
\end{equation}
in time $\mathcal{O}(n)$.
Because the pairwise distance in $\mathcal{S}_{\mathrm{obs}}(P_i)$ is at least $\varepsilon_0$ and $U^\dagger P_i U \in \mathcal{S}_{\mathrm{obs}}(P_i)$,
\begin{equation} \label{eq:Oip-star-noerror-learned}
    O^*_{i, P} = U^\dagger P_i U, \quad \forall i \in \{1, \ldots, n\}, P \in \{X, Y, Z\}
\end{equation}
with probability at least $1 - \delta$.

We now utilize Lemm~\ref{lem:sewing-const-depth} to sew the learned observables into a geometrically-local constant-depth quantum circuit.
To use the lemma, we note the following relations from Eq.~\eqref{eq:Oip-error-learned},
\begin{equation}
    A(i) := \bigcup_{P} \mathrm{supp}(O^*_{i, P}) \subseteq \bigcup_{P} \mathrm{supp}(U^\dagger P_i U) \subseteq L_d(i).
\end{equation}
Because $L_d(i)$ is a geometrically-local set, $|L_d(i)| = \mathcal{O}(1)$ and the number of qubits $j$ such that $L_d(i) \cap L_d(j) \neq \varnothing$ is at most a constant, we have $A(i)$ is a geometrically-local set, $|A(i)| = \mathcal{O}(1)$ and the number of qubits $j$ such that $A(i) \cap A(j) \neq \varnothing$ is at most a constant.
Hence Lemma~\ref{lem:sewing-const-depth} given above shows that we can find an implementation of $U_{\mathrm{sew}}(\{O^*_{i, P}\}_{i, P})$ as a geometrically-local constant-depth $2n$-qubit circuit in time $\mathcal{O}(n)$.
Given Eq.~\eqref{eq:Oip-star-noerror-learned}, we can use Lemma~\ref{lem:sew-form-Heisenberg} on the form of the sewed Heisenberg-evolved Pauli observables to yield
\begin{equation}
    \mathcal{U}_{\mathrm{sew}}(\{O^*_{i, P}\}_{i, P}) = \mathcal{U} \otimes \mathcal{U}^\dagger.
\end{equation}
Finally, define an $n$-qubit channel $\hat{\mathcal{E}}$ as follows,
\begin{equation}
    \hat{\mathcal{E}}(\rho) := \Tr_{> n}\left(\mathcal{U}_{\mathrm{sew}}(\{O^*_{i, P}\}_{i, P})(\rho \otimes \ketbra{0^n}{0^n})\right),
\end{equation}
which can be implemented as a geometrically-local constant-depth quantum circuit over $2n$ qubits.
Because Eq.~\eqref{eq:Oip-error-learned} holds with probability at least $1 - \delta$, we have
\begin{equation}
    \hat{\mathcal{E}} = \mathcal{U}
\end{equation}
with probability at least $1 - \delta$.
This concludes the proof of Theorem~\ref{thm:geo-SU4-gates}.

\subsection{Learning shallow circuits on $k$-dimensional lattice with optimized circuit depth (Proof of Theorem~\ref{thm:geo-kD-lattice-optimized})}
\label{sec:geo-kD-lattice-optimized-proof}

Here we develop an approach to optimize the depth of the learned circuit. The main idea is to design a coloring scheme for the $k$-dimensional lattice with the fewest colors possible, such that gates supported on the same color can be implemented simultaneously.

\begin{definition}[$k+1$-coloring of $k$-dimensional lattice with distance $R$] Consider a graph representing a $k$-dimensional lattice (Fig.~\ref{fig:geometry}(a) shows $k=2$). Each vertex is assigned a color, and the entire lattice is divided into many small regions with different colors. A $k+1$-coloring of $k$-dimensional lattice with distance $R$ satisfies the following properties:
\begin{enumerate}
    \item There are $k+1$ colors in total;
    \item Each small region has constant size;
    \item The distance between two regions with the same color is at least $R$.
\end{enumerate}
\end{definition}

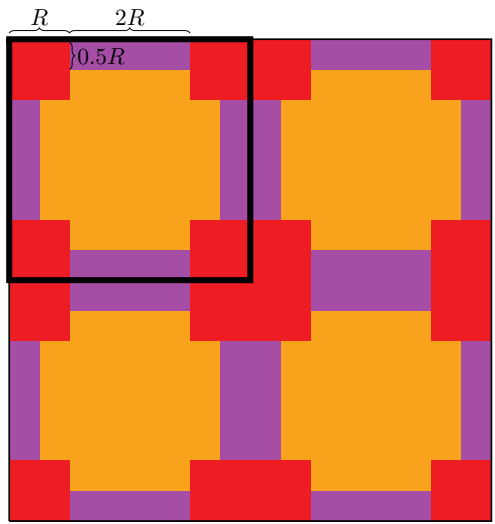
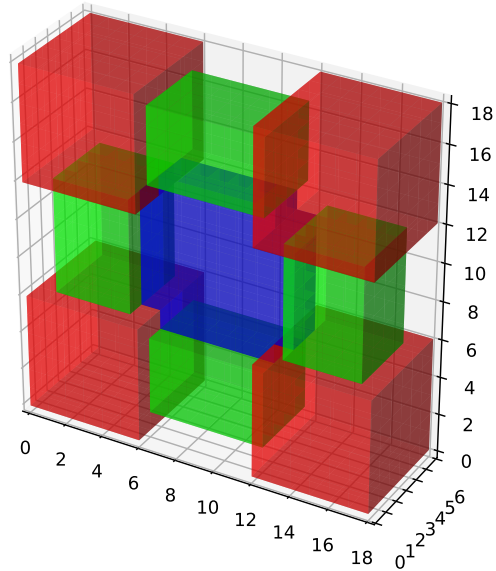
\begin{figure}[t]
    \centering
    \begin{subfigure}[b]{0.4\textwidth}
    \resizebox{\textwidth}{!}{
    \begin{tikzpicture}
  \def\size{1} % Size of each square in centimeters

  % Define colors
  \colorlet{color1}{Red}   % Red
  \colorlet{color2}{YellowOrange}   % Green
  \colorlet{color4}{violet!70}   % Blue
  \colorlet{color3}{Lavender} % Yellow

    \usetikzlibrary{calc}
  % Loop to draw the squares
  \fill[color=color2] (0,0) rectangle (8,8);

\fill[color=color4] (0,0) rectangle (0.5,8);
\fill[color=color4] (3.5,0) rectangle (4.5,8);
\fill[color=color4] (7.5,0) rectangle (8,8);
\fill[color=color4] (0,0) rectangle (8,0.5);
\fill[color=color4] (0,3.5) rectangle (8,4.5);
\fill[color=color4] (0,7.5) rectangle (8,8);

  \foreach \x in {0,4,8}{
    \foreach \y in {0,4,8}{
      \fill[color=color1] ({max(\x*\size-\size,0)}, {max(\y*\size-\size,0)}) rectangle ({min(\x*\size+\size,8)}, {min(\y*\size+\size,8)});
    }
  }

  \draw[black,thick] (0,0) -- (8,0) -- (8,8) -- (0,8) -- (0,0);
  \draw[black,line width=1mm] (-0.05,8) -- (4,8) -- (4,4) -- (0,4) -- (0,8);
\draw [decorate,
    decoration = {brace}] (0,8.1) --  (1,8.1) node[pos=0.5,anchor=south]{$R$};
    \draw [decorate,
    decoration = {brace}] (1,8.1) --  (3,8.1) node[pos=0.5,anchor=south]{$2R$};
    \draw [decorate,
    decoration = {brace}] (1,7.95) -- (1,7.5) node[pos=0.5,anchor=west]{$0.5R$};
\end{tikzpicture}}
\caption{2D}
\end{subfigure}
\quad\quad\quad\quad\quad
\begin{subfigure}[b]{0.4\textwidth}
    \includegraphics[width=\textwidth]{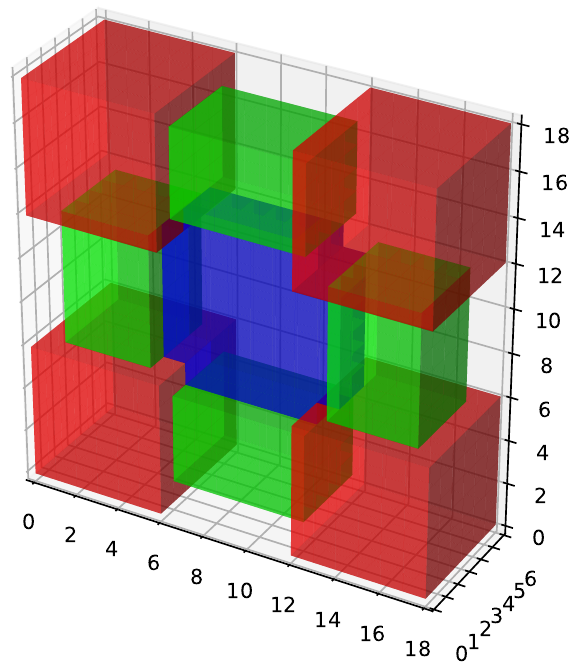}
\caption{3D}
\end{subfigure}
    \caption{A coloring of $k$-dimensional lattice with $k+1$ colors, where different regions of the same color are separated by distance at least $R$. (a) A coloring of 2-dimensional lattice. (b) A coloring of 3-dimensional lattice (the fourth color is not shown).
    }
    \label{fig:latticecoloring}
\end{figure}

Here we give a construction of the above coloring (see Fig.~\ref{fig:latticecoloring}). Similar approaches have been used in e.g.~\cite{Brandao2019finite}, although explicit constructions in 3D or above are not provided.
The construction is based on ``fattening'' different $t$-cells in the lattice, from small to large $t$.\footnote{We thank Jeongwan Haah for teaching this argument at PCMI 2023 Graduate Summer School.} Consider a $k$-dimensional cube of length $2kR$ (the volume of the cube is $(2kR)^k$). Then we do the following:
\begin{itemize}
    \item Fatten each 0-cell (vertices) to length $kR$, assign color 1.
    \item Fatten each 1-cell (edges) to length $(k-1)R$, assign color 2.
    \item Fatten each 2-cell (faces) to length $(k-2)R$, assign color 3.
    \item \dots
    \item Fill in the remaining $k$-cell with color $k+1$.
\end{itemize}
This is repeated in a translation-invariant way across the entire lattice.

This construction is illustrated in Fig.~\ref{fig:latticecoloring} for $k=2,3$. First, consider $k=2$. A 2-dimensional square of size $4R\times 4R$ is shown in the top left corner (thick black box) of Fig.~\ref{fig:latticecoloring}(a). In the first step, we fatten each of the 4 vertices into red squares of size $2R\times 2R$. Only a quarter of each red square remains within the original square. Next, we fatten each of the 4 edges into purple rectangles of size $R\times 2R$. This can be viewed as ``growing'' the edge until it has thickness $R$, but the regions that were colored red remain unchanged. Note the fact that the purple edges have a thickness of $R$, while the red vertices have a thickness of $2R$. This is crucial as it ensures that different purple regions are separated by a distance of at least $R$. Finally, the remaining regions are colored orange. Note that different orange regions are also separated by a distance of at least $R$ due to the thickness of the purple edges.

The coloring of 3-dimensional lattices is shown in Fig.~\ref{fig:latticecoloring}(b). Here we assign colors to a 3-dimensional cube of size $6R\times 6R\times 6R$, and Fig.~\ref{fig:latticecoloring}(b) illustrates one of the six faces of that cube, which is the result of fattening the red vertices, green edges, and the blue face (the final coloring of the 3-cell is not shown in the figure). The thickness of the red vertices is larger than the thickness of the green edges, which guarantees that different green edges are separated by distance $R$.
Similarly, the decrease in the thickness of the blue faces relative to the green edges guarantees the separation of different blue faces.

Choose $R=3d$ in the above coloring scheme, and suppose the system is divided into $L$ small regions $A_1, \ldots, A_L$ ($\sum_i |A_i|=n$). Two regions $A_i$, $A_j$ that have the same color are separated by distance at least $3d$. Let $A_i'$ be the ancilla system associated with $A_i$ (see Fig.~\ref{fig:geometry}), and let $S_{A_i}$ be the SWAP operator across $A_i$ and $A_i'$. Let $S=\prod_{i=1}^L S_{A_i}$ be the global swap between system and ancilla.
We are now ready to describe the learning algorithm. We separate the proof into two-qubit gates over $\mathrm{SU}(4)$ and over a finite gate set.

\subsubsection{Arbitrary $\mathrm{SU}(4)$ gates}

The learning algorithm proceeds in the same way as in Theorem~\ref{thm:geo-SU4-gates}; the only difference is that we need to learn Heisenberg-evolved Pauli operator $U^\dag P U$ for $P$ supported on each small regions in the coloring scheme instead of on each of the single qubits.

Our goal is to learn to implement the unitary
\begin{equation}
    U\otimes U^\dag = S\left[\prod_{i=1}^L (U^\dag\otimes I)S_{A_i}(U\otimes I)\right].
\end{equation}
The algorithm learns each of the operators $W_{A_i}:=(U^\dag\otimes I)S_{A_i}(U\otimes I)$ and then multiply them together, followed by the global swap. The key idea to optimize the circuit depth of the learned circuit is to utilize the coloring scheme in the following sense:
\begin{lemma}[Disjointness of supports]
    Let $A_i$, $A_j$ be two regions with the same color. Then $W_{A_i}$ and $W_{A_j}$ have disjoint support.
\end{lemma}
\begin{proof}
    Recall that the operator $W_{A_i}$ is supported on $L(A_i)\cup A_i'$, where $L(A_i)$ is the lightcone of $A_i$ according to Definition~\ref{def:light-cone-geo}. Therefore, $W_{A_i}$ does not overlap with $W_{A_j}$ when the lightcones $L(A_i)$ and $L(A_j)$ do not overlap. The coloring scheme has the property that $A_i$, $A_j$ are separated by distance at least $3d$. Note that the lightcone of a region spreads the region by distance $d$. This implies that $L(A_i)$ and $L(A_j)$ are still separated by distance at least $d$ and therefore do not overlap.
\end{proof}
Using the above lemma, we can construct the learned circuit by applying the learned operators $\{W_{A_i}\}$ with the same color simultaneously.
\begin{lemma}\label{lemma:implementationordering}
    There is an implementation of $U\otimes U^\dag$ via applying the operators $\{W_{A_i}\}$ in an appropriate order, such that the total circuit depth is $(k+1)(2d+1)+1$.
\end{lemma}
\begin{proof}
    We would like to implement
    \begin{equation}
        U\otimes U^\dag = S\prod_{i=1}^L W_{A_i}.
    \end{equation}
    Note that the operators $\{W_{A_i}\}$ pairwise commute, and we apply them in the following order: for each color $j\in\{1,2,\dots,k+1\}$, apply all operators $W_{A_i}$ that has color $j$ simultaneously. Finally, apply the global swap $S$.

    Note that by definition, $W_{A_i}=(U^\dag\otimes I)S_{A_i}(U\otimes I)$ can be viewed as a depth-$(2d+1)$ circuit acting on $L(A_i)\cup A_i'$. The total circuit depth is therefore $(k+1)(2d+1)+1$ (the final $+1$ comes from the global swap).
\end{proof}

The learning algorithm has two steps: learning and compiling.
\begin{enumerate}
    \item (Learning) Learn an approximate classical description $\hat W_{A_i}$ for each $W_{A_i}$, such that $\|\hat W_{A_i}-W_{A_i}\|_\infty\leq\varepsilon_1$ for all $i$ with high probability.
    \item (Compiling) Compile the learned unitaries $\hat W_{A_i}$ from step one into depth-$(2d+1)$ circuits $\hat W_{A_i}'$, such that $\|\hat W_{A_i}-\hat W_{A_i}'\|_\infty\leq2\varepsilon_1$ for all $i$.
\end{enumerate}
The diamond distance between the learned circuit and the true circuit is at most $3L\varepsilon_1\leq 3n\varepsilon_1$.\\

\noindent\textbf{Step 1: Learning.} The goal is to learn an approximation $\hat O_{i,P_{A_i}}$ of each operator $U^\dag P_{A_i}U$, such that the following,
\begin{equation}\label{eq:approxheisenberg}
    \left\|\hat O_{i,P_{A_i}}-U^\dag P_{A_i}U\right\|_\infty \leq \frac{\varepsilon_1}{2^{|A_i|+1}}, \quad\forall i\in\{1,2,\dots,L\}, \quad P_{A_i}\in\{I,X,Y,Z\}^{|A_i|},
\end{equation}
holds with probability at least $1-\delta$.

Using the fact that $S_{A_i}=\frac{1}{2^{|A_i|}}\sum_{P\in\{I,X,Y,Z\}^{|A_i|}}P\otimes P$, we have
\begin{equation}
    W_{A_i}=\frac{1}{2^{|A_i|}}\sum_{P\in\{I,X,Y,Z\}^{|A_i|}}U^\dag P U\otimes P.
\end{equation}
Meanwhile, let
\begin{equation}
    \hat W_{A_i}:=\mathrm{Proj_U}\left(\frac{1}{2^{|A_i|}}\sum_{P\in\{I,X,Y,Z\}^{|A_i|}}\hat O_{i,P_{A_i}}\otimes P_{A_i}\right).
\end{equation}
From the lattice coloring scheme, we have $|L(A_i)| + |A_i| \leq 2 (8 kd)^k$.
Hence, using Corollary~\ref{cor:unitary-synthesis-geo} on exact unitary synthesis with geometrically-local circuits,
we can implement $\hat W_{A_i}$ by a geometrically-local circuit with a circuit depth of
\begin{equation}
    4 (8 kd)^k 4^{2 (8kd)^k} \leq 4^{3 (8kd)^k + 1} \leq 4^{4 (8kd)^{k}}.
\end{equation}
Conditioned on Eq.~\eqref{eq:approxheisenberg} succeeds, the approximation error is bounded as follows:
\begin{equation}
    \begin{aligned}
    \left\|\hat W_{A_i}-W_{A_i}\right\|_\infty&\leq 2\left\|\frac{1}{2^{|A_i|}}\sum_{P\in\{I,X,Y,Z\}^{|A_i|}}\left(\hat O_{i,P_{A_i}}-U^\dag P_{A_i} U\right)\otimes P_{A_i}\right\|_\infty\\
    &\leq \frac{2}{2^{|A_i|}}\sum_{P\in\{I,X,Y,Z\}^{|A_i|}}\left\|\hat O_{i,P_{A_i}}-U^\dag P_{A_i} U\right\|_\infty\\
    &\leq \varepsilon_1.
    \end{aligned}
\end{equation}
Here in the first line we use the same argument as in Eq.~\eqref{eq:projuerror}.

It remains to bound the time and query complexity to achieve the learning guarantee in Eq.~\eqref{eq:approxheisenberg}. Given a randomized measurement dataset
\begin{equation}
    \mathcal{T}_U(N) = \left\{ \ket{\psi_\ell} = \bigotimes_{i=1}^n \ket{\psi_{\ell, i}}, \ket{\phi_\ell} = \bigotimes_{i=1}^n \ket{\phi_{\ell, i}} \right\}_{\ell=1}^N,
\end{equation}
for a Pauli operator $P\in\{I,X,Y,Z\}^{|A_i|}$ with weight $w\leq |A_i|$ (the weight of a Pauli operator is the number of non-identity elements), let
\begin{equation}
    v_\ell^{U^\dag P_{A_i} U}:=3^w \expval{P}{\phi_{\ell,A_i}},
\end{equation}
where we let $\ket{\phi_{\ell,A_i}}:=\otimes _{j\in A_i}\ket{\phi_{\ell,j}}$. The same argument in Lemma~\ref{lem:reusing-RMdata} shows that
\begin{equation}
    \E \left[v_\ell^{U^\dag P_{A_i} U}\right]=\expval{U^\dag P_{A_i} U}{\psi_\ell}.
\end{equation}
Let $m:=\max_i |L(A_i)|\leq (8kd)^{k}$ be the maximum support of the operators $U^\dag P_{A_i} U$. Using Lemma~\ref{lem:learn-few-body-obs-kno-supp}, with a dataset size of
\begin{equation}
    N=\frac{2^{\mc O(m)}\log(n/\delta)}{\varepsilon_1^2},
\end{equation}
Eq.~\eqref{eq:approxheisenberg} is achieved with success probability at least $1-\delta$.\\

\noindent\textbf{Step 2: Compiling.} Given a classical description of $\hat W_{A_i}$ as unitary acting on $L(A_i)\cup A_i'$, which can be implemented with a circuit depth of at most $4^{4 (8kd)^{k}}$, we would like to find a depth-$(2d+1)$ circuit $\hat W_{A_i}'$ that is close to $\hat W_{A_i}$. To do this, we construct an $\varepsilon$-net for the circuit lightcone and perform a brute force search.

\begin{definition}[$\varepsilon$-net for circuits]\label{def:epsnet}
Consider a graph $G=(V,E)$. Let $U$ be some unitary generated by $d$ layers of 2-qubit gates where each gate is chosen from $\mathrm{SU}(4)$ and acts on an edge in $E$. An $\varepsilon$-net for circuits is a set of depth-$d$ circuits defined on $G$, denoted as $\mc N_\varepsilon(G)$, such that for any choice of $U$, there exists $V\in \mc N_\varepsilon(G)$, such that $\|V-U\|_\infty\leq\varepsilon$.
\end{definition}

\begin{lemma}\label{lemma:epsnet}
    Let $G=(V,E)$ be a graph with $s=|V|$ vertices and maximum degree $\kappa$. An $\varepsilon$-net for depth-$d$ circuits defined on $G$, denoted as $\mc N_\varepsilon (G)$, can be constructed with size at most $\left(\frac{\kappa sd}{\varepsilon}\right)^{\mc O(sd)}$ and in time $\left(\frac{\kappa sd}{\varepsilon}\right)^{\mc O(sd)}$.
\end{lemma}
\begin{proof}
    There are at most $sd/2$ 2-qubit gates in the circuit. We construct the $\varepsilon$-net by first enumerating all possible circuit architectures and then enumerate each 2-qubit gate using a $\frac{2\varepsilon}{sd}$-net for $\mathrm{SU}(4)$.
    In each layer, each qubit can interact with one of the $\kappa$ neighboring qubits. This implies that the number of possible circuit architectures in one layer is at most $\kappa^s$. Therefore, the number of possible circuit architectures with depth $d$ is at most $\kappa^{sd}$.

    An $\varepsilon_1$-net for $\mathrm{SU}(4)$ can be constructed with $\left(\frac{c_0}{\varepsilon_1}\right)^{c_1}$ elements, where $c_0$, $c_1$ are absolute constants. Plugging in $\varepsilon_1=\frac{2\varepsilon}{sd}$, the size of $\mc N_\varepsilon(G)$ is at most
    \begin{equation}
        \kappa^{sd}\cdot \left(\frac{\mc O(1)\cdot sd}{\varepsilon}\right)^{\mc O(1)\cdot sd}=\left(\frac{\kappa sd}{\varepsilon}\right)^{\mc O(sd)}.
    \end{equation}
    This concludes the proof.
\end{proof}

Let $G_{L(A_i)}$ be the subgraph of $k$-dimensional lattice induced by vertices in $L(A_i)$. The lattice coloring scheme guarantees that the size of $L(A_i)$ is at most $(8kd)^{k}$. Let $\mc N_{\varepsilon_2}(G_{L(A_i)})$ be an $\varepsilon_2$-net for depth-$d$ circuits acting on $L(A_i)$, which has size at most
\begin{equation}
    \left(\frac{(8kd)^{k+1}}{\varepsilon_2}\right)^{\mc O(1)\cdot (8kd)^{k+1}}.
\end{equation}
By definition, there is an element $V\in\mc N_{\varepsilon_2}(G_{L(A_i)})$ which is a depth-$d$ circuit acting on $L(A_i)$, such that
\begin{equation}
    \|(U^\dag\otimes I)S_{A_i}(U\otimes I)-(V^\dag\otimes I)S_{A_i}(V\otimes I)\|_\infty\leq 2\varepsilon_2,
\end{equation}
which implies that
\begin{equation}
    \|\hat W_{A_i}-(V^\dag\otimes I)S_{A_i}(V\otimes I)\|_\infty\leq \varepsilon_1+2\varepsilon_2.
\end{equation}
Therefore, enumerating over all elements in $\mc N_{\varepsilon_2}(G_{L(A_i)})$, we are guaranteed to find one element $\hat V$ that satisfies
\begin{equation}
    \|\hat W_{A_i}-(\hat  V^\dag\otimes I)S_{A_i}(\hat V\otimes I)\|_\infty\leq \varepsilon_1+2\varepsilon_2.
\end{equation}
Let $\varepsilon_2=\varepsilon_1/2$ and define $\hat W_{A_i}':=(\hat  V^\dag\otimes I)S_{A_i}(\hat V\otimes I)$, we have $\|\hat W_{A_i}-\hat W_{A_i}'\|_\infty\leq 2\varepsilon_1$.\\

\noindent\textbf{Putting everything together.} To achieve diamond distance $\varepsilon$ between the learned circuit $S\prod_{i=1}^L \hat W_{A_i}'$ and the true circuit $U\otimes U^\dag$, it suffices to choose $\varepsilon_1=\frac{\varepsilon}{3n}$. With probability at least $1-\delta$, we can learn all operators $\hat W_{A_i}$ within sufficient precision, using a dataset size of
\begin{equation}
    N=\frac{2^{\mc O( (8kd)^k)}n^2\log(n/\delta)}{\varepsilon^2}.
\end{equation}
Next, each $\hat W_{A_i}$ is classically compiled into a circuit, and they are combined together according to the order in Lemma~\ref{lemma:implementationordering}, such that the learned circuit has total depth $(k+1)(2d+1)+1$. This classical postprocessing procedure takes a total time of
\begin{equation}
    \mathcal{O}(nN) + \left(n/\varepsilon\right)^{\mc O(8kd)^{k+1})},
\end{equation}
which is polynomial in $n$ and $1 / \varepsilon$.
If we do not compile $\hat W_{A_i}$ to the shorter-depth circuit $\hat W_{A_i}'$ and use $\hat W_{A_i}$ directly, then the classical postprocessing procedure only requires a computational time of
\begin{equation}
    \mathcal{O}(nN),
\end{equation}
but the learned circuit will have a total depth of $(k+1) 4^{4 (8kd)^{k}} + 1$.
This concludes the proof of the first part of Theorem~\ref{thm:geo-kD-lattice-optimized}.

\subsubsection{Finite gate sets}

The algorithm and the proof closely follow that of arbitrary $\mathrm{SU}(4)$ gates.
When one considers a finite gate set with a constant size, a key simplification is the following: for any given $i \in \{1, \ldots, L\}$ and $P_{A_i} \in \{I, X, Y, Z\}^{|A_i|}$,
$U^\dagger P_{A_i} U$ only takes on a constant number of options.
Let $\varepsilon_{i, P_{A_i}} = \Omega(1)$ be the minimum distance in spectral norm between any pair of distinct $U^\dagger P_{A_i} U$.

From the same algorithm and proof in \textit{Step 1: Learning}, we can ensure that
\begin{equation}
    \left\|\hat O_{i,P_{A_i}}-U^\dag P_{A_i}U\right\|_\infty \leq \frac{\varepsilon_{i, P_{A_i}}}{3}, \quad\forall i\in\{1,2,\dots,L\}, \quad P_{A_i}\in\{I,X,Y,Z\}^{|A_i|},
\end{equation}
holds with probability at least $1-\delta$ using a sample complexity of
\begin{equation} \label{eq:samp-comp-finite-gate-optimized-depth}
    N = \mathcal{O}\left( \frac{\log(n / \delta)}{\varepsilon_{i, P_{A_i}}^2} \right) = \mathcal{O}\left( \log(n / \delta) \right).
\end{equation}
From the definition of $\varepsilon_{i, P_{A_i}}$, we can identify $U^\dagger P_{A_i} U$ exactly from $\hat O_{i,P_{A_i}}$.
This enables us to exactly reconstruct
\begin{equation}
    W_{A_i} = \frac{1}{2^{|A_i|}}\sum_{P\in\{I,X,Y,Z\}^{|A_i|}}U^\dag P U\otimes P = U^\dagger S_{A_i} U.
\end{equation}
Because $U$ is a quantum circuit of depth $d = \mathcal{O}(1)$ on a constant-dimensional lattice over a finite gate set of a constant size, we can perform a constant-time brute-force search to find a $(2d+1)$-depth circuit implementation for $W_{A_i}$ instead of searching through the $\varepsilon$-net as in \textit{Step 2: Compiling}.
The computational time of the compiling step is improved from $\left(n/\varepsilon\right)^{\mc O(8kd)^{k+1})}$ to $\mathcal{O}(n)$.
Following the rest of the proof for the case of $\mathrm{SU}(4)$ gates, we can learn $U$ exactly with a learned circuit of depth $(k+1)(2d+1) + 1$.
The sample complexity is given in Eq.~\eqref{eq:samp-comp-finite-gate-optimized-depth}, and the computational time is dominated by reading the classical dataset, which is of $\mathcal{O}(nN) = \mathcal{O}(n \log (n / \delta))$.
This concludes the proof of Theorem~\ref{thm:geo-kD-lattice-optimized}.

\section{Learning shallow quantum circuits from quantum queries}
\label{sec:learning-unitary-coherent}

We consider quantum learning algorithms that can access an unknown $n$-qubit unitary $U$ through coherent quantum queries, which interleave the unitary $U$ with quantum computation.

\begin{definition}[Coherent quantum queries]
    The learning algorithm is a quantum algorithm with general coherent query access to the unknown unitary $U$. The quantum learning algorithm can interleave multiple accesses to the unknown unitary $U$ with polynomial-size quantum circuits.
\end{definition}

We show the following result for learning geometrically-local shallow quantum circuits over finite gate sets with asymptotically optimal query complexity and time complexity.
We only need to consider proving the matching upper bounds.
The matching lower bounds to the query and time complexity are trivial: learning anything about $U$ requires $\Omega(1)$ queries to $U$; writing down $U$ requires $\Omega(n)$ time.

\begin{theorem}[Learning geometrically-local shallow quantum circuits over a finite gate set] \label{thm:geo-finite-gates}
    Given an unknown geometrically-local constant-depth $n$-qubit circuit $U$ over a finite gate set. From
    \begin{equation}
        N = \Theta(1)
    \end{equation}
    queries to $U$, we can learn an $n$-qubit quantum channel $\hat{\mathcal{E}}$ that can be implemented by a geometrically-local constant-depth $2n$-qubit circuit, such that
    \begin{equation}
        \hat{\mathcal{E}} = \mathcal{U},
    \end{equation}
    with probability $1$.
    The computational time to learn $\hat{\mathcal{E}}$ is $\Theta(n)$.
\end{theorem}

\subsection{Learning local inversion using coherent quantum queries}
\label{sec:local-inv-coherent}

When there is only a finite choice of possible unitaries, we can find the local inversion perfectly with $\mathcal{O}(1)$ queries, even if there is incoherent noise coming from the environment.
This lemma is useful for showing the $\mathcal{O}(1)$ query complexity for learning $n$-qubit shallow quantum circuits with a finite gate set and a fixed geometric structure.
The idea is to store multiple output quantum states in a quantum memory and utilize entangled quantum data processing.
The formal statement is given below.
We use the subscript on identity $I$ or $\mathcal{I}$ to denote the number of qubits the identity acts on.

\begin{lemma}[Perfect local inversion among finite choices]
    \label{lem:perfect-local-identity-check-finite}
    Consider $k, l, m = \mathcal{O}(1)$, unitaries $U_1, \ldots, U_m$ over $k$ qubits, and unitaries $W_1, \ldots W_m$ over $(k-1)+l$ qubits.
    Let CPTP maps $\mathcal{E}_x$ from $k$ to $k+l$ qubits be
    \begin{equation}
        \mathcal{E}_x(\rho) := (\mathcal{I}_1 \otimes \mathcal{W}_x) (\mathcal{U}_x \otimes \mathcal{I}_l)(\rho \otimes I/2^{l}), \quad \forall x = 1, \ldots, m.
    \end{equation}
    Given an unknown $\mathcal{E}_x$.
    Using $\mathcal{O}(1)$ queries to $\mathcal{E}_x$, we can find a perfect local inversion $V_x$ of $U_x$ on the first qubit. Furthermore, $V_x = U_i^\dagger$ for some $i$.
\end{lemma}

In order to prove the above lemma, we use a perfect local identity check for two choices given in Lemma~\ref{lem:perfect-local-identity-check-two}.
The proof of Lemma~\ref{lem:perfect-local-identity-check-finite} is given after the proof of Lemma~\ref{lem:perfect-local-identity-check-two}.

\begin{lemma}[Perfect local identity check among two choices]
    \label{lem:perfect-local-identity-check-two}
    Consider $k, l \geq 1$, two unitaries $U_1, U_2$ over $k$ qubits, and two unitaries $V_1, V_2$ over $k+l-1$ qubits.
    Given CPTP maps from $k$ qubits to $k+l$ qubits,
    \begin{equation}
        \mathcal{E}_x(\rho) := (\mathcal{I}_1 \otimes \mathcal{V}_x) (\mathcal{U}_x \otimes \mathcal{I}_l)(\rho \otimes I/2^{l}), \quad \forall x = 1, 2.
    \end{equation}
    Assume that $k, l$ are constants, $U_1$ acts as identity on the first qubit $U_1 = I_1 \otimes \tilde{U}_1$, and $U_2$ is constant far from CPTP maps that act as an identity on the first qubit,
    \begin{equation} \label{eq:c-U2diff-IA}
        c:= \min_{\mathcal{E}} \norm{\mathcal{U}_2 - \mathcal{I}_1 \otimes \mathcal{E}}_\diamond = \Omega(1).
    \end{equation}
    Given an unknown $\mathcal{E}_x$.
    Using $\mathcal{O}(1)$ queries to $\mathcal{E}_x$, we can perfectly distinguish between $\mathcal{E}_1$ and $\mathcal{E}_2$.
\end{lemma}
\begin{proof}
    Let $\ket{\Omega_k}$ be the maximally entangled state over two copies of a $k$-qubit system.
    We define the following density matrices over $(k + l) + k$ qubits,
    \begin{equation}
        \rho_x := (\mathcal{I}_k \otimes \mathcal{E}_x)(\ketbra{\Omega_k}{\Omega_k}), \quad \forall x = 1, 2.
    \end{equation}
    The support of a density matrix $\rho$ is defined as
    \begin{equation}
        \mathrm{supp}(\rho) := \big\{ \ket{\psi} \big| \bra{\psi} \rho \ket{\psi} > 0 \big\}.
    \end{equation}
    From the definition of $\rho_x$, we have
    \begin{equation}
        \mathrm{supp}(\rho_x) = \left\{ (I_{k+1} \otimes V_x) (I_k \otimes U_x \otimes I_{l})(\ket{\Omega_k} \otimes \ket{\psi}), \,\, \forall \ket{\psi} \right\}.
    \end{equation}
    The maximal fidelity between two density matrices is defined as
    \begin{equation}
        \tilde{F}(\rho_1, \rho_2) := \max\left( |\braket{\phi_1 | \phi_2}| \,\, \big| \,\, \ket{\phi_x} \in \mathrm{supp}(\rho_x), \,\, x = 1, 2 \right).
    \end{equation}
    The maximal fidelity behaves similarly to fidelity and is multiplicative under tensor product
    \begin{equation} \label{eq:multipli-max-fid}
        \tilde{F}(\rho_1 \otimes \sigma_1, \rho_2 \otimes\ \sigma_2) = \tilde{F}(\rho_1, \rho_2) \tilde{F}(\sigma_2, \sigma_2).
    \end{equation}
    From the above definition, we see that there exists $\ket{\psi_1}, \ket{\psi_2}$ such that
    \begin{equation}
        \tilde{F}(\rho_1, \rho_2)^2 = \left| (\bra{\Omega_k} \otimes \bra{\psi_1}) (I_{k} \otimes U_2^\dagger \otimes I_l) (I_{k+1} \otimes (V_2^\dagger V_1 (\tilde{U}_1 \otimes I_{l}))) (\ket{\Omega_k} \otimes \ket{\psi_2}) \right|^2.
    \end{equation}
    We now consider two states associated with the above,
    \begin{align}
        \sigma_1 &:= (I_{k+1} \otimes (V_2^\dagger V_1 (\tilde{U}_1 \otimes I_{l}))) (\ketbra{\Omega_k}{\Omega_k} \otimes \ketbra{\psi_2}{\psi_2}) (I_{k+1} \otimes ((\tilde{U}_1^\dagger \otimes I_{l}) V_1^\dagger V_2 )) \\
        \sigma_2 &:= (I_{k} \otimes U_2 \otimes I_l) (\ketbra{\Omega_k}{\Omega_k} \otimes \ketbra{\psi_1}{\psi_1}) (I_{k} \otimes U_2^\dagger \otimes I_l)
    \end{align}
    The Fuchs–van de Graaf inequalities show that $\tilde{F}(\rho_1, \rho_2)^2 = \Tr(\sigma_1 \sigma_2) \leq 1 - \tfrac{1}{4} \norm{\sigma_1 - \sigma_2}_1^2$.
    We now consider a lower bound of the trace norm $\norm{\sigma_1 - \sigma_2}_1$ by tracing out the last $l$ qubits,
    \begin{equation}
        \norm{\sigma_1 - \sigma_2}_1 \geq \norm{ (\mathcal{I}_k \otimes \mathcal{I}_1 \otimes \mathcal{E}) (\ketbra{\Omega_k}{\Omega_k}) - (\mathcal{I}_{k} \otimes \mathcal{U}_2)(\ketbra{\Omega_k}{\Omega_k}) }_1,
    \end{equation}
    where $\mathcal{E}$ is a CPTP map that acts on the last $k-1$ qubits.
    Recall that the $1$-norm distance in the Choi states upper bounds the diamond distance in the CPTP maps up to the dimension factor $1 / 2^k$. From the definition of $c$ in Eq.~\eqref{eq:c-U2diff-IA}, we have the following inequality,
    \begin{equation}
        \norm{\sigma_1 - \sigma_2}_1 \geq \frac{1}{2^k} \norm{\mathcal{I}_1 \otimes \mathcal{E}  - \mathcal{U}_2 }_\diamond \geq \frac{c}{2^k}.
    \end{equation}
    Therefore, we have
    \begin{equation} \label{eq:rho1rho2-maxfid}
        \tilde{F}(\rho_1, \rho_2) \leq \sqrt{1 - (c / 2^{k+2})^2 } < 1,
    \end{equation}
    which is a key result that will be used later.

    We need to consider another pair of states.
    Consider the Pauli decomposition of $U_2$ on the first qubit,
    \begin{equation}
        U_2 = \sum_{P \in \{I, X, Y, Z\}} P \otimes \tilde{U}_{2, P},
    \end{equation}
    where $\tilde{U}_{2, P}$ is a complex matrix of dimension $2^{k-1}$. Because $U_2$ does not act as identity on the first qubit, we have $c' := \sum_{P \neq I} \Tr(\tilde{U}_{2, P}^\dagger \tilde{U}_{2, P}) > 0$ is a positive constant.
    Consider the following matrix,
    \begin{equation}
        M := \sum_{P \in \{X, Y, Z\}} P \otimes \tilde{U}_{2, P},
    \end{equation}
    and define two $2k$-qubit pure states,
    \begin{align}
        \ket{\psi_1} &:= \ket{\Omega_k},\\
        \ket{\psi_2} &:= I_k \otimes \left( U_2^\dagger \frac{M}{\sqrt{\Tr(M^\dagger M) / 2^k}} \right) \ket{\Omega_k}.
    \end{align}
    By the definition of $c'$ and $M$, we have $\Tr(M^\dagger M) = 2 c' > 0$ and
    \begin{equation} \label{eq:maxfid-psi1psi2}
        \tilde{F}(\ketbra{\psi_1}{\psi_1}, \ketbra{\psi_2}{\psi_2}) = |\braket{\psi_1 | \psi_2}|^2 = 2c' / 2^k > 0.
    \end{equation}
    Furthermore, the overlap between $\mathcal{E}_x(\ketbra{\psi_x}{\psi_x})$ satisfies
    \begin{align}
        &\Tr\left( \mathcal{E}_1(\ketbra{\psi_1}{\psi_1}) \mathcal{E}_2(\ketbra{\psi_2}{\psi_2}) \right) = \frac{1}{2c' / 2^k} \cdot \frac{1}{2^k} \cdot \frac{1}{2^k} \cdot \\
        & \sum_{P, Q \in \{X, Y, Z\}} \Tr( \Tr_{\leq k}(P \otimes ( (\tilde{U}_{2, P}^\dagger \otimes I_{l}) V_2^\dagger V_1 (\tilde{U}_1 \otimes I_l) )) \Tr_{\leq k}(Q \otimes ( (\tilde{U}_1^\dagger \otimes I_l) V_1^\dagger V_2 (\tilde{U}_{2, Q} \otimes I_{l}) ))) = 0,
    \end{align}
    which implies that there exists a two-outcome projective measurement $\mathcal{M}$ that could perfectly distinguish between the two states $\mathcal{E}_1(\ketbra{\psi_1}{\psi_1})$ and $\mathcal{E}_2(\ketbra{\psi_2}{\psi_2})$.

    Consider $N$ queries to $\mathcal{E}_x$ to obtain $\rho_x^{\otimes N}$, where the number of queries is
    \begin{equation}
        N := \max\left(1, \left\lceil\frac{\log \left( (2c' / 2^k) \right)}{\log \left( \sqrt{1 - (c / 2^{k+2})^2 } \right)} \right\rceil \right) = \mathcal{O}(1).
    \end{equation}
    Using Eq.~\eqref{eq:multipli-max-fid}, \eqref{eq:rho1rho2-maxfid}, and \eqref{eq:maxfid-psi1psi2}, we have
    \begin{equation}
        \tilde{F}(\rho_1^{\otimes N}, \rho_2^{\otimes N}) = \tilde{F}(\rho_1, \rho_2)^N \leq \sqrt{1 - (c / 2^{k+2})^2 }^N \leq (2c' / 2^k) = \tilde{F}(\ketbra{\psi_1}{\psi_1}, \ketbra{\psi_2}{\psi_2}).
    \end{equation}
    From Lemma~1 of \cite{duan2009perfect}, there exists a CPTP map $\mathcal{T}$ that takes $\rho_x$ to $\ketbra{\psi_x}{\psi_x}$ for $x = 1, 2$.
    We apply $\mathcal{T}$ to $\rho_x$.
    And we evoke one additional query to $\mathcal{E}_x$ to obtain $\mathcal{E}_x(\ketbra{\psi_x}{\psi_x})$.
    Finally, we perform the two-outcome projective measurement $\mathcal{M}$ to perfectly distinguish between $\mathcal{E}_1(\ketbra{\psi_1}{\psi_1})$ and $\mathcal{E}_2(\ketbra{\psi_2}{\psi_2})$.
    Together, with $N+1 = \mathcal{O}(1)$ queries to $\mathcal{E}_x$, we can perfectly distinguish between $\mathcal{E}_1$ and $\mathcal{E}_2$.
\end{proof}

We are now ready to prove Lemma~\ref{lem:perfect-local-identity-check-finite}.
The central idea is a bipartite tournament with a potential local inversion on one side and all possible non-local inversion on the other side.

\begin{proof}[Proof of Lemma~\ref{lem:perfect-local-identity-check-finite}]
    Each query to $\mathcal{E}_x$ allows us to create $1$ query to any one of the following CPTP maps,
    \begin{equation}
        \mathcal{E}_{x, i} = (\mathcal{E}_x \circ \mathcal{U}^\dagger_i), \,\, \forall i = 1, \ldots, m.
    \end{equation}
    The algorithm proceeds by going through all of $i$ one by one.
    For each $i$, the algorithm creates two sets,
    \begin{align}
        S_i &:= \left\{ y \in \{1, \ldots, m\} \,\, | \,\, U_y U_i^\dagger \,\, \mbox{acts as identity on the first qubit} \right\},\\
        T_i &:= \{1, \ldots, m\} \setminus S_i.
    \end{align}
    Note that by definition, $i \in S_i$ and $i \not\in T_i$.
    For each $y \in T_i$, the algorithm uses the algorithm given in the proof of Lemma~\ref{lem:perfect-local-identity-check-two} to test whether $\mathcal{E}_{x, i}$ is equal to $\mathcal{E}_{y, i}$ or $\mathcal{E}_{i, i}$.
    If $\mathcal{E}_{x, i}$ is indeed equal to one of them, then the algorithm in Lemma~\ref{lem:perfect-local-identity-check-two} is guaranteed to output the one that is equal to $\mathcal{E}_{x, i}$.
    If not, then the algorithm in Lemma~\ref{lem:perfect-local-identity-check-two} will output $\mathcal{E}_{y, i}$ or $\mathcal{E}_{i, i}$ arbitrarily.
    After going through all $y \in T_i$, if between $\mathcal{E}_{y, i}$ and $\mathcal{E}_{i, i}$,  $\mathcal{E}_{i, i}$ is always chosen for all $y \in T_i$,
    then the algorithm sets $i^* := i$ and terminates the for-loop over $i$.
    The algorithm outputs $U_{i^*}^\dagger$ as the claimed perfect local inversion of $U_x$ on the first qubit.

    By construction, the total number of queries to $\mathcal{E}_x$ in the above algorithm is a constant.
    We now prove that (a) $i^*$ can always be found by the above algorithm and (b) $U_{i^*}^\dagger$ is a perfect local inversion of $U_x$ on the first qubit.
    The proof is separated into the following two paragraphs addressing each claim.

    \vspace{0.35em}
    \paragraph{$i^*$ can always be found.}
    When $i = x$, for each $y \in T_i$, we are testing whether $\mathcal{E}_{x, x}$ is equal to $\mathcal{E}_{y, x}$ or $\mathcal{E}_{x, x}$.
    Because $U_y U_x^\dagger$ does not act as identity on the first qubit by definition of $T_x$, Lemma~\ref{lem:perfect-local-identity-check-two} shows that the algorithm will always return $\mathcal{E}_{x, x}$ when deciding between $\mathcal{E}_{y, x}$ and $\mathcal{E}_{x, x}$.
    Hence when $i = x$, the algorithm will set $i^* := i$ and terminate the for-loop over $i$.
    The algorithm could also terminate earlier for some $i < x$ but will always terminate when $i = x$.
    Therefore, $i^*$, as defined by the algorithm previously, can always be found.

    \vspace{0.35em}
    \paragraph{$U_{i^*}^\dagger$ is a perfect local inversion of $U_x$ on the first qubit.}
    We first show by contradiction that $x \not\in T_{i^*}$.
    Suppose that $x \in T_{i^*}$.
    For $y = x \in T_{i^*}$, we would be testing whether $\mathcal{E}_{x, i^*}$ is equal to $\mathcal{E}_{x, i^*}$ or $\mathcal{E}_{i^*, i^*}$.
    Recall that $i^* \not\in T_{i^*}$, thus $x \neq i^*$.
    Lemma~\ref{lem:perfect-local-identity-check-two} thus implies that the algorithm will always return $\mathcal{E}_{x, i^*}$ when deciding between $\mathcal{E}_{x, i^*}$ and $\mathcal{E}_{i^*, i^*}$.
    As a result, the condition defining $i^*$ is not satisfied, which is a contradiction.
    Because $S_{i^*} \cup T_{i^*} = \{1, \ldots, m\}$, we have $x \in S_{i^*}$. which means have $U_x U_{i^*}^\dagger$ acts as identity on the first qubit.
    As a result, $U_{i^*}^\dagger$ is a perfect local inversion of $U_x$ on the first qubit.
\end{proof}

\subsection{Learning geometrically-local shallow circuits over a finite gate set (Proof of Theorem~\ref{thm:geo-finite-gates})} \label{sec:geo-finite-gates-proof}

We present the algorithm for learning an unknown geometrically-local shallow quantum circuit $U$ over a finite gate set.
Let the geometry over $n$ qubits be represented by a graph $G = (V, E)$ with degree $\kappa = \mathcal{O}(1)$, the depth of $U$ be $d = \mathcal{O}(1)$, and the finite gate set be $\mathcal{G}$ with $|\mathcal{G}| = \mathcal{O}(1)$.
This algorithm requires coherent quantum queries to the unknown unitary $U$.
The key ideas are constructing $n$ CPTP maps $\mathcal{E}^U_i, \forall i \in \{1, \ldots, n\}$ from $\mathcal{O}(1)$ queries to $U$, utilizing Lemma~\ref{lem:perfect-local-identity-check-finite} to find perfect local inversion among finite choices, and using Definition~\ref{def:sew-local-inv} and Lemma~\ref{lem:sewedlocalinv} to sew the local inversion unitaries together.

We consider the lightcone $L_d(i)$ of the geometry for qubit $i$ under the unknown depth-$d$ geometrically-local circuit $U$ in Definition~\ref{def:light-cone-geo} and the properties of the lightcones given in Lemma~\ref{lem:lightcone-property-geo}.

For each qubit $i$ in the $n$-qubit system, we can always decompose the depth-$d$ geometrically-local quantum circuit $U$ as the following,
\begin{equation} \label{eq:decomp-U}
    U = \left(I_i \otimes W^{(i)} \otimes I_{\notin L_{2d}(i)} \right) \left(U^{(i)} \otimes \tilde{W}^{(i)} \right),
\end{equation}
where $U^{(i)}$ acts on qubits in the set $L_{d}(i)$, $\tilde{W}^{(i)}$ acts on qubits not in the set $L_{d}(i)$, $W^{(i)}$ acts on qubits in the set $L_{2d}(i) \setminus \{i\}$, and $I_i, I_{\notin L_{3d}(i)}$ are identity matrices acting on qubit $i$ and qubits not in $L_{3d}(i)$, respectively.
Furthermore, $U^{(i)}, W^{(i)}, \tilde{W}^{(i)}$ are all subcircuits (circuits containing a subset of gates) of the unknown depth-$d$ geometrically-local circuits $U$.
% \robert{TODO: Add a figure to illustrate this.}
We define the CPTP map $\mathcal{E}^U_i$,
\begin{align}
    \mathcal{E}^U_i(\rho) &:= \Tr_{\notin L_{2d}(i)}\left(U \left(\rho \otimes \frac{I_{\notin L_d(i)}}{2^{n - |L_d(i)|}} \right) U^\dagger \right)\\
    &= \left(\mathcal{I}_i \otimes \mathcal{W}^{(i)} \right) \left(\mathcal{U}^{(i)} \otimes \mathcal{I}_{L_{2d}(i) \setminus L_{d}(i)} \right) \left(\rho \otimes \frac{I_{L_{2d}(i) \setminus L_{d}(i)}}{2^{|L_{2d}(i)| - |L_d(i)|}} \right),
\end{align}
where $\rho$ is a density matrix for qubits in $L_d(i)$, $I_{\notin L_d(i)}$ is the identity matrix over qubits not in $L_d(i)$, $I_{\notin L_d(i)} / 2^{n - |L_d(i)|}$ is the maximally mixed state for qubits not in $L_d(i)$, and $\Tr_{\notin L_{2d}(i)}$ traces out all qubits not in $L_{2d}(i)$.
Because $\mathcal{E}^U_i(\rho)$ uses a single query to $U$, naively, one would expect that to obtain a query to $\mathcal{E}^U_i$ for every qubit $i$ requires $n$ queries to $U$.
The following lemma shows that we can do much more efficiently than what one would naively expect.

\begin{lemma}[Queries to every $\mathcal{E}^U_i$ from only $\mathcal{O}(1)$ queries to $U$] \label{lem:const-query-toU}
    We can construct a query to every $\mathcal{E}^U_i, 1 \leq i \leq n$ from only $\mathcal{O}(1)$ queries to the unknown constant-depth geometrically-local circuit $U$.
\end{lemma}
\begin{proof}
Let $d = \mathcal{O}(1)$ be the depth of the circuit $U$.
We consider a graph $G^{(3d)}$ over $n$ qubits, where each pair of qubits is connected by an edge if their distance in $G$ is at most $3d$.
The degree of $G^{(3d)}$ is at most $(\kappa + 1)^{3d} = \mathcal{O}(1)$.
The graph only has $\mathcal{O}(n)$ edges and can be constructed as an adjacency list in time $\mathcal{O}(n)$.
Let us define a coloring of the graph $G^{(3d)}$.
By the standard greedy coloring algorithm, we can find a color $c^{(3d)}(i)$ for each qubit $i$ in graph $G^{(3d)}$, where no adjacent vertices can have the same color, and there are only $\chi^{(3d)}$ distinct colors with
\begin{equation}
    \chi^{(3d)} \leq (\kappa + 1)^{3d} + 1 = \mathcal{O}(1).
\end{equation}
The greedy coloring algorithm runs in time linear in the number of edges in $G^{(3d)}$, which is linear in the number $n$ of qubits.

For each color $c = 1, \ldots, \chi^{(3d)}$, we consider the set of qubits with color $c$.
We can construct one query to every $\mathcal{E}^U_i$ for qubits $i$ with color $c^{(3d)}(i) = c$ from only one query to $U$.
By the construction of the graph coloring, for two distinct qubits $i \neq j$ with the same color $c$, $L_{3d}(i) \cap L_{3d}(j) = \varnothing$.
We now define the following sets of qubits for the color $c$,
\begin{equation}
    A(c) := \left\{ i \in \{1, \ldots, n\} \,\, \big| \,\, c^{(3d)}(i) = c \right\}, \quad B_{q}(c) := \bigcup_{i: c^{(3d)}(i) = c} L_{q}(i),
\end{equation}
for any integer $q \geq 1$.
Given the definition of $U^{(i)}, W^{(i)}$ in Eq.~\eqref{eq:decomp-U} for each qubit $i$.
We can further decompose the shallow circuit $U$ as
\begin{equation}
    U = \left[\left(I_{A(c)} \otimes \bigotimes_{i: c^{(3d)}(i) = c} W^{(i)} \right) \otimes I_{\notin B_{2d}(c)} \right] \left[\left( \bigotimes_{i: c^{(3d)}(i) = c} U^{(i)} \right) \otimes \tilde{W}^{(c)}\right],
\end{equation}
where $\tilde{W}^{(c)}$ acts on qubits not in $B_{d}(c)$.
Consider initializing the qubits not in $B_{d}(c)$ as the maximally mixed state, evolving under $U$, and tracing out any qubits not in $B_{2d}(c)$.
The resulting CPTP map $\mathcal{E}_c^U$ from qubits in $B_{d}(c)$ to qubits in $B_{2d}(c)$ can be written as
\begin{equation}
    \mathcal{E}_c^U(\rho) = \left(\mathcal{I}_{A(c)} \otimes \bigotimes_{i: c^{(3d)}(i) = c} \mathcal{W}^{(i)} \right) \left( \bigotimes_{i: c^{(3d)}(i) = c} \mathcal{U}^{(i)} \otimes \mathcal{I}_{B_{2d}(i) \setminus B_{d}(i)} \right) \left(\rho \otimes \frac{I_{B_{2d}(c) \setminus B_{d}(c)}}{2^{|B_{2d}(c)| - |B_{d}(c)|}} \right),
\end{equation}
where $\rho$ is a density matrix over qubits in $B_{d}(c)$. It is not hard to see that
\begin{equation}
    \mathcal{E}_c^U = \bigotimes_{i: c^{(3d)}(i) = c} \mathcal{E}_i^U.
\end{equation}
Because $\mathcal{E}_c^U$ only requires one query to $U$, we can create $\mathcal{E}_i^U$ for all qubit $i$ with color $c$ from one query to $U$.
Since there is only $\chi^{(3d)} = \mathcal{O}(1)$ colors, we can create a query to every $\mathcal{E}^U_i, 1 \leq i \leq n$ from only $\mathcal{O}(1)$ queries to the unknown circuit $U$.
\end{proof}

Because $U$ is over a finite gate set with size $\mathcal{O}(1)$, we have $U^{(i)}$ and $W^{(i)}$ only have a constant number of choices.
Furthermore, both $U^{(i)}$ and $W^{(i)}$ act on a constant number of qubits because $|L_d(i)| = \mathcal{O}(1), |L_{2d}(i)| = \mathcal{O}(1)$ for a constant depth $d$.
From Lemma~\ref{lem:perfect-local-identity-check-finite}, for each qubit $i$, through $\mathcal{O}(1)$ queries to $\mathcal{E}^U_i$, we can learn a perfect local inversion $V_i$ of $U^{(i)}$ on qubit $i$ with no failure probability.
The local inversion unitary $V_i$ is the inverse of one of the possible choices for $U^{(i)}$. Hence, $V_i$ is a geometrically-local depth-$d$ circuit that only acts on qubits in $L_d(i)$.
Combining with Lemma~\ref{lem:const-query-toU}, from only $\mathcal{O}(1)$ queries to $U$, we can learn $V^{(i)}, \forall i = 1, \ldots, n$, such that
\begin{equation}
    \mathcal{U}^{(i)} \mathcal{V}_i = \mathcal{I}^{(i)} \otimes \mathcal{E}^{\mathcal{U}^{(i)} \mathcal{V}_i}_{\neq i},
\end{equation}
where $\mathcal{I}^{(i)}$ is the identity map on qubit $i$ and $\mathcal{E}^{\mathcal{U}^{(i)} \mathcal{V}_i}_{\neq i}$ is the reduced channel of $\mathcal{U}^{(i)} \mathcal{V}_i$ with qubit $i$ removed.
The quantum computational time is given by $\mathcal{O}(n)$.
We now show that $V_i$ is also the perfect local inversion unitary for $U$ on qubit $i$.
To see this, recall the decomposition in Eq.~\eqref{eq:decomp-U}, we have
\begin{align}
    \mathcal{U} \mathcal{V}_i &= \left( \mathcal{I}_i \otimes \mathcal{W}^{(i)} \otimes \mathcal{I}_{\notin L_{2d}(i)} \right) \left(\mathcal{U}^{(i)} \mathcal{V}_i \otimes \tilde{\mathcal{W}}^{(i)} \right)\\
    &= \mathcal{I}^{(i)} \otimes \left( \left(\mathcal{W}^{(i)} \otimes \mathcal{I}_{\notin L_{2d}(i)}\right) \left( \mathcal{E}^{\mathcal{U}^{(i)} \mathcal{V}_i}_{\neq i} \otimes \tilde{\mathcal{W}}^{(i)} \right) \right)\\
    &= \mathcal{I}^{(i)} \otimes \mathcal{E}^{\mathcal{U} \mathcal{V}_i}_{\neq i}.
\end{align}
We can now use Definition~\ref{def:sew-local-inv} and Lemma~\ref{lem:sewedlocalinv} to sew the perfect local inversion unitaries together.
This gives the following $2n$-qubit unitary,
\begin{equation} \label{eq:Usew-def-linv}
    U_{\mathrm{sew}}(V_1, \ldots, V_n) = S \left[\prod_{i=1}^n \left(V_i^{(1)}\right) S_{i} \left(V_i^{(1)}\right)^\dagger \right] = U \otimes U^\dagger,
\end{equation}
where $V_i^{(1)}$ is the unitary $V_i$ acting on the first set of $n$ qubits.

We now show that there exists a sewing ordering such that $U_{\mathrm{sew}}(V_1, \ldots, V_n)$ is a constant-depth geometrically-local circuit.
Given the geometry over $n$ qubits represented by a graph $G = (V, E)$.
Consider a graph $G^{(2d)}$ over $n$ qubits, where each pair $(i, j)$ of qubits are connected by an edge if $i, j$ is of distance at most $2d$ in the geometric graph $G$.
Hence, equivalently, for all $(i, j)$ not connected by an edge in $G^{(2d)}$, we have
\begin{equation}
    L_d(i) \cap L_d(j) = \varnothing.
\end{equation}
The degree of $G^{(2d)}$ is bounded above by $(\kappa+1)^{2d}$.
And $G^{(2d)}$ can be constructed as an adjacency list in time $\mathcal{O}(n)$.
Because the graph has a constant degree, we can use a $\mathcal{O}(n)$-time greedy graph coloring algorithm to color the $n$-qubit graph $G^{(2d)}$ using only a constant number of colors.
For each node/qubit $i$, we consider $c(i)$ to be the color.
The sewing order for the $n$ local inversion unitaries $V_i$ is given by the greedy graph coloring, where we order from the smallest color to the largest color.
By the definition of graph coloring, for any pair $i, j$ of qubits with the same color, we have $L_d(i) \cap L_d(j) = \varnothing.$
Furthermore, $V_i$ is a constant-depth geometrically-local circuit that only acts on a constant number of qubits.
Therefore, for any color $c'$, we can find an implementation of the $2n$-qubit unitary
\begin{equation}
    \prod_{i: c(i) = c'} \left(V_i^{(1)}\right) S_{i} \left(V_i^{(1)}\right)^\dagger
\end{equation}
with a constant-depth geometrically-local quantum circuit in time $\mathcal{O}(n)$.
Since there is only a constant number of colors, the $2n$-qubit unitary $U_{\mathrm{sew}}(V_1, \ldots, V_n)$ in Eq.~\eqref{eq:Usew-def-linv} with the color-based ordering can be implemented with a constant-depth geometrically-local quantum circuit in time $\mathcal{O}(n)$.
Finally, define an $n$-qubit channel $\hat{\mathcal{E}}$ as follows,
\begin{equation}
    \hat{\mathcal{E}}(\rho) := \Tr_{> n}\left(\mathcal{U}_{\mathrm{sew}}(V_1, \ldots, V_n)(\rho \otimes \ketbra{0^n}{0^n})\right),
\end{equation}
which can be implemented as a geometrically-local constant-depth quantum circuit over $2n$ qubits.
Because $U_{\mathrm{sew}}(V_1, \ldots, V_n) = U \otimes U^\dagger$ from Eq.~\eqref{eq:Usew-def-linv}, we have
\begin{equation}
    \mathcal{E} = \mathcal{U}
\end{equation}
with probability one. This concludes the proof of Theorem~\ref{thm:geo-finite-gates}.

\section{Hardness for learning log-depth quantum circuits}

We have seen from the previous appendices that learning general constant-depth quantum circuits can be done efficiently.
A natural follow-up question is whether one could efficiently learn log-depth quantum circuits.
In the following, we show that learning log-depth quantum circuits to a constant diamond distance is exponentially hard, even when we allow coherent quantum queries to $U$.
Hence, the problem of learning quantum circuits transitions from being polynomially easy to exponentially hard when we go from $\mathcal{O}(1)$-depth to $\mathcal{O}(\log n)$-depth.

\begin{proposition}[Hardness for learning log-depth circuits] \label{prop:hardnes-log-depth}
    Consider an unknown $n$-qubit unitary $U$ generated by a $\mathcal{O}(\log n)$-depth circuit over arbitrary two-qubit gates with $n$ ancilla qubits. We have
    \begin{itemize}
        \item Learning $U$ to $1/3$ diamond distance with high probability requires $\exp(\Omega(n))$ queries.
        \item Distinguishing whether $U$ equals to the identity $I$ or is $1/3$-far from the identity $I$ in diamond distance with high probability requires $\exp(\Omega(n))$ queries.
    \end{itemize}
\end{proposition}
\begin{proof}
    Without loss of generality, we consider $n$ to be $2^k$ for an integer $k$.
    Consider the unknown unitary $U$ to be $I$ or one of $U_x, \forall x \in \{0, 1\}^{n}$.
    The unitary $U_x$ is defined to be
    \begin{equation} \label{eq:GroverOracle}
        U_x \ket{y} = \begin{cases}
            1, & x = y,\\
            -1, & x \neq y,
        \end{cases}
    \end{equation}
    for any $y \in \{0, 1\}^{n}$.
    The $n$-qubit unitary $U_x$ can be constructed as follows,
    \begin{equation}
        U_x = \left(\prod_{\substack{1 \leq i \leq n\\ x_i = 0}} X_{i}\right) C^{n}Z \left(\prod_{\substack{1 \leq i \leq n\\ x_i = 0}} X_{i}\right),
    \end{equation}
    where $X_i$ is the $X$ gate on the $i$-th qubit, and $C^{n}Z$ is a controlled-Z gate controlled on all qubits.
    The circuit $\prod_{\substack{i: \, x_i = 0}} X_{i}$ can be implemented in one layer.
    We can implement $C^{n}Z$ using $n$ ancilla qubits in depth $\mathcal{O}(\log n)$.
    To see this, we first construct a $(2^k + 2^k - 1)$-qubit unitary $V$ recursively as follows:
    \begin{enumerate}
        \item Set the $n = 2^k$ qubits to be the first set of control qubits. Set $j \leftarrow k$.
        \item Consider the $2^j$ control qubits as $2^{j-1}$ pairs of two control qubits. Include $2^{j-1}$ new ancilla qubits initialized at $\ket{0}^n$.
        \item For each pair of control qubits, implement a $\mathrm{CCX}$ gate on each newly added ancilla qubit controlled on the two control qubits.
        \item Set the new $2^{j-1}$ ancilla qubits as the set of control qubits. Set $j \leftarrow j-1$.
        \item If $j > 0$, repeat Step 2.
    \end{enumerate}
    We can compile the $\mathrm{CCX}$ gate acting on three qubits to be a sequence with a constant number of two-qubit gates. The depth of $V$ is $\mathcal{O}(\log n)$.
    The unitary $V$ computes whether all $n$ qubits are one and stores the result in the $2n-1$ qubit.
    We can implement the $n$-qubit unitary $C^{n}Z$ using a $2n$-qubit $\mathcal{O}(\log n)$-depth circuit with $n$ ancilla qubits,
    \begin{equation}
         C^{n}Z \otimes \ket{0^n} = (V \otimes I)^\dagger \, X_{2n} \, \mathrm{CZ}_{2n-1, 2n} \, X_{2n} \, (V \otimes I) \, (I_n \otimes \ket{0^n}),
    \end{equation}
    where $X_{2n}$ is the NOT gate on the one ancilla qubit not acted by $V$, $I$ is a single-qubit identity, $I_n$ is an $n$-qubit identity, and $\mathrm{CZ}_{2n-1, 2n}$ is controlled on the last ancilla qubit added in the recursive construction of $V$ and acts on the one ancilla qubit not acted by $V$.

    If one could learn $U$ up to $1/3$ error in the diamond distance with high probability or if one could distinguish whether $U$ equals to the identity $I$ or is $1/3$-far from the identity $I$ in the diamond distance with high probability, then one could successfully distinguish between the identity map $I$ and the unitary $U_x$.
    Distinguishing $I$ or one of $U_x, \forall x \in \{0, 1\}^{n}$ is the well-known Grover search problem.
    Hence, from the well-known Grover lower bound \cite{bennett1997strengths}, we have the number of queries must be at least $\Omega(2^{n/2}) = \exp(\Omega(n))$.
    This concludes the proof.
\end{proof}

\section{Learning quantum states generated by shallow circuits in 2D}

Given copies of an unknown quantum state $\ket{\psi}=U\ket{0^n}$, with the promise that $U$ is a depth-$d$ circuit acting on a 2-dimensional lattice. In this section, we present an algorithm to learn a description of a shallow circuit that prepares $\ket{\psi}$ up to a desired precision. The algorithm can be viewed as first collecting a sufficiently large randomized measurement dataset \cite{huang2020predicting, elben2022randomized} from the unknown state and then classically reconstructing the circuit based on the dataset.

% , potentially acting on $n+m$ qubits ($m$ is the number of ancilla qubits), such that
% \begin{equation}
%     \left|\bra{\psi}\otimes \bra{0^m} V\ket{0^{n+m}}\right|^2\geq 1-\varepsilon
% \end{equation}
% with high probability, with sample and time complexity that scales as $\poly(n,1/\varepsilon)$. Note that the constant factor as well as the degree of the polynomial could depend (badly) on circuit depth $d$.

\begin{definition}[Randomized measurement dataset for an unknown state] \label{def:random-measure-data-state}
    The learning algorithm accesses the unknown state via a randomized measurement dataset of the following form,
    \begin{equation} \label{eq:random-measure-data-state}
    \mathcal{T}_{\ket{\psi}}(N) = \left\{\ket{\phi_\ell} = \bigotimes_{i=1}^n \ket{\phi_{\ell, i}} \right\}_{\ell=1}^N.
    \end{equation}
    A randomized measurement dataset of size $N$ is constructed by obtaining $N$ samples from the unknown state $\ket{\psi}$.
    One sample is obtained from one experiment given as follows: measure every qubit of $\ket{\psi}$ under a random Pauli basis. The measurement collapses the state $\ket{\psi}$ to a state $\ket{\phi_\ell} = \bigotimes_{i=1}^n \ket{\phi_{\ell, i}}$, where $\ket{\phi_{\ell, i}}$ is a single-qubit stabilizer state in $\mathrm{stab}_1$.

    Together, $N$ copies of $\ket{\psi}$ construct a dataset $\mathcal{T}_{\ket{\psi}}(N)$ with $N$ samples.
    The dataset can be represented efficiently on a classical computer with $\mathcal{O}(Nn)$ bits.
    \end{definition}

% % Our current result is summarized as follows.
% \begin{theorem}\label{thm:2dstate}
% There is an algorithm that, given copies of an unknown quantum state $$ $$\ket{\psi}=U\ket{0^n}$, where $U$ is an at most depth-$d$ circuit acting on a 2D lattice with gates from a finite set, returns a depth-$2^{O(d^2)}$ circuit $V$ acting on the same lattice, such that $\left|\bra{\psi}V\ket{0^{n}}\right|^2\geq 0.99$ with probability at least 0.99. The algorithm has $\poly(n)$ sample and time complexity, and only uses single-copy measurements.
% \end{theorem}

\begin{theorem}[Learning quantum states generated by shallow circuits in 2D] \label{thm:2dstate}
Given copies of an unknown state $\ket{\psi}$, with the promise that $\ket{\psi}=U\ket{0^n}$ for an unknown $n$-qubit circuit $U$ with circuit depth $d$ acting on a 2-dimensional lattice, then the following holds.
\begin{enumerate}
    \item Suppose each two-qubit gate in $U$ is chosen from $\mathrm{SU}(4)$. With a randomized measurement dataset $\mathcal{T}_{\ket{\psi}}(N)$ of size
    \begin{equation}
    N=\frac{2^{\mathcal O(d^2)}n^{50}}{\varepsilon^{64}}\log\frac{n}{\delta},
    \end{equation}
    we can learn a quantum circuit $V$ with depth $3d$ acting on $n+m$ qubits on an extended 2-dimensional lattice, such that
    \begin{equation}
        \frac{1}{2}\norm{\Tr_B\left(V\ketbra{0^n}_A\otimes \ketbra{0^m}_B V^\dag\right) - \ketbra{\psi}}_1 \leq \varepsilon,
    \end{equation}
    with probability at least $1 - \delta$. The computational time to learn $V$ is $\left(\frac{n d^3}{\varepsilon}\right)^{\mc O(d^3)}$. The number of ancilla qubits can be chosen as $m=tn$ for an arbitrarily small constant $t>0$.
    \item In addition, if each two-qubit gate in $U$ is chosen from a finite gateset of constant size and $d=\mc O(1)$, then there is an algorithm that learns an exact preparation circuit $V$ with depth $3d$ acting on $n+m$ qubits, such that $V\ket{0^n}_A\ket{0^m}_B=\ket{\psi}_A\otimes\ket{\mathrm{junk}}_B$ with probability $1-\delta$, with sample complexity $N=\mc O(\log (n/\delta))$ and time complexity $\mc O(n \log (n/\delta))$. The number of ancilla qubits can be chosen as $m=tn$ for an arbitrarily small constant $t>0$.
    \item In addition, if each two-qubit gate in $U$ is chosen from a finite gateset of constant size and $d=\mc O(1)$, then there is an algorithm that learns a circuit $V$ with depth $2^{c\cdot d^2}$ (for some universal constant $c$) acting on $n$ qubits (without using any ancilla), such that $\left|\bra{0^n}V^\dag\ket{\psi}\right|^2\geq 1-\varepsilon$ with probability $1-\delta$, with query complexity $N=\mc O(\log (n/\delta))$ and time complexity $(n/\varepsilon)^{\mc O(1)}$.
\end{enumerate}
\end{theorem}

\begin{remark}
    The first claim in Theorem~\ref{thm:2dstate} holds for any gateset and any circuit depth $d$ (which may not be a constant), while the second and third claims are specialized to the simpler setting of finite gateset and constant depth.

    In particular, the first claim implies that when $d=\mathrm{polylog}(n)$, the state $\ket{\psi}$ can be learned within $\varepsilon$ trace distance with sample complexity $N=\frac{2^{\mathrm{polylog}(n)}}{\varepsilon^{\mc O(1)}}\log\frac{n}{\delta}$, in time $(n/\varepsilon)^{\mathrm{polylog}(n)}$.
\end{remark}

% \begin{theorem}[2D state, finite gate-set, with ancilla]
    % There is an algorithm that, given copies of an unknown quantum state $\ket{\psi}=U\ket{0^n}$ where $U$ is an at most depth-$d$ circuit acting on a 2D lattice with gates from a finite set, returns a depth-$O(d)$ circuit $V$ acting on $n+m$ qubits (m=O(n)), such that $V\ket{0^{n+m}}=\ket{\psi}\otimes \ket{\mathrm{junk}}$ with probability at least 0.99. The algorithm has $\poly(n)$ sample and time complexity, and only uses single-copy measurements.
    % \end{theorem}

% The main difference is that the first result does not require ancilla, but blows up the circuit depth exponentially; while the second result blows up the circuit depth by a constant factor but requires ancilla qubits.

% Note that the learned circuit does not require any additional ancillas. Below we give an informal argument for this result. A careful error analysis is needed to complete the proof.

We prove Theorem~\ref{thm:2dstate} in the remainder of this section. Next we give a detailed presentation of the argument outlined in Section~\ref{sec:overview2ddisentangle} and \ref{sec:overview1dfinitecorrelated}. We start by assuming a finite gate set, and address general $\mathrm{SU}(4)$ gates in Section~\ref{sec:2dlearningrobustness}.

\subsection{Learning 1D states by solving a constraint satisfaction problem}
\label{sec:1dlearning}

We start by assuming $U$ is a depth-$d$ circuit acting on a 1D lattice, for some constant $d=\mc O(1)$. The learning problem is equivalent to finding a low-depth circuit $V$ such that $V\ket{\psi}=\ket{0^n}$. Consider Fig.~\ref{fig:1dlearning} where $A$, $B$, $C$ are contiguous regions of size $3d$. Suppose we want to locally invert the qubits in region $A$ back to $\ket{0}_A$. We can do so by undoing the gates within the lightcone of $A$, i.e. apply a depth-$d$ circuit of the blue shape (that acts on $4d$ qubits) on top of $\ket{\psi}$. As we do not know what is the correct circuit to apply, we enumerate over all possible circuits of the blue shape (we can do it because its size is small). There are $2^{\mc O(d^2)}$ such circuits in total, and for each circuit we apply it to $\ket{\psi}$ and test if the state on $A$ actually equals to $\ket{0}_A$ (we can do it by measuring many copies, and seeing the outcome all-0 with high probability). For now we assume that all local inversion circuits can be found exactly; this is addressed in more detail later.

At the end of this procedure, we end up with a list of candidate circuits $\mc C_A$ of the blue shape, such that each of them is a valid local inversion of $A$, i.e., for all $V_A\in\mc C_A$ we have $V_A\ket{\psi}=\ket{0}_A\otimes \ket{\psi'}$. The inverse of the lightcone of $A$ in the unknown circuit $U$ is among them, but we don't know which one. We repeat the same procedure for each region $A$, $B$, $C$, ... and get a list of candidate local inversions $\mc C_A$, $\mc C_B$, $\mc C_C$, ... for each region.

Note that in this construction shown in Fig.~\ref{fig:1dlearning}, only the local inversions acting on neighboring regions could overlap. For example, the blue and green circuit does not overlap because $A$ and $C$ are separated by distance $3d$, and each circuit could ``spread'' into region $B$ for distance at most $d$.

The next observation is that there are certain blue circuits in $\mc C_A$ that share the same overlapping region with certain red circuits in $\mc C_B$, i.e. they share the same gates in the overlapping triangle of blue and red. For example, the inverse of the lightcone of $A$ in $U$ and the inverse of the lightcone of $B$ in $U$ share the same overlap. We call such circuits ``consistent'' with each other. Note that if two circuits are consistent, they can be merged into a bigger one. For example, take a blue circuit and a red circuit that are consistent, then they can be merged by considering the union of the gates, and applying the merged circuit to $\ket{\psi}$ will simultaneously invert both regions $A$ and $B$. If we can find a local inversion for each region such that all nearest neighbors are consistent, then they can be merged into a depth-$d$ circuit $V$ that satisfies $V\ket{\psi}=\ket{0^n}$.

\begin{figure}[t]
    % \begin{subfigure}[b]{0.4\textwidth}
    \centering
    \begin{tikzpicture}
  \draw[black] (0,0) rectangle (15,1);
  \draw[blue,thick] (0,1.2) -- (0,2.2) -- (3,2.2) -- (4,1.2) -- (0,1.2);
  \draw[red,thick] (2,1.15) -- (3,2.15) -- (6,2.15) -- (7,1.15) -- (2,1.15);
  \draw[ForestGreen,thick] (5,1.2) -- (6,2.2) -- (9,2.2) -- (10,1.2) -- (5,1.2);

\draw [decorate, decoration = {brace}] (-0.05,0) --  (-0.05,1) node[pos=0.5,anchor=east]{$d$};
\draw [decorate, decoration = {brace}] (-0.05,1.2) --  (-0.05,2.2) node[pos=0.5,anchor=east]{$d$};
\draw [decorate, decoration = {brace}] (0,2.25) --  (3,2.25) node[pos=0.5,anchor=south]{$3d$};
\draw[dashed] (3,-0.5) -- (3,2.7);
\draw[dashed] (6,-0.5) -- (6,2.7);
\draw[dashed] (9,-0.5) -- (9,2.7);
\node[anchor=north] at (1.5,0) {$A$};
\node[anchor=north] at (4.5,0) {$B$};
\node[anchor=north] at (7.5,0) {$C$};
\end{tikzpicture}
    \caption{Efficient learning of quantum states generated by a shallow circuit in 1D. For each local region $A,B,C,\dots$ we find a list of local inversion circuits, and merge them together by solving a constraint satisfaction problem.}
    \label{fig:1dlearning}
\end{figure}
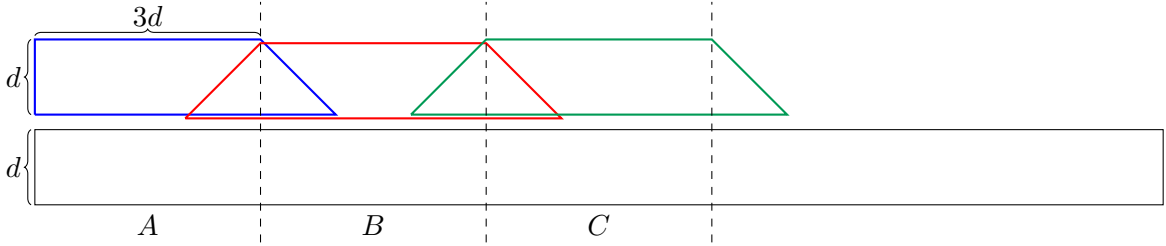

Now the task can be viewed as a constraint satisfaction problem: for each region, find a local inversion circuit  among all candidate local inversions (there are at most $2^{\mc O(d^2)}$ choices), such that each pair of nearest neighbor circuits are consistent. This can be solved efficiently by a simple dynamic programming algorithm in time $n\cdot 2^{\mc O(d^2)}$.

To be more specific, suppose the system is divided into $L=\frac{n}{3d}$ regions of size $3d$ as in Fig.~\ref{fig:1dlearning}, and suppose we have found at most $M=2^{\mc O(d^2)}$ local inversions for each region. These circuits are stored in an array $C$, where $C[i][j]$ denotes the $j$th local inversion circuit for the $i$th region. Define an arrays $cost$, where $cost[i][j]=0$ if there exists a consistent assignment at locations $1,2,\dots,i$ where $C[i][j]$ is used at location $i$; and $cost[i][j]\geq 1$ otherwise (let $cost[0][j]=0$ for all $j$). Also define an array $prev$, where $prev[i][j]$ is an index $k$, such that there exists a consistent assignment at locations $1,2,\dots,i$ where $C[i][j]$ is used at location $i$ and $C[i-1][k]$ is used at location $i-1$. $prev[i][j]$ is not defined when $cost[i][j]\geq 1$.

Once these arrays are constructed, we can take any circuit $j$ such that $cost[L][j]=0$, and construct a consistent assignment by tracing back through the $prev$ array. Let $temp$ be an array of size $M$. The following pseudocode shows how to construct these arrays in time $\mc O(L M^2)$.

\begin{algorithmic}[1]
\For{$i=1,2,\dots,L$}
    \For{$j=1,2,\dots,M$}
        \For{$k=1,2,\dots,M$}
            \State $temp[k]=cost[i-1][k]+1\left[C[i][j]\text{ is not consistent with }C[i-1][k]\right]$
        \EndFor
    \State $cost[i][j]=\min_k temp[k]$
    \If{$cost[i][j]=0$}
        \State $prev[i][j]=\argmin_k temp[k]$
    \EndIf
    \EndFor
\EndFor
\end{algorithmic}

Finally, note that the above procedure can be implemented by a two-step process:
\begin{enumerate}
    \item Learn reduced density matrices of $\ket{\psi}$ supported on the lightcone of each small region $A,B,C,\dots$.
    \item Find local inversions classically using the learned classical descriptions of the reduced density matrices, and then solve the constraint satisfaction problem.
\end{enumerate}
This is because to find local inversions, say for the $B$ region, we only need access to the reduced density matrix of $\ket{\psi}$ on the lightcone of $B$, which has $5d$ qubits, since the local inversion only acts on the reduced density matrix.

We need to learn $\frac{n}{3d}$ reduced density matrices of size at most $5d$. The following general lemma shows the complexity for learning reduced density matrices which we use throughout this section.

\begin{lemma}[Learning reduced density matrices]\label{lemma:reduceddensitymatrix}
    Let $\rho$ be an unknown $n$-qubit mixed state. Suppose we would like to learn its reduced density matrices $\rho_{A_1},\dots,\rho_{A_m}$ where $A_i$ are subsystems of size at most $k$. Given a randomized measurement dataset $\mathcal{T}_{\rho}(N)$ of size $ N=\frac{2^{\mathcal O(k)}}{\varepsilon^2}\log\frac{m}{\delta}$, we can learn a list of Hermitian matrices (not necessarily density matrices) $\{\sigma_{A_i}\}$ such that with probability at least $1-\delta$, we have $\|\rho_{A_i}-\sigma_{A_i}\|_1\leq\varepsilon$ for all $i$.
\end{lemma}
\begin{proof}
    Fix some $i$, we can write $\rho_{A_i}=\sum_{P\in\{I,X,Y,Z\}^{|A_i|}}\alpha_P P$. It suffices to learn the Pauli coefficients $\alpha_P=\frac{1}{2^{|A_i|}}\Tr(\rho_{A_i}P)=\frac{1}{2^{|A_i|}}\Tr(\rho P)$. Suppose we have learned these coefficients (denote as $\{\beta_P\}$) to within $\varepsilon_1$ precision. Let $\sigma_{A_i}:=\sum_{P\in\{I,X,Y,Z\}^{|A_i|}} \beta_P P$, then
\begin{equation}
    \left\|\rho_{A_i}-\sigma_{A_i}\right\|_1^2\leq 2^{|A_i|}\Tr(\rho-\sigma)^2=2^{2 |A_i|}\sum_P (\alpha_P-\beta_P)^2\leq 2^{4k}\varepsilon_1^2,
\end{equation}
which gives $\left\|\rho_{A_i}-\sigma_{A_i}\right\|_1\leq 2^{2k} \varepsilon_1$. Thus to achieve $\left\|\rho_{A_i}-\sigma_{A_i}\right\|_1\leq \varepsilon$ it suffices to learn $\{\Tr(\rho P)\}$ within accuracy $\varepsilon/2^k$; there are at most $m\cdot 4^k$ $k$-local Pauli operators that we need to learn.

By the main result of~\cite{huang2020predicting}, given a randomized measurement dataset of size
\begin{equation}
    N=\frac{2^{\mathcal O(k)}}{\varepsilon^2}\log\frac{m}{\delta},
\end{equation}
with probability at least $1-\delta$, we can learn all observables $\Tr(\rho P)$ for the $m\cdot 4^k$ $k$-local Pauli operators within accuracy $\varepsilon/2^k$; this is sufficient to obtain Hermitian matrices $\{\sigma_{A_i}\}$ that satisfy $\|\rho_{A_i}-\sigma_{A_i}\|_1\leq\varepsilon$ for all $i$.
\end{proof}

Note that when the gates in the unknown circuit are assumed to come from a constant-size gate set, the reduced density matrices only have $2^{\mc O(d^2)}=\mc O(1)$ choices. Therefore, choosing $\varepsilon$ to be some small constant in Lemma~\ref{lemma:reduceddensitymatrix} suffices to learn all the reduced density matrices \emph{exactly}. This allows us to find the exact local inversions by classically processing the reduced density matrices.

In summary, we have shown an algorithm that learns a depth-$d$ circuit $V$ that satisfies $\ket{\psi}=V^\dag \ket{0^n}$ with success probability $1-\delta$, using a randomized measurement dataset of size $N=\mc O(\log (n/\delta))$, in time $\mc O(n)$.

\subsection{Disentangling a 2D state}\label{sec:disentangle2dstate}

\begin{figure}[t]
    \begin{subfigure}[b]{0.45\textwidth}
    \centering
    \resizebox{\textwidth}{!}{
    \begin{tikzpicture}
\useasboundingbox (-0.3,-0.3) rectangle (8.3,8.3);

  \fill[blue!40!white] (3.75,7) rectangle node[black]{$A$} (4.25,8) ;
  \fill[red!40!white] (3.75,6) rectangle node[black]{$B$} (4.25,7);
  \fill[ForestGreen!40!white] (3.75,5) rectangle node[black]{$C$} (4.25,6);

  \draw[blue,densely dashed] (3.65,6.9) rectangle (4.35,8);
  \draw[red,densely dashed] (3.64,5.9) rectangle (4.36,7.1);
\draw[ForestGreen,densely dashed] (3.65,4.9) rectangle (4.35,6.1);

  \draw[black] (0,0) rectangle (8,8);
  \draw[black] (3.75,0) -- (3.75,8);
  \draw[black] (4.25,0) -- (4.25,8);

\draw [decorate, decoration = {brace}] (4.25,0) -- (3.75,0)node[pos=0.5,anchor=north]{$5d$};
\draw [decorate, decoration = {brace}] (3.65,8) -- (4.35,8)node[pos=0.5,anchor=south]{$7d$};
\node at (1.875,4) {$L$};
\node at (4,4) {$M$};
\node at (6.125,4) {$R$};

\end{tikzpicture}}
    \caption{}
    \end{subfigure}
    \hfill
    \begin{subfigure}[b]{0.45\textwidth}
    \resizebox{\textwidth}{!}{
    \centering
\begin{tikzpicture}
  \useasboundingbox (-0.3,-0.3) rectangle (8.3,8.3);

  \foreach \x in {1,2,...,7} {
        \fill[gray!40!white] (\x-0.25,0) rectangle node[black]{$B_\x$} (\x+0.25,8);
    }
\foreach \x in {1,2,...,8} {
        \node at (\x-0.5, 4) {$A_\x$};
    }

  \draw[black] (0,0) rectangle (8,8);

\end{tikzpicture}}
    \caption{}
    \end{subfigure}
    \caption{Learning to disentangle a quantum state generated by a shallow circuit in 2D. (a) The middle region $M$ can be inverted by solving a similar 1D constraint satisfaction problem as in Fig.~\ref{fig:1dlearning}. (b) After inverting all the gray $B_i$ regions, the remaining white $A_i$ regions are disentangled into a tensor product of pure states.}
    \label{fig:2dlearning}
\end{figure}
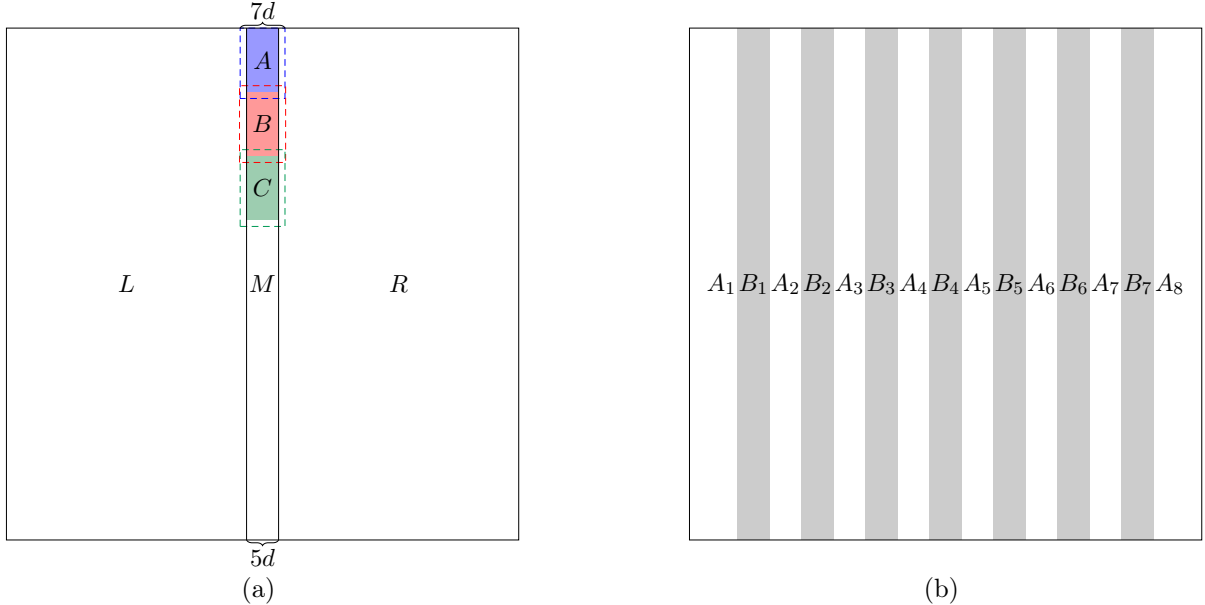

Next we use the 1D techniques developed above to disentangle a state $\ket{\psi}=U\ket{0^n}$, where $U$ is a depth-$d$ circuit acting on a 2D lattice, for some constant $d=\mc O(1)$.

For this purpose we need to introduce a general property for quantum states generated by low depth circuits, that is they have finite correlation length.

\begin{lemma}[Finite correlation length]\label{lemma:lowdepthimpliesfinitecorrelation}
Let $\ket{\psi}$ be a state generated by a depth-$d$ geometrically-local circuit (Definition~\ref{def:geo-local-circuit}). Let $A$, $B$ be two regions that are separated by distance at least $2d$ in the connectivity graph. Then $I(A:B)_{\psi}=0$. In other words, let $\rho_{AB}$, $\rho_{A}$, $\rho_{B}$ be the reduced density matrices of $\ket{\psi}$ on $AB$, $A$ and $B$, then $\rho_{AB}=\rho_A\otimes \rho_B$.
\end{lemma}
\begin{proof}
As $A$ and $B$ are separated by distance $2d$, their lightcones $L(A)$ and $L(B)$ are disjoint. $\rho_{AB}=\rho_A\otimes \rho_B$ follows from the fact that $\rho_{AB}$ is generated by the gates in $L(AB)$, which is a tensor product between $L(A)$ and $L(B)$.
\end{proof}

Fig.~\ref{fig:2dlearning} (a) shows a quantum state $\ket{\psi}$ (let $\rho=\ketbra{\psi}$) prepared by a depth-$d$ circuit on a 2D lattice, divided into three regions $L,M,R$. Since $L$ and $R$ are separated by distance $5d$, Lemma~\ref{lemma:lowdepthimpliesfinitecorrelation} implies that $\rho_{LR}=\rho_L\otimes \rho_R$. Although subsystems $L$ and $R$ are not entangled with each other, they both could be entangled with $M$. Therefore we develop an argument to invert the qubits in $M$, so that the state on $L$ and $R$ could become a tensor product of pure states.

Note that $M$ is a 1D-like region. Our goal is to find a depth-$d$ circuit $V$ acting on a slightly wider strip (of width $7d$) around $M$, such that $V\ket{\psi}=\ket{0}_M\otimes\ket{\psi'}$. Such a circuit exists since we can undo the lightcone of $M$, and we can find such a circuit using the same argument as in the previous section. In Fig.~\ref{fig:2dlearning} (a), the blue, red and green regions play the same role as in Fig.~\ref{fig:1dlearning}. For example, we can find a set of local inversions $\mc C_A$ for the shaded blue region $A$, by first learning the reduced density matrix on the dotted blue region, and then enumerating over all depth-$d$ circuits acting on the dotted blue region. After learning a set of local inversions for each local region, we can find a desired depth-$d$ circuit that inverts $M$ by solving a 1D constraint satisfaction problem.

Now, we have effectively reduced the problem of learning $\ket{\psi}$ to the following problem: given copies of a state $\ket{\psi_1}$ with the promise that
\begin{enumerate}
    \item it is prepared by a depth-$2d$ circuit (defined on a 2D lattice) acting on $\ket{0^n}$;
    \item its reduced density matrix on $M$ equals $\ketbra{0}_M$.
\end{enumerate}
The goal is to learn the state $\ket{\psi_1}$. Note that in this new state $\sigma=\ketbra{\psi_1}$, even though its circuit depth has increased from $d$ to $2d$, the reduced state on $L$ and $R$ is still in tensor product, i.e. $\sigma_{LR}=\sigma_L\otimes \sigma_R$, due to the fact that $M$ (with width $5d$) is sufficiently wide. The main purpose of inverting the $M$ region is that now $\sigma_L$ and $\sigma_R$ are guaranteed to be pure states, as shown by the following.

\begin{lemma}\label{lemma:exactdisentangle}
Let $\rho_{ABC}$ be a pure state such that the following two properties hold:
\begin{enumerate}
    \item $\rho_B=\ketbra{0}_B$,
    \item $\rho_{AC}=\rho_A\otimes \rho_C$.
\end{enumerate}
Then $\rho_A$ and $\rho_C$ are both pure states.
\end{lemma}
\begin{proof}
    This is a special case of Lemma~\ref{lemma:approxdisentangle}.
\end{proof}

Next, we apply the above argument across the entire system. In Fig.~\ref{fig:2dlearning} (b), the system is divided into many vertical strips of width $5d$. By repeating the above argument, we can learn a inverting circuit $V_i$ for each shaded $B_i$ region. Note that each $V_i$ acts on a width-$7d$ strip around $B_i$ and therefore different $V_i$s do not overlap. By combining these different inverting circuits, overall we have learned a depth-$d$ circuit $V$ such that $V\ket{\psi}=\ket{0}_B\otimes \ket{\psi'}$ where $B$ denotes the union of $B_i$.

Finally, by repeatedly applying Lemma~\ref{lemma:exactdisentangle}, we know that the reduced density matrix of $V\ket{\psi}$ on each region $A_i$ is a pure state. This means that overall the state can be written as $V\ket{\psi}=\ket{0}_B\otimes (\otimes_i \ket{\phi}_{A_i})$ for some pure states $\ket{\phi}_{A_i}$.

Now, we have disentangled the state $\ket{\psi}$ into a tensor product of many 1D-like pure states, and the problem of learning $\ket{\psi}$ is reduced to the following problem: \\

\noindent\textbf{Problem 1.} We are given copies of a state $\ket{\psi_2}$ with the promise that
\begin{enumerate}
    \item it is prepared by a depth-$2d$ circuit (defined on a 2D lattice) acting on $\ket{0^n}$;
    \item its reduced density matrix on each of the $B_i$ regions in Fig.~\ref{fig:2dlearning} (b) equals $\ketbra{0}_{B_i}$; in particular, this implies that $\ket{\psi_2}=\ket{0}_B\otimes (\otimes_i \ket{\phi}_{A_i})$ for some pure states $\ket{\phi}_{A_i}$.
\end{enumerate}
The goal is to learn the state $\ket{\psi_2}$, and it suffices to learn each of the individual states $\ket{\phi}_{A_i}$.

\subsection{Learning finite correlated states in 1D}
\label{sec:1dfinitecorrelated}

Next we show how to learn a state $\ket{\phi}$ (abbreviating the subscript $A_i$) on a specific region $A_i$ that came from Problem 1. Besides the fact that $\ket{\phi}$ is a pure state, the learning algorithm heavily relies on the property that $\ket{\phi}$ is part of a larger state that is prepared by a depth-$2d$ circuit. Note that this does not imply that $\ket{\phi}$ itself can be prepared by a depth-$2d$ circuit acting on $A_i$. Instead, we will use this property to derive useful facts about $\ket{\phi}$, presented as two different viewpoints. Each of them leads to a learning algorithm that is similar to the approach in Section~\ref{sec:1dlearning}.\\

\noindent\textbf{Viewpoint 1.} By Lemma~\ref{lemma:lowdepthimpliesfinitecorrelation}, the state $\ket{\phi}$ is a finite correlated state with correlation length $\ell=4d$. That is, let $\sigma=\ketbra{\phi}$ and let $R_1,R_2\subseteq A_i$ be two regions that are separated by distance at least $4d$, then $\sigma_{R_1 R_2}=\sigma_{R_1}\otimes \sigma_{R_2}$.\\

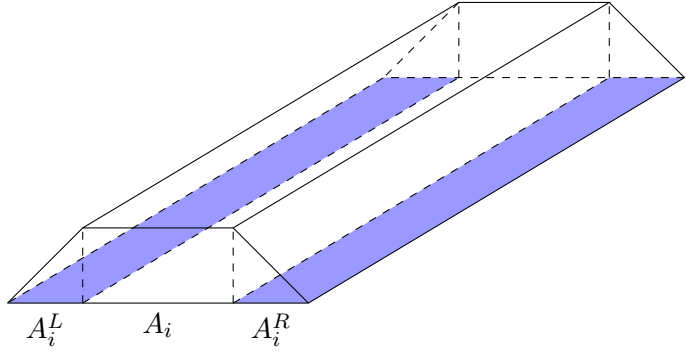
\begin{figure}[t]
    \centering
    \begin{tikzpicture}
\fill[blue!40!white]  (0,0) -- (1,0) -- (6,3)
 -- (5,3) -- cycle;
 \fill[blue!40!white]  (3,0) -- (4,0) -- (9,3)
 -- (8,3) -- cycle;
    \draw[black] (0,0) -- (4,0) -- (9,3);
    \draw[black] (1,1) -- (3,1) -- (8,4) -- (6,4) -- (1,1);
    \draw[black] (0,0) -- (1,1);
    \draw[black] (4,0) -- (3,1);
    \draw[black] (9,3) -- (8,4);
    \draw[black,dashed] (6,4) -- (5,3);
    \draw[black,dashed] (9,3) -- (5,3) -- (0,0);
    \draw[black,dashed] (1,0) -- (1,1);
    \draw[black,dashed] (3,0) -- (3,1);
    \draw[black,dashed] (8,3) -- (8,4);
    \draw[black,dashed] (6,3) -- (6,4);
    \draw[black,dashed] (1,0) -- (6,3);
    \draw[black,dashed] (3,0) -- (8,3);
    \node[anchor=north] at (0.5,0) {$A_i^L$};
    \node[anchor=north] at (2,0) {$A_i$};
    \node[anchor=north] at (3.5,0) {$A_i^R$};
\end{tikzpicture}
    \caption{Each of the states on the white $A_i$ regions in Fig.~\ref{fig:2dlearning} (b) can be viewed as being prepared by a depth-$2d$ circuit acting on $A_i$ (white) as well as ancilla qubits $A_i^L$ and $A_i^R$ (blue).}
    \label{fig:1dancilla}
\end{figure}

\noindent\textbf{Viewpoint 2.} $\ket{\phi}$ can be prepared by a depth-$2d$ circuit acting on $A_i$ as well as some ancilla qubits $A_i^L$ and $A_i^R$, shown in Fig.~\ref{fig:1dancilla}. To see this, recall that $\ket{\phi}$ is part of a state that is prepared by a depth-$2d$ circuit. Now, imagine that we \emph{undo} all the gates in that circuit, except for those in the \emph{backward lightcone} of $A_i$. This procedure does not affect the state on $A_i$, and the resulting circuit (denote as $W_i$) has exactly the same shape as in Fig.~\ref{fig:1dancilla}, where $A_i^L$, $A_i^R$ both has width $2d$. Moreover, since $\ket{\phi}$ is a pure state, it is disentangled with the ancilla qubits, which means
\begin{equation}\label{eq:preparewithancilla}
    W_i\ket{0}_{A_i^L}\ket{0}_{A_i}\ket{0}_{A_i^R}=\ket{\mathrm{junk}}_{A_i^L}\otimes\ket{\phi}\otimes \ket{\mathrm{junk}'}_{A_i^R}.
\end{equation}

Clearly, Viewpoint 2 is a much stronger characterization of $\ket{\phi}$ and derives Viewpoint 1 as a corollary; however, it involves additional ancilla qubits. In the following, we show that each of these Viewpoints itself is sufficient to derive a learning algorithm; in particular,
\begin{itemize}
    \item Using Viewpoint 1, we show that the state $\ket{\phi}$ can be prepared by a depth-$2^{\mc O(d^2)}$ circuit acting on $A_i$ (without ancilla), therefore it can be learned using the techniques in Section~\ref{sec:1dlearning}.
    \item Using Viewpoint 2, we show how to learn a depth-$2d$ circuit $W_i$ that prepares the state $\ket{\phi}$ using ancilla qubits, according to Eq.~\eqref{eq:preparewithancilla}.
\end{itemize}

Central to both of these results is a technique that allows us to disentangle a finite correlated state in 1D. For simplicity, below we present this technique for a 1D system on a line with no width.

\begin{figure}[t]
    \centering
    \begin{equation*}
\begin{aligned}
\ket{\phi}\enspace&=\enspace\begin{tikzpicture}[baseline={(0, 0-\MathAxis pt)}]
  \draw[black,very thick] (0,0) -- (7,0);
  \draw[black,densely dashed] (3,0.5) -- (3,-0.5);
  \draw[black,densely dashed] (4,0.5) -- (4,-0.5);
  \node[anchor=north] at (1.5,0) {$A$};
  \node[anchor=south] at (1.5,0) {$\rho_A$};
  \node[anchor=north] at (5.5,0) {$C$};
  \node[anchor=south] at (5.5,0) {$\rho_C$};
  \node[anchor=north] at (3.5,0) {$B$};
\end{tikzpicture}
\enspace=\enspace
    \begin{tikzpicture}[baseline={(0, 0-\MathAxis pt)}]
  \draw[black,very thick] (0,0) -- (3.4,0);
  \draw[black,very thick] (3.6,0) -- (7,0);
  \draw[black,densely dashed] (3,0.4) -- (3,-0.5);
  \draw[black,densely dashed] (4,0.4) -- (4,-0.5);
  \node[anchor=north] at (1.5,0) {$A$};
  \node[anchor=south] at (1.5,0) {$\rho_A$};
  \node[anchor=north] at (5.5,0) {$C$};
  \node[anchor=south] at (5.5,0) {$\rho_C$};
  \node[anchor=north] at (3.2,0) {$B_1$};
  \node[anchor=north] at (3.8,0) {$B_2$};
  \draw[black] (3,1.2) rectangle node {$U$} (4,0.5);
  \draw[black] (3.2,0) -- (3.2,0.5);
  \draw[black] (3.8,0) -- (3.8,0.5);
\end{tikzpicture}\\
&=\enspace \begin{tikzpicture}[baseline={(0, 0-\MathAxis pt)}]
  \draw[black,very thick] (0,0) -- (1.4,0);
  \draw[black,very thick] (1.6,0) -- (3.4,0);
  \draw[black,very thick] (3.6,0) -- (5.4,0);
  \draw[black,very thick] (5.6,0) -- (7,0);
  \node[anchor=north] at (0.5,0) {$A_1$};  \node[anchor=north] at (2.5,0) {$A_2$};
  \node[anchor=north] at (4.5,0) {$A_3$};
  \node[anchor=north] at (6.5,0) {$A_4$};
  \node[anchor=north] at (1.5,0) {$B_1$};  \node[anchor=north] at (3.5,0) {$B_2$};
  \node[anchor=north] at (5.5,0) {$B_3$};

  \foreach \x in {1,2,3} {
        \draw[black] (\x*2-1,1.2) rectangle node {$U_{\x}$} (\x*2,0.5);
        \draw[black] (\x*2-1+0.2,0) -- (\x*2-1+0.2,0.5);
        \draw[black] (\x*2-1+0.8,0) -- (\x*2-1+0.8,0.5);
        \draw[black,densely dashed] (\x*2-1,0.4) -- (\x*2-1,-0.5);
  \draw[black,densely dashed] (\x*2,0.4) -- (\x*2,-0.5);
    }
    \draw [decorate, decoration = {brace}] (1.4,-0.6) --  (0,-0.6) node[pos=0.5,anchor=north]{$R_1$};
    \draw [decorate, decoration = {brace}] (3.4,-0.6) --  (1.6,-0.6) node[pos=0.5,anchor=north]{$R_2$};
    \draw [decorate, decoration = {brace}] (5.4,-0.6) --  (3.6,-0.6) node[pos=0.5,anchor=north]{$R_3$};
    \draw [decorate, decoration = {brace}] (7,-0.6) --  (5.6,-0.6) node[pos=0.5,anchor=north]{$R_4$};
\end{tikzpicture}
\end{aligned}
\end{equation*}
    \caption{Disentangling a finite correlated state in 1D.}
    \label{fig:disentangle1dfinitecorrelated}
\end{figure}
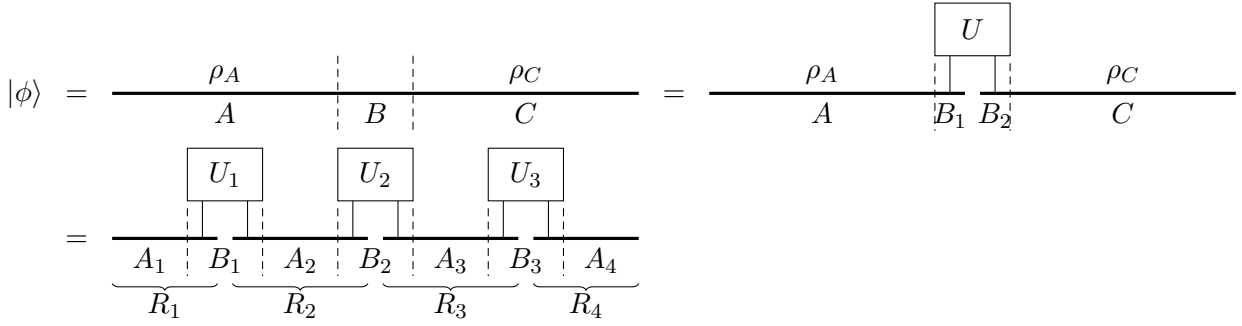

\begin{lemma}[Disentangling finite correlated states in 1D]\label{lemma:disentanglefinitecorrelated}
Let $\ket{\phi}$ be a state defined on a line with correlation length $\ell$, that is, every two regions $R_1$, $R_2$ that are separated by distance at least $\ell$ have zero mutual information, i.e. $\rho_{R_1 R_2}=\rho_{R_1}\otimes \rho_{R_2}$, where $\rho=\ketbra{\phi}$. Divide the 1D line into contiguous regions of size $\ell$, denote as $A_1,B_1,A_2,B_2,\dots,B_{L-1},A_L$ (Fig.~\ref{fig:disentangle1dfinitecorrelated}). Then for each $i$ there exists a unitary $U_i$ acting on the $B_i$ region, such that $\prod_{i=1}^{L-1}U_i \ket{\phi}$ is a tensor product of $L$ pure states.

% Then $\ket{\psi}$ can be generated by a circuit shown in Fig.~\ref{fig:1dfinitecorrelation}, where each unitary acts on $O(\ell)$ qubits. This implies that $\ket{\psi}$ can be generated by a geometrically-local circuit in 1D with depth $2^{O(\ell)}$.
\end{lemma}
\begin{proof}
We start with three subsystems $A,B,C$ (first line of Fig.~\ref{fig:disentangle1dfinitecorrelated}), where $B$ has size $\ell$. Then we have
    \begin{equation}
        \rank(\rho_B)=\rank(\rho_{AC})=\rank(\rho_A\otimes \rho_C)=\rank(\rho_A)\cdot \rank(\rho_C)\leq \dim(B).
    \end{equation}
    Purifying the state $\rho_A$ ($\rho_C$) requires an ancilla system with dimension $\rank(\rho_A)$ ($\rank(\rho_C)$). Therefore we can partition $B$ into two systems $B_1$, $B_2$, such that there exists pure states $\ket{\phi_1}_{AB_1}$ and $\ket{\phi_2}_{B_2C}$, such that $\ket{\phi_1}_{AB_1}$ is a purification of $\rho_A$, and $\ket{\phi_2}_{B_2C}$ is a purification of $\rho_C$. This implies that $\ket{\phi_1}_{AB_1}\otimes \ket{\phi_2}_{B_2C}$ is a purification of $\rho_{AC}$. Since $\ket{\phi}_{ABC}$ is also a purification of $\rho_{AC}$, by Uhlmann's theorem there exists a unitary $U_B$ such that $\ket{\phi}_{ABC}=U_B\ket{\phi_1}_{AB_1}\otimes \ket{\phi_2}_{B_2C}$.

    Applying this argument independently at different $B_i$ regions (bottom line of Fig.~\ref{fig:disentangle1dfinitecorrelated}), we have that for each $i=1,2,\dots,L-1$, there exists a partition of the system $B_i$ as two systems $B_i^L$ and $B_i^R$, as well as a unitary $U_i$ acting on $B_i=B_i^L\cup B_i^R$, such that
    \begin{equation}\label{eq:disentangleabc}
        \ket{\phi}=U_i\ket{\phi_1}_{A_1\dots B_i^L}\otimes \ket{\phi_2}_{B_i^R A_{i+1}\cdots A_L},
    \end{equation}
    or equivalently, $U_i^\dag\ket{\phi}=\ket{\phi_1}_{A_1\dots B_i^L}\otimes \ket{\phi_2}_{B_i^R A_{i+1}\cdots A_L}$, for some pure states $\ket{\phi_1}$ and $\ket{\phi_2}$. Next, we relabel the systems according to
    \begin{equation}
        R_i:= B_{i-1}^R\cup A_i\cup B_i^L.
    \end{equation}
    Intuitively, after applying all $U_i^\dag$s, the system must be disentangled across all the $R_i$ regions. To prove this we use a simple argument based on the strong subadditivity of quantum entropy (Lemma~\ref{lemma:ssa}).

    Let $\sigma:=\left(\prod_{i=1}^{L-1}U_i^\dag\right) \ketbra{\phi}\left(\prod_{i=1}^{L-1}U_i\right)$ be the final (pure) state. Fix some $i$, our goal is to prove that $\sigma_{R_i}$ is pure, i.e., $S(\sigma_{R_i})=0$. The strong subadditivity of quantum entropy gives
    \begin{equation}
        S(\sigma_{R_i})\leq S(\sigma_{R_1\dots R_i})+S(\sigma_{R_i\dots R_L})-S(\sigma)=S(\sigma_{R_1\dots R_i})+S(\sigma_{R_i\dots R_L}).
    \end{equation}
    Note that when calculating $S(\sigma_{R_1\dots R_i})$ we can \emph{undo} all the unitaries $U_j^\dag$ for $j<i$ due to the invariance of entropy under unitary. Then $S(\sigma_{R_1\dots R_i})=0$ immediately follows from Eq.~\eqref{eq:disentangleabc}, and a similar argument shows $S(\sigma_{R_i\dots R_L})=0$, which concludes the proof.
\end{proof}

\begin{lemma}[Strong subadditivity of quantum entropy~\cite{Lieb2003proof}]\label{lemma:ssa}
Let $\rho$ be a mixed state defined on three systems $A,B,C$. Let $S(\rho):=-\Tr(\rho\log \rho)$ be the von Neumann entropy. Then we have
\begin{equation}
    S(\rho_{ABC})+S(\rho_B)\leq S(\rho_{AB})+S(\rho_{BC}).
\end{equation}
\end{lemma}

\noindent\textbf{Learning under Viewpoint 1.} A corollary of Lemma~\ref{lemma:disentanglefinitecorrelated} is that any finite correlated state in 1D can be prepared by a low-depth circuit, because each of the small pure state on the $R_i$ regions in the bottom line of Fig.~\ref{fig:disentangle1dfinitecorrelated} can be prepared by a local unitary acting on $\mc O(\ell)$ qubits. Applying this argument to the state $\ket{\phi}_{A_i}$ shown in Fig.~\ref{fig:1dancilla}, we conclude that it can be prepared by two layers of unitaries acting on $\mc O(d^2)$ qubits, acting on the $A_i$ region only. This implies that the state $\ket{\phi}_{A_i}$ can be prepared by a depth-$2^{\mc O(d^2)}$ circuit acting on $A_i$, and thus can be learned by applying the argument in Section~\ref{sec:1dlearning}.\\

\noindent\textbf{Learning under Viewpoint 2.} The main drawback of the above argument is that the learned circuit depth has an exponential blowup. To reduce this blowup we use additional structure of the state $\ket{\phi}_{A_i}$, described in Viewpoint 2 and Fig.~\ref{fig:1dancilla}. Note that there is a key difference between learning the state $\ket{\phi}_{A_i}$ and learning 1D states discussed in Section~\ref{sec:1dlearning}. Here, while the state $\ket{\phi}_{A_i}$ has a low-depth property shown in Fig.~\ref{fig:1dancilla}, this property relies on ancilla qubits (the $\ket{\mathrm{junk}}$ states in Eq.~\eqref{eq:preparewithancilla}) that \emph{we do not have access to}. Therefore we cannot directly apply the techniques in Section~\ref{sec:1dlearning}, which requires access to all qubits prepared by the low-depth circuit.

The main idea is to learn a mixed state $\rho$ that is \emph{locally consistent} with the state $\ketbra{\phi}$, i.e., they have the same local reduced density matrices, and then show that this forces the two states to be globally the same.

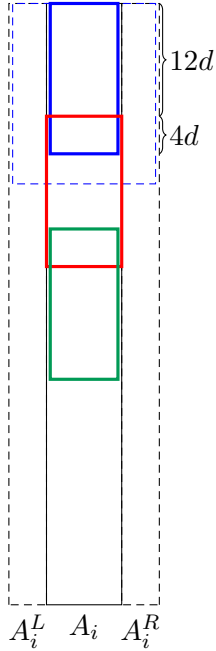
\begin{figure}[t]
    \centering
    \begin{tikzpicture}
\draw[black] (0,0) rectangle (1,8);
\draw[black,densely dashed] (-0.5,0) rectangle (0,8);
\draw[black,densely dashed] (1.5,0) rectangle (1,8);
\draw[blue,densely dashed] (-0.45,5.6) rectangle (1.45,8);

  \draw[blue,very thick] (0.05,6) rectangle (0.95,8);
  \draw[red,very thick] (0,4.5) rectangle (1,6.5);
  \draw[ForestGreen,very thick] (0.05,5) rectangle (0.95,3);
\node[anchor=north] at (-0.25,0) {$A_i^L$};
\node[anchor=north] at (0.5,0) {$A_i$};
\node[anchor=north] at (1.25,0) {$A_i^R$};
\draw [decorate, decoration = {brace}] (1.5,6.5)--(1.5,6)   node[pos=0.5,anchor=west]{$4d$};
\draw [decorate, decoration = {brace}] (1.5,8)--(1.5,6.5)   node[pos=0.5,anchor=west]{$12d$};
\end{tikzpicture}
    \caption{Learning a quantum state generated by a depth-$2d$ circuit with ancilla.}
    \label{fig:1doverlapping}
\end{figure}

The argument is illustrated in Fig.~\ref{fig:1doverlapping}, where we learn to locally \emph{prepare} the state instead of \emph{invert} the state. Consider the state $\ket{\phi}$ on the $A_i$ region shown in Fig.~\ref{fig:1doverlapping}, and suppose we have learned its reduced density matrix $\rho_{\mathrm{blue}}$ on the solid blue region. Due to the fact that $\ket{\phi}$ is prepared by a depth-$2d$ circuit acting on $A_i^L,A_i,A_i^R$, we know that there exists a depth-$2d$ circuit acting on the dotted blue region that prepares $\rho_{\mathrm{blue}}$ (the circuit looks like a small piece of Fig.~\ref{fig:1dancilla}), by undoing all the gates except for those in the backward lightcone of the solid blue region. We can perform a brute force search over all depth-$2d$ circuits acting on the dotted blue region, and for each of them we can test whether it prepares $\rho_{\mathrm{blue}}$. In this way we obtain a list of depth-$2d$ circuits acting on the dotted blue region that prepares $\rho_{\mathrm{blue}}$.

By repeating the above procedure we can obtain a list of local preparation circuits for each of the solid colored regions. A key point here is that the neighboring colored regions overlap by distance $4d$. Moreover, the local preparation circuits for the blue and green regions do not overlap, since the red region is sufficiently big. This enables us to solve a constraint satisfaction problem of the same nature as in Section~\ref{sec:1dlearning}, where we can choose a local preparation circuit for each region, such that neighboring circuits are consistent and can be merged together. Overall we have learned a depth-$2d$ circuit $W$ acting on $A_i^L,A_i,A_i^R$, that simultaneously prepares all the local reduced density matrices.

Let $\rho:=\Tr_{A_i^L A_i^R}(W\ketbra{0}_{A_i^L A_i A_i^R}W^\dag)$ be the learned density matrix on $A_i$. At this point we know that $\rho$ and $\ketbra{\phi}$ are locally the same on the solid blue, red, and green regions (and so on), but this does not directly imply that $\rho=\ketbra{\phi}$. For example, a Haar random pure state and the maximally mixed state are locally very close but globally very far. Next, we show that the finite correlation property forces $\rho$ and $\ketbra{\phi}$ to be globally equal.

% Finally we have a useful lemma showing that local consistency implies global consistency.
\begin{lemma}[Local consistency implies global consistency]\label{lemma:consistency}
    Let $\ket{\psi}$ be a state defined on a 1D line with correlation length $\ell$ and let $\sigma=\ketbra{\psi}$. Suppose the system is partitioned into contiguous regions $A_1,\dots,A_L$ where $|A_i|\geq \ell$. Suppose $\rho$ is a mixed state that satisfies $\rho_{A_i A_{i+1}}=\sigma_{A_i A_{i+1}}$ for all $i$, then $\rho=\sigma$.
\end{lemma}
\begin{proof}
    We show this for 3 subsystems; generalizing to more subsystems is straightforward. Let $\rho$ be a mixed state satisfying $\rho_{A_1 A_2}=\sigma_{A_1 A_2}$ and $\rho_{A_2 A_3}=\sigma_{A_2 A_3}$. Following the proof of Lemma~\ref{lemma:disentanglefinitecorrelated}, there exists a unitary $U$ acting on $A_2$ such that
    \begin{equation}
        U_{A_2}\ket{\psi}_{A_1 A_2 A_3} = \ket{\phi_1}_{A_1 A_{21}}\otimes \ket{\phi_2}_{A_{22} A_3},
    \end{equation}
    where $A_{21},A_{22}$ is a partition of $A_2$, and $\ket{\phi_1}_{A_1 A_{21}}$, $\ket{\phi_2}_{A_{22} A_3}$ are some pure states. Equivalently, we have
    \begin{equation}\label{eq:consistencydisentangle}
        U_{A_2} \sigma U_{A_2}^\dag=\ketbra{\phi_1}_{A_1 A_{21}}\otimes \ketbra{\phi_2}_{A_{22} A_3}.
    \end{equation}

    Let $\tau := U_{A_2} \rho U_{A_2}^\dag$, we will show that $\tau = \ketbra{\phi_1}_{A_1 A_{21}}\otimes \ketbra{\phi_2}_{A_{22} A_3}$, which implies $\rho=\sigma$.

    First, taking the partial trace over $A_3$ on both sides of Eq.~\eqref{eq:consistencydisentangle}, we have
    \begin{equation}
        U \sigma_{A_1 A_2} U^\dag=\ketbra{\phi_1}_{A_1 A_{21}}\otimes \Tr_{A_3}\ketbra{\phi_2}.
    \end{equation}
    Then, notice that
    \begin{equation}
        \tau_{A_1 A_2}=U \rho_{A_1 A_2} U^\dag=U \sigma_{A_1 A_2} U^\dag=\ketbra{\phi_1}_{A_1 A_{21}}\otimes \Tr_{A_3}\ketbra{\phi_2}.
    \end{equation}
    Tracing out $A_{22}$ on both sides, we have $\tau_{A_1 A_{21}}=\ketbra{\phi_1}_{A_1 A_{21}}$; similarly, $\tau_{A_{22} A_{3}}=\ketbra{\phi_2}_{A_{22} A_{3}}$. Since $\tau_{A_1 A_{21}}$ and $\tau_{A_{22} A_{3}}$ are both pure states, this implies that the global state $\tau$ is a tensor product
    \begin{equation}
        \tau=\tau_{A_1 A_{21}}\otimes \tau_{A_{22} A_{3}}=\ketbra{\phi_1}_{A_1 A_{21}}\otimes \ketbra{\phi_2}_{A_{22} A_3}.
    \end{equation}
    Thus we have $\tau = U \sigma U^\dag$, which implies $\rho=\sigma$.
\end{proof}

\noindent\textbf{Summary of our progress so far.} So far we have developed all technical ingredients for learning a quantum state $\ket{\psi}=U\ket{0^n}$, under the simplified setting that $U$ is a depth $d=\mc O(1)$ circuit acting on a 2D lattice, and each gate in $U$ is from a constant size gate set.

Note that all the above arguments can be viewed as first learning the local reduced density matrices of $\ket{\psi}$ followed by classically reconstructing the circuit. As we have discussed before in Section~\ref{sec:1dlearning}, a reduced density matrix of constant size can be learned \emph{exactly} as it only has a constant number of choices. In the disentangling step shown in Fig.~\ref{fig:2dlearning}, we can learn $\mc O(n)$ reduced density matrices on the dotted regions of size $\mc O(d^2)$, and then classically reconstruct a depth-$d$ circuit $V$ in time $\mc O(n)$, such that $V\ket{\psi}=\ket{0}_B\otimes (\otimes_i \ket{\phi}_{A_i})$ where the pure states $\ket{\phi}_{A_i}$ live on the white regions of Fig.~\ref{fig:2dlearning} (b).\\

\noindent\textbf{Proof of second claim of Theorem~\ref{thm:2dstate}.} Next, we start with Viewpoint 2. As shown in Fig.~\ref{fig:1doverlapping}, learning a state $\ket{\phi}_{A_i}$ requires learning its reduced density matrices of size $5d\times 16d$. This can be achieved by experimentally applying $V$ to $\ket{\psi}$ and then learning the reduced density matrices. Equivalently, say we want to learn the reduced density matrix of $\ket{\phi}_{A_i}$ on a region $M$ of size $5d\times 16d$, then it suffices to learn a reduced density matrix of $\ket{\psi}$ of size $7d\times 18d$ on a region surrounding $M$, then \emph{classically} apply the gates of $V$ within the backward lightcone of $M$, and then classically trace out the qubits outside $M$. In other words, the reduced density matrices of $\ket{\phi}_{A_i}$ can be simulated by slightly larger reduced density matrices of $\ket{\psi}$. Using these reduced density matrices, for each $i$ we can learn a depth-$2d$ circuit $W_i$ such that
\begin{equation}
    W_i \ket{0}_{A_i^L A_i A_i^R} = \ket{\phi}_{A_i}\otimes \ket{\mathrm{junk}}_{A_i^L A_i^R},
\end{equation}
which takes total time $\mc O(n)$. The entire process requires $\mc O(n)$ reduced density matrices of $\ket{\psi}$ of size $\mc O(d^2)$, which can be learned exactly with probability at least $1-\delta$, using a randomized measurement dataset of size $N=\mc O(\log (n/\delta))$.

The state $\ket{\psi}$ can be prepared as follows:
\begin{enumerate}
    \item Initialize registers $A_i,B_i,A_i^L,A_i^R$ in the state $\ket{0}$. Let $A=\cup_i A_i$ and $B=\cup_i B_i$.
    \item For each $i$, apply the depth-$2d$ circuit $W_i$ to $A_i^L A_i A_i^R$.
    \item Apply the depth-$d$ circuit $V^\dag$ to $AB$, and the state $\ket{\psi}$ lives on $AB$.
\end{enumerate}
Overall the learned circuit has depth $3d$ and can be implemented on an extended 2D lattice, where the qubits in $A_i$ can interact with its ancilla qubits $A_i^L,A_i^R$ as well as neighboring $B_i$ regions.

In Fig.~\ref{fig:1doverlapping} we have chosen the width of $A_i$ to be $5d$. Note that the width of $A_i^L$ and $A_i^R$ are both $2d$, regardless of the width of $A_i$. In fact we could have chosen the width of $A_i$ to be $Cd$ for some large constant $C$, and the number of ancilla qubits is at most $n/(Cd)\cdot 4d=\frac{4}{C}n$, which can be made arbitrarily small.\\

\noindent\textbf{Proof of third claim of Theorem~\ref{thm:2dstate}.} Using Viewpoint 1, the state $\ket{\phi}_{A_i}$ can be prepared by a depth-$2^{\mc O(d^2)}$ circuit acting on $A_i$, and thus can be learned by applying the argument in Section~\ref{sec:1dlearning}. Let $\ket{\phi}_{A_i}=W\ket{0}_{A_i}$ for some depth-$2^{\mc O(d^2)}$ circuit $W$ acting on $A_i$. A technical issue here is that we no longer have the guarantee that $W$ consists of gates from a finite gate set as in $U$, because the existence of $W$ comes from the disentangling argument in Lemma~\ref{lemma:disentanglefinitecorrelated}, instead of coming from the original circuit $U$ as in Viewpoint 2. Below we discuss how to find this circuit $W$.

Let $d'=2^{\mc O(d^2)}$ be the circuit depth of $W$. Following Section~\ref{sec:1dlearning}, we can learn reduced density matrices of $\sigma:=\ketbra{\phi}_{A_i}$ of size $5d\times 5d'$ (which can be done exactly, as discussed above) and then classically find local inversions for regions of size $5d\times 3d'$. Following Fig.~\ref{fig:1dlearning}, let $A$ be a region of size $5d\times 3d'$, and let $AA_1$ be the lightcone of $A$ with size $5d\times 4d'$. Then there is a depth-$d'$ circuit $W_{AA_1}$ acting on $AA_1$ such that
\begin{equation}
    \Tr_{A_1}\left(W_{AA_1}\sigma_{AA_1} W_{AA_1}^\dag\right)=\ketbra{0}_A.
\end{equation}
To find the local inversion $W_{AA_1}$ we use an $\varepsilon_0$-net over depth-$d'$ circuits acting on $AA_1$, denoted as $\mc N_{\varepsilon_0}(AA_1)$ (see Definition~\ref{def:epsnet} and Lemma~\ref{lemma:epsnet}), which has size at most
\begin{equation}
    S=\left(\frac{d'^3}{\varepsilon_0}\right)^{\mc O(d'^3)}.
\end{equation}
By definition, there exists $\hat W_{AA_1}\in \mc N_{\varepsilon_0}(AA_1)$ such that $\|\hat W_{AA_1}- W_{AA_1}\|_\infty\leq\varepsilon_0$, which gives
\begin{equation}
    \expval{\Tr_{A_1}\left(\hat W_{AA_1}\sigma_{AA_1} \hat W_{AA_1}^\dag\right)}{0_A}\geq 1-2\varepsilon_0.
\end{equation}
By enumerating over every element in $\mc N_{\varepsilon_0}(AA_1)$, we can find a list of circuits which satisfy the above equation. Following the argument in Section~\ref{sec:1dlearning}, we repeat the same procedure for each local region and merge the local circuits into a global depth-$d'$ circuit $\hat W_{i}$, which approximately inverts each local region up to $1-2\varepsilon_0$ fidelity. By union bound, we have
\begin{equation}
    \left|\bra{0_{A_i}}\hat W_{i}\ket{\phi}_{A_i}\right|^2\geq 1-2\sqrt{n}\varepsilon_0.
\end{equation}
After learning each region $A_i$, the state $\ket{\psi}$ can be approximately prepared as follows:
\begin{enumerate}
    \item Initialize registers $A_i,B_i$ in the state $\ket{0}$. Let $A=\cup_i A_i$ and $B=\cup_i B_i$.
    \item For each $i$, apply the depth-$d'$ circuit $\hat W_i^\dag$ to $A_i$.
    \item Apply the depth-$d$ circuit $V^\dag$ to $AB$, and the state on $AB$, which is $\ket{\hat\psi}=V^\dag(\otimes_i \hat W_i^\dag)\ket{0^n}$, approximately equals to $\ket{\psi}$.
\end{enumerate}
We bound the approximation error as follows.
\begin{equation}
    \left|\braket{\hat\psi|\psi}\right|^2=\left|\bra{0^n}(\otimes_i \hat W_i) V\ket{\psi}\right|^2=\prod_i \left|\bra{0_{A_i}}\hat W_{i} \ket{\phi}_{A_i}\right|^2\geq 1-2n\varepsilon_0.
\end{equation}
Therefore to achieve $1-\varepsilon$ fidelity it suffices to choose $\varepsilon_0=\frac{\varepsilon}{2n}$, which gives total running time $n\cdot S=(n/\varepsilon)^{\mc O(1)}$.

\subsection{Robustness to imprecision}
\label{sec:2dlearningrobustness}

In the previous sections we have been focusing on a finite gateset, which allows us to learn reduced density matrices exactly, and therefore the disentangling procedure in Fig.~\ref{fig:2dlearning} can be performed exactly. However, it's not clear that this argument still works for general $\mathrm{SU}(4)$ gates, because in this case each step can only be performed \emph{approximately}. In particular, we can only approximately disentangle the state using the procedure in Fig.~\ref{fig:2dlearning}, and learning the remaining 1D states poses new technical challenges as they are no longer pure.

In this section we address this issue. In the following we first outline the argument and develop key technical lemmas, before going into the full proof of the first claim in Theorem~\ref{thm:2dstate}.

We start with the disentangling step in Fig.~\ref{fig:2dlearning}. Here, instead of exhaustively enumerating small circuits acting on local regions, we can only enumerate over an $\varepsilon$-net of the circuit. Therefore, we are only able to find circuits that approximately invert each $B_i$ region shown in Fig.~\ref{fig:2dlearning} (b). This means that after the disentangling step, the reduced density matrix on $B$ will be \emph{close} to $\ketbra{0}_B$, instead of being \emph{exactly equal} to $\ketbra{0}_B$.

Now the question is what happens to the remaining $A_i$ regions. Note that the state is still in tensor product across different $A_i$ regions due to the finite correlation length property, but the reduced density matrices on each $A_i$ region will not be pure. The following lemma shows that these states are \emph{approximately pure}.

\begin{lemma}\label{lemma:approxdisentangle}
Let $\rho_{A_1 A_2\dots A_L B}$ be a pure state such that the following two properties hold:
\begin{enumerate}
    \item $\expval{\rho_B}{0_B}\geq 1-\varepsilon$,
    \item $\rho_{A_1 A_2\dots A_L}=\rho_{A_1}\otimes\cdots\otimes  \rho_{A_L}$.
\end{enumerate}
Then for each $i=1,\dots,L$ there exists a pure state $\ket{\phi}_{A_i}$ such that $\expval{\rho_{A_i}}{\phi_{A_i}}\geq 1-\varepsilon$.
\end{lemma}
\begin{proof}
Consider the operator norm $\|\rho\|_\infty:=\lambda_{\max}(\rho)=\max_{\ket{\psi}}\expval{\rho}{\psi}$. Condition 1 gives $\|\rho_B\|_\infty\geq 1-\varepsilon$. Using condition 2 we have
\begin{equation}
    \|\rho_B\|_\infty=\|\rho_{A_1\dots A_L}\|_\infty=\|\rho_{A_1}\otimes\cdots\otimes \rho_{A_L}\|_\infty=\prod_{i=1}^L\|\rho_{A_i}\|_\infty\geq 1-\varepsilon,
\end{equation}
which implies that $\lambda_{\max}(\rho_{A_i})\geq 1-\varepsilon$ for any $i$.
\end{proof}

Next, we discuss how to learn these states $\{\rho_{A_i}\}$ that are approximately pure. Again, we still have the property that each $\rho_{A_i}$ is a 1D-like state with finite correlation length. However, our previous techniques developed in Section~\ref{sec:1dfinitecorrelated} only work for \emph{exactly} pure states. We develop new techniques by examining the robustness of the key technical lemma developed in Section~\ref{sec:1dfinitecorrelated}, Lemma~\ref{lemma:disentanglefinitecorrelated}.

There are two key ingredients in the proof of Lemma~\ref{lemma:disentanglefinitecorrelated}:
\begin{enumerate}
    \item The use of Uhlmann's theorem to prove the existence of a local disentangling unitary;
    \item The use to entropy inequalities (in particular, strong subadditivity) to prove that the state is disentangled into many local pieces after applying Uhlmann's unitaries across the entire system.
\end{enumerate}

Fortunately, both ingredients are robust. First, Uhlmann's theorem says that if two mixed states are close, then there exists a unitary (acting on the purifying system) that approximately maps between their purifications. Second, entropy inequalities are robust, thanks to the continuity of entropy given below.

\begin{lemma}[Fannes–Audenaert inequality]
Let $\rho$, $\sigma$ be two $n$-qubit density matrices, and let $\varepsilon:=\frac{1}{2}\|\rho-\sigma\|_1$. Then
\begin{equation}
    \left|S(\rho)-S(\sigma)\right|\leq n\varepsilon+h(\varepsilon),
\end{equation}
where $h(\cdot)$ is the binary entropy function and can be upper bounded as $h(\varepsilon)\leq 2\sqrt{\varepsilon}$.
\end{lemma}

We formalize the above intuitions as the following main technical lemma, which is a robust version of Lemma~\ref{lemma:consistency}.

\begin{lemma}\label{lemma:approxconsistency}
    Let $\rho$ be an $n$-qubit mixed state defined on systems $A_1,\dots,A_L$, with the following properties:
\begin{enumerate}
    \item there exists an $n$-qubit pure state $\ket{\psi}$, such that $\expval{\rho}{\psi}\geq 1-\varepsilon$.
    \item for any $i=2,3,\dots,L-1$, it holds that $I(A_1\cdots A_{i-1}:A_{i+1}\cdots A_L)_{\rho}=0$.
\end{enumerate}
For simplicity we assume that $L$ is odd.
Let $\sigma$ be another $n$-qubit mixed state that satisfies
\begin{equation}
    \frac{1}{2}\|\sigma_{A_{2i} A_{2i+1} A_{2i+2}}-\rho_{A_{2i} A_{2i+1} A_{2i+2}}\|_1\leq\delta,\quad \forall i=0,1,\dots,(L-1)/2,
\end{equation}
Then
\begin{equation}
    \frac{1}{2}\left\|\sigma-\rho\right\|_1\leq 13n\varepsilon^{1/16}+4n\delta^{1/4}.
\end{equation}
\end{lemma}

\begin{proof}
    The above condition says that $\rho$ and $\sigma$ are close on local regions $A_1 A_2,A_2 A_3 A_4, A_4 A_5 A_6,\dots,A_{L-1} A_L$. The goal is to prove that they are globally close.

Let $\tau:=\ketbra{\psi}$ denote the density matrix of $\ket{\psi}$. For any $j\in\{1,2,\dots,(L-1)/2\}$, define three regions $L^{(j)}:=A_{\leq 2j-1}$, $M^{(j)}:=A_{2j}$, $R^{(j)}:=A_{\geq 2j+1}$ (the superscript $(j)$ is abbreviated when there is no confusion).

Note that for any subsystem $W$, we have
\begin{equation}
    \frac{1}{2}\|\tau_W-\rho_W\|_1\leq \frac{1}{2}\|\tau-\rho\|_1\leq \sqrt{1-\expval{\rho}{\psi}}\leq \sqrt{\varepsilon}.
\end{equation}
Therefore,
\begin{equation}
\begin{aligned}
    \|\tau_{LR}-\tau_L\otimes \tau_R\|_1&\leq \|\tau_{LR}-\rho_{LR}\|_1+\|\rho_{LR}-\rho_L\otimes \rho_R\|_1+\|\rho_L\otimes \rho_R-\tau_L\otimes \tau_R\|_1\\
    &\leq \|\tau_{LR}-\rho_{LR}\|_1+\|\rho_L-\tau_L\|_1+\|\rho_R-\tau_R\|_1\\
    &\leq \varepsilon_1
\end{aligned}
\end{equation}
where we let $\varepsilon_1:=6\sqrt{\varepsilon}$. Then, the relationship between fidelity and trace distance implies that
\begin{equation}
    F(\tau_{LR},\tau_L\otimes \tau_R)\geq 1-\|\tau_{LR}-\tau_L\otimes \tau_R\|_1\geq 1-\varepsilon_1.
\end{equation}
Let $\ket{\phi_1}_{LM_1^{(j)}}$ be a purification of $\tau_L$, and let $\ket{\phi_2}_{M_2^{(j)} R}$ be a purification of $\tau_R$. Note that $\dim(M_1^{(j)})\leq \dim(L)$ and $\dim(M_2^{(j)})\leq \dim(R)$. Let $M'^{(j)}$ be an ancilla space with dimension $\dim(M_1^{(j)})\dim(M_2^{(j)})/\dim(M^{(j)})$. Here $M'^{(j)}$ is needed in case $M^{(j)}$ is smaller than $M_1^{(j)} M_2^{(j)}$. Now, $\ket{\psi}_{LMR}\ket{0}_{M'^{(j)}}$ is a purification of the state $\tau_{LR}$, while $\ket{\phi_1}_{LM_1^{(j)}}\otimes \ket{\phi_2}_{M_2^{(j)}R}$ is a purification of the state $\tau_L \otimes \tau_R$, and they have the same dimension. Then by Uhlmann's theorem, there exists a unitary $U^{(j)}:M^{(j)}M'^{(j)}\to M_1^{(j)} M_2^{(j)}$, such that
\begin{equation}\label{eq:disentangle}
    U^{(j)}_{M^{(j)}M'^{(j)}}\ket{\psi}_{LM^{(j)}R}\ket{0}_{M'^{(j)}}\approx_{\varepsilon_1}\ket{\phi_1}_{LM_1^{(j)}}\otimes \ket{\phi_2}_{M_2^{(j)}R}.
\end{equation}
Here, $\ket{u}\approx_\varepsilon\ket{v}$ means $|\braket{u|v}|^2\geq 1-\varepsilon$.

The above argument shows the existence of a unitary $U^{(j)}$ acting on $M^{(j)}=A_{2j}$ (as well as an ancilla system $M'^{(j)}$), that approximately disentangles the state $\ket{\psi}$ into a tensor product between $LM_1^{(j)}$ and $M_2^{(j)}R$, where $M_1^{(j)}$, $M_2^{(j)}$ are ancilla systems associated with $A_{2j}$. We apply all such unitaries $U^{(j)}$ ($j\in\{1,2,\dots,(L-1)/2\}$) to $\ket{\psi}$, and obtain
\begin{equation}
    \eta:=\left(\prod_{j=1}^{(L-1)/2}U^{(j)}\right)\ketbra{\psi}\otimes\ketbra{0}_{M'}\left(\prod_{j=1}^{(L-1)/2}U^{(j)\dag}\right),
\end{equation}
where $M'$ represents the union of all $M'^{(j)}$. Note that $\eta$ supports on $A_1,A_3,A_5,\dots,A_L$ as well as $M_1^{(j)},M_2^{(j)}$ for $j\in\{1,2,\dots,(L-1)/2\}$. Now, we relabel the systems according to
\begin{equation}
    B_j:=M_2^{(j-1)}\cup A_{2j-1}\cup M_1^{(j)},\quad j\in\{1,2,\dots,(L+1)/2\},
\end{equation}
and the state $\eta$ supports on $B_j$, $j\in\{1,2,\dots,(L+1)/2\}$, and we want to prove that it is approximately a tensor product across all $B_j$ regions via upper bounding the relative entropy
\begin{equation}
    D(\eta||\otimes_j \eta_{B_j})=\sum_j S(\eta_{B_j})-S(\eta)=\sum_j S(\eta_{B_j}).
\end{equation}
By the strong subadditivity of quantum entropy,
\begin{equation}
    S(\eta_{B_j})\leq S(\eta_{B_{\leq j}})+S(\eta_{B_{\geq j}})-S(\eta)=S(\eta_{B_{\leq j}})+S(\eta_{B_{\geq j}}).
\end{equation}
Focusing on the entropy of $S(\eta_{B_{\leq j}})$, we can ignore the unitaries that are applied on regions other than $A_{2j}$. Note that Eq.~\eqref{eq:disentangle} implies that
\begin{equation}
    \frac{1}{2}\left\|\Tr_{M_2^{(j)} R}(U^{(j)}\ketbra{\psi}\otimes \ketbra{0}_{M'^{(j)}}U^{(j)\dag})-\ketbra{\phi_1}_{LM_1^{(j)}}\right\|_1\leq \sqrt{\varepsilon_1}.
\end{equation}
Therefore by the Fannes-Audenaert inequality,
\begin{equation}
    S(\eta_{B_{\leq j}})=S(\Tr_{M_2^{(j)} R}(U^{(j)}\ketbra{\psi}\otimes \ketbra{0}_{M'^{(j)}}U^{(j)\dag}))\leq 2|L|\sqrt{\varepsilon_1}+2\varepsilon_1^{1/4}\leq 2n\sqrt{\varepsilon_1}+2\varepsilon_1^{1/4}.
\end{equation}
A similar argument holds for $S(\eta_{B_{\geq j}})$. Therefore we have
\begin{equation}
    S(\eta_{B_j})\leq 4n\sqrt{\varepsilon_1}+4\varepsilon_1^{1/4},\quad \forall j\in\{1,2,\dots,(L+1)/2\}.
\end{equation}
Let
\begin{equation}
    \omega:=\left(\prod_{j=1}^{(L-1)/2}U^{(j)}\right)\sigma\otimes\ketbra{0}_{M'}\left(\prod_{j=1}^{(L-1)/2}U^{(j)\dag}\right),
\end{equation}
then $\left\|\sigma-\ketbra{\psi}\right\|_1=\left\|\omega-\eta\right\|_1$. Note that for any $j$, $\eta_{B_j}$ only depends on the reduced density matrix $\tau_{A_{2j-2}A_{2j-1}A_{2j}}$; similarly, $\omega_{B_j}$ only depends on the reduced density matrix $\sigma_{A_{2j-2}A_{2j-1}A_{2j}}$. Therefore,
\begin{equation}
\begin{aligned}
    \left\|\omega_{B_j}-\eta_{B_j}\right\|_1&\leq \left\|\sigma_{A_{2j-2}A_{2j-1}A_{2j}}-\tau_{A_{2j-2}A_{2j-1}A_{2j}}\right\|_1\\
    &\leq \left\|\sigma_{A_{2j-2}A_{2j-1}A_{2j}}-\rho_{A_{2j-2}A_{2j-1}A_{2j}}\right\|_1+\left\|\rho_{A_{2j-2}A_{2j-1}A_{2j}}-\tau_{A_{2j-2}A_{2j-1}A_{2j}}\right\|_1\\
    &\leq 2\delta+2\sqrt{\varepsilon}.
\end{aligned}
\end{equation}
Note that $|B_j|\leq 3n$, by the Fannes-Audenaert inequality,
\begin{equation}
    S(\omega_{B_j})\leq S(\eta_{B_j})+3n(\delta+\sqrt{\varepsilon})+2\sqrt{\delta+\sqrt{\varepsilon}}.
\end{equation}
This implies that
\begin{equation}
\begin{aligned}
    D(\omega||\otimes_j \omega_{B_j})&=\sum_j S(\omega_{B_j})-S(\omega)\\&\leq \sum_j S(\omega_{B_j})\\
    &\leq \sum_j S(\eta_{B_j})+3n^2(\delta+\sqrt{\varepsilon})+2n\sqrt{\delta+\sqrt{\varepsilon}}.
\end{aligned}
\end{equation}
Then
\begin{equation}
    \begin{aligned}
        \left\|\sigma-\rho\right\|_1&\leq \left\|\sigma-\tau\right\|_1+\left\|\tau-\rho\right\|_1\\
        &\leq \left\|\omega-\eta\right\|_1+2\sqrt{\varepsilon}\\
        &\leq \left\|\omega-\otimes_j \omega_{B_j}\right\|_1+\left\|\otimes_j\omega_{B_j} -\otimes_j\eta_{B_j}\right\|_1+\left\|\otimes_j\eta_{B_j}-\eta\right\|_1+2\sqrt{\varepsilon}\\
        &\leq \sqrt{2D(\omega||\otimes_j \omega_{B_j})}+2n\delta+2n\sqrt{\varepsilon}+\sqrt{2D(\eta||\otimes_j \eta_{B_j})}+2\sqrt{\varepsilon}\\
        &\leq \sqrt{8n(n\sqrt{\varepsilon_1}+\varepsilon_1^{1/4})+6n^2(\delta+\sqrt{\varepsilon})+4n\sqrt{\delta+\sqrt{\varepsilon}}}\\&+\sqrt{8n(n\sqrt{\varepsilon_1}+\varepsilon_1^{1/4})}+2n\delta+2(n+1)\sqrt{\varepsilon}.
    \end{aligned}
\end{equation}
Here in the fourth line we use the quantum Pinsker inequality, which says that $\|\rho-\sigma\|_1\leq\sqrt{2D(\rho\|\sigma)}$ for two density matrices $\rho,\sigma$. Using the fact that
$\varepsilon_1=6\sqrt{\varepsilon}$, we have
\begin{equation}
\begin{aligned}
    \frac{1}{2}\left\|\sigma-\rho\right\|_1&\leq n\delta+2n\sqrt{\varepsilon}+\sqrt{8}n\varepsilon_1^{1/4}+\sqrt{8}\sqrt{n}\varepsilon_1^{1/8}+\frac{\sqrt{6}}{2}n\sqrt{\delta}+\frac{\sqrt{6}}{2}n\varepsilon^{1/4}+\sqrt{n}\delta^{1/4}+\sqrt{n}\varepsilon^{1/8}\\
    &\leq \sqrt{8}n\varepsilon_1^{1/4}+\sqrt{8}\sqrt{n}\varepsilon_1^{1/8}+5n\varepsilon^{1/8} + 4n\delta^{1/4}\\
    &\leq 13n\varepsilon^{1/16}+4n\delta^{1/4}.
\end{aligned}
\end{equation}
\end{proof}

Finally, the next technical lemma bounds the distance between the learned state and the unknown state $\ket{\psi}$.

\begin{lemma}\label{lemma:boundfinalerror}
Let $\ket{\psi}_{A_1\dots A_L B}$ be a pure state, and let $\rho_{A_1\dots A_L B}=\ketbra{\psi}_{A_1\dots A_L B}$. Suppose the following two properties hold:
\begin{enumerate}
    \item $\expval{\rho_B}{0_B}= 1-\varepsilon$,
    \item $\rho_{A_1\dots A_L }=\rho_{A_1}\otimes\cdots\otimes \rho_{A_L}$.
\end{enumerate}
Suppose $\{\sigma_{A_i}\}$ are density matrices that satisfies $\frac{1}{2}\left\|\rho_{A_i}-\sigma_{A_i}\right\|_1\leq\delta$ for any $i$. Then
\begin{equation}
    \frac{1}{2}\left\|(\otimes_{i=1}^L \sigma_{A_i})\otimes\ketbra{0}_B-\ketbra{\psi}\right\|_1\leq \sqrt{2\varepsilon+L\delta}.
\end{equation}
\end{lemma}
\begin{proof}
    The state $\ket{\psi}_{A_1\dots A_L B}$ can be written as
    \begin{equation}
        \ket{\psi}_{A_1\dots A_L B}=\sqrt{1-\varepsilon}\ket{0}_B \ket{\phi}_{A_1\dots A_L} + \sqrt{\varepsilon}\ket{\mathrm{else}}_{A_1\dots A_L B},
    \end{equation}
    where $\bra{0}_B\ket{\mathrm{else}}_{A_1\dots A_L B}=0$. This implies that
    \begin{equation}
        \rho_{A_1\dots A_L}=\Tr_B \rho_{A_1\dots A_L B}=(1-\varepsilon)\ketbra{\phi}_{A_1\dots A_L }+\varepsilon\Tr_B \ketbra{\mathrm{else}}.
    \end{equation}
    Note that
    \begin{equation}
    \begin{aligned}
        \frac{1}{2}\left\|\rho_{A_1\dots A_L}-\sigma_{A_1}\otimes \cdots\otimes \sigma_{A_L}\right\|_1&=\frac{1}{2}\left\|\rho_{A_1}\otimes\cdots\otimes\rho_{A_L}-\sigma_{A_1}\otimes \cdots\otimes \sigma_{A_L}\right\|_1\\
        &\leq \frac{1}{2}\sum_{i=1}^L \left\|\rho_{A_i}-\sigma_{A_i}\right\|_1\\
        &\leq L\delta.
    \end{aligned}
    \end{equation}
    Therefore,
    \begin{equation}
        \begin{aligned}
            \bra{\psi}_{A_1\dots A_L B}\sigma_{A_1}\otimes\cdots\otimes\sigma_{A_L}\otimes \ketbra{0}_B\ket{\psi}_{A_1\dots A_L B}&\geq \bra{\psi}\rho_{A_1\dots A_L}\otimes \ketbra{0}_B\ket{\psi}-L\delta\\
            &\geq (1-\varepsilon)^2-L\delta\\
            &\geq 1-2\varepsilon-L\delta.
        \end{aligned}
    \end{equation}
This implies that
\begin{equation}
    \frac{1}{2}\left\|(\otimes_{i=1}^L \sigma_{A_i})\otimes\ketbra{0}_B-\ketbra{\psi}\right\|_1\leq \sqrt{2\varepsilon+L\delta}.
\end{equation}
\end{proof}

\noindent\textbf{Proof of first claim of Theorem~\ref{thm:2dstate}.} Next we show how to use the above techniques to learn an unknown quantum state $\ket{\psi}=U\ket{0^n}$, with the promise that $U$ is a depth-$d$ circuit acting on a 2D lattice (here $d$ is treated as a generic parameter which is not necessarily a constant) with arbitrary $\mathrm{SU}(4)$ gates.

We work with Viewpoint 2 described in Section~\ref{sec:1dfinitecorrelated}. As discussed at the end of Section~\ref{sec:1dfinitecorrelated}, the learning process requires $\mc O(n)$ reduced density matrices of $\ket{\psi}$ of size $\mc O(d^2)$. Suppose all of these reduced density matrices are learned to within $\varepsilon_0$ trace distance with probability $1-\delta$, then by Lemma~\ref{lemma:reduceddensitymatrix} it suffices to take a randomized measurement dataset $\mathcal{T}_{\ket{\psi}}(N)$ of size
\begin{equation}
    N=\frac{2^{\mathcal O(d^2)}}{\varepsilon_0^2}\log\frac{n}{\delta}.
\end{equation}

Next we proceed with the disentangling step shown in Fig.~\ref{fig:2dlearning}. We have learned the reduced density matrices on the dotted regions shown in Fig.~\ref{fig:2dlearning} (a) to within $\varepsilon_0$ trace distance. Denote the dotted blue region as $AA_1$ where $A$ is the colored blue region, and let $\rho_{AA_1}$ be the reduced density matrix of $\ket{\psi}$ on $AA_1$. We know that there exists a depth-$2d$ circuit $V_{AA_1}$ such that
\begin{equation}\label{eq:localinversionrdm}
    V_{AA_1} \rho_{AA_1}V_{AA_1}^\dag=\ketbra{0}_A\otimes \sigma_{A_1}
\end{equation}
for some density matrix $\sigma_{A_1}$. We have learned a density matrix $\hat\rho_{AA_1}$ such that $\|\hat\rho_{AA_1}-\rho_{AA_1}\|_1\leq\varepsilon_0$. To find an approximate local inversion for the region $A$, we perform a brute force search over an $\varepsilon_0$-net for depth-$2d$ circuits acting on $AA_1$, denoted as $\mc N_{\varepsilon_0}(AA_1)$, which is constructed by discretizing each $\mathrm{SU}(4)$ gate (see Definition~\ref{def:epsnet} and Lemma~\ref{lemma:epsnet}), which has size at most
\begin{equation}
    S=\left(\frac{d^3}{\varepsilon_0}\right)^{\mc O(d^3)}.
\end{equation}
Note that Eq.~\eqref{eq:localinversionrdm} together with $\|\hat\rho_{AA_1}-\rho_{AA_1}\|_1\leq\varepsilon_0$ implies that
\begin{equation}
    \Tr(\bra{0}_A V_{AA_1}\hat \rho_{AA_1}V^\dag_{AA_1}\ket{0}_A)\geq 1-\varepsilon_0.
\end{equation}
By definition of $\varepsilon_0$-net, there exists a unitary $\hat V_{AA_1}\in\mc N_{\varepsilon_0}(AA_1)$ that satisfies $\|\hat V_{AA_1}-V_{AA_1}\|_\infty\leq\varepsilon_0$, which gives
\begin{equation}
    \Tr(\bra{0}_A \hat V_{AA_1}\hat \rho_{AA_1}\hat V_{AA_1}^\dag\ket{0}_A)\geq 1-2\varepsilon_0.
\end{equation}
The algorithm is to enumerate over all elements in $\mc N_{\varepsilon_0}(AA_1)$ and find the ones which satisfy the above equation. Each of these circuits is an approximate local inversion in the sense that
\begin{equation}\label{eq:approxlocalinversion}
    \Tr(\bra{0}_A \hat V_{AA_1} \rho_{AA_1}\hat V_{AA_1}^\dag\ket{0}_A)\geq 1-3\varepsilon_0.
\end{equation}
Using the same argument as in Section~\ref{sec:disentangle2dstate}, in Fig.~\ref{fig:2dlearning} (a) we can find a depth-$d$ circuit $\hat V$ acting on the width-$7d$ strip around $M$, such that Eq.~\eqref{eq:approxlocalinversion} is satisfied for all local colored regions. There are at most $\sqrt{n}$ such regions. Let $\rho=\ketbra{\psi}$, by union bound,
\begin{equation}
    \Tr(\bra{0}_M \hat V \rho\hat V^\dag\ket{0}_M)\geq 1-3\sqrt{n}\varepsilon_0.
\end{equation}
Repeat the same procedure for all vertical $B_i$ strips shown in Fig.~\ref{fig:2dlearning} (b). There are at most $\sqrt{n}$ different vertical strips. Let $B=\cup_i B_i$, and let $V$ denote the union of all learned inversion circuits across different regions, we have
\begin{equation}\label{eq:approxinversion}
    \Tr(\bra{0}_B  V \rho V^\dag\ket{0}_B)\geq 1-3n\varepsilon_0.
\end{equation}
Now, the problem reduces to learning the state $V\ket{\psi}$, which can be formulated as follows.\\

\noindent\textbf{Problem 2.} We are given copies of a state $\sigma=\ketbra{\phi}$ with the promise that
\begin{enumerate}
    \item it is prepared by a depth-$2d$ circuit (defined on a 2D lattice) acting on $\ket{0^n}$;
    \item its reduced density matrix on each of the $B_i$ regions in Fig.~\ref{fig:2dlearning} (b) is close $\ketbra{0}_{B_i}$, i.e. $\expval{\sigma_B}{0_B}\geq 1-\varepsilon_1$.
\end{enumerate}
The goal is to (approximately) learn the state $\ket{\phi}$.\\

Let $\ket{\phi}:=V\ket{\psi}$ and let $\varepsilon_1:=3n\varepsilon_0$. Consider dividing the state $\sigma=\ketbra{\phi}$ into regions $A_1,A_2,\dots,A_L$ and $B=\cup_i B_i$ as in Fig.~\ref{fig:2dlearning} (b). As the regions $\{A_i\}$ are sufficiently far from each other, the reduced density matrix on $A=\cup_i A_i$ is a tensor product across each region, i.e., $\sigma_{A_1\dots A_L}=\sigma_{A_1}\otimes\cdots\otimes \sigma_{A_L}$. By Eq.~\eqref{eq:approxinversion}, we have $\expval{\sigma_B}{0_B}\geq 1-\varepsilon_1$. By Lemma~\ref{lemma:approxdisentangle}, for each $i=1,\dots,L$ there exists a pure state $\ket{\phi}_{A_i}$ such that $\expval{\sigma_{A_i}}{\phi_{A_i}}\geq 1-\varepsilon_1$.

Next we discuss how to learn the state $\sigma_{A_i}$ for a fixed $i$. This is similar to the earlier situation in Viewpoint 2, but with the critical difference that here $\sigma_{A_i}$ is no longer pure. So we list the updated Viewpoint below.\\

\noindent\textbf{Viewpoint 2'.} $\sigma_{A_i}$ can be prepared by a depth-$2d$ circuit acting on $A_i$ as well as some ancilla qubits $A_i^L$ and $A_i^R$, shown in Fig.~\ref{fig:1dancilla}. To see this, recall that $\sigma_{A_i}$ is part of a state that is prepared by a depth-$2d$ circuit. Now, imagine that we \emph{undo} all the gates in that circuit, except for those in the \emph{backward lightcone} of $A_i$. This procedure does not affect the state on $A_i$, and the resulting circuit (denote as $W_i$) has exactly the same shape as in Fig.~\ref{fig:1dancilla}, where $A_i^L$, $A_i^R$ both has width $2d$. Note that here $\sigma_{A_i}$ could be entangled with the ancilla qubits, and we have
\begin{equation}\label{eq:preparewithancillaentangled}
\Tr_{A_i^L A_i^R}\left(W_i \ketbra{0}_{A_i^L A_i A_i^R}W_i^\dag\right)=\sigma_{A_i}.
\end{equation}

Using the same argument as the end of Section~\ref{sec:1dfinitecorrelated}, the reduced density matrices of $\sigma_{A_i}$ can be simulated by reduced density matrices of $\ketbra{\psi}$ on slightly larger regions. Therefore we can obtain reduced density matrices of $\sigma_{A_i}$ within trace distance $\varepsilon_0$. Let $C$ be the solid blue region in Fig.~\ref{fig:1doverlapping}, and let $CC_1$ be the dotted blue region. We have learned a reduced density matrix $\hat\sigma_{C}$ such that $\|\hat\sigma_{C}-\sigma_{C}\|_1\leq\varepsilon_0$. From Viewpoint 2', we know that there is a depth-$2d$ circuit $W_{CC_1}$ acting on $CC_1$, such that
\begin{equation}
    \Tr_{C_1}\left(W_{CC_1}\ketbra{0}_{CC_1}W_{CC_1}^\dag\right)=\sigma_{C}.
\end{equation}
Consider an $\varepsilon_0$-net for depth-$2d$ circuits acting on $CC_1$, denoted as $\mc N_{\varepsilon_0}(CC_1)$. By definition, there exists a unitary $\hat W_{CC_1}$ that satisfies $\|\hat W_{CC_1}-W_{CC_1}\|_\infty\leq\varepsilon_0$, which means that
\begin{equation}
\begin{aligned}
    &\left\|\Tr_{C_1}\left(\hat W_{CC_1}\ketbra{0}_{CC_1}\hat W_{CC_1}^\dag\right)-\hat \sigma_{C}\right\|_1\\\leq &\left\|\Tr_{C_1}\left(\hat W_{CC_1}\ketbra{0}_{CC_1}\hat W_{CC_1}^\dag\right)-\sigma_{C}\right\|_1+\left\|\sigma_{C}-\hat \sigma_{C}\right\|_1\\
    \leq & 2\varepsilon_0.
\end{aligned}
\end{equation}
By enumerating over every element in $\mc N_{\varepsilon_0}(CC_1)$, we can find a list of circuits $\{\hat W_{CC_1}'\}$ that satisfy $\left\|\Tr_{C_1}\left(\hat W_{CC_1}'\ketbra{0}_{CC_1}\hat W_{CC_1}'^\dag\right)-\hat \sigma_{C}\right\|_1\leq 2\varepsilon_0$. Any such circuit $\hat W_{CC_1}'$ will also satisfy
\begin{equation}
    \left\|\Tr_{C_1}\left(\hat W_{CC_1}'\ketbra{0}_{CC_1}\hat W_{CC_1}'^\dag\right)-\sigma_{C}\right\|_1\leq 3\varepsilon_0.
\end{equation}

Using the same argument as in Section~\ref{sec:1dfinitecorrelated}, we can merge these learned local circuits into a global depth-$2d$ circuit $\hat W_i$. Let $\hat \sigma_{A_i}:=\Tr_{A_i^L A_i^R}\left(\hat W_i \ketbra{0}_{A_i^L A_i A_i^R}\hat W_i^\dag\right)$ be the learned reduced density matrix on $A_i$, then the local reduced density matrices of $\hat \sigma_{A_i}$ and $\sigma_{A_i}$ are $3\varepsilon_0$ close in trace distance on solid colored regions in Fig.~\ref{fig:1doverlapping}. This allows us to invoke the main technical lemma, Lemma~\ref{lemma:approxconsistency}, which gives
\begin{equation}
    \frac{1}{2}\left\|\hat\sigma_{A_i}-\sigma_{A_i}\right\|_1\leq 13 n \varepsilon_1^{1/16}+8 n \varepsilon_0^{1/4}\leq 22 n^{17/16}\varepsilon_0^{1/16}.
\end{equation}

The state $\ket{\psi}$ can be approximately prepared as follows:
\begin{enumerate}
    \item Initialize registers $A_i,B_i,A_i^L,A_i^R$ in the state $\ket{0}$. Let $A=\cup_i A_i$ and $B=\cup_i B_i$.
    \item For each $i$, apply the depth-$2d$ circuit $\hat W_i$ to $A_i^L A_i A_i^R$. The reduced density matrix on $AB$ equals $(\otimes_i \hat \sigma_{A_i})\otimes \ketbra{0}_B$
    \item Apply the depth-$d$ circuit $V^\dag$ to $AB$, and the reduced density matrix on $AB$ is $\hat\rho=V^\dag (\otimes_i \hat \sigma_{A_i})\otimes \ketbra{0}_B V$, which approximately equals to $\ketbra{\psi}$.
\end{enumerate}
Similar to the proof of second claim of Theorem~\ref{thm:2dstate} at the end of Section~\ref{sec:1dfinitecorrelated}, we can choose the $A_i$ regions to be sufficiently wide, such that the number of ancilla qubits equals to $tn$ for an arbitrarily small constant $t$.

The final task is to bound the error between the learned density matrix and $\ketbra{\psi}$. Using Lemma~\ref{lemma:boundfinalerror}, the trace distance can be bounded as
\begin{equation}
\begin{aligned}
    \frac{1}{2}\left\|V^\dag (\otimes_i \hat \sigma_{A_i})\otimes \ketbra{0}_B V-\ketbra{\psi}\right\|_1&=\frac{1}{2}\left\|(\otimes_i \hat \sigma_{A_i})\otimes \ketbra{0}_B -V\ketbra{\psi}V^\dag\right\|_1\\
    &\leq \sqrt{2\cdot 3 n \varepsilon_0+\sqrt{n}\cdot 22 n^{17/16}\varepsilon_0^{1/16}}\\
    &\leq 6 n^{25/32}\varepsilon_0^{1/32}.
\end{aligned}
\end{equation}
Therefore, to achieve trace distance $\varepsilon$, it suffices to choose $\varepsilon_0=\mc O(\frac{\varepsilon^{32}}{n^{25}})$. The total sample complexity is
\begin{equation}
N=\frac{2^{\mathcal O(d^2)}}{\varepsilon_0^2}\log\frac{n}{\delta}=\frac{2^{\mathcal O(d^2)}n^{50}}{\varepsilon^{64}}\log\frac{n}{\delta}.
\end{equation}
The total running time is
\begin{equation}
    n\cdot S = \left(\frac{n d^3}{\varepsilon}\right)^{\mc O(d^3)}.
\end{equation}

\section{Verifying learned shallow circuits under average-case distance}
\label{sec:verify-learned-shallow-circuit}

From the previous appendices, we have seen that given an $n$-qubit CPTP map $\mathcal{C}$ promised to be a unitary $U$ generated by a constant-depth quantum circuit, we can learn a constant-depth $2n$-qubit circuit $\hat{V}$, such that $\hat{V}$ is close to $U \otimes U^\dagger$, and the reduced channel $\hat{\mathcal{E}} := \mathcal{E}^{\hat{V}}_{\leq n}$ of $\hat{V}$ on the first $n$ qubits is close to $\mathcal{C} = U(\cdot) U^\dagger = \mathcal{U}$ in the diamond distance.
In this section, we answer the question: What happens if there is no promise that $\mathcal{C}$ is a unitary generated by a shallow quantum circuit, and, furthermore, $\mathcal{C}$ may not even be unitary?

Given an arbitrary CPTP map $\mathcal{C}$, the proposed algorithm can still learn a constant-depth $2n$-qubit circuit $\hat{V}$ with an associated $n$-qubit CPTP map $\hat{\mathcal{E}} := \mathcal{E}^{\hat{V}}_{\leq n}$.
However, without the promise on $\mathcal{C}$, the learned map $\hat{\mathcal{E}}$ could be arbitrary.
This raises the question: can we verify that $\hat{\mathcal{E}}$ is close to $\mathcal{C}$?
From the previous section on the hardness for learning log-depth circuits, we see that even if $\mathcal{C}$ is an $n$-qubit unitary $U$ generated by a log-depth circuit, one already needs $\exp(\Omega(n))$ queries to check if $U$ is close to $I$ in the diamond distance or not.
Hence, when the learning algorithm outputs $\hat{\mathcal{E}} = \mathcal{I}$, which is very likely in this case as the unitary $U_x$ in Eq.~\eqref{eq:GroverOracle} is almost identity, we cannot efficiently check if $\hat{\mathcal{E}}$ is close to $\mathcal{C}$ in the diamond distance.
The exponential hardness stems from the definition of diamond distance, which considers the worst case over all possible input states.

To circumvent the exponential hardness, we consider closeness under the average-case distance $\mathcal{D}_{\mathrm{ave}}$ (see Definition~\ref{def:ave-dist}) instead of the worst-case distance $\mathcal{D}_{\diamond}$.
We give a verification algorithm that verifies the learned map $\hat{\mathcal{E}}$ by outputting \textsc{pass} or \textsc{fail} as follows:
\begin{enumerate}
    \item the verification algorithm outputs \textsc{fail} with high probability if the learned map $\hat{\mathcal{E}}$ is not close to $\mathcal{C}$ under the average-case distance $\mathcal{D}_{\mathrm{ave}}$;
    \item the verification algorithm outputs \textsc{pass} with high probability if the learned map $\hat{\mathcal{E}}$ is close to $\mathcal{C}$ under the average-case distance $\mathcal{D}_{\mathrm{ave}}$ and the unknown map $\mathcal{C}$ is close to a unitary.
\end{enumerate}
The verification algorithm only needs access to a randomized measurement dataset $\mathcal{T}_{\mathcal{C}}(N)$ generalizing Definition~\ref{def:random-measure-data} by replacing the unitary $U$ with the map $\mathcal{C}$.
Formally, we have the following theorem.

\begin{theorem}[Verifying the learned shallow circuit]
    Given a failure probability~$\delta$, a verification error~$\varepsilon$, a learned constant-depth $2n$-qubit circuit $\hat{V}$, the associated $n$-qubit CPTP map $\hat{\mathcal{E}} = \mathcal{E}^{\hat{V}}_{\leq n}$, and an unknown $n$-qubit CPTP map $\mathcal{C}$. With a randomized measurement dataset $\mathcal{T}_{\mathcal{C}}(N)$ of size
    \begin{equation} \label{eq:sample-size-verification}
        N = \mathcal{O}\left( \frac{n^2 \log(n / \delta)}{\varepsilon^2}  \right),
    \end{equation}
    the verification algorithm outputs \textsc{pass} or \textsc{fail} such that
    \begin{enumerate}
        \item if $\mathcal{D}_{\mathrm{ave}}(\hat{\mathcal{E}}, {\mathcal{C}}) > \varepsilon$, the output is \textsc{fail} with probability $\geq 1 - \delta$.
        \item if $\mathcal{D}_{\mathrm{ave}}(\hat{\mathcal{E}}, {\mathcal{C}}) \leq \frac{\varepsilon}{12n}$ and $\norm{{\mathcal{C}}^\dagger \mathcal{C} - \mathcal{I}}_\diamond \leq \frac{\varepsilon}{12n}$, the output is \textsc{pass} with probability $\geq 1 - \delta$;
    \end{enumerate}
    The computational time of the verification algorithm is $\mathcal{O}(nN)$.
\end{theorem}
\begin{proof}
    The verification algorithm is based on the concept of weak approximate local identity presented in Section~\ref{sec:weak-localid}.
    Let us define the $n$-qubit CPTP map
    \begin{equation}
        \hat{\mathcal{I}} := \hat{\mathcal{E}}^\dagger \mathcal{C}.
    \end{equation}
    Note that $\hat{\mathcal{E}}^\dagger(\rho)$ can be implemented by appending $n$-qubit maximally mixed state to $\rho$, evolving $\rho \otimes (I_n / 2^n)$ under the unitary $\hat{V}^\dagger$, then tracing out the appended $n$ ancilla qubits, i.e.,
    \begin{equation} \label{eq:hatcEdag}
        \hat{\mathcal{E}}^\dagger(\rho) = \Tr_{> n} \left(\hat{V}^\dagger (\rho \otimes I_n / 2^n) \hat{V} \right),
    \end{equation}
    where $I_n$ is an $n$-qubit identity.
    The verification algorithm uses the randomized measurement dataset $\mathcal{T}_{\mathcal{C}}(N)$ to estimate $\hat{o}_i$ approximating $\mathcal{D}_{\mathrm{ave}}(\mathcal{E}_i^{\hat{\mathcal{I}}}, \mathcal{I})$ up to $\varepsilon / (3n)$ error for all $i$ from $1$ to $n$ with probability at least $1 - \delta$.
    Then the verification algorithm outputs
    \begin{equation}
        \begin{cases}
            \textsc{pass}, & \text{if} \,\,\,\, \frac{3}{2} \sum_{i=1}^n \hat{o}_i \leq \varepsilon/2, \\
            \textsc{fail}, & \text{if} \,\,\,\, \frac{3}{2} \sum_{i=1}^n \hat{o}_i > \varepsilon/2. \\
        \end{cases}
    \end{equation}
    From Lemma~\ref{lem:check-weak-local-id-algo} presented at the end of this section, we can show that the dataset size $N$ stated in Eq.~\eqref{eq:sample-size-verification} is sufficient to guarantee the desired property on $\hat{o}_i$ and the computational time to estimate $\hat{o}_i$ for all $i$ is $\mathcal{O}(n N)$.
    We define the event that
    \begin{equation} \label{eq:event-Estar}
        \left|\hat{o}_i - \mathcal{D}_{\mathrm{ave}}(\mathcal{E}_i^{\hat{\mathcal{I}}}, \mathcal{I})\right| \leq \frac{\varepsilon}{6n}, \quad \forall i = 1, \ldots, n
    \end{equation}
    to be event $E^*$.
    Conditioning on event $E^*$,
    we show that the desired outputs, \textsc{fail} and \textsc{pass}, must be given by the verification algorithm in the two scenarios stated in the theorem, respectively.

    \paragraph{Case $1$: $\mathcal{D}_{\mathrm{ave}}(\hat{\mathcal{E}}, {\mathcal{C}}) > \varepsilon$.}
    When conditioning on event $E^*$, we claim that the algorithm always outputs \textsc{fail}.
    We prove this claim by contradiction.
    Assume that the algorithm outputs \textsc{pass}.
    From the definition of fidelity $F(\rho, \sigma) = \Tr(\sqrt{\sqrt{\sigma} \rho \sqrt{\sigma}})^2$ given in Definition~\ref{def:fidelity}, we can see that $F(\rho, \sigma) \geq \Tr(\rho \sigma)$. Hence, from Definition~\ref{def:ave-dist} on $\mathcal{D}_{\mathrm{ave}}$, we have
    \begin{equation} \label{eq:Dave-IvsI}
        \varepsilon < \mathcal{D}_{\mathrm{ave}}(\hat{\mathcal{E}}, {\mathcal{C}}) \leq \mathcal{D}_{\mathrm{ave}}(\hat{\mathcal{E}}^\dagger {\mathcal{C}}, {\mathcal{I}}) = \mathcal{D}_{\mathrm{ave}}(\hat{\mathcal{I}}, \mathcal{I}).
    \end{equation}
    If the algorithm outputs $\textsc{pass}$, we have
    \begin{equation}
        \frac{3}{2} \sum_{i=1}^n \hat{o}_i \leq \frac{\varepsilon}{2}.
    \end{equation}
    Because in the event $E^*$, Eq.~\eqref{eq:event-Estar} ensures
    \begin{equation}
        \left| \hat{o}_i - \mathcal{D}_{\mathrm{ave}}(\mathcal{E}_i^{\hat{\mathcal{I}}}, \mathcal{I}) \right| \leq \frac{\varepsilon}{6n},
    \end{equation}
    we can conclude that
    \begin{equation}
        \frac{3}{2} \sum_{i=1}^n \mathcal{D}_{\mathrm{ave}}(\mathcal{E}_i^{\hat{\mathcal{I}}}, \mathcal{I}) \leq \frac{3}{4} \varepsilon.
    \end{equation}
    Using Lemma~\ref{thm:IDcheck-weak} on global identity check from weak local identity check, we have
    \begin{equation}
        \mathcal{D}_{\mathrm{ave}}(\hat{\mathcal{I}}, \mathcal{I}) \leq \frac{3}{2} \sum_{i=1}^n \mathcal{D}_{\mathrm{ave}}(\mathcal{E}_i^{\hat{\mathcal{I}}}, \mathcal{I}) \leq \frac{3}{4} \varepsilon.
    \end{equation}
    This inequality contradicts the one in Eq.~\eqref{eq:Dave-IvsI}.
    Hence, if $\mathcal{D}_{\mathrm{ave}}(\hat{\mathcal{E}}, {\mathcal{C}}) > \varepsilon$, the output of the verification algorithm is \textsc{fail} with probability at least $1 - \delta$.

    \paragraph{Case $2$: $\mathcal{D}_{\mathrm{ave}}(\hat{\mathcal{E}}, {\mathcal{C}}) \leq \varepsilon/(24 n)$ and $\norm{{\mathcal{C}}^\dagger \mathcal{C} - \mathcal{I}}_\diamond \leq \varepsilon/(12 n)$.}
    When conditioning on event $E^*$, we claim that the algorithm always outputs \textsc{pass}.
    We begin by noting that the fidelity $F(\rho, \sigma) \leq F(\mathcal{E}(\rho), \mathcal{E}(\sigma))$ for any CPTP map $\mathcal{E}$ from Fact~\ref{fact:fidelity}.
    Therefore, we have
    \begin{equation} \label{eq:Dave-CdagECdagC}
        \mathcal{D}_{\mathrm{ave}}(\mathcal{C}^\dagger \hat{\mathcal{E}}, \mathcal{C}^\dagger {\mathcal{C}}) \leq \mathcal{D}_{\mathrm{ave}}(\hat{\mathcal{E}}, {\mathcal{C}}) \leq \frac{\varepsilon}{24 n}.
    \end{equation}
    We now consider the following derivations,
    \begin{align} \label{eq:DaveIhatI}
        \mathcal{D}_{\mathrm{ave}}(\hat{\mathcal{I}}, \mathcal{I}) &= \mathcal{D}_{\mathrm{ave}}(\hat{\mathcal{E}}^\dagger \mathcal{C}, \mathcal{I})\\
        &= \Exp_{\ket{\psi}: \mathrm{Unif}} \big[ 1 - \mathcal{F}( (\hat{\mathcal{E}}^\dagger \mathcal{C})(\ketbra{\psi}{\psi}), \ketbra{\psi}{\psi} ) \big]\\
        &= \Exp_{\ket{\psi}: \mathrm{Unif}} \big[ 1 - \Tr( \mathcal{C}(\ketbra{\psi}{\psi}) \hat{\mathcal{E}}(\ketbra{\psi}{\psi}) ) \big]\\
        &= \Exp_{\ket{\psi}: \mathrm{Unif}} \big[ 1 - F((\mathcal{C}^\dagger \hat{\mathcal{E}})(\ketbra{\psi}{\psi}), \ketbra{\psi}{\psi}) \big].
    \end{align}
    Using the triangle inequality for Fubini-Study metric $\Theta$ from Fact~\ref{fact:fidelity}, we have
    \begin{align}
        &\sqrt{1 - F((\mathcal{C}^\dagger \hat{\mathcal{E}})(\ketbra{\psi}{\psi}), \ketbra{\psi}{\psi})}\\
        &\leq \sin\left( \Theta\left((\mathcal{C}^\dagger \hat{\mathcal{E}})(\ketbra{\psi}{\psi}), (\mathcal{C}^\dagger \mathcal{C})(\ketbra{\psi}{\psi}) \right) + \Theta\left((\mathcal{C}^\dagger {\mathcal{C}})(\ketbra{\psi}{\psi}), \ketbra{\psi}{\psi} \right) \right)\\
        &\leq \sin\left( \Theta\left((\mathcal{C}^\dagger \hat{\mathcal{E}})(\ketbra{\psi}{\psi}), (\mathcal{C}^\dagger \mathcal{C})(\ketbra{\psi}{\psi}) \right) \right) + \sin\left(\Theta\left((\mathcal{C}^\dagger {\mathcal{C}})(\ketbra{\psi}{\psi}), \ketbra{\psi}{\psi} \right) \right)\\
        &\leq \sqrt{1 - F((\mathcal{C}^\dagger \hat{\mathcal{E}})(\ketbra{\psi}{\psi}), (\mathcal{C}^\dagger \mathcal{C})(\ketbra{\psi}{\psi}))} + \sqrt{1 - F((\mathcal{C}^\dagger {\mathcal{C}})(\ketbra{\psi}{\psi}), \ketbra{\psi}{\psi})}.
    \end{align}
    From $1 - F(\rho, \psi) \leq \frac{1}{2} \norm{\rho - \psi}_1$ for any state $\rho$ and pure state $\psi$ from Fact~\ref{fact:fidelity}, we have
    \begin{equation}
        1 - F((\mathcal{C}^\dagger {\mathcal{C}})(\ketbra{\psi}{\psi}), \ketbra{\psi}{\psi}) \leq \frac{1}{2} \norm{ (\mathcal{C}^\dagger {\mathcal{C}})(\ketbra{\psi}{\psi}) - \ketbra{\psi}{\psi}}_{\mathrm{tr}} \leq \frac{\varepsilon}{24 n}.
    \end{equation}
    From the two inequalities above, we see that
    \begin{equation}
        \sqrt{1 - F((\mathcal{C}^\dagger \hat{\mathcal{E}})(\ketbra{\psi}{\psi}), \ketbra{\psi}{\psi})} \leq \sqrt{1 - F((\mathcal{C}^\dagger \hat{\mathcal{E}})(\ketbra{\psi}{\psi}), (\mathcal{C}^\dagger \mathcal{C})(\ketbra{\psi}{\psi}))} + \sqrt{\frac{\varepsilon}{24 n}}.
    \end{equation}
    Using Jensen's inequality, the above inequality, and Eq.~\eqref{eq:DaveIhatI}, we obtain
    \begin{align}
        &\mathcal{D}_{\mathrm{ave}}(\hat{I}, I)\\
        &= \Exp_{\ket{\psi}: \mathrm{Unif}} \big[ 1 - F((\mathcal{C}^\dagger \hat{\mathcal{E}})(\ketbra{\psi}{\psi}), \ketbra{\psi}{\psi}) \big] \\
        &\leq \Exp_{\ket{\psi}: \mathrm{Unif}} \big[ 1 - F((\mathcal{C}^\dagger \hat{\mathcal{E}})(\ketbra{\psi}{\psi}), (\mathcal{C}^\dagger {\mathcal{C}})(\ketbra{\psi}{\psi})) \big] + \frac{\varepsilon}{24 n}\\
        &\quad + 2 \sqrt{\frac{\varepsilon}{24 n}} \sqrt{\Exp_{\ket{\psi}: \mathrm{Unif}} \big[ 1 - F((\mathcal{C}^\dagger \hat{\mathcal{E}})(\ketbra{\psi}{\psi}), (\mathcal{C}^\dagger {\mathcal{C}})(\ketbra{\psi}{\psi})) \big]}\\
        &= \mathcal{D}_{\mathrm{ave}}(\mathcal{C}^\dagger \hat{\mathcal{E}}, \mathcal{C}^\dagger {\mathcal{C}}) +  \frac{\varepsilon}{24 n} + 2 \sqrt{\frac{\varepsilon}{24 n}} \sqrt{\mathcal{D}_{\mathrm{ave}}(\mathcal{C}^\dagger \hat{\mathcal{E}}, \mathcal{C}^\dagger {\mathcal{C}})} \leq \frac{\varepsilon}{6n}.
    \end{align}
    The last inequality follows from Eq.~\eqref{eq:Dave-CdagECdagC}.
    Using Lemma~\ref{lem:nonIDcheck-weak} on weak local identity from global identity check through average-case distance, we have
    \begin{equation}
        \mathcal{D}_{\mathrm{ave}}(\mathcal{E}_i^{\hat{\mathcal{I}}}, \mathcal{I}) \leq \frac{\varepsilon}{6n}
    \end{equation}
    for all $i$ from $1$ to $n$.
    When event $E^*$ occurs, we can combine the above with Eq.~\eqref{eq:event-Estar} to show that
    \begin{equation}
        \hat{o}_i \leq \frac{\varepsilon}{3n}, \quad \forall i = 1, \ldots, n.
    \end{equation}
    As a result, we can see that $\frac{3}{2} \sum_{i=1}^n \hat{o}_i \leq \varepsilon/2$. Hence, in this case, the output of the verification algorithm is \textsc{pass} with probability at least $1 - \delta$.
\end{proof}

\noindent From the theorem, the verification algorithm outputs \textsc{pass} with high probability if the promise on~$\mathcal{C}$ is satisfied, and one uses our proposed learning algorithm to learn $\hat{\mathcal{E}}$.
Furthermore, whenever the verification algorithm outputs \textsc{pass}, we can be certain that $\hat{\mathcal{E}}$ is close to $\mathcal{C}$ (under the average-case distance).
Together, our proposed learning algorithm and verification algorithm enable one to learn a verifiable shallow quantum circuit approximation to an arbitrary unknown CPTP map $\mathcal{C}$.

\begin{lemma}[Checking weak approximate local identity] \label{lem:check-weak-local-id-algo}
    Given a failure probability~$\delta$, a verification error~$\varepsilon$, a learned constant-depth $2n$-qubit circuit $\hat{V}$, the associated $n$-qubit CPTP map $\hat{\mathcal{E}} = \mathcal{E}^{\hat{V}}_{\leq n}$, and an unknown $n$-qubit CPTP map $\mathcal{C}$. With a randomized measurement dataset $\mathcal{T}_{\mathcal{C}}(N)$ of size
    \begin{equation} \label{eq:sample-size-verification-check}
        N = \mathcal{O}\left( \frac{n^2 \log(n / \delta)}{\varepsilon^2}  \right),
    \end{equation}
    we can estimate $\hat{o}_i, \forall i$ in time $\mathcal{O}(n N)$ such that
    \begin{equation}
        \left| \hat{o}_i - \mathcal{D}_{\mathrm{ave}}(\mathcal{E}_i^{\hat{\mathcal{E}}^\dagger \mathcal{C}}, \mathcal{I}) \right| \leq \frac{\varepsilon}{3n}, \quad \forall i = 1, \ldots, n,
    \end{equation}
    with probability at least $1 - \delta$.
\end{lemma}
\begin{proof}
    Recall from Eq.~\eqref{eq:hatcEdag} that the CPTP map $\hat{E}^\dagger$ is given by
    \begin{equation}
        \hat{\mathcal{E}}^\dagger(\rho) = \Tr_{> n} \left(\hat{V}^\dagger (\rho \otimes I_n / 2^n) \hat{V} \right).
    \end{equation}
    Hence, we have the following identity for the single-qubit CPTP map,
    \begin{equation}
        \mathcal{E}_i^{\hat{\mathcal{E}}^\dagger \mathcal{C}}(\rho_i) = \Tr_{\neq i} \left(\hat{V}^\dagger ( \mathcal{C}(\rho_i \otimes I_{n-1} / 2^{n-1}) \otimes I_n / 2^n) \hat{V} \right),
    \end{equation}
    where $\rho_i$ is a single-qubit density matrix, $\rho_i \otimes I_{n-1} / 2^{n-1}$ is an $n$-qubit density matrix equal to $\rho_i$ on the $i$-th qubit and maximally mixed on all other qubits, and $\Tr_{\neq i}$ traces out all qubits except for the $i$-th qubit.
    Because $\hat{V}$ is a constant-depth quantum circuit, $\mathcal{E}_i^{\hat{\mathcal{E}}^\dagger \mathcal{C}}$ depends only on a reduced channel $\mathcal{E}^{\mathcal{C}}_{S_i}$ of $\mathcal{C}$ on a subset $S_i$ of qubits with $|S_i| = \mathcal{O}(1)$ and $i \in S_i$, i.e.,
    \begin{equation} \label{eq:cEihatEdagCrhoi}
        \mathcal{E}_i^{\hat{\mathcal{E}}^\dagger \mathcal{C}}(\rho_i) = \Tr_{\neq i} \left(\hat{V}^\dagger \left( \left(\mathcal{E}^{\mathcal{C}}_{S_i} \otimes \mathcal{I}_{[n] \setminus S_i}\right) (\rho_i \otimes I_{n-1} / 2^{n-1}) \otimes I_n / 2^n \right) \hat{V} \right),
    \end{equation}
    where $\mathcal{I}_{[n] \setminus S_i}$ is the identity CPTP map over qubit $1$ to qubit $n$ not in set $S_i$.
    For any $i = 1, \ldots, n$, from the results in \cite{surawy2022projected, levy2021classical, huang2022learning, kunjummen2023shadow}, one could use $\mathcal{T}_{\mathcal{C}}(N)$ with the specified size to learn $\hat{\mathcal{E}}^{\mathcal{C}}_{S_i}$ such that
    \begin{equation}
        \norm{\hat{\mathcal{E}}^{\mathcal{C}}_{S_i} - {\mathcal{E}}^{\mathcal{C}}_{S_i}}_\diamond \leq \frac{\varepsilon}{3n},
    \end{equation}
    with probability at least $1 - (\delta / n)$.
    By the union bound, we have
    \begin{equation}
        \norm{\hat{\mathcal{E}}^{\mathcal{C}}_{S_i} - {\mathcal{E}}^{\mathcal{C}}_{S_i}}_\diamond \leq \frac{\varepsilon}{3n}, \quad \forall i = 1, \ldots, n,
    \end{equation}
    with probability at least $1 - \delta$.
    Hence, from Eq.~\eqref{eq:cEihatEdagCrhoi}, we can learn $\hat{\mathcal{E}}_i^{\hat{\mathcal{E}}^\dagger \mathcal{C}}$ for all $i$ such that
    \begin{equation}
        \norm{\hat{\mathcal{E}}_i^{\hat{\mathcal{E}}^\dagger \mathcal{C}} - {\mathcal{E}}_i^{\hat{\mathcal{E}}^\dagger \mathcal{C}}}_\diamond \leq \frac{\varepsilon}{3n}, \quad \forall i = 1, \ldots, n,
    \end{equation}
    with probability at least $1 - \delta$.
    By defining
    \begin{equation}
        \hat{o}_i := \mathcal{D}_{\mathrm{ave}}\left(\hat{\mathcal{E}}_i^{\hat{\mathcal{E}}^\dagger \mathcal{C}}, \mathcal{I}\right) = \Exp_{\ket{\psi}: \mathrm{Unif}} \left[ 1 - \bra{\psi} \hat{\mathcal{E}}_i^{\hat{\mathcal{E}}^\dagger \mathcal{C}}(\ketbra{\psi}{\psi}) \ket{\psi} \right], \quad \forall i = 1, \ldots, n,
    \end{equation}
    we can obtain the desired claim.
\end{proof}

\section{Exponentially many local minima in parameterized shallow quantum circuits}
\label{sec:exp-many-local-minima-param-shallow}

In this section, we study the optimization landscape of training 1D shallow parameterized quantum circuits to learn an unknown unitary.
In particular, we will show that there are exponentially many strictly suboptimal local minima, where each local minimum is the minimum over an exponentially sized neighborhood.
Consider a simple 1D shallow parameterized quantum circuit,
\begin{equation}
    U(\vec{\theta}) := \prod_{j} \exp(i \theta_{1,j} \SWAP_{2j+1, 2j+2}) \prod_{j} \exp(i \theta_{2,j} \SWAP_{2j, 2j+1}) \prod_{j} \exp(i \theta_{3,j} \SWAP_{2j+1, 2j+2}),
\end{equation}
where $\vec{\theta} = (\theta_{1,j}, \theta_{2,j}, \theta_{3,j})$ is a vector of all the real-valued parameters.
We consider an unknown unitary $U$ over $n$ qubits to be given by the tensor product of SWAP operators over some pairs of qubits, i.e.,
\begin{equation}
    U_S = \prod_{i \in S} \SWAP_{i, i+3},
\end{equation}
for some subset $S \subseteq \{0, 1, 2, \ldots, \lfloor n / 4 \rfloor - 1 \}$ of qubits with $|S| = \Theta(n)$.
For any such subset $S$, there exists a parameter vector $\vec{\theta}$ such that $U_S = U(\vec{\theta})$.

To avoid barren plateaus in the optimization landscape, we consider the local cost function \cite{cerezo2021cost},
\begin{equation} \label{eq:cost-local-theta}
    C_S(\vec{\theta}) := \mathop{\mathbb{E}}_{\ket{\psi} = \bigotimes_{i=1}^n \ket{\psi_i} \in \mathrm{stab}_1^{\otimes n}} \sum_{i=1}^n \left( 1 -  \Tr\left(\bra{\psi_i} U(\vec{\theta})^\dagger U_S \ketbra{\psi}{\psi} U_S^\dagger U(\vec{\theta}) \ket{\psi_i}\right) \right) \geq 0.
\end{equation}
It is well known that the local cost function is faithful \cite{cerezo2021cost, caro2022out}, i.e., if the local cost function is at most~$\varepsilon$, then $U$ is close to $U(\vec{\theta})$ up to average-case distance (equiv. to normalized Frobenius norm; See Prop.~\ref{prop:ave-dist-Frob}) of $\mathcal{O}(\varepsilon)$, and when $U_S$ is $\varepsilon$-close to $U(\vec{\theta})$ in the average-case distance, the local cost function is bounded above by $\mathcal{O}(n \varepsilon)$.
The local cost function does not suffer from the barren plateau problem when $U(\vec{\theta})$ and $U_S$ can both be implemented by shallow quantum circuits.
For those unfamiliar with barren plateau, it is an overwhelmingly large region in the parameter space with a large cost function and a near-zero gradient \cite{mcclean2018barren, cerezo2021cost}.
When a barren plateau is present, one can easily randomly initialize on the barren plateau and cannot escape the plateau.

While no barren plateau is present in training shallow parameterized circuits, we show that there are exponentially many strictly suboptimal local minima in the optimization landscape.
Furthermore, these suboptimal local minima are minima over neighborhoods with an exponentially large volume $(2 \pi / 4)^{\mathcal{O}(n)} \approx 1.57^{\mathcal{O}(n)}$. This is formally stated below.

\begin{proposition}[Exponentially many strictly suboptimal local minima]
    Consider
    \begin{equation}
    S \subseteq \{0, 1, 2, \ldots, \lfloor n / 4 \rfloor - 1 \}
    \end{equation}
    with $|S| = \Theta(n)$.
    For the cost function $C_S(\vec{\theta})$ in Eq.~\eqref{eq:cost-local-theta}, there are exponentially many strictly suboptimal local minima $\{\vec{\theta}_x\}_{x=0}^{2^{|S|} - 2}$, i.e.,
    \begin{align}
        C_S(\vec{\theta}_x) &\geq 1 + \min_{\vec{\theta}} C_S(\vec{\theta}), & \text{(strictly suboptimal)}\\
        C_S(\vec{\theta}_x) &\leq C_S(\vec{\theta}), \quad \forall \norm{\vec{\theta} - \vec{\theta}_x}_\infty < \pi / 4, & \text{(local minimum)}
    \end{align}
    for all $x = 0, \ldots, 2^{|S|} - 2$.
\end{proposition}
\begin{proof}
    Without loss of generality, we consider $n$ to be divisible by $4$. If $n$ is not divisible by $4$, we neglect the last $n \, \mathrm{mod} \, 4$ qubits.
    For convenience, we group and name the parameters $\vec{\theta}$ as follows.
    \begin{align}
        \vec{\theta}_{B, j} &:= (\theta_{1, 2j + 1}, \theta_{1, 2j + 2}, \theta_{2, 2j+1}, \theta_{3, 2j + 1}, \theta_{3, 2j + 2}), \quad \forall j = 0, \ldots, (n / 4) - 1,\\
        \theta_{L, j} &:= \theta_{2, 2j+2}, \quad  \forall j = 0, \ldots, (n / 4) - 2.
    \end{align}
    Here, $\vec{\theta}_{B, j}$ corresponds to a block of $5$ gates acting on $4$ qubits.
    And, $\theta_{L, j}$ corresponds to a single gate linking two blocks.
    Each integer $x \in \{0, \ldots, 2^{|S|} - 1\}$ corresponds to a local minimum $\vec{\theta}_x$.
    Let $b_0(x), \ldots, b_{|S|-1}(x)$ be the binary representation of the integer $x$ using $|S|$ bits.
    We sort the set $S$ from small to large and consider a mapping $\mathrm{id}$ from $j \in S$ to the index in $S$, which is between $0$ to $|S|-1$.
    The local minimum $\vec{\theta}_x$ is defined as follows. For each $j = 0, \ldots, (n / 4) - 1$,
    \begin{equation}
        \vec{\theta}_{x, B, j} := (\pi / 2) \times \begin{cases}
            (1, 1, 1, 1, 1) & \text{if} \,\, j \in S \,\, \mathrm{and} \,\, b_{\mathrm{id}(j)}(x) = 1\\
            (0, 0, 0, 0, 0) & \text{else}
        \end{cases}
    \end{equation}
    And for all $j = 0, \ldots, (n / 4) - 2$, $\theta_{x, L, j} := 0$.
    It is not hard to verify that
    \begin{align}
        C_S(\vec{\theta}_{x}) &= 0, &\mathrm{for} \,\, x = 2^{|S|} - 1,\\
        C_S(\vec{\theta}_{x}) &= n - \left(b_0(x) + \ldots + b_{|S|-1}(x)\right) \geq 1, &\mathrm{for} \,\, x = 0, \ldots 2^{|S|} - 2.
    \end{align}
    Hence, $\vec{\theta}_{2^{|S|} - 1}$ is the global minimum.
    And for all $x = 0, \ldots, 2^{|S|} - 2$, $\vec{\theta}_{x}$ is suboptimal.
    This establishes the first statement of this proposition.

    We are now ready to prove the statement that $\vec{\theta}_x$ is a local minimum for all $x = 0, \ldots, 2^{|S|} - 2$.
    Consider $\vec{\theta}$ such that $\norm{\vec{\theta} - \vec{\theta}_x}_\infty < \pi / 4$.
    We now consider the cost function for each four-qubit block.
    For block $j \in \{0, \ldots, (n/4)-1\}$, we have a block of qubits
    \begin{equation}
        a := 4j+1, b := 4j+2, c := 4j+3, d := 4j+4.
    \end{equation}
    The associated cost function is
    \begin{equation}
        C_{S, j}(\vec{\theta}) := \mathop{\mathbb{E}}_{\ket{\psi} = \bigotimes_{i=1}^n \ket{\psi_i} \in \mathrm{stab}_1^{\otimes n}} \sum_{i \in \{a, b, c, d\}} \left( 1 -  \Tr\left(\bra{\psi_i} U(\vec{\theta})^\dagger U_S \ketbra{\psi}{\psi} U_S^\dagger U(\vec{\theta}) \ket{\psi_i}\right) \right) \geq 0.
    \end{equation}
    If $j \notin S$, or $j \in S$ and $b_{\mathrm{id}(j)}(x) = 1$, we have
    \begin{equation}
        C_{S, j}(\vec{\theta}_x) = 0 \leq C_{S, j}(\vec{\theta}).
    \end{equation}
    So we only need to consider the case when $j \in S$ and $b_{\mathrm{id}(j)}(x) = 0$, which is the case when $U(\vec{\theta}_x)$ acts as identity on block $j$ and $U_S$ acts as a SWAP gate between the first and fourth qubits in block $j$.
    In this case, we have the following cost function at $\vec{\theta}_x$,
    \begin{equation}
        C_{S, j}(\vec{\theta}_x) = 1.
    \end{equation}
    For each qubit $i$, we have the following identity,
    \begin{align}
        &\mathop{\mathbb{E}}_{\ket{\psi} = \bigotimes_{i=1}^n \ket{\psi_i} \in \mathrm{stab}_1^{\otimes n}} \left( 1 -  \Tr\left(\bra{\psi_i} U(\vec{\theta})^\dagger U_S \ketbra{\psi}{\psi} U_S^\dagger U(\vec{\theta}) \ket{\psi_i}\right)\right)\\
        &= \frac{2}{3}\left( 1 - \frac{1}{4}\Tr_{\neq i}\left(\Tr_i\left(U(\vec{\theta}) U_S\right)^\dagger \left(\frac{I_{n-1}}{2^{n-1}}\right) \Tr_i\left(U(\vec{\theta})^\dagger U_S\right)^\dagger \right)\right), \label{eq:tensor-contract-fidelity}
    \end{align}
    where $\frac{I_{n-1}}{2^{n-1}}$ is the maximally mixed state over $n-1$ qubits.
    By the definition of $U_S$ and $U(\vec{\theta})$, $U(\vec{\theta})^\dagger U_S$ is a linear combination of permutation operators with complex-valued weights.
    For $i = a$, we can rewrite the tensor contractions in Eq.~\eqref{eq:tensor-contract-fidelity} using the three gates associated with parameters $\theta_{B, j, 2}, \theta_{B, j, 3}, \theta_{B, j, 4}$.
    By first treating the maximally mixed states and the tracing operation $\Tr_{\neq i}$, we can rewrite the three gates as depolarizing channels, which gives rise to the following identity.
    \begin{equation}
        \frac{1}{4}\Tr_{\neq a}\left(\Tr_a\left(U(\vec{\theta}) U_S\right)^\dagger \left(\frac{I_{n-1}}{2^{n-1}}\right) \Tr_a\left(U(\vec{\theta})^\dagger U_S\right)^\dagger \right) = \lambda_a + (1 - \lambda_a) \frac{1}{4},
    \end{equation}
    where $\lambda_a := \sin(\theta_{B, j, 2})^2 \sin(\theta_{B, j, 3})^2 \sin(\theta_{B, j, 4})^2$.
    Similarly, for $i = d$, we have
    \begin{equation}
        \frac{1}{4}\Tr_{\neq d}\left(\Tr_d\left(U(\vec{\theta}) U_S\right)^\dagger \left(\frac{I_{n-1}}{2^{n-1}}\right) \Tr_d\left(U(\vec{\theta})^\dagger U_S\right)^\dagger \right) = \lambda_d + (1 - \lambda_d) \frac{1}{4},
    \end{equation}
    where $\lambda_d := \sin(\theta_{B, j, 1})^2 \sin(\theta_{B, j, 3})^2 \sin(\theta_{B, j, 5})^2$.
    For $i = b$, the tensor contractions in in Eq.~\eqref{eq:tensor-contract-fidelity} using the four gates associated with parameters $\theta_{B, j, 1}, \theta_{B, j, 3}, \theta_{B, j, 4}, \theta_{L, j-1}$.
    We can rewrite the two gates associated with $\theta_{L, j-1}$ and $\theta_{B, j, 3}$ in terms of depolarizing channels on qubit $a, b$, respectively.
    By enumerating all possible terms, we have
    \begin{align}
        &\frac{1}{4}\Tr_{\neq b}\left(\Tr_b\left(U(\vec{\theta}) U_S\right)^\dagger \left(\frac{I_{n-1}}{2^{n-1}}\right) \Tr_b\left(U(\vec{\theta})^\dagger U_S\right)^\dagger \right)\\
        &= \cos(\theta_{B, j, 1})^2 \cos(\theta_{B, j, 4})^2 \left(\cos(\theta_{B, j, 3})^2 + \frac{1}{4}\sin(\theta_{B, j, 3})^2 \right) \\
        &+ \sin(\theta_{B, j, 1})^2 \sin(\theta_{B, j, 4})^2 \left(\cos(\theta_{L, j-1})^2 + \frac{1}{4} \sin(\theta_{L, j-1})^2 \right) \\
        &+ \frac{1}{4}\left( \cos(\theta_{B, j, 1})^2 \sin(\theta_{B, j, 4})^2 + \sin(\theta_{B, j, 1})^2 \cos(\theta_{B, j, 4})^2 \right) \\
        &- \frac{3}{2} \cos(\theta_{B, j, 1}) \sin(\theta_{B, j, 1}) \cos(\theta_{B, j, 4}) \sin(\theta_{B, j, 4}) \cos(\theta_{L, j-1})^2  \cos(\theta_{B, j, 3})^2 \\
        &\leq \cos(\theta_{B, j, 1})^2 \cos(\theta_{B, j, 4})^2 \left(1 - \frac{3}{4}\sin(\theta_{B, j, 3})^2 \right) + \sin(\theta_{B, j, 1})^2 \sin(\theta_{B, j, 4})^2\\
        &+ \frac{1}{4}\left( \cos(\theta_{B, j, 1})^2 \sin(\theta_{B, j, 4})^2 + \sin(\theta_{B, j, 1})^2 \cos(\theta_{B, j, 4})^2 \right)\\
        &+ \frac{3}{2} \left| \cos(\theta_{B, j, 1}) \sin(\theta_{B, j, 1}) \cos(\theta_{B, j, 4}) \sin(\theta_{B, j, 4}) \right| \left(1 - \sin(\theta_{B, j, 3})^2\right).
    \end{align}
    Because $\norm{\vec{\theta} - \vec{\theta}_x}_\infty < \pi / 4$, we have $\cos(\theta_{B, j, 1}) \geq 0, \cos(\theta_{B, j, 4}) \geq 0$ and
    \begin{equation}
        |\sin(\theta_{B, j, 1})| = \sin(|\theta_{B, j, 1}|), \quad |\sin(\theta_{B, j, 4})| = \sin(|\theta_{B, j, 4}|).
    \end{equation}
    We can use trigonometric identities to obtain
    \begin{align}
        &\frac{1}{4}\Tr_{\neq b}\left(\Tr_b\left(U(\vec{\theta}) U_S\right)^\dagger \left(\frac{I_{n-1}}{2^{n-1}}\right) \Tr_b\left(U(\vec{\theta})^\dagger U_S\right)^\dagger \right)\\
        &\leq 1 - \frac{3}{4} \sin\left(|\theta_{B, j, 1}| - |\theta_{B, j, 4}|\right)^2 - \frac{3}{4} \cos(\theta_{B, j, 1})^2 \cos(\theta_{B, j, 4})^2 \sin(\theta_{B, j, 3})^2\\
        &- \frac{3}{2} \left| \cos(\theta_{B, j, 1}) \sin(\theta_{B, j, 1}) \cos(\theta_{B, j, 4}) \sin(\theta_{B, j, 4}) \right|\sin(\theta_{B, j, 3})^2 \\
        &\leq 1 - \frac{3}{4} \sin(\theta_{B, j, 3})^2 \cos(\theta_{B, j, 1})^2 \cos(\theta_{B, j, 4})^2.
    \end{align}
    Similarly, we have
    \begin{align}
        &\frac{1}{4}\Tr_{\neq c}\left(\Tr_c\left(U(\vec{\theta}) U_S\right)^\dagger \left(\frac{I_{n-1}}{2^{n-1}}\right) \Tr_c\left(U(\vec{\theta})^\dagger U_S\right)^\dagger \right)\\
        &\leq 1 - \frac{3}{4} \sin(\theta_{B, j, 3})^2 \cos(\theta_{B, j, 2})^2 \cos(\theta_{B, j, 5})^2.
    \end{align}
    Combining all four upper bounds on
    \begin{equation}
    \frac{1}{4}\Tr_{\neq i}\left(\Tr_i\left(U(\vec{\theta}) U_S\right)^\dagger \left(\frac{I_{n-1}}{2^{n-1}}\right) \Tr_i\left(U(\vec{\theta})^\dagger U_S\right)^\dagger \right)
    \end{equation}
    for $i = a, b, c, d$, we can obtain the cost function associated to this block,
    \begin{align}
        C_{S, j}(\vec{\theta}) \geq 1 &- \frac{1}{2} \sin(\theta_{B, j, 2})^2 \sin(\theta_{B, j, 3})^2 \sin(\theta_{B, j, 4})^2 - \frac{1}{2} \sin(\theta_{B, j, 1})^2 \sin(\theta_{B, j, 3})^2 \sin(\theta_{B, j, 5})^2\\
        &+ \frac{1}{2} \sin(\theta_{B, j, 3})^2 \cos(\theta_{B, j, 1})^2 \cos(\theta_{B, j, 4})^2 + \frac{1}{2} \sin(\theta_{B, j, 3})^2 \cos(\theta_{B, j, 2})^2 \cos(\theta_{B, j, 5})^2.
    \end{align}
    From $\norm{\vec{\theta} - \vec{\theta}_x}_\infty < \pi / 4$, we have
    \begin{align}
        &|\sin(\theta_{B, j, k})| < 0.5, \quad \forall k = 1, 2, 3, 4, 5,\\
        &|\cos(\theta_{B, j, k})| > 0.5, \quad \forall k = 1, 2, 3, 4, 5.
    \end{align}
    Hence, $C_{S, j}(\vec{\theta}) \geq 1 = C_{S, j}(\vec{\theta}_x)$.
    Together with the fact that
    \begin{equation}
        C_{S}(\vec{\theta}) = \sum_{j=0}^{(n/4)-1} C_{S, j}(\vec{\theta}),
    \end{equation}
    we have established the claim $C_S(\vec{\theta}_x) \leq C_S(\vec{\theta})$.
\end{proof}

% \bibliographystyle{unsrt}
% \bibliography{references,ref2}

% \clearpage
\printbibliography

\end{document}